%% file: main.tex
\def\isarxiv{1}
\documentclass[11pt]{article}
\usepackage{amsmath}
\usepackage{amsthm}
\allowdisplaybreaks
\usepackage{amssymb}
\usepackage{algorithm}
\usepackage{subfig}
\usepackage{color}
\usepackage[dvipsnames]{xcolor}
\usepackage{graphicx}
\usepackage{wrapfig,epsfig}
\usepackage{epstopdf}
\usepackage{url}
\usepackage{graphicx}
\usepackage{color}
\usepackage{epstopdf}
\usepackage{algpseudocode}
\usepackage{scrextend}
\usepackage[T1]{fontenc}
\usepackage{bbm}
\usepackage{comment}
\usepackage{multicol}
\usepackage{dsfont}
\usepackage{mathtools}
\usepackage{enumitem}

\let\C\relax
\usepackage{tikz}
\usepackage{hyperref}
\hypersetup{colorlinks=true,citecolor=blue,linkcolor=red}

\usetikzlibrary{arrows}
\usepackage[margin=1in]{geometry}
\graphicspath{{./figs/}}

\definecolor{b2}{RGB}{51,153,255}
\definecolor{mygreen}{RGB}{80,180,0}
\definecolor{ToCgreen}{RGB}{0, 128, 0}

\theoremstyle{plain}
\newtheorem{theorem}{Theorem}[section]
\newtheorem{lemma}[theorem]{Lemma}

\newtheorem{fact}[theorem]{Fact}
\newtheorem{claim}[theorem]{Claim}
\newtheorem{corollary}[theorem]{Corollary}

\newtheorem{definition}[theorem]{Definition}

\newtheorem{problem}[theorem]{Problem}

\newtheorem{remark}[theorem]{Remark}

\DeclareMathOperator*{\E}{{\mathbb{E}}}

\graphicspath{{./figs/}}

\newcommand{\wh}{\widehat}
\newcommand{\wt}{\widetilde}
\newcommand{\ov}{\overline}
\newcommand{\eps}{\varepsilon}
\renewcommand{\epsilon}{\varepsilon}
\renewcommand{\phi}{\varphi}

\newcommand{\rect}{\mathrm{rect}}
\newcommand{\N}{\mathcal{N}}
\newcommand{\R}{\mathbb{R}}

\newcommand{\Z}{\mathbb{Z}}
\newcommand{\C}{\mathbb{C}}

\newcommand{\RHS}{\mathrm{RHS}}
\newcommand{\LHS}{\mathrm{LHS}}
\newcommand{\tr}{\mathrm{tr}}
\newcommand{\polylog}{\mathrm{polylog}}
\renewcommand{\i}{\mathbf{i}}
\renewcommand{\tilde}{\wt}
\renewcommand{\hat}{\wh}
\renewcommand{\bar}{\ov}
\renewcommand{\d}{\mathrm{d}}
\newcommand{\poly}{\mathrm{poly}}
\DeclareMathOperator{\sinc}{sinc}
\newcommand{\argmin}{\mathrm{argmin}}

\newcommand{\supp}{\mathrm{supp}}

\newcommand{\SVP}{\mathrm{SVP}}
\newcommand{\Tmat}{{\cal T}_\mathrm{mat}}

\makeatletter
\newcommand*{\RN}[1]{\expandafter\@slowromancap\romannumeral #1@}
\makeatother

\title{Improved Reconstruction for Fourier-Sparse Signals}%

\ifdefined\isarxiv
\author{
Yeqi Gao\thanks{\texttt{a916755226@gmail.com}. The University of Washington.}
\and
Zhao Song\thanks{\texttt{zsong@adobe.com}. Adobe Research. }
\and 
Baocheng Sun\thanks{\texttt{woafrnraetns@gmail.com}. Weizmann Institute of Science.
} 
\and
Omri Weinstein\thanks{\texttt{omri@cs.columbia.edu}. The Hebrew University and Columbia University.}
\and 
Ruizhe Zhang\thanks{\texttt{ruizhe@utexas.edu}. Simons Institute for the Theory of Computing.}
}
\fi

\date{}

\begin{document}

 \begin{titlepage}
     \maketitle
     \begin{abstract}
         \input{abstract}

     \end{abstract}
     \thispagestyle{empty}
 \end{titlepage}

\newpage

{\hypersetup{linkcolor=black}
	\thispagestyle{empty}
	\tableofcontents
	\thispagestyle{empty}
}

\newpage
\setcounter{page}{1} 

\input{intro}

\input{tech}

\input{prob_def}

\input{preli}

\input{lattice}
\input{energy_bound}
\input{discretization}

\input{well_balance_sampling}
\input{distillation}

\input{oned}
\input{highd}

\input{discrete_set_query}

\input{lowerror}

\section*{Acknowledgements}
The authors would like to thank Michael Kapralov, Eric Price, Kshiteej Sheth, and Lichen Zhang for their helpful discussions. OW's research is supported by NSF CAREER award CCF-1844887, an ISF grant \#3011005535, and ERC Starting Grant \#101039914. Part of this research was performed while RZ was visiting the Institute for Pure and Applied Mathematics (IPAM), which is supported by the National Science Foundation (Grant No. DMS-1925919).

\newpage
\appendix
\section*{Appendix}

\input{noiseless}
\input{exact_exponential_algo}

\input{constant_gap}

\input{lowerror_app}

\addcontentsline{toc}{section}{References}
\bibliographystyle{alpha}
\bibliography{ref}
\end{document}

%% file: abstract.tex
We revisit the classical problem of Fourier-sparse signal reconstruction---a variant of the \emph{Set Query} problem 
--- which asks to efficiently reconstruct (a subset of) a $d$-dimensional Fourier-sparse signal ($\|\widehat{x}(t)\|_0 \leq k$),  from minimum \emph{noisy} samples of $x(t)$ in the time domain.  We present a unified framework for this problem by developing a theory of sparse Fourier transforms (SFT) for frequencies lying on a \emph{lattice}, which can be viewed as a ``semi-continuous'' version of SFT in between discrete and continuous domains. Using this framework, we obtain the following results:  

\begin{itemize}
\item  \bf Dimension-free Fourier sparse recovery \rm We present a sample-optimal  discrete Fourier Set-Query algorithm with $O(k^{\omega+1})$ reconstruction time in one dimension, \emph{independent} of the signal's length ($n$)  and $\ell_\infty$-norm. This complements the state-of-art algorithm of [Kapralov, STOC 2017], whose reconstruction time is $\tilde{O}(k \log^2 n \log R^*)$, where $R^* \approx \|\wh{x}\|_\infty$ is a signal-dependent parameter, and the algorithm is limited to low dimensions. By contrast, our algorithm works for arbitrary $d$ dimensions, mitigating the $\exp(d)$ blowup in decoding time to merely linear in $d$. A key component in our algorithm is fast spectral sparsification of the Fourier basis. 
\item 
 \bf High-accuracy Fourier interpolation \rm 
In one dimension, we design a  poly-time $(3+ \sqrt{2} +\epsilon)$-approximation algorithm for continuous Fourier interpolation.  This bypasses a barrier of all previous algorithms [Price and Song, FOCS 2015, Chen, Kane, Price and Song, FOCS 2016], which only achieve $c>100$ approximation for this basic problem. Our main contribution is a new analytic tool for hierarchical frequency decomposition based on \emph{noise cancellation}. 
\end{itemize}

%% file: intro.tex
\section{Introduction}\label{sec:intro}

The fast Fourier transform (FFT) \cite{ct65} is a fundamental tool in engineering, signal processing, mathematics, and theoretical computer science, with profound applications in theory and practice. 
Over the years, many variations of FFT have been studied and developed,  depending on the underlying domain and time-invariance properties of the signal \cite{ow97,o02,opp11}. %
In this paper, we study the \emph{sparse Fourier transform} (SFT), where the signal is either discrete or continuous in time domain, but $k$-\emph{sparse} in the frequency domain, i.e., $\hat{x}$ is a discrete set of size $k$: \begin{align*}
    x(t) = &~ \sum_{j=1}^k v_j e^{2\pi\i \langle f_j, t\rangle}.
\end{align*}
Band-limited (i.e., Fourier-sparse) signals arise in many real-world datasets and applications, from image compression and analysis \cite{Watson94}, to compressed sensing \cite{Don06} and (deep) learning with frequency-invariant kernels \cite{MMM21}; for a broader exposition of  SFT and its applications, we refer the reader to the survey \cite{GIIS14}.

A prototypical problem in this setting is \emph{band-limited signal interpolation} \cite{ckps16}
and the related\footnote{The Fourier interpolation literature typically focuses on  frequency estimation followed by magnitude estimation. The magnitude estimation can be formulated as a Set-Query problem, where we are given a set of locations and we only try to recover the Fourier coefficients $\hat{x}$ at the given locations. %
In our setup, frequencies are assumed to lie on a lattice, hence the problems are essentially equivalent, see Section~\ref{sec:FT_problem}.}
\emph{Fourier Set-Query} problem \cite{pri11}, which ask how to reconstruct (a subset of)  
the signal from few (ideally $\sim k$) \emph{noisy} samples of $x(t)$ in a time domain $[0,T]^d$. In this model, the algorithm has access to samples $x(t) +g(t)$, where the 
signal-to-noise ratio is guaranteed to be above a certain constant threshold (e.g., $\| x \|_T \gtrsim \| g \|_T$, where $\|x\|_T^2:=T^{-d}\int_{[0,T]^d} |x(t)|^2\d t$ is the \emph{energy} of the signal). 
In general, sparse-recovery problems have two computational aspects: the sample complexity, i.e., number of (noisy) samples required by the algorithm, and the reconstruction time (decoding the signal from the measurements).  
This problem has a long history in signal-processing and TCS \cite{ct65,r89,assn08,v11,hikp12a,hikp12,ghikps13,ik14,ikp14,b15,kap16,kap17,kvz19,nsw19,jls23,sswz23}.  
A fundamental fact, pointed out in \cite{moi15}, is that when the frequency gap is small 
($\eta := \min_{i\ne j\in [k]}|f_i-f_j| < 1/T$), exact recovery of the signal is informational-theoretically impossible.  Complementing this negative result, \cite{ps15} gave a $k\cdot \polylog(k,FT/\delta)$-time 
$\delta$-error reconstruction algorithm for one-dimensional signals where $F$ is the band-limit, assuming the time domain satisfies  $T>\Omega(\log^2(k/\delta)/\eta)$, 
and that the frequency gap $\eta$ is known. \cite{ckps16} strengthened this result by showing that even if the frequency gap is unknown, \emph{approximate} %
reconstruction of one-dimensional signals in $\poly(k,\log(FT))$-samples %
and time is possible, in the sense that the output signal is close to the original signal in the time domain albeit with worse sparsity in the frequency domain\footnote{More precisely, the error guarantee is  $\|y(t)-x^*(t)\|_T\le O(\|g(t)\|_T+\delta \|x^*(t)\|_T)$, where $x^*(t)$ is the original signal, $y(t)$ is the reconstructed signal, and $g(t)$ is the noise distribution.}. Subsequent works \cite{cp19_icalp,cp19_colt, llm21} have improved this result, both in sample-complexity and decoding time. %
Recently, \cite{llm21} improved the sparsity of the output signal from $\poly(k)$ to $k \poly\log(k)$, 
settling for a somewhat weaker notion of approximation\footnote{$\| y(t) - x^*(t) \|_{(1-c)T} \leq \poly(\log(k/c\delta)) \cdot \| g(t) \|_T +\delta \| x^* \|_T$.} than that of \cite{ckps16}. 

In the discrete setting, 
\cite{pri11} defined and studied the Set Query problem in the standard compressed-sensing model \cite{Don06}, where the design of the sensing matrix is unrestricted.  
\cite{kap17} defined and studied the Set Query problem in Fourier domain, where the sensing matrix is applied to the \emph{Fourier transform} of $x$  (i.e., measurements are $S\hat{x} = S\cdot \mathsf{FFT} \cdot x$). As such, Fourier Set Query is a more challenging problem than the former one. The current best discrete Fourier set query algorithm in 1D
is due to Kapralov \cite{kap17}, who gave an algorithm with near-linear sample complexity ($O(k/\eps)$) and decoding time $\tilde{O}( \epsilon^{-1} k \log^{2.001} n \log R^*)$, where $n$ is the length of signal and $R^*$ is (roughly) the $\ell_\infty$ norm of the signal in the time domain. 

Unfortunately, much less was known in higher dimensions  -- The ``curse of dimensionality'' of the Filter function \cite{kap17} drastically deteriorates the sample complexity, which grows \emph{exponentially} with the dimension, 
hence filter-based algorithms (a-la \cite{kap17}) are near-optimal only for small $d$. Indeed, this drawback was one of the principal  motivations of this work.    

\vspace{.3cm}

Given the abundance of largely incomparable results and diversity of techniques mentioned above, one might wonder whether there is a systematic, unified framework for analyzing band-limited  signal reconstruction.  
We propose such a framework, which \emph{decouples} the band-limited signal interpolation problem into two sub-problems:
\begin{enumerate}
    \item {\bf Frequency Estimation: } Find $L$  net-frequencies $\{f_i'\}_{i\in [L]}$ such that each $f_i$ is close to some $f'_j$ for $j\in [L]$. 
    \item {\bf Signal Estimation: } Based on the net-frequencies, approximate the original signal. 
\end{enumerate}
We note that almost all the previous works \cite{ps15,ckps16,llm21} fall into this framework, yet differ in the techniques used to solve these two sub-problems. Our starting point is the observation that the assumption of an $\eta$-frequency gap of the signal, is roughly equivalent to assuming that the frequencies lie on the grid $\eta\cdot \mathbb{Z}^d$. It is therefore natural to  generalize this assumption by considering the case where signal frequencies lie on a \emph{$d$-dimensional lattice} ${\cal L}=\Lambda({\cal B})$, where ${\cal B}$ is a basis of ${\cal L}$ and is given. We interchangeably call this problem \emph{lattice
Fourier interpolation}, or \emph{semi-continuous signal reconstruction}, as lattice frequencies can be viewed as interpolating between discrete and continuous domains. %
Solving this problem is our key tool en-route to faster sparse-recovery, but it is also interesting on its own right.

In this framework, the Fourier interpolation problem can be reduced to sequentially solve frequency estimation and signal estimation problems. We focus on the second step, that is, given the result of frequency estimation for a semi-continuous signal, how to efficiently reconstruct the coefficients. It can be formulated as a semi-continuous Fourier set-query problem, where the queried frequency set $L$ (a set of lattice points) is obtained from Frequency Estimation, and our goal is to recover the coefficients of these frequencies.   
The generality of our framework allows us to apply it to the \emph{discrete} Fourier Set-Query problem as well, where the goal is to (approximately) compute DFT only on a subset $S\subset [n]$ of $k$ coordinates. (See Remark~\ref{rmk:dis_set_query_to_semi_continuous} for more detailed discussions of the relationship between semi-continuous signal estimation and discrete Fourier set-query.)%

\subsection{Our Results}\label{sec:our_results}
\paragraph{Semi-continuous signal estimation}Our first main result shows that for semi-continuous signals (whose frequencies lie on a lattice), Signal Estimation can be efficiently reduced to Frequency Estimation:  
\begin{theorem}[1D Semi-continuous Signal Estimation, informal version of Theorems~\ref{thm:reduction_freq_signal_1d} and~\ref{thm:reduction_freq_signal_1d_clever}]\label{thm_semi_cont_decouple_thm}
Let $\Lambda(\mathcal{B}) \subset \R$ denote the lattice  $ \Lambda(\mathcal{B}) = \{ z \in \R | z =  c \eta, c \in \mathbb{Z} \}$, and 
suppose $x^*(t) = \sum_{j=1}^k v_j e^{2\pi\i  f_j  t  }$ with $f_j\in \Lambda({\cal B})$. 
Given observations of the form $x(t) = x^*(t) + g(t)$ for $t\in [0, T]$ for arbitrary noise distribution $g(t)$, let $L$ be the result of Frequency Estimation for $x(t)$ such that for each $f_i$, there is an $f'_i\in L$ with  $|f_i-f'_i|\le D/T$. Then given the set $L$ as input, we can obtain:  
\begin{itemize}
    \item {\bf A Sample-optimal algorithm: }There is an algorithm that takes $O(\wt{k})$ samples and outputs a $\wt{k}$-sparse signal $y(t)$ in $O(\wt{k}^{\omega+1})$-time\footnote{$\omega\approx 2.373$ is the fast matrix-multiplication exponent \cite{wil12,aw21}}, such that $\|y-x^*\|_T^2\leq O(\|g\|_T^2)$ holds with high probability, where 
    $\wt{k}:=O(|L|(1+D/(T\eta)))$ is the output sparsity. 
    
    \item {\bf A High-accuracy algorithm: } For any $\eps\in (0,1)$, there is an algorithm that takes $\wt{O}(\epsilon^{-1}\wt{k})$ samples and outputs a $\wt{k}$-sparse signal $y(t)$ in $O(\epsilon^{-1}\wt{k}^\omega)$-time such that $\|y-x^*\|_T^2\leq (1+\eps) \|g\|_T^2$ holds with high probability.
\end{itemize}

\end{theorem}

We note that our algorithms can be easily adapted to high-dimensional signal estimation.

\paragraph{Discrete Fourier set query}Our next  result provides an efficient, high-accuracy algorithm for the discrete Fourier Set Query problem in any dimension.  %

\begin{theorem}[Discrete Fourier Set Query, informal version of Theorems \ref{thm:main_fourier_ours_clever_discrete}]\label{thm:set_query_intro}
For any $d\geq 1$, let $n=p^d$ for some $p\in \mathbb{N}_+$. Given a vector $x \in (\C^p)^{\otimes d}$, for $k \geq 1$, any $S \subseteq [n]$, $|S| = k$, there exists an algorithm that takes $O(\eps^{-1}k)$ samples from $x$, runs in $O(\eps^{-1}k^{\omega+1} + \eps^{-1} k^{\omega-1} d)$ time, and outputs a vector $x' \in \C^n$ such that
\begin{align*}
\| (x' - \wh{x})_S \|_2^2 \leq \epsilon \|  \wh{x}_{[n] \backslash S} \|_2^2 
\end{align*}
holds with probability at least $0.9$. Note that $\wh{x}$ is the  $d$-dimensional discrete Fourier transform of $x$, $\wh{x}_f=\sum_{t\in[p]^d} x_t e^{-2\pi\i \langle f, t \rangle / p}, f\in [p]^d$.%
\end{theorem}

\begin{remark}\label{rmk:dis_set_query_to_semi_continuous}
Discrete Fourier set query can be viewed as a special case of semi-continuous signal estimation, where the queried frequencies (supported on $S$) lie on the integer lattice $\Z^d$. And the remaining part of the signal with frequencies outside of $S$ (i.e., $\wh{x}_{[n]\backslash S}$) corresponds to the noise $g$ in the semi-continuous signal estimation problem. One difference between these two problems is that for semi-continuous signal estimation, we assume the signal has a continuous time domain $[0,T]$, while for discrete Fourier set query, the signal has a discrete time domain $[p]^d$.
\end{remark}

\begin{remark}\label{rem_kap}
In one-dimension ($d=1$), the runtime of our set query algorithm can be simplified to $O(\eps^{-1}k^{\omega+1})$.  And prior to this work, the only known result for $\ell_2/\ell_2$ Fourier set query (due to \cite{kap17}) achieves the same sample complexity but runs in time $O( \epsilon^{-1} k \log^{2.001} n \log R^*)$. Here, $R^*$ is an upper bound on the $\|\cdot\|_{\infty}$ norm of the vector, typically assumed to be  $\poly(n)$. 
We emphasize that our approach immediately yields a $\poly(k)$-time algorithm, independent of $\log(n)$ and $R^*$, and thus generalizes to arbitrary dimension (see next result). This, however, comes at a price of a slower dependence on $k$, as opposed to the linear dependence of  \cite{kap17}. 

In high dimensions ($d>1$), the decoding time of \cite{kap17} scales as $\log^d(n)$ due to the curse of dimensionality of the Filter function  used\footnote{See \cite{nsw19} for more detailed discussions about the time complexity in high dimensions.}, while our algorithm
scales linearly with $d$. 
\end{remark}

\paragraph{High-accuracy Fourier interpolation} 
All previous algorithms for the Fourier interpolation problem provide rather large approximation guarantees on the reconstruction error 
\[ \|y-x^*\|_T\leq C\|g\|_T + \delta \|x^*\|_T, \] 
where $C$ is an absolute constant around 100 \cite{ckps16}, and is never better than $3$ due to repeated applications of the \emph{triangle inequality} for the errors from different sources.   
We develop a \emph{sharper noise control} technique, together with a high sensitivity frequency estimation method and an efficient signal estimation algorithm, which allows us to break this barrier and obtain a $C = (1+\sqrt{2} + \eps)$-approximation in $\poly(k, 1/\delta, 1/\eps)$ time. Our high-accuracy Fourier interpolation algorithm relies on a new error analysis. 

\begin{theorem}[High-accuracy Fourier interpolation, informal version of Theorem~\ref{thm:main_ours_better}]\label{thm:intro_improve_ckps}
Let $x^*(t)$ be a $k$-Fourier sparse signal with frequencies in $[-F, F]$. Given observations $x(t)=x^*(t)+g(t)$ in time duration $[0,T]$, where $g$ is arbitrary noise. For $\epsilon, \delta\in (0,1)$, there exists an algorithm that uses $\poly(k, \epsilon^{-1},\log(1/\delta))\log(FT)$ samples and runtime, and outputs a $\poly(k, \epsilon^{-1},\log(1/\delta)) $-Fourier-sparse signal $y(t)$ such that
\begin{align*}
    \|y-x^*\|_T\leq (3 +\sqrt{2}+\epsilon)\|g\|_T+\delta\|x^*\|_T.
\end{align*}
\end{theorem}

%% file: tech.tex
\section{Technical Overview}\label{sec:tech_overview}
The following section contains a streamlined technical overview of the main ideas and techniques required to prove our results in Section \ref{sec:our_results}. Section~\ref{sec:tech_framework} develops a unified framework for a wide range of Fourier set query-type problems. Section~\ref{subsec_tech_sig_est} focuses on implementing the framework for semi-continuous signals estimation and proving Theorem~\ref{thm_semi_cont_decouple_thm}. Section~\ref{sec:tech_ov_dis} focuses on implementing the framework for discrete Fourier set query and proving Theorem~\ref{thm:set_query_intro}. Section~\ref{sec:tech_ov_high_acc_interpolate} shows how to apply our signal estimation algorithm and obtain a high-accuracy Fourier interpolation algorithm as in Theorem~\ref{thm:intro_improve_ckps}.

\subsection{A general framework for Fourier set query-type  problems}\label{sec:tech_framework}
Many Fourier-related problems can be expressed as the following set-query problem:
Given a signal $x(t)$  (either continuous or discrete),  we wish to recover the Fourier coefficients $\wh{x}$ only at the coordinates of a predetermined set $S$. The observed signal $x(t)$ can be accessed only through 
samples, either noiseless or noisy. In the latter case, one is only allowed to access $x(t)+g(t)$ for an arbitrary function $g(t)$. 

A natural approach to such problem is to use \emph{linear regression} -- Notice that the observed signal $x(t)$ can be decomposed into $x_S(t) + (x_{\ov{S}}(t) + g(t))$, where $x_S(t)$ is a portion of the noiseless signal $x(t)$ with frequencies in $S$, and $x_{\ov{S}}(t)$ is the remaining part such that $(x_{\ov{S}}+g)(t)$ together can be treated as noise. In this terminology, recovering $\wh{x}_S$ is equivalent to solving a linear regression problem in the subspace spanned by the \emph{Fourier basis} functions with frequencies in $S$. Implementing this approach, however, has two substantial challenges:
\begin{enumerate}
    \item {\it Sample complexity: }How should one select the sample points $\{t_1,\dots,t_s\}$ in the time domain to solve the linear regression? Recall we wish to use as few samples as possible.
    \item {\it Estimation accuracy: } How can one guarantee that the solution of the linear regression will not be corrupted by the noise? In other words, how can one ensure the recovered signal is close to the relevant projection $x_S(t)$ of the true signal? 
\end{enumerate}

To resolve these two issues, our framework consists of the following three steps, and we note that this framework can be applied to both continuous and discrete signals.

\paragraph{Step 1: Oblivious sketching}
A straightforward approach for choosing the sample points for the linear regression is to uniformly\footnote{In the one-dimensional case, it is possible to design a more clever  \emph{nonuniform} oblivious distribution that achieves near-linear sample complexity \cite{cp19_colt}. However, in $d\geq 2$ dimensions, the best-known construction is a uniform oblivious sketch with polynomially many samples; more on this in Section \ref{subsec_tech_sig_est}.} sample points in the time domain.  
Let $S_0$ denote the set of i.i.d uniform samples. 
Then, we need to know how many samples are sufficient to guarantee that $S_0$ is a good sketch for the signal, i.e., $\|x_S\|_{S_0}\approx \|x_S\|_T$. The size of $S_0$ can be bounded by an important quantity---\emph{the Fourier energy bound}, which shows how far the maximum magnitude of a signal in a family ${\cal F}$ can deviate from its energy in a given time duration. More specifically, the energy bound $R$ of the family ${\cal F}$ is defined as:
\begin{align*}
    R:=\sup_{f\in {\cal F}}\sup_{t\in [0,T]}\frac{ |f(t)|^2}{\|f\|_T^2},
\end{align*}
In \cite{be06, k08, ckps16, cp19_icalp}, energy bounds for one-dimensional continuous Fourier-sparse signals are proved. We further show (nearly) tight energy bounds for high-dimensional discrete and continuous signals. Although the size of $S_0$ obtained in this way may be quite large, an appealing property is that it only depends on the Fourier sparsity instead of the actual signal $x(t)$. In this sense, the oblivious sketching step acts as a ``preconditioner" for our regression problem. %

\paragraph{Step 2: Sketch distillation}
We will reduce the size of the oblivious sketching by 
sub-sampling a \emph{linear-sized} subset $S_1$ from $S_0$, in a data-dependent fashion, such that $S_1$ is still a good sketch of $x_S$. More specifically, suppose the frequencies of $x_S$ are given. We can use %
a \emph{well-balanced sampling procedure} defined by \cite{cp19_colt} to sample a set of $s=O(k)$ points such that the $i$-th point $t_i$ is sampled from a distribution $D_i$ supported by $S_0$, and a coefficient vector $\alpha\in \R^s$. Then, this sub-sampler satisfies the following properties:
\begin{itemize}
    \item For the weights $w_i:=\alpha_i \frac{1/|T|}{D_i(t_i)}$ for $i\in [s]$, we have $\|x\|_{S_1,w}\approx \|x\|_T$ for all signals $x\in {\cal F}$, where $\|x\|_{S_1,w}:=\sum_{t\in S_1} w_t\cdot |x(t)|^2$ and ${\cal F}$ is the signal family spanned by the frequencies in $S$.
    \item The sum of coefficients $\alpha_i$ is small and each distribution $D_i$ is not ``ill-conditioned''\footnote{Formal definition is in Definition~\ref{def:procedure_agnostic_learning}.}.  
\end{itemize}

\begin{figure}[!ht]
    \centering
    \includegraphics[scale=0.6]{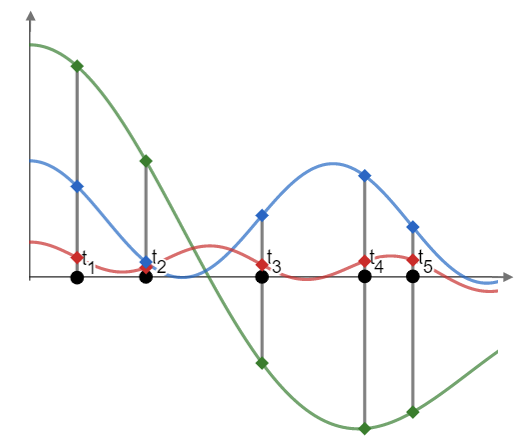}
    \caption{An illustration of sketching the set query signals. The green curve is the queried signal $x_S$, the red curve is the signal with remaining frequencies $x_{\ov{S}}$, and the blue curve is the noise $g$. Suppose $S_1=\{t_1,\dots,t_5\}$ are the sampling points. Then, we need to guarantee that $\|x_S\|_{S_1,w}\approx \|x_S\|_T$ as well as $\|x_{\ov{S}}+g\|_{S_1,w}\lesssim \|x_{\ov{S}}+g\|_{T}$.}
    \label{fig:my_label}
\end{figure}

On the one hand, %
the first property guarantees that the sketch distillation outputs a sample set $S_1$ of linear size and $\|x_S\|_{S_1,w}\approx \|x_S\|_T$.
On the other hand, the sketch distillation process is \emph{robust to noise}. %
An easy information-theoretic argument shows that the estimation error must be proportional to the energy of the noise, i.e., $\|x_{\ov{S}}+g\|_T$. However, %
after the sketch distillation the weighted energy of the noise $\|x_{\ov{S}}+g\|_{S_1,w}$ could be amplified (i.e., $\|x_{\ov{S}}+g\|_{S_1,w}\gg \|x_{\ov{S}}+g\|_T$), resulting in a large estimation error. %
Fortunately, the second property ensures that it will not happen and in expectation $\|x_{\ov{S}}+g\|_{S_1,w}\lesssim \|x_{\ov{S}}+g\|_T$.

\paragraph{Step 3: Weighted linear regression}
The last step is to solve a (weighted) linear regression with the noisy samples $\{\wt{x}(t_i)\}_{i\in[s]}$ and weight $w\in \R^s_{>0}$. For simplicity, take the one-dimensional continuous $k$-Fourier-sparse signal as an example. Suppose the set-query frequencies of $x_S$ are $S=\{f_1,\dots,f_k\}$. Then, we consider the following weighted linear regression:  
\begin{align}
\underset{v' \in \C^{k} }{\min}~\left\|\sqrt{w} \circ (A v' - b )\right\|_2,  \label{eq:liner_regression_intro}
\end{align}
where $\sqrt{w}:=(\sqrt{w_1},\dots,\sqrt{w_s})$, the coefficients matrix $A\in \C^{s\times k}$ and the target vector $b\in \C^s$ are defined as follows:
\begin{align*}
    A_{i,j}:=&~\exp(2\pi \i f_j t_i) ~~~\forall (i,j)\in [s]\times [k], ~~\text{and}\\
    b_i := &~ \wt{x}(t_i)~~~\forall i\in [s].
\end{align*}
Let $v'$ denote an optimal solution of Eq.~\eqref{eq:liner_regression_intro}. Then, we output a signal $y(t):=\sum_{i=1}^k v'_i\exp(2\pi\i f_i t)$.

Using the error analysis in sketch distillation (Step 2), we can upper-bound the estimation error by $\|y-x_S\|_T\lesssim \|x_{\ov{S}}+g\|_T$.

\subsection{Our techniques for signal estimation algorithms} \label{subsec_tech_sig_est}
In this section, we show how to instantiate the framework and get Signal Estimation algorithms in the semi-continuous setting (Theorem~\ref{thm_semi_cont_decouple_thm}). 

We first note that given the output of Frequency Estimation algorithm, a set $L$ such that all the true frequencies are close to $L$, the Signal Estimation problem can be formulated as a Fourier set-query problem. The idea is that by the guarantee of $L$ and the semi-continuous assumption, there will not be too many lattice points that are close to $L$, and we can efficiently find all of them. Let's denote this set of candidate frequencies by $\wt{S}$ of size $\wt{k}$. Then, Signal Estimation problem is reduced to a set-query for $\wh{x^*}_{\wt{S}}$.
In Section~\ref{sec:sample_opt_high_acc_algo}, we discuss how to apply our set-query framework to obtain sample-optimal and high-accuracy algorithms for 1-D signals and how to generalize to higher dimensions. In Section~\ref{sec:speed_up_bss}, we show how to implement Sketch Distillation very efficiently by speeding up the randomized spectral sparsification, which may be of independent interest. 
\subsubsection{Sample-optimal and high-accuracy algorithms}\label{sec:sample_opt_high_acc_algo}
Based on the aforementioned reduction, we briefly show the instantiation of our framework for the Signal Estimation problem. For convenience, suppose $\wt{k}=O(k)$ throughout this part.
\paragraph{Sample-optimal signal estimation}
For one-dimensional $k$-Fourier-sparse signals, \cite{k08} proved that the energy bound is $O(k^2)$, which implies that  the uniform sketching for the signal needs at least $\Omega(k^3\log k)$ samples. Although the sketch size can be reduced to $O(k)$ via Sketch Distillation, the algorithm has to first sample $\wt{O}(k^3)$ points in $[0,T]$, which already takes $O(k^3)$-time. It is possible to improve this straightforward approach using \emph{weighted oblivious sketching}~\cite{cp19_colt}, namely, we sample a set $S_0$ of points in $[0,T]$ from a carefully-chosen non-uniform distribution $D$, and assign each point a weight (this works only in 1D, see last section). Using the following distribution constructed by \cite{cp19_colt}:%
\begin{align*}
D(t):=
\begin{cases}
{c}/(1-|t/T|), & \text{ for } |t| \le T(1-{1}/k)\\
c \cdot k, & \text{ for } |t|\in [T(1-{1}/k), T]
\end{cases} 
\end{align*}
we only need to take $|S_0|=O(k\log k)$ samples to guarantee that $\|x^*\|_{S_0, w_0}\approx \|x^*\|_T$ holds with high probability. 
In this way, oblivious sketching in Step 1 will not be a time-consuming step. In Step 2, the sketch $S_0$ will be distilled to a subset $S_1$ of size $O(k)$. 
Finally, in Step 3, we sample the signal at the points in $S_1$ and solve the weighted linear regression to recover a $k$-sparse signal $y(t)$. This gives us a \emph{linear-sample} reduction from Frequency Estimation to Signal Estimation with $O(1)$-estimation error with high probability. %

\paragraph{High-accuracy signal estimation}
For one-dimensional semi-continuous signals, we also discover a ``sample-accuracy trade-off''. If we can use nearly-linear (i.e., $\wt{O}(k)$) samples, then we can skip Step 2 and directly use $(S_0,w_0)$ to solve the linear regression in Step 3. 

The advantage of this approach is that the sampling procedure for $(S_0,w_0)$ is well-balanced, and a sharper error analysis shows that we can achieve much smaller errors.
The main observation is that the noise can be decomposed into $g^\parallel\in {\cal F}$ and $g^\bot$ orthogonal to ${\cal F}$. 
Since we will solve a linear regression (in Step 3) in the space ${\cal F}$, the contribution of $g^\parallel$ to the final estimation error will not blow up. %
For the orthogonal part, we find an orthonormal basis $\{u_1,\dots, u_k\}$ for ${\cal F}$ and look at the weighted projection $\langle g^\bot, u_i\rangle_{S_1,w} :=\sum_{t\in S_1} w_t \ov{u_i(t)}g^\bot (t)$, whose 
magnitude indicates how much the noise is amplified due to sketching in each direction. The well-balanced sampling procedure gives that the total magnitudes $\sum_{i=1}^k |\langle g^\bot, u_i\rangle_{S_1,w}|^2 \leq \epsilon\cdot \|g^\bot\|_T^2$
holds with high probability, which will imply %
an $\eps\cdot \|x_{\ov{S}}+g\|_{T}^2$-estimation error for any small $\epsilon$.

\begin{figure}[!ht]
    \centering
    \subfloat[The function space with norm $\|\cdot \|_T$]{\includegraphics[width=0.48\textwidth]{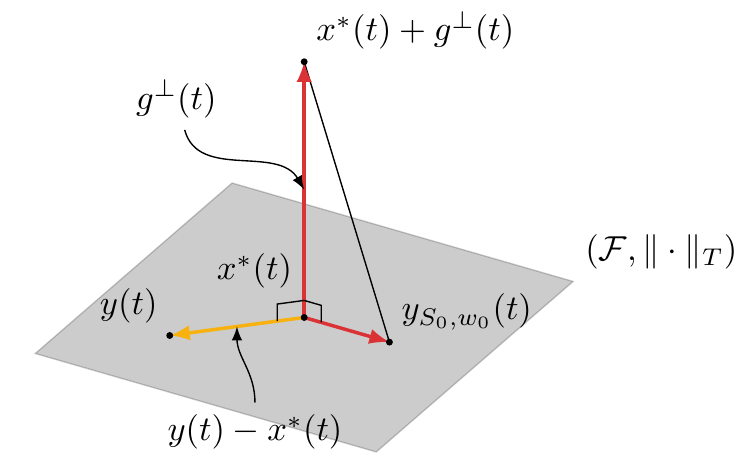}}\label{fig:compose_well_balance_a}
    \vspace{1mm}
    \subfloat[The function space with norm $\|\cdot \|_{S_0,w_0}$]{\includegraphics[width=0.48\textwidth]{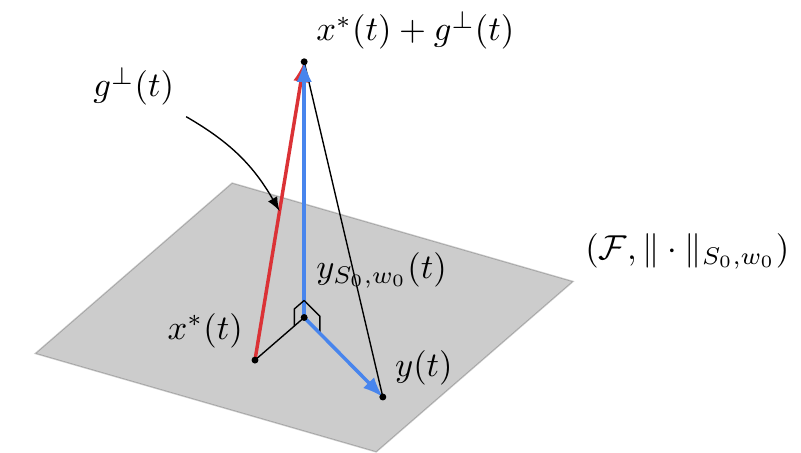}} \label{fig:compose_well_balance_b}
    \vspace{1mm}
    \subfloat[The function space with norm $\|\cdot \|_{S_1, w}$]{\includegraphics[width=0.48\textwidth]{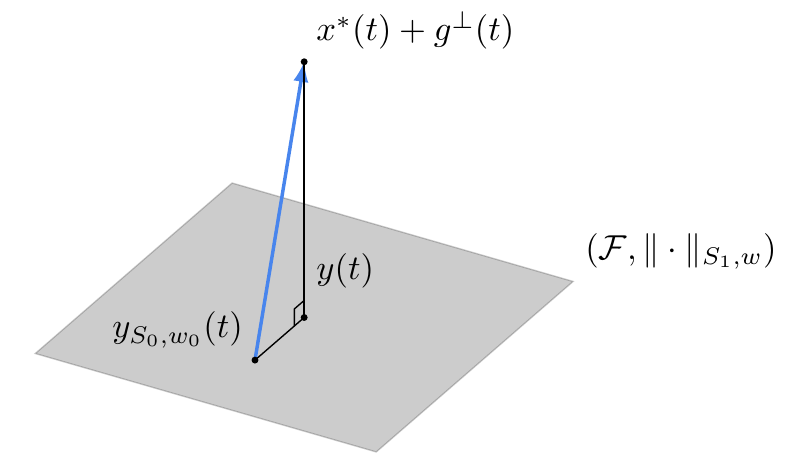}} \label{fig:compose_well_balance_c}
    \caption{
    The composition of two WBSPs may not be a WBSP. In (a), the central point is the ground-truth signal $x^*(t)$, and the yellow vector $y(t)$ is the final output of the composited WBSPs. Our goal is to bound the estimation error $ \|y(t)-x^*(t)\|_T$. Let $y_{S_0, w_0}(t)\in {\cal F}$ be the optimal approximation of the observing signal in the basis $(S_0, w_0)$, which is obtained by Step 1 of the framework. In (b), we show that by the property of the first WBSP,  the estimation error $\|y_{S_0, w_0}(t)-x^*(t)\|_T$ can be bounded by the projection of $g^\bot$ in the space $({\cal F}, \|\cdot\|_{S_0,w_0})$, which is about $\eps \|g^\bot(t)\|_T$. In (c), the space is reduced to $({\cal F}, \|\cdot\|_{S_1,w_1})$ by Step 2 of the framework, and the estimation error $\|y_{S_0,w_0}-y\|_{S_1,w_1}$ is bounded by the projection of $(x^*(t)+g^\bot-y_{S_0,w_0})$, which is about $\eps\|x^*(t)+g^\bot-y_{S_0,w_0}(t)\|_{S_0, w_0}$.  However, these two error bounds $\eps \|g^\bot(t)\|_T$ and $\eps\|x^*(t)+g^\bot-y_{S_0,w_0}(t)\|_{S_0, w_0}$ could not imply any bound on $\|y(t)-x^*(t)\|_T$.
        }
    \label{fig:compose_well_balance}
\end{figure}

Then, in step 3, we can also decompose the final estimation error $\|y-x^*\|_T^2$ into two parts: one contributed by $g^\parallel$ with energy $\|g^\parallel\|_T^2$, and another contributed by $g^\bot$ with energy at most $(1+\epsilon)\|g^\bot\|_T^2$. %
Combining the two parts together, we achieve a %
a high-accuracy guarantee:
\begin{align*}
    \|y-x^*\|_T^2\leq (1+\eps)\|g\|_T^2.
\end{align*}
We remark that this error analysis does not work for the sample-optimal algorithm. Roughly speaking, the weighted sketch $(S_1,w)$ in the sample-optimal algorithm is generated by a two-step sampling procedure: one in the oblivious sketching (Step 1) and another in the sketch distillation (Step 2). Even if both of them are well-balanced, when they are composited together and considered as a single sampling procedure, they may not be well-balanced. %

\paragraph{High-dimensional signal estimation}
For high dimensional signals, we prove a $k^{O(d)}$-energy bound for $k$-Fourier-sparse signals in $d$-dimensions in Section~\ref{sec:hd_energy_bound}, where we also provide a nearly-matching lower bound. %
We view it as a new generic tool in Fourier analysis and are quite a tour-de-force. 
Thus, we basically follow the three-step framework. One tricky thing in high-dimension is how to bound the output signal's sparsity, which is equivalent to the number of $d$-dimensional lattice points close to a fixed set $L$. We further reduce it to a clean math problem about lattices: let $\Lambda({\cal B})$ be a lattice in $\R^d$. For $r>0$, how many  lattice points can be within an $r$-radius ball, i.e.,  $\sup_{x\in \R^d}|\Lambda({\cal B})\cap B_d(x,r)|$. We upper-bound this quantity via different approaches, which might be of independent interest. More details are deferred to the appendix.

\subsubsection{Speed up randomized spectral sparsification}\label{sec:speed_up_bss}

In Section~\ref{sec:sample_opt_high_acc_algo}, we do not discuss how to implement a well-balanced sampling procedure. %
\cite{cp19_colt} proved that \emph{Randomized BSS} algorithm in \cite{bss12,ls15} yields a well-balanced sampling procedure. However, this algorithm is slow when the sampling domain $S_0$ is very large. One contribution of this work is to improve the 
 time and space costs of the randomized BSS algorithm. 
 
We observe that the bottleneck of each iteration in the original Randomized BSS algorithm is to sample a point $t\in S_0$ from the distribution $D_j$ defined by $D_j(t)= v(t)^\top E_j v(t)$, where $v(t)$ is a $k$-dimensional vector and $E_j$ is a $k$-by-$k$ positive semi-definite matrix determined by the potential function value at the $j$-th iteration. Suppose $E_j$ and $\{v(t)\}_{t\in S_0}$ have already been computed. A naive approach to sampling from $D_j$ is to compute the probability $D_j(t)$ for each $t\in S_0$, which takes $O(nk^2)$-time per iteration, where $n=|S_0|$ is the size of the sampling domain. To improve the algorithm, we consider a more general data structure problem---\emph{Online Quadratic-Form Sampling}. In this problem, we are given $n$ vectors $v_1,\dots,v_n\in \R^k$.
In each query, the input is a positive semi-definite matrix $A\in \R^{k\times k}$ and we need to sample an $i\in [n]$ with probability proportional to $v_i^\top A v_i$. (See Problem~\ref{prob:oqfs} for formal definition). The na\"ive space and query time for generating a sample is $O(nk^2)$. 
We design two data structures with substantially faster time-space tradeoff:%

\begin{theorem}[Online Quadratic-Form Sampling, informal version of Theorems~\ref{thm:main_preserve_distance} and~\ref{thm:main_preserve_distance_tradeoff}]\label{thm:fast_bss_intro}
The Online Quadratic-Form Sampling problem admits  the following two data structures:
\begin{itemize}
    \item {\bf Data Structure 1:} $O(nk^2)$ preprocessing time of the vectors $\{v_i\}_{i\in [n]}$ , $O(k^2\log n)$ query-time for generating a sample, and  $O(nk^2)$-space.
    \item {\bf Data Structure 2:}  $O(nk^{\omega-1})$  preprocessing time, $O(k^2\log (n/k)+k^\omega)$  query-time for generating a sample, and   $O(nk)$-space.
\end{itemize}
\end{theorem}
The main idea is to %
construct a range search tree for $\{1,2,\dots, n\}$, and for a node corresponding to range $[l,r]$, stores the matrix $\sum_{i=l}^r v_iv_i^\top$. Subsequently, for each query matrix $A$, we traverse the root-to-leaf path in the tree (which corresponds to an element in $[n]$), and output this element as a sample. The rule for descendin the tree resembles a form of rejection sampling: At a node with range $[l,r]$, we know that its left child has range $[l,m]$ and right child has right $[m+1, r]$, where $m=\lfloor (l+r)/2\rfloor$. Then, we decide whether move to the left or the right subtree by tossing a coin with probability: %
\begin{align*}
    p_\mathsf{left}:=\frac{\langle \sum_{i=l}^m v_i v_i^\top, A\rangle}{\langle \sum_{i=l}^r v_i v_i^\top, A\rangle}, ~~\text{and}~~p_\mathsf{right}:=1-p_\mathsf{left},
\end{align*}
where $\langle \sum_{i=l}^r v_iv_i^\top , A\rangle := \sum_{i=l}^r v_i^\top Av_i$ is the trace-product of matrices. Therefore, the probability $p_\mathsf{left}$ equals to the \emph{conditional} probability $\Pr_{{\bf i}\sim {\cal D}_A}[{\bf i}\in [l,m]~|~{\bf i}\in [l,r]]$. By the chain rule of conditional probability, we get that the output distribution of this procedure is exactly equal to ${\cal D}_A$. This data structure can be built in $O(nk^2)$ and each query only takes $O(k^2\log n)$-time. To store the matrices in each node, this data structure uses $O(nk^2)$-space. We can further improve the space complexity to $O(nk)$ by trading-off the preprocessing time and the query time. 

By plugging-in this data structure to the Randomized BSS algorithm, we can improve the time and space complexity by a factor of $k$ when $n$ is large.

\begin{figure}[ht]
    \centering
    \includegraphics[scale=0.8]{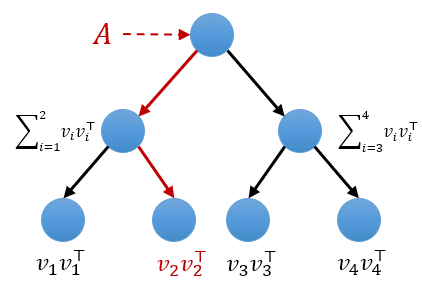}
    \caption{An example of the outer-product range tree with $n=4$. For a query matrix $A$, the sampling probability of the red path is $\frac{\langle \sum_{i=1}^2v_iv_i^\top, A\rangle}{\langle \sum_{i=1}^4v_iv_i^\top\rangle}\cdot \frac{\langle v_2v_2^\top, A\rangle}{\langle \sum_{i=1}^2v_iv_i^\top, A\rangle}=\frac{v_2^\top Av_2}{\sum_{i=1}^4 v_i^\top A v_i}$.}
    \label{fig:ds}
\end{figure}

\subsection{Our techniques for discrete Fourier set query}\label{sec:tech_ov_dis}
For simplicity, we consider one-dimensional discrete Fourier signal $x(t)$, which is a length-$n$ vector such that $x(t)=\sum_{j=1}^n \wh{x}_j e^{2\pi \i jt/n}$ for any $t\in [n]$.
And the set query problem asks to recover $\wh{x}_S$, for a given $k$-subset $S\subset [n]$. Let $x_S$ be the part of signal with frequencies in $S$, i.e., $x_S(t)=\sum_{f\in S}\wh{x}_f e^{2\pi \i ft/n}$, which is a $k$-Fourier-sparse signal. 

The high-level idea of obtaining Theorem~\ref{thm:set_query_intro} is as follows.
We prove that the energy bound for discrete $k$-sparse signal is $k$, which implies that uniformly sample a set $S_0\subset [n]$ of $O(k\log k)$ points  can form a good sketch of the signal $x_S$. Then, by Sketch Distillation, we can find a subset $S_1\subset S_0$ of linear size together with a weight vector $w$ such that $\|x_S\|_{S_1,w}\approx \|x_S\|_2/n$. Finally, we can recover $\wh{x}_S$ by solving a weighted linear regression on the samples $\{x(t)\}_{t\in S_1}$. By a direct error analysis, it is easy to see that this algorithm can achieve $O(1)$-estimation error. In the followings, we will show that it can perform much better.

\paragraph{Composition of well-balanced samplers} The key step to proving a $(1+\eps)$-error guarantee is to show that the final weighted sketch $(S_1, w)$ can be generated by a well-balanced sampling procedure. In general, compositing two well-balanced sampling procedures may not be well-balanced. More specifically, we can show that the composition sampler satisfies the first property of well-balanced; that is, the output set and weight can well-approximate the energy of every function in the family ${\cal F}$. However, the second property about the sum of the composition sampler's coefficients and the condition number of the composition sampling distribution may not hold. 

The above discrete set query algorithm is a very special case of composition in the sense that the first sampler samples each element from the same, simplest distribution---uniform distribution over $[n]$. Then, we can prove that in the composition sampler, each sample is equivalent to directly sampled from the uniform distribution. Furthermore, for discrete Fourier-sparse signal, we also have a tight energy bound of $R= k$. Using these results, we can show that the composition sampler in our algorithm is well-balanced! Then, we are able to apply our sharper error analysis and get that
\begin{align*}
    \|\wh{y}_S-\wh{x}_S\|_2\leq \epsilon\cdot \| \wh{x}_{\bar{S}} \|_2^2 
\end{align*}
holds with high probability. Therefore, we obtain a linear-sample and high-accuracy algorithm for discrete Fourier set query. 

\begin{remark}
This result is a fundamental departure from \cite{cp19_colt} since they assumed that the noise has a zero-mean (i.e., $\mathbb{E}[g(t)]=0$), which makes the WBSPs composition much easier. However, we do not have such an assumption. A natural fix is to ensure the orthogonality of the noise $g(t)$ to our basis functions $\{\exp(-2 \pi \mathbf{i} f_i t / n)\}$. Unfortunately, this  only holds for the first WBSP, but not for the second one (since sampling from the weighted output of the first WBSP can break the orthogonality). 

Our main novelty is to directly find an equivalent sampling procedure to the combination of two WBSPs (just for the analysis). Then, using the special properties of DFT, we prove that the composition also works in our setting (see Section \ref{sec:equiv_sample_procedure} Lemma \ref{lem:equiv_sampling} for more details). We believe this technique will be useful for other set-query or active learning problems.
\end{remark}

\begin{algorithm}[ht]\caption{Discrete One-Dimensional Signal Set-Query Algorithm (Informal)} \label{algo:set_query_intro}
\begin{algorithmic}[1]

\Procedure{\textsc{SetQuery}}{$x$, $n$, $k$, $S$, $\eps$} \Comment{Theorem~\ref{thm:set_query_intro}}
\State $\{f_1,f_2,\cdots,f_{k}\} \gets S$
\Statex {\color{blue}/*Step 2: Oblivious Sketching*/}
\State $S_0\gets$ $O(\epsilon^{-2}k\log(k))$ i.i.d. samples from $\text{Uniform}([n])$
\Statex {\color{blue}/*Step 3: Sketch Distillation*/}
\State Let ${\mathcal F} = \{\sum_{j\in [k]} v_j \exp(2\pi\i f_j t / n) ~|~ v_j\in\C \}$.
\State $\{t_1,t_2,\cdots, t_s\},w \leftarrow\textsc{RandBSS+}(k, \mathcal{F},\mathrm{Uniform}(S_0),(\eps/4)^2)$ \Comment{Algorithm~\ref{alg:BSS_faster}}
\Statex {\color{blue}/*Step 4: Weighted linear regression*/}
\State $A_{i, j} \leftarrow\exp(2\pi\i f_j t_i/n)$ for each $(i,j)\in [s]\times [k]$
\State $b_i  \leftarrow x(t_i)$ for each $i\in [s]$\Comment{Observe the signal at time $t_i$}
\State $v' \gets \underset{v' \in \C^{ k } }{\arg\min}\| \sqrt{w} \circ (A v'- b)\|_2$  
\State \Return $v'$
\EndProcedure
\end{algorithmic}
\end{algorithm}

\subsection{High-accuracy Fourier interpolation}\label{sec:tech_ov_high_acc_interpolate}

In this section, we introduce how to apply our four-step Fourier set-query framework together with some new techniques to obtain a high-accuracy one-dimensional Fourier interpolation algorithm (Theorem~\ref{thm:intro_improve_ckps}), which improves the constant-accuracy algorithm by \cite{ckps16}. 

Let us briefly summarize the previous algorithm in \cite{ckps16}. The high-level idea is to first find some small intervals in the frequency domain such that each contains some significant frequencies of the signal $x^*$. (These intervals are called  ``heavy-clusters'' in their paper.) Then, they use some filter techniques (also used in \cite{ps15}) to reduce the problem of reconstructing $x^*$, a signal with multiple heavy-clusters to several single heavy-cluster signals. Then, for each single heavy-cluster signal, since the band-limit is small, they can efficiently estimate its frequencies. Finally, they reconstruct a $\poly(k)$-sparse signal that is close to $x^*$ via a robust polynomial learning algorithm. More specifically, their algorithm consists of the following steps: %

\begin{enumerate}
    \item They show that the ground-truth signal $x^*(t)=\sum_{j=1}^k v_j e^{2\pi \i f_j t}$ can be approximated by $x_{S}(t)=\sum_{j\in S} v_j e^{2\pi \i f_j t}$, where $S:=\{j\in [k]:f_j\in \text{some heavy-cluster}~C_i\}$ is the set of frequencies in the heavy-clusters. This step will cause an approximation error $E_1:=\|x^*-x_S\|_T\leq 1.2{\cal N}$\footnote{Due to the noisy observations, not every frequency in the heavy-clusters is recoverable. This gap causes an extra implicit error term in \cite{ckps16}, which is about $12{\cal N}$. Section~\ref{sec:intro_7_esp} has a more detailed discussion.}, where ${\cal N}^2:=\|g\|_T^2+\delta \|x^*\|_T^2$ appears in the approximation error of the Fourier interpolation problem. 
    However, the frequency in the heavy-cluster can't be recovered correctly based on the method they proposed. Due to the $S$ not having an exact property of one-cluster (See Lemma~\ref{lem:findfrequency}), we will only focus on the frequency with a High signal-noise ratio, which can correctly satisfy the one-cluster property. Furthermore, we define the frequency we focus as
    $x_{S_f}(t)=\sum_{j\in S_f} v_j e^{2\pi \i f_j t}$, where $S_f :=\{j\in [k]:f_j\in \text{some heavy-cluster}~C_i~\text{and}~\text{ bucket}~j~\text{has high signal-noise ratio}\}$
    \item They solve a Frequency Estimation problem for $x_S$ using the filter techniques and multiple-to-one heavy-cluster reduction, and get a list $L$ of $O(k)$ candidate frequencies so that for each $j\in S$, $f_j$ is close to some $\wt{f}_{p_j}\in L$.    
    \item The signal $x_S(t)$ can be decomposed into $\sum_{i=1}^{|L|}e^{2\pi\i \wt{f}_it}\cdot x^*_i(t)$, where $x^*_i(t):=\sum_{j:p_j=i}e^{2\pi \i (f_j-\wt{f}_i)}$ is a one-cluster signal with small band-limit. For each $x^*_i(t)$, they prove that there exists a low-degree polynomial $P_i(t)$ that can approximate it. Let's denote the polynomial-approximated signal $\sum_{i=1}^{|L|} e^{2\pi \wt{f}_i t}\cdot P_i(t)$ by $x_{S,\mathsf{poly}}(t)$, which has an approximation error $E_2:=\|x_S - x_{S,\mathsf{poly}}\|_T\leq \delta \|x_S\|_T$.\footnote{We remark that even if the ground-truth signal $x^*(t)$ can be well-approximated by a mixed Fourier-polynomial signal $\wt{x}(t)$, we are unable to recover \emph{every} basis of $\wt{x}(t)$ due to the limitation of the frequency estimation procedure. Thus, directly applying linear regression to \emph{partially} reconstruct $\wt{x}(t)$ will not guarantee to be a $(1+\epsilon)$-approximation of $x^*(t)$.} 
    \item It remains to reconstruct $x_{S,\mathsf{poly}}(t)$, which is a variant of Signal Estimation problem. They use a sampling-and-regression approach to obtain a $\poly(k)$-Fourier-sparse signal $y(t)$ with an approximation error $E_3:=\|y-x_{S,\mathsf{poly}}\|_T\leq 2200{\cal N}$. 
\end{enumerate}
By triangle inequality, the total approximation error is $\|y-x^*\|_T\leq E_1+E_2+E_3\leq C\cdot {\cal N}$, where $C>1000$ is an absolute constant. 

However, the frequency in the heavy-cluster can't be recovered correctly based on the method they proposed. Due to the $S$ not having an exact property of a one-cluster (See Lemma~\ref{lem:findfrequency}), we will only focus on the frequency with a high signal-to-noise ratio that can correctly satisfy the one-cluster property.
Furthermore, we will focus on
$
x_{S_f}(t) = \sum_{j\in S_f} v_j e^{2\pi \i f_j t},
$
where $S_f := \{j \in [k]: f_j \in \text{some heavy-cluster}~C_i~\text{and the bucket}~j~\text{has a high signal-to-noise ratio}\}$.
In the remainder of this section, we first introduce our techniques to achieve a $(9 +\eps)$-approximation error. Then, we show how to overcome the barrier and achieve a $( 3 +\sqrt{2}+\eps)$-approximation error.

\subsubsection{\texorpdfstring{$(9+\eps)$-approximation algorithm}{}~}\label{sec:intro_7_esp}

\paragraph{High sensitivity frequency estimation} 
We first improve $E_1$ from $1.2{\cal N}$ to $(1+\eps){\cal N}$ by proposing a high sensitivity frequency estimation method. More specifically, to identify the heavy-clusters of the signal $x^*$, \cite{ckps16} designed a filter function $H$ such that $H\cdot x^*$ has high energy in each heavy cluster $C_i$; that is, 
\begin{align}\label{eq:intro_heavy_cluster}
    \int_{C_i} |\wh{H\cdot x^*}(f)|^2\d f\geq \frac{T}{k}{\cal N}^2.
\end{align}
Moreover, $H$'s frequencies are contained in a small interval of length $\Delta$. These two properties imply that for any true frequency $f_i\in C_i$, the signal $H\cdot x^*$ with frequency domain restricted to $[f_i-\Delta, f_i+\Delta]$ is a one heavy-cluster signal with small band-limit, which allows us to use \cite{ps15}'s approach to estimate $f_i$. The filter function $H$ in \cite{ckps16} is only $O(1)$-sensitive, which means it can concentrate a constant fraction of the signal's energy. And for those less important frequencies, they cannot be clustered by $H$ and will be lost in the frequency estimation procedure. 

We manage to modify their filter construction and obtain a $(1-\eps)$-sensitive filter function $H$ such that the signal $x_{S^*}$ consisted of the frequencies in the new heavy-clusters has about $(1-\eps)$-fraction of energy of $x^*$. More specifically, we have $E^{\mathsf{new}}_1=\|x^*-x_{S^*}\|_T\leq (1+\eps){\cal N}$.

To prove that we can actually estimate the frequencies in the new heavy-clusters, we observe a subtle point: the energy condition of heavy-cluster and the energy condition of frequency estimation are inconsistent due to the \emph{noise} in observations. To be able to estimate the one heavy-cluster signal's frequency, it is required that 
\begin{align}\label{eq:intro_heavy_cluster_noise}
    \int_{C_i}|\wh{H\cdot x}(f)|^2\d f\geq \frac{T}{k}{\cal N}^2,
\end{align}
which is different from Eq.~\eqref{eq:intro_heavy_cluster}.
In other words, not all frequencies in $S^*$ are recoverable, but only most of them. Since \cite{ckps16} only wants a constant approximation error, they may simply ignore this difference by losing a constant factor in accuracy. For us, however, we need to make it precise. We define $S$ to be a subset of $S^*$ containing the frequencies in the heavy-clusters satisfying Eq.~\eqref{eq:intro_heavy_cluster_noise}. We analyze the effect of $H\cdot g$ and show that by strengthening the RHS of heavy-cluster's energy condition (Eq.~\eqref{eq:intro_heavy_cluster}) to $\frac{4T}{k}{\cal N}^2$, we can bound the unrecoverable part's energy by $$E_{1.5}:=\|x_{S^*}-x_S\|_T\leq (2+\eps){\cal N}.$$ 

Given the definition of $S_f$ as provided earlier, we also need to establish the upper bound of 
\begin{align*}
E_{S_f} := \|x_{S_f} - x_S \|_T \leq (1 + \epsilon) \mathcal{N}.  
\end{align*}

For the recoverable part $x_{S_f}$, we can just follow \cite{ckps16}'s approach to estimating the frequencies in each heavy-cluster.

\paragraph{Generalized high-accuracy signal estimation} We apply our four-step Fourier set-query framework to solve the Signal Estimation problem in the three-step of \cite{ckps16}'s algorithm and improve $E_3$ from $2200{\cal N}$ to $(4+\eps){\cal N}$. We first define the problem more formally. By frequency estimation, we obtain a list of candidate frequencies of $x_S$ and in the third step, we know that it can be approximated by $x_{S_f,\mathsf{poly}}(t):=\sum_{i=1}^{|L|} e^{2\pi \i \wt{f}_i t}\cdot P_i(t)$ with very tiny error $E_2$, where $P_i(t)$ are some degree-$d$ polynomials. We can rewrite $x_{S,\mathsf{poly}}$ in the \emph{Fourier-monomial mixed basis}:
\begin{align*}
    x_{S,\mathsf{poly}}(t) = \sum_{i=1}^{|L|}\sum_{j=0}^d v_{i,j}\cdot e^{2\pi \i \wt{f}_i t} t^{j},
\end{align*}
where $v_{i,j}\in \C$ and $\wt{f}_i\in L$ are known. That is, we need to learn $\{v_{i,j}\}$ given noisy observations $x_{S_f,\mathsf{poly}}(t)+g'(t)$, which is a Signal Estimation problem for the following family of signals:
\begin{align*}
{\cal F}_{\mathsf{mix}}:=\mathrm{span}\left\{e^{2\pi \i \widetilde{f}_{i} t} \cdot t^j~\bigg{|}~ i\in [|L|],j \in \{0,\cdots,d\} \right\}.    
\end{align*}

We apply our three-step framework as follows. In Step 1, we need an energy bound for ${\cal F}_\mathsf{mix}$. \cite{ckps16} showed that $R_\mathsf{mix}:=\sup_{u(t)\in {\cal F}_\mathsf{mix}}\frac{\sup_t |u(t)|^2}{\|u\|_T^2}\leq \wt{O}(|L|^4d^4)$. Then,  we get that uniformly sample $\wt{O}(|L|^4d^4\epsilon^{-1})$ points in $[0,T]$ gives an oblivious sketching for $x_{S,\mathsf{poly}}$. Furthermore, we can show that this sampler is $\epsilon$-well-balanced. In Step 2, since we aim at achieving high accuracy, we do not distill the sketch but directly apply the sharper error analysis to control the energy of the orthogonal part of the noise, as we did for our high-accuracy signal estimation algorithm  (Theorem~\ref{thm_semi_cont_decouple_thm}). Finally, in Step 3, we solve a weighted linear regression to estimate the coefficients and obtain a signal $y'(t)\in {\cal F}_{\mathsf{mix}}$ such that 
\begin{align*}
    E^{\mathsf{new}}_3:=\|y'-x_{S_f,\mathsf{poly}}\|_T\leq (1+\eps)\|x-x_{S_f,\mathsf{poly}}\|_T\leq (4+\eps){\cal N}.
\end{align*}
Then, we can  transform $y$ back to a $\poly(k)$-Fourier-sparse signal with error $\|y-y'\|_T\leq E_2$. 

Combining them together and re-scaling $\epsilon$ and $\delta$, we get that
\begin{align*}
    \|y-x^*\|_T 
    \leq &~ \|y-y'\|_T + \|y-x_{S_f,\mathsf{poly}}\|_T + \|x_{S_f,\mathsf{poly}}-x_{S_f}\|_T + 
    \|x_{S_f} - x_{S}\|_T
    +   
    \|x_S-x_{S^*}\|_T + \|x_{S^*}-x^*\|_T\\
    \leq &~ E_2 + E_3^\mathsf{new} + E_2 + 
    E_{S_f} +E_{1.5}+E_1^{\mathsf{new}} \\
    \leq &~ (9 +\epsilon)\|g\|_T + \delta \|x^*\|_T.
\end{align*}
Therefore, we obtain a Fourier interpolation algorithm with $(9 +\eps)$-approximation error.

\subsubsection{\texorpdfstring{$(3+\sqrt{2}+\eps)$-approximation algorithm}{}~}

How can we further improve this algorithm? We observe that $E^{\mathsf{new}}_3$ can be written more precisely as $(1+\eps)\|g\|_T + (3+\eps){\cal N}$, where the first term $(1+\eps)\|g\|_T$ is identical to the Signal Estimation problem's approximation error (Theorem~\ref{thm_semi_cont_decouple_thm}) and may not be further improved. On the other hand, $E_1^{\mathsf{new}}$ and $E_{1.5}$ only depend on ${\cal N}$. If we can take a smaller value for ${\cal N}^2$, i.e., $\epsilon (\|g\|_T^2+\delta \|x\|_T^2)$, then we will improve approximation error. We show that it is possible via an \emph{ultra-high sensitivity frequency estimation method}.

\paragraph{Ultra-high sensitivity frequency estimation} To improve the sensitivity of the frequency estimation method, let ${\cal N}_1^2:=\epsilon {\cal N}^2$ and consider the heavy-clusters with parameter ${\cal N}_1$. Let $S^*_1$ denote the set of frequencies of $x^*(t)$ in the ${\cal N}_1$-heavy-clusters.  By the same analysis as in our previous frequency estimation approach, we have $E_1^{\mathsf{new}+}:=\|x^*-x_{S_1^*}\|_T\leq (1+\eps){\cal N}_1$.

However, due to the inconsistent energy conditions, only those frequencies in the heavy-clusters satisfying Eq.~\eqref{eq:intro_heavy_cluster_noise} are recoverable. Let $S_1$ denote the set of such frequencies, and we need to upper bound $\|x_{S_1^*}-x_{S_1}\|_T$. Previously, we strengthen the heavy-cluster's condition (Eq.~\eqref{eq:intro_heavy_cluster}) and get a $E_{1.5}\leq (2+\eps)\|g\|_T$ bound. Here, instead, we relax the RHS of Eq.~\eqref{eq:intro_heavy_cluster_noise} to $\epsilon\cdot \frac{T}{k}{\cal N}_1^2$. Intuitively, more frequencies will satisfy the new frequency estimation condition; and if there is a unrecoverable frequency $f^*\in S^*_1\backslash S_1$, it indicates that its contribution in filtered signal $H\cdot x^*$ is cancelled out by the filtered noise $H\cdot g$. 
Using this \emph{signal-noise cancellation effect}, we prove that:
\begin{align}\label{eq:intro_noise_cancel}
    \|H(x_{S_1} - x_{S^*_1})\|_T^2+\|H(x-x_{S_1})\|_T^2\leq (1+\sqrt{\eps})\|x-x_{S^*_1}\|_T^2,
\end{align}
which saves a factor of 2 from $E_{1.5}$ by introducing an extra term $\|H(x-x_{S_1})\|_T$. Recall $\|x-x_{S_1}\|_T$ is related to $E_3^{\mathsf{new}}$, the error of the signal estimation procedure. We can decompose it into the ``passing energy'' $\|H(x-x_{S_1})\|_T$ and ``filtered energy'' $\|(\mathrm{Id}-H)(x-x_{S_1})\|_T$ and bound them by:
\begin{align*}
    \|x-x_{S_1}\|_T\leq \|H(x-x_{S_1})\|_T + \|g\|_T + O(\epsilon)\|x^*-x_{S_1}\|_T.
\end{align*}
Thus, Eq.~\eqref{eq:intro_noise_cancel} can be considered as upper-bounding $E_{1.5}$ and $E_{3}^{\mathsf{new}}$ simultaneously. Combining them together, we get the following error guarantee for the frequency recoverable signal $x_{S_1}$: 
\begin{align}\label{eq:intro_err_decomp}
    \|x-x_{S_1}\|_T + \|x^*-x_{S_1}\|_T \leq (1+\sqrt{2}+O(\sqrt{\eps}))\|g\|_T + O(\sqrt{\delta})\|x^*\|_T.
\end{align}

\begin{figure}[!ht]
    \centering
    \subfloat[$ \|H(t)\cdot (x_{S^*_1}-x_{S_1})(t)\|_T$]{\includegraphics[width=0.4\textwidth]{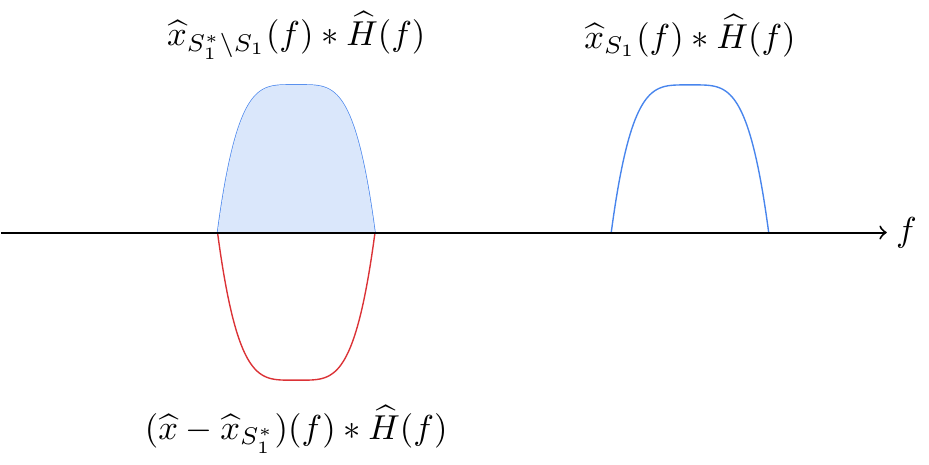}} \label{fig:cancellation_b}
    \vspace{1mm}
    \subfloat[$ \|H(t) \cdot (x-x_{S_1})(t)\|_T$]{\includegraphics[width=0.4\textwidth]{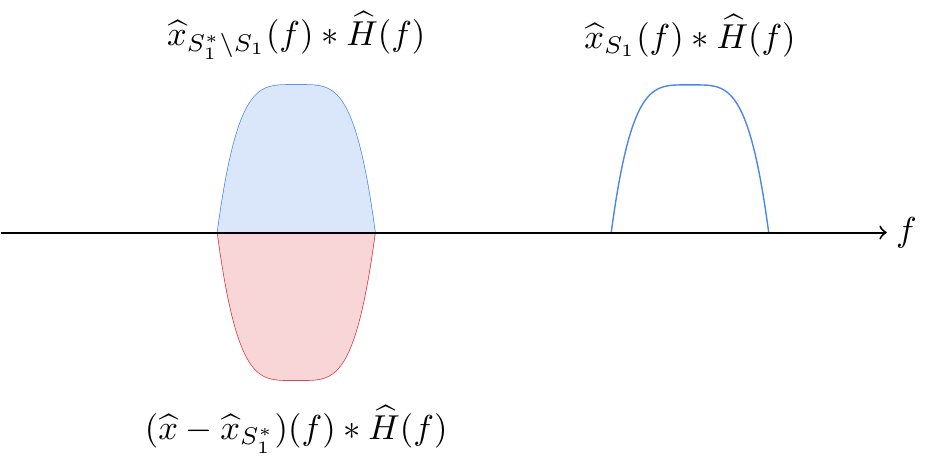}}\label{fig:cancellation_a}
    \caption{
    An illustration of the signal-noise cancellation effect (Eq.~\eqref{eq:intro_noise_cancel}). In (a), the blue region corresponds to the first term of Eq.~\eqref{eq:intro_noise_cancel}, which roughly equals the energy of $x-x_{S^*_1}$. In (b), the red and blue regions correspond to the second term, where most of the energy is canceled.  Thus, their total energy is very close to $\|x-x_{S^*_1}\|_T$. %
    }
    \label{fig:cancellation}
\end{figure}

Then, by a more careful analysis of the \textsc{HashToBins} approach used by \cite{ckps16} for Frequency Estimation, we show that $x_{S_1}$'s frequencies can be efficiently approximated, which gives an ultra-high sensitivity frequency estimation method.

The remaining part of the algorithm is almost identical to the previous one. We run the high-accuracy signal estimation algorithm to reconstruct $x_{S_1}$. Let $y(t)$ denote the output Fourier-sparse signal. By Eq.~\eqref{eq:intro_err_decomp} and re-scaling $\epsilon$ and $\delta$, we have
\begin{align*}
    \|y-x^*\|_T\leq (3+\sqrt{2}+\eps)\|g\|_T + \delta\|x^*\|_T.
\end{align*}
Therefore, we achieve a $(3 +\sqrt{2}+\eps)$-approximation error for the Fourier interpolation.

\section*{Paper Organization}
The remainder of the paper is organized as follows. 
In Sections \ref{sec:FT_problem}, we formally define and study the ``semi-continuous'' Fourier interpolation and set query problems over lattices. In Section~\ref{sec:preliminary}, we provide some preliminaries on Fourier transformation, lattices, etc. %

Then, we focus on developing Fourier set query algorithms based on our three-step framework (Section~\ref{sec:tech_framework}). We first build some technical components in Sections \ref{sec:energy_bound} - \ref{sec:sketch_distillation}.
Section \ref{sec:energy_bound} is for Step 1, where we prove energy bounds and concentration properties for high-dimensional Fourier-sparse signals and discrete Fourier-sparse signals. Section \ref{sec:offline_sketch} is for Step 2, where we describe fast \emph{oblivious} sketching methods and we use them as ``preconditioners'' for continuous and discrete Fourier-sparse signals in one and higher dimensions.  
Sections \ref{sec:fast_wbsp} and \ref{sec:sketch_distillation} are for Step 3. More specifically, in Sections \ref{sec:fast_wbsp}, we design a data structure for improving the time/space complexity of the Randomized BSS algorithm (Theorem~\ref{thm:fast_bss_intro}). Using this data structure, in Section \ref{sec:sketch_distillation} we describe the  \emph{sketch distillation} technique for different kinds of Fourier-sparse signals and analyze its robustness to noise. Finally, these components are wrapped up in Sections \ref{sec:1d-reduction} - \ref{sec:discrete_fourier_set_query}. Section \ref{sec:1d-reduction} shows sample-optimal and high-accuracy signal estimation algorithms for one-dimensional semi-continuous  signals (Theorem~\ref{thm_semi_cont_decouple_thm}). Section \ref{sec:high_d_reduction} generalizes  to high-dimensional signals using lattice theory. Section \ref{sec:discrete_fourier_set_query} gives  a discrete Fourier set query algorithm (Theorem~\ref{thm:set_query_intro}). %

In the last part of paper, in Section~\ref{sec:high_precision_interpolate}, we give a high-accuracy Fourier interpolation algorithm for one-dimensional continuous signals, based on our sharper error control and efficient Fourier set query algorithm.

In the appendix, Section \ref{sec:hd_reduction_noiseless} discusses a special case of the Signal Estimation problem where the observation has no noise. We present a straightforward algorithm for one-dimensional Fourier sparse signals. %
Section \ref{sec:exact_recover_algo} proves that any signal can be approximated by a semi-continuous signal with the same sparsity and polynomially-small frequency gap\footnote{We also show an approximation with $O(1)$-frequency gap but slightly worse sparsity.}, which implies a Fourier interpolation algorithm with optimal output-sparsity with a different error guarantee. %
Finally, Section~\ref{sec:shrinking_range} shows that the accuracy of our Fourier interpolation algorithm can be further improved, if we only care about the signal in a subset of time duration.

%% file: prob_def.tex
\newpage
\paragraph{Notations.}
For any positive integer $n$, we use $[n]$ to denote $\{1,2,\cdots,n\}$. We use $\i$ to denote $\sqrt{-1}$. For a complex number $z \in \C$ where $z = a + \i b$ and $a, b \in \R$. We use $\ov{z}$ to denote the complex conjugate of $z$, i.e., $\ov{z} = a - \i b$. Then it is obvious that $|z|^2 = z \cdot \ov{z} = a^2 +b^2$.  We use $f \lesssim g$ to denote that there exists a constant $C$ such that $f\leq Cg$, and $f\eqsim g$ to denote $ f\lesssim g \lesssim f$.  We use $\wt{O}(f)$ to denote $f\log^{O(1)}(f)$.
We say $x(t)$ is a $k$-Fourier-sparse when $x(t)=\sum_{j=1}^{k} v_j \exp(2\pi\i f_j t)$. We use $\wh{x}(f)$ to denote the Fourier transform of $x(t)$. More specifically, $\wh{x}(f) =\int_{-\infty}^{\infty} x(t) \exp(-2\pi\i f t) \d t$. 
We define our discrete norm as $\|g(t)\|^2_W= \frac{1}{|W|} \sum_{t\in W}|g(t)|^2 $ for function $g$. We define our weighted discrete norm as $\|g(t)\|^2_{S,w}=  \sum_{t\in S}w_t|g(t)|^2 $ for function $g$.  %
We define the continuous $T$-norm as $\|g(t)\|^2_T=\frac{1}{T} \int_0^T |g(t)|^2\d t $ for function $g$. %

In general, we assume $x^*(t)$ is our ground truth and is a $k$-Fourier-sparse signal. We can observe function $x(t)=x^*(t)+g(t)$ for $g(t)$ being a noise function. We can observe $x(t)$ in duration $[0, T]$. The ground truth $x^*(t)$ has frequencies in $[-F,F]$. 

\section{Definitions of Semi-Continuous Fourier Set Query and Interpolation}\label{sec:FT_problem}
In this section, we give the formal definitions of the problems studied in this paper. In Section~\ref{sec:formal_def_set_query}, we define the Fourier set query for discrete and continuous signals. In Section~\ref{sec:formal_def_interpolate}, we define the Fourier interpolation problem and its two sub-problems: frequency estimation and signal estimation.
\subsection{Formal definitions of Fourier set query}\label{sec:formal_def_set_query}

The discrete Fourier set query problem is defined as follows: 
\begin{definition}[Discrete Fourier set query problem]
Let $x\in \C^n$ and $\wh{x}$ be its discrete Fourier transformation. Let $\epsilon > 0$. Given a set $S\subseteq [n]$ and query access to $x$, the goal is to use a few queries to compute a vector $x'$ with support $\mathrm{supp}(x')\subseteq  S$ such that 
\begin{align*}
    \|(x'-\wh{x})_S\|_2^2\leq \epsilon \cdot  \|\wh{x}_{[n]\backslash S}\|_2^2.
\end{align*}
\end{definition}

We also define the continuous Fourier set query problem as follows:
\begin{definition}[Continuous Fourier set query problem]\label{prob:set_query_cont}
For $d\geq 1$, let $x^*(t)$ be a signal in time duration $[0,T]^d$. Let $\wh{x^*}(f)$ denote the continuous Fourier transformation of $x^*(t)$. Let $\epsilon > 0$. Given a set $S\subseteq \R^{d}$ of frequencies such that $\supp(\wh{x^*}) \subseteq S$, and observations of the form $x(t)=x^*(t)+g(t)$, where $g(t)$ denotes the noise. The goal is to output a Fourier-sparse signal $x'(t)$ with support $\supp( \wh{x'} )\subseteq S$ such that
\begin{align*}
    \|x'- x^*\|_T^2\leq %
     (1+\epsilon) \cdot \| g \|_T^2.
\end{align*}
\end{definition}

\subsection{Formal definitions of semi-continuous Fourier interpolation}\label{sec:formal_def_interpolate}

In this section, we provide the following formal definition of the semi-continuous Fourier interpolation problem, where we assume that the frequencies of the signal are contained in a lattice.
\begin{problem}[Semi-continuous Fourier interpolation problem]
\label{prob:recover_on_the_lattice}
Given a basis $\mathcal{B}$ of $m$ known vectors $b_1, b_2, \cdots b_m \in \R^d$, let $\Lambda(\mathcal{B}) \subset \R^d$ denote the lattice 
\begin{align*}
    \Lambda(\mathcal{B}) = \Big\{ z \in \R^d : z = \sum_{i=1}^m c_i b_i, c_i \in \mathbb{Z}, \forall i \in [m] \Big\}
\end{align*}
Suppose that $f_1, f_2, \cdots, f_k \in \Lambda(\mathcal{B})$, $\forall i\in [k], |f_i|\leq F$. Let $x^*(t) = \sum_{j=1}^k v_j e^{2\pi\i \langle f_j , t \rangle }$, and let $g(t)$ denote the noise. Given observations of the form $x(t) = x^*(t) + g(t)$, $t\in [0, T]^d$. Let $\eta=\min_{i \neq j}\|f_j-f_i\|_{\infty}$. There are three goals:

\begin{enumerate}%
    \item The first goal is to design an algorithm that output $f_1, f_2, \cdots, f_k$ \emph{exactly} given query access to the signal $x(t)$ for $t\in [0, T]^d$. 
    \item The second goal is to design an algorithm that output a set $L$ of  frequencies such that, for each $f_i$, there is $f'_i\in L$, $\|f_i-f'_i\|_2\le  D/T$. 
    \item The third goal is to design an algorithm that  output $y(t) = \sum_{j=1}^{\wt{k}} v_j' \cdot e^{2\pi \i f_j' t}$ such that $\int_{[0,T]^d} |y(t) - x(t)|^2 \d t \lesssim \int_{[0,T]^d} |g(t)|^2 \d t$.
\end{enumerate}
\end{problem}

Then, we extract two sub-problems from Problem~\ref{prob:recover_on_the_lattice}: Frequency Estimation and Signal Estimation. We give their definitions below.

We first define the $d$-dimensional frequency estimation under the semi-continuous as follows. In this problem, we want to recover each frequencies in a small range.  
\begin{problem}[Frequency estimation]
\label{prob:freq_estimation}
Given a basis $\mathcal{B}$ of $m$ known vectors $b_1, b_2, \cdots b_m \in \R^d$, let $\Lambda(\mathcal{B}) \subset \R^d$ denote the lattice 
\begin{align*}
    \Lambda(\mathcal{B}) = \Big\{ z \in \R^d : z = \sum_{i=1}^m c_i b_i, c_i \in \mathbb{Z}, \forall i \in [m] \Big\}
\end{align*}
Suppose that $f_1, f_2, \cdots, f_k \in \Lambda(\mathcal{B})$. Let $x^*(t) = \sum_{j=1}^k v_j e^{2\pi\i \langle f_j , t \rangle }$,
and let $g(t)$ denote the noise. Given observations of the form $x(t) = x^*(t) + g(t)$, $t\in [0, T]^d$. %
Let $\eta=\min_{i \neq j}\|f_j-f_i\|_{\infty}$. 

The goal is to design an algorithm that output a set $L$ of frequencies such that, for each $f_i$, there is $f'_i\in L$, $\|f_i-f'_i\|_2\le D / T$. %
\end{problem}
We remark that the recovered frequencies in $L$ are not necessary to be in $\Lambda(\mathcal{B})$, and $D$ is a parameter that can depend on $k$.

Next, we define the $d$-dimensional Signal Estimation under the semi-continuous setting as follows. In this problem, we want to recover a signal that can approximate the ground-truth signal in the time domain. %
\begin{problem}[Signal Estimation problem]
\label{prob:signal_estimation}
Given a basis $\mathcal{B}$ of $m$ known vectors $b_1, b_2, \cdots b_m \in \R^d$, let $\Lambda(\mathcal{B}) \subset \R^d$ denote the lattice 
\begin{align*}
    \Lambda(\mathcal{B}) = \Big\{ z \in \R^d : z = \sum_{i=1}^m c_i b_i, c_i \in \mathbb{Z}, \forall i \in [m] \Big\}
\end{align*}
Suppose that $f_1, f_2, \cdots, f_k \in \Lambda(\mathcal{B})$. Let $x^*(t) := \sum_{j=1}^k v_j e^{2\pi\i \langle f_j , t \rangle }$, and let $g(t)$ denote the noise. Given observations of the form $x(t) = x^*(t) + g(t)$, $t\in [0, T]^d$. Let $\eta=\min_{i \neq j}\|f_j-f_i\|_{\infty}$.

The goal is to design an algorithm that outputs $y(t) = \sum_{j=1}^{\wt{k}} v_j' \cdot e^{2\pi \i f_j' t}$ such that $$\int_{[0,T]^d} |y(t) - x(t)|^2 \d t \lesssim \int_{[0,T]^d} |g(t)|^2 \d t.$$
Note that outputting $y(t) = \sum_{j=1}^{\wt{k}} v_j' \cdot e^{2\pi \i f_j' t}$ means outputting $\{v_j',f_j'\}_{j\in [\wt{k}]}$.
\end{problem}

\begin{remark}
We note that given the solution of Frequency Estimation (Problem~\ref{prob:freq_estimation}), Signal Estimation (Problem~\ref{prob:signal_estimation}) can be formulated as a Fourier set query problem (Problem~\ref{prob:set_query_cont}). More specifically, by Frequency Estimation, we will find a set that contains all frequencies of the ground truth signal $x^*(t)$. Then, we only need to recover the coefficients with frequencies in this set, which is equivalent to a set query problem. 
\end{remark}

%% file: preli.tex
\section{Preliminaries}\label{sec:preliminary}

This section is organized as follows. In Section~\ref{sec:tool_and_ineq}, we provide some technical tools in probability theory and linear algebra. In Section~\ref{sec:basic_fourier_trans}, we review the Fourier transformation for different types of signals. In Section~\ref{sec:lattice}, we show some facts about Lattices. And in Section~\ref{sec:prelim_is}, we discuss the importance sampling method.

\subsection{Tools and inequalities}\label{sec:tool_and_ineq}
\begin{lemma}[Chernoff Bound \cite{chernoff1952}]\label{lem:chernoff_bound}
Let $X_1, X_2, \cdots, X_n$ be independent random variables. Assume that $0\leq X_i \leq 1$ always, for each $i \in [n]$. Let $X= X_1+X_2+\cdots+X_n$ and $\mu = \mathbb{E}[X] = \overset{n}{  \underset{i=1}{\sum} } \mathbb{E}[X_i]$. Then for any $\epsilon>0$,
\[\mathsf{Pr} [ X \geq (1+\epsilon) \mu ] \leq \exp(-\frac{\epsilon^2 }{2+\epsilon} \mu) \textit{ and } \mathsf{Pr} [ X \leq (1-\epsilon) \mu ] \leq \exp(-\frac{\epsilon^2 }{2} \mu).\]
\end{lemma}

\begin{definition}[$\epsilon$-net]
Let $T$ be a metric space with distance measure $d$. Consider a subset $K \subset T$ and let $\epsilon>0$. A subset ${\cal N}\subseteq K$ is called an $\epsilon$-net of $K$ if every point in $K$ is within distance $\epsilon$ of some point of ${\cal N}$, i.e.
\begin{align*}
    \forall x\in K, \exists y\in {\cal N}~\text{s.t.}~d(x,y)\leq \epsilon.
\end{align*}
\end{definition}

\begin{fact}[Fast matrix multiplication]\label{fac:matrix_multiplication}
We use $\Tmat(a,b,c)$ to denote the time of multiplying an $a \times b$ matrix with another $b \times c$ matrix.

We use $\omega$ to denote the exponent of matrix multiplication, i.e., $\Tmat(n,n,n) = n^{\omega}$. Currently $\omega \approx 2.373$ \cite{wil12,l14,aw21}. %
\end{fact}

\begin{fact}[Weighted linear regression]\label{fac:basic_l2_regression}
Given a matrix $A\in \C^{n\times d}$, a vector $b\in \C^n$ and a weight vector $w\in \R_{>0}^n$, it takes $O(nd^{\omega-1})$ time to output an $x'$ such that
\begin{equation*}
x' = \underset{x}{ \arg \min } \| \sqrt{W}(Ax - b)\|_2=(A^*WA)^{-1}A^* W b.
\end{equation*}
where $\sqrt{W}:=\mathrm{diag}(\sqrt{w_1},\dots,\sqrt{w_n})\in \R^{n\times n}$, and $\omega \approx 2.373$ is the exponent of matrix multiplication~\cite{wil12,l14,aw21}.
\end{fact}

\begin{fact}\label{fac:cosin_bound}
    For any $x\in (0,1)$, we have $\cos(x)\leq \exp(-x^2/2)$.
\end{fact}

\subsection{Basics of Fourier transformation}\label{sec:basic_fourier_trans}

The definition of high dimensional Fourier transform is as follows:
\begin{align*}
\wh{x}(f) = \int_{(-\infty,\infty)^d} {x}(t) \exp(-2\pi\i \langle f , t \rangle ) \d t, \text{~where~} f\in \R^d,
\end{align*}
and the definition of high dimensional inverse Fourier transform is as follows:
\begin{align*}
x(t) = \int_{(-\infty,\infty)^d} \wh{x}(f) \exp(2\pi\i \langle f , t \rangle ) \d f, \text{~where~} t\in \R^d.
\end{align*}

Note that when we replace $d=1$ in the definition of high dimensional Fourier transform and inverse Fourier transform above, we get the definition of one-dimensional Fourier transform and inverse Fourier transform.

The definition of discrete Fourier transform is as follows:
\begin{align*}
\wh{x}_f = \sum_{t=1}^n {x}_t \exp(-2\pi\i ft/n ) , \text{~where~} f\in [n],
\end{align*}
and the definition of discrete inverse Fourier transform is as follows:
\begin{align*}
x_t = \frac{1}{n}\sum_{f=1}^n \wh{x}_f \exp(2\pi\i ft/n) , \text{~where~} t\in [n].
\end{align*}

A continuous $k$-Fourier sparse signal $x(t):\R^d\rightarrow\C$ can be represented as follows:
\begin{align*}
x(t)=\sum_{j=1}^k v_j\exp(2\pi\i \langle f_j, t \rangle), ~v_j\in\C,f_j\in\R^d, ~\forall j\in[k]. 
\end{align*}
Thus, $\wh{x}(f)$ is:
\begin{align*}
\wh{x}(f)=\sum_{j=1}^k v_j \delta(t-f_j). 
\end{align*}

A discrete $k$-Fourier sparse signal $x\in\C^n$ can be represented as follows:
\begin{align*}
x_t = \sum_{j\in S} v_j \exp(2\pi\i j t / n), ~ S\subseteq [n], |S|=k, ~v_j\in\C, \forall j \in S.
\end{align*}
So, $\wh{x}_f$ is:
\begin{align*}
    \wh{x}_f = \left\{
    \begin{array}{ll}
        v_j & ,j\in S \\
         0 & ,\text{o.w.}
    \end{array}
    \right.
\end{align*}

%% file: lattice.tex
\subsection{Facts about Lattices} \label{sec:lattice}

\begin{definition}[Lattice]
A lattice ${\cal L}$ in $\R^d$ is defined as follows:
\begin{align*}
    {\cal L}:= \Big\{ \sum_{i=1}^k \lambda_i b_i: \lambda_1,\dots,\lambda_k\in \Z \Big\},
\end{align*}
where $b_1,\dots,b_k\in \R^d$ are linearly independent vectors. And we denote the matrix $B:=\begin{bmatrix}b_1 & \cdots & b_k\end{bmatrix}\in \R^{n\times k}$ as the basis of the lattice ${\cal L}$.
\end{definition}
\begin{definition}[Fundamental parallelepiped]
For a lattice ${\cal L}$ with basis $B$, its fundamental parallelepiped is defined to be:
\begin{align*}
   {\cal P}(B):=\{Bx~|~x\in [0,1)^d\}. 
\end{align*}
\end{definition}

\begin{fact}
For any lattice with basis $B$, we have
\begin{align*}
    \mathrm{vol}({\cal P}(B))=\sqrt{\det(B^\top B)}.
\end{align*}
In particular, if $B$ is full-rank, $\mathrm{vol}({\cal P}(B))=|\det(B)|$.
\end{fact}

\begin{lemma}[The number of lattice points within a ball \footnote{We thank Thomas Rothvoss for providing the proof of this bound. %
}]\label{lem:number_of_lattice_points}
Let ${\cal L}$ be any lattice with basis $B$ such that the spectral norm $\|B\| \leq \ell$. 
Then, the number of lattice points inside a ball centered at 0  with radius $R$ is upper bounded by:
\begin{align*}
    |{\cal L}\cap B_d(0,R)|\leq (1+\frac{\sqrt{k}\ell}{R})^k\cdot \frac{\mathrm{vol}(B_k(0,R))}{\mathrm{vol}({\cal P}(B))}.
\end{align*}
\end{lemma}
\begin{proof}
We first show that for two different lattice points $x,y\in {\cal L}\cap B_d(0,R)$, the translations of ${\cal P}$ at $x$ and at $y$ are disjoint, i.e., $(x+{\cal P})\cap (y+{\cal P})= \emptyset$.

Suppose $(x+{\cal P})\cap (y+{\cal P}) \ne \emptyset$ for some $x,y\in {\cal L}\cap B_d(0,R)$. Then, we have
\begin{align*}
    x+\sum_{i=1}^k \lambda_i b_i = y + \sum_{i=1}^k \tau_i b_i,
\end{align*}
where $\lambda_i, \tau_i\in [0,1)$. It gives
\begin{align*}
    x-y=\sum_{i=1}^k (\tau_i - \lambda_i) b_i.
\end{align*}
Note that $x-y\in {\cal L}$, which means
\begin{align*}
    \sum_{i=1}^k (\tau_i - \lambda_i)b_i = \sum_{i=1}^k c_i b_i,
\end{align*}
where $c_1,\dots,c_k\in \Z$. Since $\tau_i - \lambda_i\in (-1, 1)$, we get that $c_i=0$ for all $i\in [k]$. Thus, $x=y$.

Then, for any point $y\in x+{\cal P}$, where $x\in {\cal L}\cap B_d(0,R)$, we have
\begin{align}\label{eq:total_width}
    \|y\|_2 = \|x+z\|_2\leq \|x\|_2 + \|z\|_2\leq R + \sqrt{k}\ell,
\end{align}
where $z\in {\cal P}$ and the last step follows from $x\in B_d(0,R)$ and $\|z\|_2=\|B\lambda\|_2\leq \|B\|\|\lambda\|_2\leq \sqrt{k}\ell$, for some $\lambda\in [0,1)^k$.

Then, we have
\begin{align*}
    |{\cal L}\cap B_d(0,R)| \leq &~ \sum_{x\in {\cal L}\cap B(0,R)}\frac{\mathrm{vol}(x+{\cal P})}{\mathrm{vol}({\cal P}(B))}\\
    \leq &~ \frac{\mathrm{vol}(B_k(0,R+\sqrt{k}\ell))}{\mathrm{vol}({\cal P}(B))}\\
    \leq&~ (1+\frac{\sqrt{k}\ell}{R})^k\cdot \frac{\mathrm{vol}(B_k(0,R))}{\mathrm{vol}({\cal P}(B))},
\end{align*}
where the second step follows from the disjointness of translations and the bound on the total width (Eq.~\eqref{eq:total_width}).

The lemma is then proved.
\end{proof}

We define the shortest vector problem (SVP) as follows:
\begin{definition}[SVP]
Let ${\cal L}$ denote a Lattice. 
We define $\SVP( {\cal L} )$,
\begin{align*}
    \SVP( {\cal L} ) := \min \{ \| x \|_2 ~|~ x \in {\cal L} \backslash \{ {\bf 0} \} \} .
\end{align*}
Given a basis of ${\cal L}$, the goal is to compute $\SVP( {\cal L} )$.
\end{definition}

In fact, compute SVP (or even approximations of SVP) is an NP-hard problem.
The following theorem shows a well-known lower bound for the shortest vector length.
\begin{theorem}[Theorem 1.10 in \cite{r16}, a lower bound on shortest vector]\label{thm:svp_lower_bound}
Let ${\cal L}$ denote a lattice with basis ${\cal B}$. Let $(b_1^*, \cdots, b_n^*)$ be its Gram-Schmit orthogonalization. Then
\begin{align*}
    \SVP( {\cal L}) \geq \min_{i \in [n]} \| b_i^* \|_2
\end{align*}
\end{theorem}

\begin{fact}\label{fact:gram_running_time}
The Gram–Schmidt process takes a finite, linearly independent set of vectors $ S = \{v_1, \cdots, v_k\}$  for $k \leq n$, runs $O(nk^2)$ time, and generates an orthogonal set $S'  = \{u_1, \cdots, u_k\} $ that spans the same $k$-dimensional subspace of $\R^n$ as $S$. 
\end{fact}

\subsection{Facts about importance sampling}\label{sec:prelim_is}

Important sampling try to estimate a statistic value in one distribution by taking samples in another distribution. In particular, \cite{cp19_colt} considered the importance sampling for estimating the norm of functions in a linear family ${\cal F}$. 

In this followings, we first provide some basic definitions about linear function family. 

\begin{definition}[Condition number of sampling distribution]\label{def:cond_num}
Let $G$ be any domain and ${\cal F}$ is a linear function family from $G$ to $\C$. Let $D$ be an arbitrary distribution over $G$. Then the condition number of $D$ with respect to ${\cal F}$ is defined as follows:
\begin{equation*}
K_{D}:=\sup_{t \in G} \sup_{f \in {\mathcal F}}  \frac{|f(t)|^2}{\|f\|_D^2},
\end{equation*}
where 
\begin{align*}
    \|f\|_D^2:=\int_G D(t)\cdot |f(t)|^2\d t.
\end{align*}
\end{definition}

\begin{definition}[Orthonormal basis for linear function family]
Let $G$ be any domain. Given a linear function family $\mathcal{F}$ from $G$ to $C$, and a probability distribution $D$ over $G$. We say $\{v_1,\ldots,v_d\}$ form an orthonormal basis of $\mathcal{F}$ with respect to $D$, if they satisfy the following properties:
\begin{itemize}
    \item for any $i,j\in [d]$, $\int_G D(t) v_i(t)\ov{v_j(t)}\d t={\bf 1}_{i=j}$, and
    \item for any $f\in {\cal F}$, $f\in \mathrm{span}\{v_1,\dots,v_d\}$.
\end{itemize}%

\end{definition}

\begin{fact}\label{fac:orthonormal_norm}
Let $\{v_1,\dots,v_k\}$ be an orthonormal basis of ${\cal F}$ with respect to $D$. For any function $f \in \mathcal{F}$, let $\alpha(f)$ denote the coefficients under the basis $\{v_1,\ldots,v_d\}$, i.e., $h=\sum_{i=1}^d \alpha(h)_i \cdot v_i$. Then, $$\|\alpha(h)\|_2=\|h\|_D.$$ 
\end{fact}

For an unknown function $f\in {\cal F}$, the goal of importance sampling is to estimate $\|f\|_{D}$, given samples from another distribution $D'$. The following definition introduces the importance sampling procedure and condition number of the importance sampling distribution. 
\begin{definition}[Definition 3.1 of \cite{cp19_colt}]\label{def:importance_sampling}
For any unknown distribution $D'$ over the domain $G$ and any function $f \in {\cal F}$, let $f^{(D')}(t):=\sqrt{\frac{D(t)}{D'(t)}} \cdot f(t)$ be the importance sampling function for some known distribution $D$ such that $$
\underset{t \sim D'}{\E} \left[ |f^{(D')}(t)|^2\right]=\underset{t \sim D'}{\E} \left[\frac{D(t)}{D'(t)} |f(t)|^2 \right]= \underset{t \sim D}{\E}\left[|f(t)|^2\right]. 
$$
Then, we can use samples from $D'$ to estimate $\|f^{(D')}\|_{D'}$, which gives an estimate of $\|f\|_D$.

When the family $\mathcal{F}$ and $D$ is clear, we use $K_{\mathsf{IS},D'}$ to denote the condition number of importance sampling from $D'$: 
\begin{align}\label{eq:def_K}
K_{\mathsf{IS},D'}=\underset{t}{\sup} \left\{ \underset{f \in \mathcal{F}}{\sup} \left\{ \frac{|f^{(D')}(t)|^2}{\|f^{D'}\|_{D'}^2} \right\}\right\}=\underset{t}{\sup} \bigg\{ \frac{D(t)}{D'(t)} \cdot \underset{f \in \mathcal{F}}{\sup} \big\{ \frac{|f(t)|^2}{\|f\|_D^2} \big\} \bigg\}.
\end{align}
\end{definition}

From Definition~\ref{def:importance_sampling}, we know that the efficiency of importance sampling depends on how many samples we need to estimate $\|f^{D'}\|_{D'}$. The following lemma provide a criteria for judging whether a set of samples gives a good estimation for the norm of function. %
\begin{lemma}[Lemma 4.2 in \cite{cp19_colt}]\label{lem:operator_estimation}

For any $\eps\in (0, 1)$, let $S=\{t_1,\dotsc,t_s\}$ and the weight vector $w\in \R_{>0}^s$. Define a matrix $A\in \R^{s\times d}$ be the $s \times d$ matrix defined as $A_{i,j}=\sqrt{w_i} \cdot v_j(t_i)$, where $\{v_1,\dots,v_d\}$ is an orthonormal basis for ${\cal F}$.  Then
\[ 
\|h\|^2_{S,w}:=\sum_{j=1}^s w_j \cdot |h(x_j)|^2 \in [1 \pm \eps] \cdot \|h\|_D^2 \text{ \quad for every } h \in {\mathcal F}
\]
if and only if the eigenvalues of $A^* A$ are in $[1-\eps,1+\eps]$.
\end{lemma}

The following lemma shows that the sample complexity depends on the condition number $K_{\mathsf{IS},D'}$:

\begin{lemma}[Lemma 6.6 in \cite{cp19_colt}]\label{lem:general_distribution}
Let $D'$ be an arbitrary distribution over $G$ and let $K_{\mathsf{IS}, D'}$ be the condition number of importance sampling from $D'$ (defined by Eq.~\eqref{eq:def_K}).
There exists an absolute constant $C$ such that for any $\eps\in (0,1)$ and $\delta \in (0,1)$, let $S=\{t_1,\dots,t_s\}$ be a set of i.i.d. samples from the distribution $D'$ and let $w$ be the weight vector defined by $w_j=\frac{D(t_j)}{s \cdot D'(t_j)}$ for each $j \in [s]$. Then, as long as
\begin{align*}
    s \ge \frac{C}{\eps^2} \cdot K_{\mathsf{IS},D'} \log \frac{d}{\delta},
\end{align*}
the $s \times d$ matrix $A_{i,j}=\sqrt{w_i} \cdot v_j(t_i)$
satisfies
\[
\|A^* A- I\|_2 \le \eps \text{ with probability at least } 1-\delta.
\]
\end{lemma}

%% file: energy_bound.tex
\section{Energy Bounds for Fourier Signals}
\label{sec:energy_bound}

The energy bound shows that the maximum value of a Fourier sparse signal in a certain interval can be bounded by its energy on the interval. One interesting fact is that the approximation ratio in the energy bound is only relate to the sparsity $k$, and have no relationship with time duration $T$ and band-limit $F$. An application of energy bound is preserving the norm, that is what is the least size of set $S$, such that $\|f\|_{S}=\|f\|_T$, for any function $f$ in a certain function family. The relationship between energy bound and norm preserving can be build by Chernoff bound.

\cite{be06, k08, ckps16, cp19_icalp} proved energy bounds for sparse Fourier signal under one-dimensional continuous Fourier transform. 
We further generalize these results to discrete Fourier sparse signal under discrete Fourier transform and high-dimensional Fourier sparse signal under continuous Fourier transform.

This section is organized as follows:
\begin{itemize}
    \item Section \ref{sec:energy_bound_1d} reviews previous results for one-dimensional continuous Fourier-sparse signals.
    \item Section \ref{sec:hd_energy_bound} proves a new energy bound for high dimensional Fourier-sparse signals, and also gives a nearly matching lower bound. 
    \item Section \ref{sec:energy_bound_discrete} proves energy bound for discrete Fourier-sparse signals.
    \item Section \ref{sec:energy_bound_to_concentration} builds the connection between energy bound and the concentration property.
\end{itemize}

\subsection{Energy bound for one-dimensional signals}\label{sec:energy_bound_1d}
In this section, we review the energy bound proved in prior work \cite{be06, k08, ckps16, cp19_icalp}. %

\cite{k08} proved the following energy bound: 
\begin{theorem}[\cite{k08, ckps16}] %
\label{thm:worst_case_sFFT_improve}
Define a family of $F$-band-limit, $k$-sparse Fourier signals:
\begin{align*}
{\cal F}:=\Big\{ x(t)=\sum_{j=1}^k v_j \cdot e^{2 \pi \i f_j t} ~\Big|~ f_j \in \mathbb{R} \cap [-F,F] \Big\}    
\end{align*}
Then, for  any $t \in (-1,1)$,
\begin{align*}
    \underset{x \in {\cal F}}{\sup} \frac{|x(t)|^2}{\|x\|_D^2} \lesssim k^2.
\end{align*}
\end{theorem}

\cite{be06} also proved a time-dependent energy bound for one-dimensional signal:
\begin{theorem}[\cite{be06, cp19_colt}] %
\label{thm:bound_k_sparse_FT_x_middle_improve}
Define a family of $F$-band-limit, $k$-sparse Fourier signals:
\begin{align*}
{\cal F}:=\Big\{ x(t)=\sum_{j=1}^k v_j \cdot e^{2 \pi \i f_j t} ~\Big|~ f_j \in \mathbb{R} \cap [-F,F] \Big\}    
\end{align*}

Then, for  any $t \in (-1,1)$,
\begin{align*}
    \underset{x \in {\cal F}}{\sup} \frac{|x(t)|^2}{\|x\|_D^2} \lesssim \frac{k}{1-|t|}.
\end{align*}
\end{theorem}

\subsection{Energy bound for high-dimensional signals}\label{sec:hd_energy_bound}

The goal of this section is to prove Theorem~\ref{thm:max_is_bounded_by_energy_high_D}, which gives an energy bound for $d$-dimensional Fourier signal. It can be viewed as a $d$-dimensional version of {\cite[Lemma~5.1]{ckps16}}. We also prove a lower bound in Lemma~\ref{lem:hd_energy_lb}.
\begin{theorem}[Energy bound in $d$-dimensional]
\label{thm:max_is_bounded_by_energy_high_D}
For any $d$-dimensional $k$-Fourier-sparse signal $x(t) : \mathbb{R}^d \rightarrow \mathbb{C}$ and any duration $T$, we have
\begin{align*}
    \max_{t\in [0,T]^d } |x(t)|^2 \leq k^{O(d)}  \| x \|_T^2,
\end{align*}
where $\| x \|_T^2 = \frac{1}{T^d} \int_{ [0,T]^d } | x(t) |^2 \mathrm{d} t $.
\end{theorem}

\begin{proof}
Without loss of generality, we fix $T=1$. Then $\| x \|_T^2 = \int_{ [0,1]^d} |x(t)| \mathrm{d} t$. Because $\| x\|_T^2$ is the average over the interval $[0,T]^d$, if the maximizer $t^* = \arg\max_{t\in [0,T]^d} |x(t)|^2$ is not $0^d$ or $T = 1$, we can scale the two intervals $[0^d,t^*]$ and $[t^*,T^d]$ to $[0,1]$ and prove the desired property separately. Hence we assume that $|x(0)|^2 = \max_{t\in [0,T]} |x(t)|^2$ in the proof.

In the next a few paragraphs, we show how to use Lemma~\ref{cla:fourier_signal_linear_combination} to prove Theorem~\ref{thm:max_is_bounded_by_energy_high_D}. 

We use $0^d$ to denote a length-$d$ vector with $0$ everywhere. Due to Lemma~\ref{cla:fourier_signal_linear_combination}, we can choose $t_0 = 0^d$ such that $\forall \tau \in \R_{>0}^d$ there exist $C_1, \cdots, C_m \in \mathbb{C}$, and 
\begin{align*}
x(0^d ) = \sum_{j\in [m]} C_j \cdot x(j \cdot \tau ).
\end{align*}
By the Cauchy-Schwarz inequality, it implies that for any $\tau$,
\begin{align*}
|x(0^d)|^2 \leq m \sum_{ j \in [m] } |C_j|^2 | x( j \cdot \tau )|^2
\end{align*}

Then, we obtain
\begin{align}
\label{eq:rewrite_x_at_0_unit_box}
|x(0^d )|^2 = & ~ m^d \int_{[0,1/m]^d}  | x( 0^d ) |^2 \mathrm{d} \tau \notag  \\
\lesssim & ~ m^d \cdot \int_{[0,1/m]^d}  \left( m \sum_{j=1}^m |x(j \cdot \tau)|^2 \right) \mathrm{d} \tau \notag \\
= & ~ m^{d+1} \cdot \sum_{j=1}^m \int_{[0,1/m]^d}   |x(j \cdot \tau)|^2  \mathrm{d} \tau \notag \\
= & ~ m^{d+1} \cdot \sum_{j=1}^m \frac{1}{j^d} \int_{[0,j/m]^d} |x(\tau)|^2 \mathrm{d} \tau \notag  \\
\leq & ~ m^{d+1} \cdot \sum_{j=1}^m \frac{1}{j^d} \int_{[0,1]^d} |x(\tau)|^2 \mathrm{d} \tau \notag \\
\lesssim & ~ m^{d+1} \log m \cdot \| x \|_T^2 \notag \\
\leq & ~ k^{O(d)} \| x \|_T^2,
\end{align}
where the third step follows by moving $m$ outside of the integral and swapping the integration and the summation, the fourth step follows by replacing $j\tau$ by $\tau$, the fifth step follows by $j/m \leq 1$, the sixth step follows by $\sum_{j=1}^m 1/j^d \leq \sum_{j=1}^m 1/j = O(\log m)$ and the definition of $\| x \|_T^2$, and the last step follows from Lemma~\ref{cla:fourier_signal_linear_combination} that $m=\poly(k)$.

Thus, we have the desired bound.
\end{proof}

The following lemma shows that each point of the signal can be expressed as a linear combination of about $k^2$ equally spaced signal points. 
\begin{lemma}[$d$-dimensional signal interpolation]%

\label{cla:fourier_signal_linear_combination}
For any $k$ and $d$, there exists $m = O( k^2 \log k )$ such that for any $d$-dimensional $k$-Fourier-sparse signal $x(t)$, any $t_0 \in \R^d_{\geq 0}$ and $\tau \in \R_{>0}^d$, there always exist $C_1, C_2, \cdots, C_m \in \mathbb{C}$ such that the following properties hold,
\begin{align*}
\mathrm{Property~\RN{1}} & \quad ~ |C_j| \leq 11 ~ \text{~for~all~} j \in [m] , \\
\mathrm{Property~\RN{2}} & \quad ~ x(t_0) = \sum_{j\in [m]} C_j \cdot x(t_0 + j \cdot \tau).
\end{align*}
\end{lemma}

\begin{proof}
Consider a specific signal $x(t) := \sum_{i=1}^k v_i e^{2\pi\i f_i^\top t}$ for  $t\in \R^d$, where $f_i \in \R^d$ are given. We fix $t_0 \in \R^d$ and $\tau \in \R^d$, and then rewrite $x(t_0+j\cdot \tau)$ as a polynomial of $b_i := v_i \cdot e^{2\pi \i f_i^\top t_0}$ and $z_i := e^{2\pi \i f_i^\top \tau}$ for each $i\in [k]$.
\begin{align*}
x(t_0 + j \cdot \tau) = & ~ \sum_{i=1}^k v_i e^{2\pi \i f_i^\top (t_0 + j\tau ) } \\
= & ~ \sum_{i=1}^k v_i e^{2\pi \i f_i^\top t_0} e^{2\pi \i f_i^\top j \tau} \\
= & ~ \sum_{i=1}^k b_i \cdot z_i^j.
\end{align*}
where the last step follows from the definition of $b_i$ and $z_i$. 

Given $k$ and $z_1, \cdots, z_k$, let $P(z) = \sum_{j=0}^m c_j z^j$ be the degree $m$-polynomial in \cite[Lemma~5.4]{ckps16}.
\begin{align*}
\sum_{j=0}^m c_j x (t_0 +j \tau ) = & ~ \sum_{j=0}^m c_j \sum_{i=1}^k b_i \cdot z_i^j \\
= & ~ \sum_{i=1}^k b_i \sum_{j=0}^m c_j \cdot z_i^j \\
= & ~ \sum_{i=1}^k b_i P(z_i) \\
= & ~ 0,
\end{align*}
where the last step follows by Property \RN{1} of $P(z)$ in \cite[Lemma~5.4]{ckps16}.

By Property \RN{2} and \RN{3} in \cite[Lemma~5.4]{ckps16}, we have $x(t_0 )= -\sum_{j=1}^m c_j x(t_0 + j \tau)$.
\end{proof}

The energy bound in Theorem~\ref{thm:max_is_bounded_by_energy_high_D} for $d$-dimensional signals is nearly optimal due to a  $k^{\Omega(d)}$ lower bound as follows.\footnote{The proof is due to Yang P. Liu.} 

\begin{lemma}\label{lem:hd_energy_lb}
Given $d\geq 1, \delta \in (0, 0.1), k \in \Z $ such that $k \geq O(d^{1+\delta})$. Then, there is a $d$-dimensional $k$-Fourier-sparse signal $x(t) : \mathbb{R}^d \rightarrow \mathbb{C}$ and a duration $T$ such that,
\begin{align*}
    \max_{t\in [0,T]^d } |x(t)|^2 \geq k^{\Omega(\delta d)}  \| x(t) \|_T^2.
\end{align*}
\end{lemma}
\begin{proof}
We consider the following construction of $x(t)$:
\begin{align*}
    x(t):=2^{-k} (1+e^{2\pi\i\langle f_0, t \rangle})^k,
\end{align*}
where $f_0 = {\bf 1}/(100 d T)\in \R^d$. 

It is easy to see that $x$ is a $k$-Fourier sparse signal, and 
\begin{align*}
    |x(t)|^2 = &~2^{-2k} |1+e^{2\pi \i \langle f_0,t\rangle}|^{2k}\\
    = &~ 2^{-2k} \big((1+e^{2\pi \i \langle f_0,t\rangle}) (1+e^{-2\pi \i \langle f_0,t\rangle})\big)^k\\
    = &~ 2^{-2k} \big(2+e^{2\pi \i \langle f_0,t\rangle} + e^{-2\pi \i \langle f_0,t\rangle}\big)^k\\
    = &~ 2^{-2k} \cdot 2^k (1+\cos(2\pi\langle f_0,t\rangle))^k\\
    = &~ \Big(\frac{1+\cos(2\pi\langle f_0,t\rangle)}{2}\Big)^k\\
    = &~ \cos(\pi \langle f_0,t\rangle)^{2k},
\end{align*}
where the first step follows from the definition, the second step follows from $|z|^2=z\ov{z}$, the third step is straightforward, the fourth step follows from $e^{ia}+e^{-ia}=2\cos(a)$, the fifth step is straightforward, the last step follows from $(\cos(a)+1)/2 = \cos(a/2)^2$.

Then, we know that
\begin{align}\label{eq:energy_lb_max}
    \max_{t\in [0,T]^d}|x(t)|^2 = |x({\bf 0})|^2 = 1.
\end{align}

It remains to upper bound 
\begin{align*}
    \|x(t)\|_T^2 = T^{-d} \int_{[0,T]^d}|x(t)|^2 \d t = T^{-d} \int_{[0,T]^d} \cos(\pi \langle f_0, t\rangle)^{2k}\d t.
\end{align*}
Let $r$ be a parameter. We have
\begin{align*}
    \int_{[0,T]^d} \cos(\pi \langle f_0, t\rangle)^{2k}\d t =&~ \int_{W_d(r)} \cos(\pi \langle f_0, t\rangle)^{2k}\d t + \int_{[0,T]\backslash W_d(r)}\cos(\pi \langle f_0, t\rangle)^{2k}\d t\\
    \leq &~ \int_{W_d(r)} 1\cdot \d t + T^d \cdot \max_{t\in [0,T]\backslash W_d(r)}~\cos(\pi \langle f_0, t\rangle)^{2k},
\end{align*}
where $W_d(r):=\{t\in [0,T]^d~|~t_1+\cdots+t_d\leq r\}$. 

We first bound the second term:
\begin{align*}
    \cos(\pi \langle f_0,t\rangle)^{2k}\leq \exp\big(-\pi^2 k \langle f_0,t\rangle^2\big)\leq \exp(-\Omega(kr^2/(dT)^2)),
\end{align*}
where the first step follows from Fact~\ref{fac:cosin_bound}, and the second step follows from the definition of $f_0$ and $t\notin W_d(r)$. Hence,
\begin{align}\label{eq:eb_lb_1}
    T^d \cdot \max_{t\in [0,T]\backslash W_d(r)}~\cos(\pi \langle f_0, t\rangle)^{2k}\leq T^d\cdot \exp(-\Omega(kr^2/(dT)^2)).
\end{align}

Next, we bound the first term, which is equal to the volume of $W_d(r)$. Note that $W_d(r)$ is contained in the following simplex:
\begin{align*}
    P_d(r):=\{t\in \R_+^d~|~t_1+\cdots+t_d\leq r\}.
\end{align*}
Thus, we have
\begin{align}\label{eq:eb_lb_2}
    \mathrm{Vol}(W_d(r)) \leq &~ \mathrm{Vol}(P_d(r))\leq  r^d \cdot \mathrm{Vol}(P_d(1))\notag\\
    \leq &~  r^d/d!\notag\\
    = &~ O(r/d)^d,
\end{align}
where the first step is straightforward, the second step follows from the scaling of the volume, the third step follows from a well-known fact on the volume of a $d$-dimensional simplex $P_d(1)=1/d!$ (see e.g. \cite{ste66}), and the last step follows from Stirling's approximation.

Combining Eqs.~\eqref{eq:eb_lb_1} and \eqref{eq:eb_lb_2} together, we get that
\begin{align*}
    T^{-d}\|x(t)\|_T^2 \leq O(\frac{r}{dT})^d + \exp(-\Omega(kr^2/(dT)^2)).
\end{align*}
By taking $r=Tdk^{-\delta/6}$, we have
\begin{align*}
    O(\frac{r}{dT})^d = O(k^{-\delta / 6})^d = k^{-\Omega(\delta d)},
\end{align*}
and
\begin{align*}
    \exp(-\Omega(kr^2/(dT)^2)) =&~ \exp(-\Omega(k^{1-\delta/3}))\\
    =&~ \exp(-\Omega(d^{(1+\delta)(1-\delta/3)}))\\
    \leq &~  \exp(-\Omega(\delta d\log k))\\
    = &~ k^{-\Omega(d)},
\end{align*}
where the first step is straightforward, the second step follows from $k=O(d^{1+\delta})$, and the last step follows from $\delta\in (0,1)$.

Therefore, we have
\begin{align*}
    T^{-d}\|x(t)\|_T^2 \leq k^{-\Omega(d)} = k^{-\Omega(d)} \cdot \max_{t\in [0,T]^d}~|x(t)|^2,
\end{align*}
where the last step follows from Eq.~\eqref{eq:energy_lb_max}.

The lemma is then proved.

\end{proof}

\subsection{Energy bound for discrete Fourier signals}\label{sec:energy_bound_discrete}
Recall the definition of one-dimensional discrete sparse Fourier signal: for $t\in \{0,1,\dots, n-1\}$, %
\begin{align}\label{eq:def_discrete_signal}
    x_t = \frac{1}{n}\sum_{i=1}^{k} \hat{x}_{f_i}\exp\left(\frac{2\pi\i}{n} f_i t\right), ~f_i \in [n]
\end{align}
More generally, for $d\geq 1$, the $d$-dimensional discrete sparse Fourier signal can be defined as follows. Let $n=p^d$ where both $p$ and $d$ are positive integers. Recall the definition of high-dimensional discrete sparse Fourier signal (see e.g. \cite{nsw19}):
\begin{align}\label{eq:hd_def_discrete_signal}
    x_t = \frac{1}{n}\sum_{i=1}^{k} \hat{x}_{f_i}\exp\left(\frac{2\pi\i}{p}\langle f_i , t\rangle \right) ~~~\forall t\in [p]^d,
\end{align}
where each $f_i \in [p]^d$.

In this section, we prove the following discrete Fourier signals energy bound that works for any dimension:
\begin{theorem}[Discrete $d$-dimensional Fourier energy bound]\label{thm:discrete_energy_bound}
For $d\geq 1$ and any discrete  $d$-dimensional $k$-sparse Fourier signal $\{x_i\}_{i\in [n]}$, we have
\begin{align*}
    \|x\|_\infty^2 \leq k\cdot \frac{\|x\|_2^2}{n}
\end{align*}
\end{theorem}
\begin{proof}
\begin{align*}
    \|x\|_\infty^2 \leq \frac{1}{n^2}\|\hat{x}\|_1^2 \leq \frac{k}{n^2}\|\hat{x}\|_2^2 = \frac{k}{n}\|x\|_2^2,
\end{align*}
where the first step follows from Claim~\ref{clm:hd_l_infty_l_1}, the second step follows from Cauchy-Schwarz inequality, and the last step follows from Theorem~\ref{thm:parseval}.
\end{proof}
\begin{claim}\label{clm:hd_l_infty_l_1}
For any $d\geq 1$ and any discrete $d$-dimensional Fourier signal $x$,
\begin{align*}
    \|x\|_\infty \leq \frac{1}{n}\|\hat{x}\|_1.
\end{align*}
\end{claim}
\begin{proof}
By triangle inequality,
\begin{align*}
    |x(t)|\leq \frac{1}{n}\left|\sum_{i=1}^{k} \hat{x}_i\exp\left(\frac{2\pi\i}{p}\langle f_i, t\rangle \right)\right|\leq \frac{1}{n}\sum_{i=1}^k |\hat{x}_i|=\frac{1}{n}\|\hat{x}\|_1.
\end{align*}
\end{proof}

\begin{theorem}[Parseval's theorem]\label{thm:parseval}
For any $d\geq 1$ and any discrete $d$-dimensional Fourier signal (Eq.~\eqref{eq:hd_def_discrete_signal}),
\begin{align*}
    \|x\|_2^2 = \frac{1}{n}\|\hat{x}\|_2^2.
\end{align*}
\end{theorem}

\subsection{Energy bounds imply concentrations}\label{sec:energy_bound_to_concentration}

By using Chernoff bound, we prove the following lemma to show the performance of uniformly sampling.  

\subsubsection{Continuous case}

\begin{lemma}\label{lem:max_is_bounded_imply_sample_is_small_high_D}
Let $d\in\Z_+$. Let $R$ be a parameter. Given any function $x(t) : \mathbb{R}^d \rightarrow \mathbb{C}$ with $\underset{t\in [0,T]^d }{\max} | x(t)|^2 \leq R \| x(t) \|_T^2$. Let $S$ denote a set of points chosen uniformly at random from $[0,T]^d$. We have that
\begin{eqnarray*}
\mathsf{Pr} \left[ \left| \frac{1}{|S|} \sum_{i\in S} |x(t_i)|^2- \| x(t)\|_T^2 ]  \right|  \geq \epsilon \| x(t)\|_T^2 \right] & \leq & \exp(-\Omega (\epsilon^2 |S|/ R )),
\end{eqnarray*}
where $\| x(t) \|_T^2 = \frac{1}{T^d} \int_{ [0,T]^d } | x(t) |^2 \mathrm{d} t $.
\end{lemma}

\begin{proof}
Let $M$ denote  $\underset{t\in [0,T]^d }{\max} |x(t)|^2$.  Replacing $X_i$ by $\frac{ |x(t_i)|^2}{ M}$ and $n$ by $|S|$ in Lemma \ref{lem:chernoff_bound}, we obtain that
\begin{align*}
\mathsf{Pr}[ | X - \mu | > \epsilon \mu ] \leq 2 \exp(-\frac{\epsilon^2}{3} \mu) 
\end{align*}
The above equation implies
\begin{align*}
 \mathsf{Pr} \left[ \left| \sum_{i\in S} \frac{ |x(t_i)|^2 }{ M } -  |S|\frac{\| x(t) \|_T^2}{ M}    \right| > \epsilon  |S|\frac{\| x(t) \|_T^2}{  M} \right] \leq 2\exp(-\frac{\epsilon^2}{3}\mu) 
\end{align*}
Multiplying $M$ on the both sides
\begin{align*}
\mathsf{Pr} \left[ \left| \frac{1}{|S|} \sum_{i\in S} |x(t_i)|^2- \| x(t) \|_T^2  \right|  \geq \epsilon \|x(t)\|_T^2 \right] \leq 2\exp(-\frac{\epsilon^2}{3}\mu) 
\end{align*}
Applying bound on $\mu$ 
\begin{align*}
 \mathsf{Pr} \left[ \left| \frac{1}{|S|} \sum_{i\in S} |x(t_i)|^2- \| x(t) \|_T^2  \right|  \geq \epsilon \|x(t)\|_T^2 \right] \leq 2\exp(-\frac{\epsilon^2}{3} |S|\frac{ \|x(t)\|_T^2}{  M}  )  
\end{align*}
which is less than $2\exp(-\frac{\eps^2}{3} |S| / R)$, thus completes the proof.
\end{proof}

\subsubsection{Discrete case}

\begin{lemma}\label{lem:max_is_bounded_imply_sample_is_small_discrete}
Let $R$ be a parameter. Given any function $x\in\C^n$ with $\| x\|_\infty^2 \leq R \| x \|_2^2 / n$. Let $S$ denote a set of points chosen uniformly at random from $[n]$. We have that
\begin{eqnarray*}
\mathsf{Pr} \left[ \left| \frac{1}{|S|} \sum_{t\in S} |x_t|^2- n^{-1}\| x\|_2^2 ]  \right|  \geq \epsilon n^{-1} \| x\|_2^2 \right] & \leq & \exp(-\Omega (\epsilon^2 |S|/ R )).
\end{eqnarray*}
\end{lemma}
\begin{proof}
Let $M$ denote  $\underset{t\in [n]}{\max} |x_t|^2$. Replacing $X_i$ by $\frac{ |x_t|^2}{ M}$ and $n$ by $|S|$ in Lemma \ref{lem:chernoff_bound}, we obtain that
\begin{align*}
\mathsf{Pr}[ | X - \mu | > \epsilon \mu ] \leq 2 \exp(-\frac{\epsilon^2}{3} \mu) 
\end{align*} 
The above equation implies that
\begin{align*}
\mathsf{Pr} \left[ \left| \sum_{t\in S} \frac{ |x_t|^2 }{ M } -  |S|\frac{\| x \|_2^2}{n M}    \right| > \epsilon  |S|\frac{\| x \|_2^2}{n M}    \right] \leq 2\exp(-\frac{\epsilon^2}{3}\mu) 
\end{align*}
Multiplying the normalization factor on both sides,
\begin{align*}
\mathsf{Pr} \left[ \left| \frac{1}{|S|} \sum_{t\in S} |x_t|^2- n^{-1}\| x \|_2^2  \right|  \geq \epsilon n^{-1} \|x\|_2^2 \right] \leq 2\exp(-\frac{\epsilon^2}{3}\mu) 
\end{align*}
Applying bound on $\mu$
\begin{align*}
 \mathsf{Pr} \left[ \left| \frac{1}{|S|} \sum_{t\in S} |x_t|^2- n^{-1} \| x \|_2^2  \right|  \geq \epsilon n^{-1} \| x \|_2^2  \right] \leq 2\exp(-\frac{\epsilon^2}{3} |S|\frac{ \|x\|_2^2}{ n M}  )  
\end{align*}
which is less than $2\exp(-\frac{\eps^2}{3} |S| / R)$, thus completes the proof.
\end{proof}

%% file: discretization.tex
\section{Oblivious Sketching Fourier Sparse  Signals}\label{sec:offline_sketch}
In this section, we show an intermediate step in the reduction from Frequency estimation to Signal estimation: constructing a small sketching subset $S$ of the time domain \emph{obliviously} (without making any query to the signal), so that any signal in the family ${\cal F}$ discretized by $S$ has norm close to the original continuous signal.
More formally, we define the \emph{oblivious sketching Fourier signal  problem} as follows:

\begin{problem}[Oblivious sketching  Fourier sparse signal  problem]\label{prob:perserve_norm_high_D}
Let ${\cal F}$ be a family of Fourier-sparse signals.
Let $\epsilon \in (0,0.1)$ denote the accuracy parameter. The goal is to find a set $S=\{t_1,\dots,t_s\}$ in the time domain of size $s$ %
such that, with high probability, for any signal $x \in {\cal F}$, it holds that
\begin{align*}
(1-\epsilon)\|x\|_{T}\leq \|x\|_S \leq (1+\epsilon)\|x\|_T, 
\end{align*}
where 
\begin{align*}
    \| x \|_T^2 := \frac{1}{T^d} \int_{ [0,T]^d } | x(t) |^2 \mathrm{d} t, ~~\mathrm{and}~~ 
    \|x\|^2_{S}:=\frac{1}{|S|}\sum_{i\in [s]}|x(t_i)|^2.
\end{align*}
\end{problem}
We remark that the concentration inequalities in Section~\ref{sec:energy_bound_to_concentration} do not give a small sketching directly. To apply concentration inequalities for all signals in ${\cal F}$, we need an $\epsilon$-net, which will increase the sketching size by a factor of $\poly(k)$. 

In Section~\ref{sec:oblivious_sketching_unified}, we show a unified approach to obtain oblivious sketching using uniform samples for general signals. Then, in Section~\ref{sec:preserving_norm_by_cp19_energy_bound}, we show how to sketch one-dimensional signals with nearly-optimal size using importance sampling. %

\subsection{A unified approach via uniform sampling}\label{sec:oblivious_sketching_unified}
In this sub-section, we prove a general result of oblivious sketching for discrete and continuous Fourier-sparse signals. The technical tools we employ are Fourier energy bounds and importance sampling in Section~\ref{sec:prelim_is}.

\begin{lemma}[Oblivious sketching via uniform sampling]\label{lem:oblivious_sketching_unified}
Let ${\cal F}\subset \{x(t):{\cal G}\rightarrow \C\}$ be a family of signals of dimension $k$. Suppose the energy bound $R\geq \sup_{t\in {\cal G}}\sup_{x\in {\cal F}}\frac{|x(t)|}{\|x(t)\|_T}$ holds.  For any $\epsilon,\rho\in (0, 1)$, let $S=\{t_1,t_2,\dots,t_s\}$ be a set of i.i.d. samples chosen uniformly at random over the time domain ${\cal G}$ of size 
\begin{align*}
    s \ge \Omega\left(\epsilon^{-2}R\log(k/\rho)\right).
\end{align*}
Then, with probability at least $1-\rho$, it holds that for all $x \in {\cal F}$,
\begin{align*}
    (1-\epsilon)\|x\|_{T}\leq \|x\|_S \leq (1+\epsilon)\|x\|_T.
\end{align*}
\end{lemma}
\begin{proof}
To show that sampling from the distribution $D=\mathrm{Uniform}({\cal G})$ gives a good sketch, we apply Lemma \ref{lem:general_distribution} with $ D'=D$, $ d=k$, $ \delta=\rho$, $w_i = 1/s$, which implies that the matrix $A\in\C^{s\times k}$ defined by $A_{i,j}:=\sqrt{w_i} \cdot v_j(t_i)$ satisfies: 
\begin{align}\label{eq:A_norm_sec_6}
\|A^* A- I\|_2 \le \eps
\end{align}
with probability at least $ 1-\rho$, as long as $s= \Omega(\eps^{-2} K_{\mathsf{IS},D} \log (k/{\rho}))$. 

Under the condition that Eq.~\eqref{eq:A_norm_sec_6} holds, Lemma \ref{lem:operator_estimation} shows that
\begin{align*}
(1-\epsilon)\|x\|_T^2\leq \| {x}\|^2_{S}\leq (1+\epsilon)\|x\|_T^2~~~\forall x\in {\cal F}, 
\end{align*}
which proves the oblivious sketching guarantee.

It remains to upper-bound the condition number $K_{\mathsf{IS},D}$. By  Definition~\ref{def:importance_sampling}, we have
\begin{align*}
K_{\mathsf{IS},D} =&~ \underset{t\in {\cal G}}{\sup}~   \underset{x \in \mathcal{F}}{\sup}~  \frac{|x(t)|^2}{\|x\|_T^2}   \leq  R.
\end{align*}
Thus, we get that
\begin{align*}
    s= \Omega\left(\epsilon^{-2}R\log(k/\rho)\right).
\end{align*}
The lemma is then proved.    
\end{proof}

Using the $d$-dimensional Fourier-sparse signals' energy bound (Theorem~\ref{thm:max_is_bounded_by_energy_high_D}) and the discrete Fourier-sparse signals' energy bound (Theorem~\ref{thm:discrete_energy_bound}), we immediately obtain the following corollaries:
\begin{corollary}[Oblivious sketching high-dimensional continuous signal]\label{lem:concentration_for_any_polynomial_signal_high_D}
Let $d>1$ be the dimension of the signal. Let $f_1,\dots,f_k\in \R^d$. For any $\epsilon\in(0,1)$, let $S_d$ be a set of i.i.d. samples chosen uniformly at random over $[0,T]^d$ of size $|S_d|\ge \epsilon^{-2} k^{O(d)}\log(k/\rho)$. Let $V:=\{\exp(2\pi\i \langle f_i,t\rangle)~|~i\in [k]\}$. Then, with probability at least $1-\rho$, for all $x \in \mathrm{span}\{V\}$, we have
\begin{equation*}
(1-\epsilon) \|x\|_T\leq \| {x}\|_{S_d} \leq (1 + \epsilon)\| {x}\|_T.
\end{equation*}     
\end{corollary}

\begin{corollary}[Oblivious sketching discrete signal]\label{cla:concentration_for_any_polynomial_signal_discrete}
For $d\geq 1$, let $n=p^d$ for some positive integer $p$. Let $k\in\mathbb{N}_+$ and $f_1,\dots, f_k\in [p]^d$. Define $V:=\left\{  ( e^{2 \pi \i \langle f_i, t  \rangle / p})_{t\in[p]^d}  ~ {|}~ \forall i \in [k] \right\}\subseteq \C^{[p]^d}$.
For any $\epsilon,\rho\in (0, 1)$, let $S$ be a set of i.i.d. samples chosen uniformly at random over $[n]$ of size $|S| \ge O(\epsilon^{-2}k\log(k/\rho))$. Then, with probability at least $1-\rho$, for all $u \in \mathrm{span}\{V\}$, we have
\begin{equation*}
(1-\epsilon)\|u\|^2_2\leq n\| {u}\|^2_S \leq (1 + \epsilon)\| {u}\|^2_2,
\end{equation*} 
where $\|u\|_S^2=\sum_{i\in S} |u_i|^2 / |S|$.    
\end{corollary}

\subsection{Weighted oblivious sketching one-dimensional signals}\label{sec:preserving_norm_by_cp19_energy_bound}
For one-dimensional signals, using uniform sampling (Lemma~\ref{lem:oblivious_sketching_unified}) and the energy bound (Theorem~\ref{thm:worst_case_sFFT_improve}) gives a sketching set of size at least $\epsilon^{-2}k^2$. Motivated by \cite{cp19_colt}, we show a more efficient sketching method using only $\epsilon^{-2}k$ samples by assigning different weights to each sample point. 
In the following lemma, we give a sketch for any one-dimensional Fourier sparse signal with nearly-optimal  size:

\begin{lemma}[Nearly-optimal weighted sketch for one-dimensional signals]\label{lem:concentration_for_any_polynomial_signal_improve_by_CP19}
For $k\in\mathbb{N}_+$, define a probability distribution ${\cal D}(t)$ as follows:
\begin{align}\label{eq:def_D}
{\cal D}(t):=
\begin{cases}
{c}/(1-|t/T|) , & \text{ for } |t| \le T(1-{1}/k)\\
c \cdot k, & \text{ for } |t|\in [T(1-{1}/k), T]
\end{cases} 
\end{align} 
where $c=\Theta(T^{-1}\log^{-1}(k))$ is a normalization factor such that $\int_{-T}^T {\cal D}(t) \d t= 1$.

For any $f_1,\cdots, f_k \in [-F, F]$, let ${\cal F}:=\Big\{ x(t)=\sum_{j=1}^k v_j \cdot e^{2 \pi \i f_j t}  ~\Big|~  v_j \in \mathbb{C} \Big\}$. For any $\epsilon,\rho \in (0,1)$, let $S_{{\cal D}}=\{t_1, \cdots , t_{s}\}$ be a set of i.i.d. samples from ${\cal D}(t)$ of size $s \ge O(\eps^{-2}k\log(k)\log(k/\rho))$. Let the weight vector $w\in \R^s$ be defined by $w_i:=2/(Ts {\cal D}(t_i))$ for $i\in[s]$. Then with probability at least $1-\rho$, we have
\begin{equation*}
(1-\epsilon)\| x\|_T \leq \| x\|_{S_{\cal D},w} \leq (1 + \epsilon)\| x\|_T,
\end{equation*} 
where $\|x\|_T^2:=\frac{1}{2T}\int_{-T}^T |x(t)|^2 \d t$.\footnote{We use time duration $[-T,T]$ for convenience. It is easy to transform to $[0,T]$ by shifting and re-scaling.}
\end{lemma}
\begin{proof}
Let $\{v_1(t),v_2(t),\cdots,v_k(t)\}$ be an orthonormal basis for ${\cal F}$ with respect to the distribution ${\cal D}$, i.e.,
\begin{align*}
    \int_0^T {\cal D}(t)\cdot v_i(t)\overline{v_j(t)} \d t=~{\bf 1}_{i=j}, \quad \forall i, j\in[k].
\end{align*}

We first prove that the distribution ${\cal D}$ is well-defined. By the condition that $\int_{-T}^T {\cal D}(t) \d t= 1$, we have 
\begin{align*}
2\int_0^{T(1-{1}/({k}))}\frac{c}{(1-|t/T|)}\d t + 2\int_{T(1-{1}/({k}))}^T  c \cdot k^2 k \d t= 1,
\end{align*}
which implies that
\begin{align*}
c^{-1} = &~2\int_0^{T(1-{1}/k)}\frac{1}{(1-|t/T|)} \d t+ 2\int_{T(1-{1}/k)}^T  k^2 \d t\\
= &~ 2T {\log k} +  2T\\
= &~ \Theta(T \log(k)).
\end{align*}
Thus, we get that $c=\Theta(T^{-1}\log^{-1}(k))$.

To show that sampling from distribution ${\cal D}$ gives a good weighted sketch, we will use some technical tools in Section~\ref{sec:prelim_is}. 
Applying Lemma \ref{lem:general_distribution} with $ D'={\cal D}$, $D=\mathrm{Uniform}([-T,T])$, $ d=k$, $ \delta=\rho$, we have that, with probability at least $ 1-\rho$, the matrix $A\in\C^{s\times k}$ defined by $A_{i,j}:=\sqrt{w_i} \cdot v_j(t_i)$ satisfying 
\begin{align*}
\|A^* A- I\|_2 \le \eps,
\end{align*}
as long as $s \ge \Omega(\eps^{-2} K_{\mathsf{IS},D'} \log( k/{\rho}))$, where $w_i=\frac{D(t_i)}{s \cdot D'(t_i)}=\frac{2}{Ts{\cal D}(t_i)}$. Then, by Lemma \ref{lem:operator_estimation}, it implies that for every $ x \in {\mathcal F}$,
\begin{align*}
(1-\epsilon)\|x\|_T^2\leq \| {x}\|^2_{S_{\cal D},w}\leq (1+\epsilon)\|x\|_T^2. 
\end{align*}

It remains to upper-bound the condition number $K_{\mathsf{IS},D'}$ (see Definition~\ref{def:importance_sampling}):
\begin{align*}
K_{\mathsf{IS},D'} :=&~ \underset{t}{\sup}  \{ \frac{D(t)}{D'(t)} \cdot \underset{f \in \mathcal{F}}{\sup}  \{ \frac{|f(t)|^2}{\|f\|_D^2}  \}  \}\\
=&~\underset{t}{\sup}  \{ \frac{1}{2T {\cal D}(t)} \cdot \underset{f \in \mathcal{F}}{\sup}  \{ \frac{|f(t)|^2}{\|f\|_T^2}  \}  \} \\
\leq &~ \underset{t}{\sup}  \{ \frac{1}{2T {\cal D}(t)} \cdot {\min}\{ \frac{k }{1-|t/T|} , k^2\}  \}\\
\leq &~ {\max}  \{\frac{(1-|t/T|)}{2 c T } \frac{k }{1-|t/T|}, \frac{1}{2c T k} k^2     \}\\
= &~ \frac{k}{2 c T }\\
= &~ O(k\log k),
\end{align*}
where the first step follows from the definition, the second step follows from $D(t)=\mathrm{Uniform}([-T,T])(t)=\frac{1}{2T}$, the third step follows from Theorem \ref{thm:worst_case_sFFT_improve} and Theorem \ref{thm:bound_k_sparse_FT_x_middle_improve}, and the remaining steps follow from direct calculations. Thus, we get that
\begin{align*}
    s\geq \Omega\left(\epsilon^{-2}k\log(k)\log(k/\rho)\right).
\end{align*}

The lemma is then proved.
\end{proof}

%% file: well_balance_sampling.tex
\section{Fast Implementation of Well-Balanced Sampling Procedure}\label{sec:fast_wbsp}

Well-balanced sampling procedure was first defined in \cite{cp19_colt} to study the active linear regression problem. Our signal estimation algorithm will call it as a sub-procedure. In this section, we give a fast implementation of well-balanced sampling procedure based on the Randomized BSS algorithm \cite{bss12,ls15}.

First, we restate the definition of well-balanced sampling procedure in \cite{cp19_colt}. 

\begin{definition}[Well-balanced sampling procedure (WBSP), \cite{cp19_colt}]
\label{def:procedure_agnostic_learning}
Given a linear family $\mathcal{F}$ and underlying distribution $D$, let $P$ be a random sampling procedure that terminates in $m$ iterations ($m$ is not necessarily fixed) and provides a coefficient $\alpha_i$ and a distribution $D_i$ to sample $x_i \sim D_i$ in every iteration $i \in [m]$.

We say $P$ is an $\eps$-WBSP if it satisfies the following two properties:
\begin{enumerate}
\item With probability $0.9$, for weight $w_i=\alpha_i \cdot \frac{D(x_i)}{D_i(x_i)}$ of each $i \in [m]$, $$
\sum_{i=1}^m w_i \cdot |h(x_i)|^2 \in \left[1-10\sqrt{\eps}, 1+10\sqrt{\eps} \right] \cdot \|h\|_D^2 \quad \forall h \in \mathcal{F}.
$$

\item The coefficients always have $\sum_{i=1}^m \alpha_i \le \frac{5}{4}$ and $\alpha_i \cdot K_{\mathsf{IS},D_i} \le \frac{\eps}{2}$ for all $i \in [m]$.  
\end{enumerate}
\end{definition}
This definition describes a general sampling procedure that uses a few samples to represent the whole continuous signal, and the sampling procedure should satisfy two properties: one guarantees that the norm of any function in a function family is preserved, and another guarantees that the norm of noise is also preserved. %

In Section~\ref{sec:rbss_wbsp}, we review some results in \cite{cp19_colt} and show that WBSP can be implemented via randomized spectral sparsification. In Section~\ref{sec:fast_wbsp_sub}, we design a data structure and improve the time efficiency of the WBSP. In Section~\ref{sec:trade_pre_query}, we discover a tradeoff between the preprocessing cost and the query cost, which can improve the space complexity.

\subsection{Randomized BSS implies a WBSP}\label{sec:rbss_wbsp}
In this section, we review the result of \cite{cp19_colt}, which shows that the Randomized BSS algorithm \cite{bss12,ls15} implies a well-balanced sampling procedure.

\begin{lemma}[Lemma 5.1 in \cite{cp19_colt}]\label{lem:BSS}
Let $G$ be any domain.
  Given any dimension $d$ linear function family  $\mathcal{F}$ of function $f:G\rightarrow\C$,
  \begin{align*}
{\mathcal F} = \{f(t)=\sum_{j=1}^d v_j u_j(t) | v_j\in\C \},
\end{align*}
where $u_j:G\rightarrow\C$.
Given any distribution $D$ over $G$, and any $\eps>0$,  there exists an efficient procedure (Algorithm \ref{alg:BSS}) that runs in $O(\eps^{-1} d^3|G|+\eps^{-1}d^{\omega+1})$ time and outputs a set $S\subseteq G$ and weight $w$ such that
  \begin{itemize}
  \item $|S|=O(d/\eps),~w\in \R^{|S|}$,
  \item the procedure is an $\eps$-WBSP,
  \end{itemize}
  holds with probability $1-\frac{1}{200}$.
\end{lemma}
\begin{algorithm}[!ht]
\caption{A well-balanced sampling procedure based on Randomized BSS (see \cite{cp19_colt}) %
}\label{alg:BSS}
\begin{algorithmic}[1]
\Procedure{\textsc{RandBSS}}{$d, \mathcal{F},D,\epsilon$}
\State Find an orthonormal basis $v_1,\ldots,v_d$ of $\mathcal{F}$ under $D$ 
\State Set $\gamma\leftarrow\sqrt{\epsilon}/3$ and $\textsf{mid}\leftarrow \frac{4d/\gamma}{1/(1-\gamma)-1/(1+\gamma)}$ 
\State $j \leftarrow  0,  B_0\leftarrow 0$ 
\State $l_0\leftarrow -2d/\gamma, u_0\leftarrow 2d/\gamma$ 
\While {$u_{j+1}-l_{j+1}<8 d/\gamma$} 
\State $\Phi_j \leftarrow  \tr [ (u_j I - B_j)^{-1} ] + \tr [ (B_j - l_j I)^{-1} ] $  \hfill $\triangleright$ The potential function at iteration $j$.
\State Set the coefficient $\alpha_j\leftarrow \frac{\gamma}{\Phi_j} \cdot \frac{1}{\textsf{mid}}$ 
\State Set $v(x)\leftarrow \big(v_1(x),\ldots,v_d(x) \big)$ 
\For{$x \in \supp(D)$} %
\State \label{step:set_dist_bss} Set the distribution $$D_j(x)\leftarrow D(x) \cdot \bigg(v(x)^\top (u_j I - B_j)^{-1} v(x) + v(x)^\top (B_j - l_j I)^{-1} v(x) \bigg)/\Phi_j $$
\EndFor
\State Sample $x_j \sim D_j$ and set a scale $s_j\leftarrow \frac{\gamma}{\Phi_j} \cdot \frac{D(x_j)}{D_j(x_j)}$ 
\State $B_{j+1}\leftarrow B_j + s_j \cdot v(x_j) v(x_j)^\top$ 
\State $u_{j+1}\leftarrow u_j + \frac{\gamma}{\Phi_j (1-\gamma)},  \quad l_{j+1}\leftarrow l_j + \frac{\gamma}{\Phi_j ( 1 +\gamma)}$ 
\State $j\leftarrow j+1$ 
\EndWhile
\State $m\leftarrow j$ 
\State Assign the weight $w_j\leftarrow s_j/\textsf{mid}$ for each $x_j$ 
\State \Return $\{x_1,x_2,\cdots, x_m\},w$ 
\EndProcedure
\end{algorithmic}
\end{algorithm}

\subsection{Fast implementation of WBSP}\label{sec:fast_wbsp_sub}

In this section, we give a fast implementation of Algorithm~\ref{alg:BSS}:

\begin{theorem}[Fast implementation of WBSP]\label{thm:fast_WBSP}
Let $G$ be any domain.
Given any dimension $d$ linear function family  $\mathcal{F}$ of function $f:G\rightarrow\C$,
  \begin{align*}
{\mathcal F} = \{f(t)=\sum_{j=1}^d v_j u_j(t) | v_j\in\C \},
\end{align*}
where $u_j:G\rightarrow\C$.
Given any distribution $D$ over $G$, and any $\eps>0$,  there exists an efficient procedure (Algorithm \ref{alg:BSS_faster}) that runs in $O( d^2 |G| + \eps^{-1}d^3 \log |G| + \eps^{-1} d^{\omega+1})$ time and outputs a set $S\subseteq G$ and weight $w\in \R^{|S|}$ such that the following properties hold with probability at least 0.995:
  \begin{itemize}
  \item $|S|=O(d/\eps)$,
  \item the procedure is an $\eps$-WBSP.
  \end{itemize}
\end{theorem}

\begin{algorithm}[ht]
\caption{Our fast implementation of well-balanced sampling procedure}\label{alg:BSS_faster}
\begin{algorithmic}[1]
\Procedure{\textsc{RandBSS+}}{$d, \mathcal{F},D,\epsilon$}\Comment{Theorem~\ref{thm:fast_WBSP}}
\State {\color{blue}/*Preprocessing*/}
\State Find an orthonormal basis $v_1,\ldots,v_d$ of $\mathcal{F}$ under $D$ 
\State $\gamma\leftarrow\sqrt{\epsilon}/3$ and $\textsf{mid}\leftarrow \frac{4d/\gamma}{1/(1-\gamma)-1/(1+\gamma)}$ 
\State $j \leftarrow  0,  B_0\leftarrow 0$ 
\State $l_0\leftarrow -2d/\gamma, u_0\leftarrow 2d/\gamma$ 
\State $\delta\gets 1/\poly(d)$
\State \Comment{ Let $v(x) =  \big(v_1(x),\ldots,v_d(x) \big) \in \R^d$ }
\State \label{step:init} $\mathrm{DS}.\textsc{Init}( |D|, d , \{ v(x_1), \cdots, v(x_{|D|}) \} \subset \R^{d}, \{D(x_1),\dots,D(x_{|D|})\}\subset \R)$ \Comment{Algorithm~\ref{algo:ip_ds}} 
\State {\color{blue}/*Iterative step*/}
\While {$u_{j+1}-l_{j+1}<8 d/\gamma$}
\State $\Phi_j \leftarrow  \tr [ (u_j I - B_j)^{-1} ] + \tr [ (B_j - l_j I)^{-1} ] $  \hfill $\triangleright$ The potential function at iteration $j$.
\State  $\alpha_j\leftarrow \frac{\gamma}{\Phi_j} \cdot \frac{1}{\textsf{mid}}$ 
\State \label{step:inverse} $E_j \gets (u_j I - B_j)^{-1} + (B_j - l_j I)^{-1}$
\State \label{step:query} $q \gets \mathrm{DS}.\textsc{Query}( E_j / \Phi_j )$ \Comment{$q \in [|D|]$, Algorithm~\ref{algo:ip_ds}}
\State ${\sf x}_j\gets x_q$ and set a scale $s_j\leftarrow \frac{\gamma}{v({\sf x}_j)^\top E_j v({\sf x}_j)}$ 
\State $B_{j+1}\leftarrow B_j + s_j \cdot v({\sf x}_j) v({\sf x}_j)^\top$ 
\State $u_{j+1}\leftarrow u_j + \frac{\gamma}{\Phi_j (1-\gamma)},  \quad l_{j+1}\leftarrow l_j + \frac{\gamma}{\Phi_j ( 1 +\gamma)}$ 
\State $j\leftarrow j+1$ 
\EndWhile
\State $m\leftarrow j$ 
\State Assign the weight $w_j\leftarrow s_j/\textsf{mid}$ for each $x_j$ 
\State \Return $\{{\sf x}_1,{\sf x}_2,\cdots, {\sf x}_m\},w$ 
\EndProcedure
\end{algorithmic}
\end{algorithm}

Our algorithm is based on a data structure for solving the \emph{online quadratic-form sampling problem} defined as follows: 

\begin{problem}[Online Quadratic-Form Sampling Problem]\label{prob:oqfs}
Given $n$ vectors $v_1,\dots,v_n\in \R^d$ and $n$ coefficients $\alpha_1,\dots,\alpha_n$, for any PSD matrix $A\in \R^{d\times d}$, output a sample $i\in [n]$ from the following distribution ${\cal D}_{A}$:
\begin{align}\label{eq:def_D_A}
    \Pr_{{\cal D}_A}[i]:=\frac{\alpha_i \cdot v_i^\top Av_i}{\sum_{j=1}^n \alpha_j \cdot v_j^\top Av_j}~~~\forall i\in [n].
\end{align}
\end{problem}

\begin{theorem}\label{thm:main_preserve_distance}
There is a data structure (Algorithm~\ref{algo:ip_ds}) that uses $O(nd^2)$ spaces for the Online Quadratic-Form Sampling Problem with the following procedures:
\begin{itemize}
    \item \textsc{Init}$( n,d, \{v_1, \dots, v_n\}\subset \R^d, \{\alpha_1,\dots,\alpha_n\}\subset \R)$: the data structure preprocesses in time $O(nd^2)$.
    
    \item \textsc{Query}$(A \in \R^{d\times d})$: Given a PSD matrix $A$, the \textsc{Query} operation samples $i\in [n]$ exactly from the probability distribution ${\cal D}_A$ defined in Problem~\ref{prob:oqfs} in $O(d^2\log n)$-time.
\end{itemize}
\end{theorem}
\begin{proof}
The pseudo-code of the algorithm is given as Algorithm~\ref{algo:ip_ds}. The idea is to build a binary tree such that each node has an interval in $[l,\dots, r]\subset [1,\dots, n]$ and stores a matrix $\sum_{i=l}^r \alpha_i \cdot v_iv_i^\top$. For each internal node with interval $[l,\dots, r]$, its left child node has interval $[l,\dots,\lfloor (l+r)/2\rfloor]$, and its right child node has interval $[\lfloor (l+r)/2\rfloor+1,\dots, r]$.

We first prove the correctness. Suppose the output of \textsc{Query} is $i\in [n]$. We compute its probability. Let $u_0=\mathsf{root}, u_1,\dots,u_t$ be the path from the root of the tree to the leaf with $\mathrm{id}=i$. Then, we have
\begin{align*}
    \Pr[u_t]=\prod_{j=1}^t \Pr[u_j|u_{j-1}]=\prod_{j=1}^t \frac{\sum_{k=l_j}^{r_j} \alpha_k \cdot v_k^\top A v_k}{\sum_{k=l_{j-1}}^{r_{j-1}} \alpha_k \cdot v_k^\top A v_k}=\frac{\alpha_i \cdot v_i^\top A v_i}{\sum_{k=1}^{n} \alpha_k \cdot v_k^\top A v_k},
\end{align*}
where $[l_j,\dots,r_j]$ is the range of the node $u_j$, the first step follows from the conditional probability, the second step follows from Line~\ref{ln:sample_prob} in Algorithm~\ref{algo:ip_ds}, and the last step follows from the telescoping products. Hence, we get that
\begin{align*}
    \Pr[\textsc{Query}(A)=i]=\Pr_{{\cal D}_A}[i]~~~\forall i\in [n].
\end{align*}
Hence, the sampling distribution is the same as the Online Quadratic-Form Sampling Problem's distribution.

For the running time, in the preprocessing stage, we build the binary tree recursively. It is easy to see that the number of nodes in the tree is $O(n)$ and the depth is $O(\log n)$. For a leaf node, we take $O(d^2)$-time to compute the matrix $\alpha_i\cdot v_iv_i^\top\in \R^{d\times d}$. For an internal node, we take $O(d^2)$-time to add up the matrices of its left and right children. Thus, the total preprocessing time is $O(nd^2)$. 

In the query stage, we walk along a path from the root to a leaf, which has $O(\log n)$ steps. In each step, we compute the inner product between $A$ and the current node's matrix, which takes $O(d^2)$-time. And we compute the inner product between $A$ and its left child node's matrix, which also takes $O(d^2)$-time. Then, we toss a coin and decide which subtree to move. Hence, each query takes $O(d^2\log n)$-time.

The theorem is then proved.
\end{proof}

\begin{algorithm}[!ht]\caption{Quadratic-form sampling data structure}\label{algo:ip_ds}
\begin{algorithmic}[1]
\State {\bf structure} Node
\State \quad $V\in \R^{d\times d}$
\State \quad $\mathrm{left}, \mathrm{right}$ \Comment{Point to the left/right child in the tree}
\State {\bf end structure}
\State {\bf data structure} \textsc{DS}
\State {\bf members}
\State \quad $n\in \mathbb{N}$\Comment{The number of vectors}
\State \quad $v_1,\dots,v_n\in \R^d$ \Comment{$d$-dimensional vectors}
\State \quad $\alpha_1,\dots,\alpha_n\in \R$ \Comment{Coefficients}
\State \quad $\mathsf{root}$: Node \Comment{The root of the tree}
\State {\bf end members}
\Procedure{BuildTree}{$l, r$}\Comment{$[l,\dots,r]$ is the range of the current node}
\State ${\sf p}\gets \mathbf{new}$ Node
\If{$l=r$}\Comment{Leaf node}
    \State ${\sf p}.V \gets \alpha_l \cdot v_l v_l^\top$ \Comment{It takes $O(d^2)$-time}
\Else\Comment{Internal node}
    \State $mid\gets \lfloor (l+r)/2\rfloor$
    \State ${\sf p}.\mathrm{left}\gets \textsc{BuildTree}(l, mid)$
    \State ${\sf p}.\mathrm{right}\gets \textsc{BuildTree}(mid+1, r)$
    \State ${\sf p}.V\gets ({\sf p}.\mathrm{left}).V + ({\sf p}.\mathrm{right}).V$ \Comment{It takes $O(d^2)$-time}
\EndIf
\State \Return ${\sf p}$
\EndProcedure
\Procedure{Init}{$n, d, \{v_i\}_{i\in [n]}\subseteq \R^d, \{\alpha_i\}_{i\in [n]}\subseteq \R$}
    \State $v_i\gets v_i$, $\alpha_i\gets \alpha_i$ for $i\in [n]$
    \State $\mathsf{root}\gets \textsc{BuildTree}(1, n)$
\EndProcedure
\Procedure{Query}{$A\in \R^{d\times d}$}
\State ${\sf p}\gets {\sf root}, ~l\gets 1, ~r\gets n$
\State $s\gets 0$
\While{$l\ne r$} \Comment{There are $O(\log n)$ iterations}
    \State $w\gets \langle {\sf p}.V, A\rangle$\Comment{It takes $O(d^2)$-time}
    \State $w_\ell \gets \langle ({\sf p}.\mathrm{left}).V, A\rangle$
    \State Sample $c$ from $\mathrm{Bernoulli}(w_\ell / w)$\label{ln:sample_prob}
    \If{$c=0$}
        \State ${\sf p}\gets {\sf p}.\mathrm{left}$, $r\gets \lfloor (l+r)/2\rfloor$
    \Else
        \State ${\sf p}\gets {\sf p}.\mathrm{right}$, $l\gets \lfloor (l+r)/2\rfloor+1$
    \EndIf
\EndWhile
\State \Return $l$
\EndProcedure
\State {\bf end data structure}
\end{algorithmic}
\end{algorithm}

\begin{lemma}[Running time of Procedure \textsc{RandBSS+} in Algorithm \ref{alg:BSS_faster}]\label{lem:time_bss_faster}
Algorithm \ref{alg:BSS_faster} runs in
\begin{itemize}
    \item $O( |D| d^2  )$-time for preprocessing,
    \item $O( d^2 \log (|D|) + d^{\omega})$-time per iteration, and
    \item $O(\eps^{-1} d)$ iterations. %
\end{itemize}
Thus, the total running time is,
\begin{align*} 
O(|D| d^2  + \eps^{-1}d \cdot ( d^2 \log |D| +  d^{\omega} ) ).  
\end{align*}
\end{lemma}
\begin{proof}
In each call of the Procedure \textsc{RandBSS+} in Algorithm \ref{alg:BSS_faster},  

\begin{itemize}
    \item Finding orthonormal basis takes $O(|D| d^2)$.
    \item In the line \ref{step:init}, it runs $O(|D|d^2)$ times.
    \item The while loop repeat $O(\eps^{-1}d)$ times.
    \begin{itemize}
        \item Line \ref{step:inverse} is computing $(u_jI-B_j)\in \C^{d \times d}$, $(u_jI-B_j)^{-1}$. This part takes $O(d^\omega)$ time\footnote{Note that this step seems to be very difficult to speed up via the Sherman-Morrison formula since $u_j$ changes in each iteration and the update is of high rank.}. 
        \item Note that line \ref{step:query} of Procedure \textsc{RandBSS+} in Algorithm \ref{alg:BSS_faster} runs $ O(d^2 \log |D|)$ times.
    \end{itemize}
\end{itemize}

So, the time complexity of Procedure \textsc{RandBSS+} in Algorithm \ref{alg:BSS_faster} is 
\begin{align*} 
O(|D| d^2 + \eps^{-1}d \cdot ( d^2 \log |D| +  d^{\omega} ) )  .
\end{align*}
\end{proof}

\begin{lemma}[Correctness of Procedure \textsc{RandBSS+} in Algorithm \ref{alg:BSS_faster}]\label{lem:randbss+_correct}
Given any dimension $d$ linear space ${\cal F}$, any distribution $D$ over the domain of ${\cal F}$, and any $\epsilon>0$, $\textsc{RandBSS+}(d, \mathcal{F},D,\epsilon)$ is an $\epsilon$-WBSP that terminates in $O(d/\epsilon)$ rounds with probability $1-1/200$.
\end{lemma}
\begin{proof}
We first claim that, for each $j\in [m]$, ${\sf x}_j$ has the same distribution as $D_j$, where
\begin{align*}
    D_j(x)=D(x) \cdot ( v(x)^\top E_j v(x) )/\Phi_j~~~\forall x\in D
\end{align*}
Notice that sampling from distribution $D_j$ can be reformulated as an Online Quadratic-Form Sampling Problem: the vectors are $\{v(x)\}_{x\in D}$ , the coefficients are $\{D(x)\}_{x\in D}$, and the query matrix is $E'_j:=E_j/\Phi_j$. Then, we have $D_j={\cal D}_{E'_j}$ defined in Problem~\ref{prob:oqfs}. 
Hence, by Theorem~\ref{thm:main_preserve_distance}, we can use the data structure (Algorithm~\ref{algo:ip_ds}) to efficiently sample from $D_j$.

Therefore, the sample ${\sf x}_j$ in each iteration is generated from the same distribution as the original randomized BSS algorithm (Algorithm~\ref{alg:BSS}). Then, the WBSP guarantee and the number of iterations immediately follow from the proof of \cite[Lemma 5.1]{cp19_colt}. 

The proof of the lemma is then completed.
\end{proof}

\begin{proof}[Proof of Theorem~\ref{thm:fast_WBSP}]
The running time of the algorithm follows from Lemma~\ref{lem:time_bss_faster}, and the correctness follows from Lemma~\ref{lem:randbss+_correct}.
\end{proof}

\subsection{Trade-off between preprocessing and query}\label{sec:trade_pre_query}

In this section, we consider the preprocessing and query trade-off in the data structure for quadratic form sampling problem. In the following theorem, we give a new data structure that takes less time in preprocessing and more time for each query than Theorem~\ref{thm:main_preserve_distance}, and the space complexity is also reduced from $O(nd^2)$ to $O(nd)$.

\begin{theorem}\label{thm:main_preserve_distance_tradeoff}
There is a data structure (Algorithms~\ref{algo:ip_ds_tradeoff} and \ref{algo:ip_ds_tradeoff_query}) that uses $O(nd)$ spaces for the Online Quadratic-Form Sampling Problem with the following procedures:
\begin{itemize}
    \item \textsc{Init}$( n,d, \{v_1, \dots, v_n\}\subset \R^d, \{\alpha_1,\dots,\alpha_n\}\subset \R)$: the data structure preprocesses in time $O(nd^{\omega-1})$.
    
    \item \textsc{Query}$(A \in \R^{d\times d})$: Given a PSD matrix $A$, the \textsc{Query} operation samples $i\in [n]$ exactly from the probability distribution ${\cal D}_A$ defined in Problem~\ref{prob:oqfs} in $O(d^2\log(n/d)+d^\omega)$-time.
\end{itemize}
\end{theorem}
\begin{proof}
The time and space complexities follow from Lemma~\ref{lem:tradeoff_ts_cost}. And the correctness follows from Lemma~\ref{lem:tradeoff_correct}.
\end{proof}

\begin{algorithm}[ht]\caption{Quadratic-form sampling with preprocessing-query trade-off: Preprocessing}\label{algo:ip_ds_tradeoff}
\begin{algorithmic}[1]
\State {\bf structure} Node
\State \quad $V_1,V_2\in \R^{d\times d}$
\State \quad $\mathrm{left}, \mathrm{right}$ \Comment{Point to the left/right child in the tree}
\State {\bf end structure}
\State {\bf data structure} \textsc{DS+}\Comment{Theorem~\ref{thm:main_preserve_distance_tradeoff}}
\State {\bf members}
\State \quad $n\in \mathbb{N}$\Comment{The number of vectors}
\State \quad $m\in \mathbb{N}$\Comment{The number of blocks}
\State \quad $v_1,\dots,v_n\in \R^d$ \Comment{$d$-dimensional vectors}
\State \quad $\mathsf{root}$: Node \Comment{The root of the tree}
\State {\bf end members}
\Procedure{BuildTree}{$l, r$}\Comment{$[l,\dots,r]$ is the range of the current node}
\State ${\sf p}\gets \mathbf{new}$ Node
\If{$l=r$}\Comment{Leaf node}
    \State ${\sf p}.V_2\gets \begin{bmatrix}v_{(l-1)d+1}&\cdots& v_{ld}\end{bmatrix}$\label{ln:tradeoff_leaf_v2}
    \State ${\sf p}.V_1 \gets ({\sf p}.V_2)\cdot  ({\sf p}.V_2)^\top$ \Comment{It takes $O(d^\omega)$-time}\label{ln:tradeoff_leaf_v1}
    \State \Comment{${\sf p}.\mathrm{mat1}=\sum_{i=(l-1)d+1}^{ld} v_iv_i^\top$}
\Else\Comment{Internal node}
    \State $mid\gets \lfloor (l+r)/2\rfloor$
    \State ${\sf p}.\mathrm{left}\gets \textsc{BuildTree}(l, mid)$
    \State ${\sf p}.\mathrm{right}\gets \textsc{BuildTree}(mid+1, r)$
    \State ${\sf p}.V_1\gets ({\sf p}.\mathrm{left}).V_1 + ({\sf p}.\mathrm{right}).V_1$ \Comment{It takes $O(d^2)$-time}\label{ln:tradeoff_internal_node}
\EndIf
\State \Return ${\sf p}$
\EndProcedure
\Procedure{Init}{$n, d, \{v_i\}_{i\in [n]}\subseteq \R^d, \{\alpha_i\}_{i\in [n]}\subseteq \R$}
    \State $v_i\gets v_i\cdot \sqrt{\alpha_i}$ for $i\in [n]$
    \State $m\gets n/d$\Comment{We assume that $n$ is divisible by $d$}
    \State Group $\{v_i\}_{i\in [n]}$ into $m$ blocks $B_1,\dots, B_m$\Comment{$B_i=\{v_{(i-1)d+1},\dots, v_{id}\}$ for $i\in [m]$}
    \State $\mathsf{root}\gets \textsc{BuildTree}(1, m)$
\EndProcedure
\State {\bf end data structure}
\end{algorithmic}
\end{algorithm}

\begin{algorithm}[ht]\caption{Quadratic-form sampling with preprocessing-query trade-off: Query}\label{algo:ip_ds_tradeoff_query}
\begin{algorithmic}[1]
\State {\bf data structure} \textsc{DS+}\Comment{Theorem~\ref{thm:main_preserve_distance_tradeoff}}
\State {\bf members}
\State \quad $n\in \mathbb{N}$\Comment{The number of vectors}
\State \quad $m\in \mathbb{N}$\Comment{The number of blocks}
\State \quad $v_1,\dots,v_n\in \R^d$ \Comment{$d$-dimensional vectors}
\State \quad $\mathsf{root}$: Node \Comment{The root of the tree}
\State {\bf end members}

\Procedure{BlockSampling}{${\sf p}$, $l\in \mathbb{N}$, $A\in \R^{d\times d}$}\Comment{${\sf p}$ is a leaf node with index $l$}
\State $U\gets ({\sf p}.V_2)^\top \cdot A\cdot ({\sf p}.V_2)$\Comment{It takes $O(d^\omega)$-time}\label{ln:tradeoff_U}
\State Define a distribution ${\cal D}_l$ over $[d]$ such that ${\cal D}_l(i)\propto U_{i,i}$
\State Sample $i\in [d]$ from ${\cal D}_l$\Comment{It takes $O(d)$-time}\label{ln:tradeoff_sample}
\State \Return $(l-1)d+i$
\EndProcedure

\Procedure{Query}{$A\in \R^{d\times d}$}
\State ${\sf p}\gets {\sf root}, ~l\gets 1, ~r\gets m$
\State $s\gets 0$
\While{$l\ne r$} \Comment{There are $O(\log m)$ iterations}\label{ln:tradeoff_while_loop}
    \State $w\gets \langle {\sf p}.V_1, A\rangle$\Comment{It takes $O(d^2)$-time}
    \State $w_\ell \gets \langle ({\sf p}.\mathrm{left}).V_1, A\rangle$
    \State Sample $c$ from $\mathrm{Bernoulli}(w_\ell / w)$%
    \If{$c=0$}
        \State ${\sf p}\gets {\sf p}.\mathrm{left}$, $r\gets \lfloor (l+r)/2\rfloor$
    \Else
        \State ${\sf p}\gets {\sf p}.\mathrm{right}$, $l\gets \lfloor (l+r)/2\rfloor+1$
    \EndIf
\EndWhile
\State \Return \textsc{BlockSampling}(${\sf p}$, $l$, $A$)
\EndProcedure
\State {\bf end data structure}
\end{algorithmic}
\end{algorithm}

\begin{lemma}[Time and space complexities of Algorithms~\ref{algo:ip_ds_tradeoff} and \ref{algo:ip_ds_tradeoff_query}]\label{lem:tradeoff_ts_cost}
The \textsc{Init} procedure takes $O(nd^{\omega-1})$-time. The \textsc{Query} procedure takes $O(d^2\log(n/d)+d^\omega)$-time. The data structure uses $O(nd)$-space. 
\end{lemma}
\begin{proof}
We prove the space and time complexities of the data structure as follows:
\\
\textbf{Space complexity: } Let $m=n/d$. It is easy to see that there are $O(m)$ nodes in the data structure. And each node has two $d$-by-$d$ matrices. Hence, the total space used by the data structure is $O(n/d)\cdot O(d^2)=O(nd)$.

\paragraph{Time complexity: } In the preprocessing stage, the time-consuming step is the call of \textsc{BuildTree}. There are $O(m)$ internal nodes and $O(m)$ leaf nodes. Each internal node takes $O(d^2)$-time to construct the matrix $V_1$ (Line~\ref{ln:tradeoff_internal_node}). For each leaf node, it takes $O(d^2)$-time to form the matrix $V_2$ (Line~\ref{ln:tradeoff_leaf_v2}). And it takes $O(d^\omega)$-time to compute the matrix $V_1$ (Line~\ref{ln:tradeoff_leaf_v1}). Hence, the total running time of \textsc{BuildTree} is $O(md^\omega)=O(nd^{\omega-1})$.

In the query stage, the While loop in the \textsc{Query} procedure (Line~\ref{ln:tradeoff_while_loop}) is the same as in Algorithm~\ref{algo:ip_ds}. Since there are $O(m)$ nodes in the tree, it takes $O(d^2\log m)$-time. Then, in the \textsc{BlockSampling} procedure, it takes $O(d^\omega)$-time to compute the matrix $U$ (Line~\ref{ln:tradeoff_U}), and it takes $O(d)$-time to sample an index from the distribution ${\cal D}_l$ (Line~\ref{ln:tradeoff_sample}). Hence, the total running time for each query is $O(d^2\log m + d^\omega)=O(d^2\log(n/d)+d^\omega)$.

The proof of the lemma is then completed.
\end{proof}

\begin{lemma}[Correctness of Algorithm~\ref{algo:ip_ds_tradeoff_query}]\label{lem:tradeoff_correct}
The distribution of the output of the \textsc{Query}($A$) is ${\cal D}_A$ defined by Eq.~\eqref{eq:def_D_A}.
\end{lemma}
\begin{proof}
For simplicity, we assume that all the coefficients $\alpha_i=1$.

Let $u_0=\mathsf{root}, u_1,\dots,u_t$ be the path in the While loop (Line~\ref{ln:tradeoff_while_loop}) from the root of the tree to the leaf with index $l\in [m]$. By the construction of leaf node, we have
\begin{align*}
    V_1=V_2 V_2^\top = \begin{bmatrix}v_{(l-1)d+1}&\cdots & v_{ld}\end{bmatrix} \begin{bmatrix}v_{(l-1)d+1}^\top \\\vdots \\ v_{ld}^\top\end{bmatrix}=\sum_{i=(l-1)d+1}^{ld} v_iv_i^\top,
\end{align*}
which is the same as the $V$-matrix in Algorithm~\ref{algo:ip_ds}.
Hence, similar to the proof of Theorem~\ref{thm:main_preserve_distance}, we have
\begin{align*}
    \Pr[u_t]=\prod_{j=1}^t \Pr[u_j|u_{j-1}]=\frac{\sum_{i=(l-1)d+1}^{ld} v_i^\top Av_i}{\sum_{i=1}^n v_i^\top A v_i}.
\end{align*}
where $\{(l-1)d+1,\dots,ld\}$ is the range of the node $u_t$ and $\{1,\dots,n\}$ is the range of $u_0$.

Then, consider the \textsc{BlockSampling} procedure. Let $\{v_1,\dots,v_d\}$ be the vectors in the input block. At Line~\ref{ln:tradeoff_U}, we have
\begin{align*}
    U=V_2^\top A V_2 = \begin{bmatrix}v_1^\top \\\vdots\\v_d^\top\end{bmatrix} A \begin{bmatrix} v_1& \cdots & v_d\end{bmatrix}.
\end{align*}
For $i\in [d]$, the $i$-th element in the diagonal of $U$ is
\begin{align*}
    U_{i,i}=v_i^\top A v_i.
\end{align*}
Hence,
\begin{align*}
    \Pr[\textsc{BlockSampling}=i]=\frac{v_{i}^\top A v_{i}}{\sum_{j=1}^{d} v_j^\top A v_j}.
\end{align*}

Therefore, for any $k\in [n]$, if $k=(l-1)d+r$ for some $l,r\in \mathbb{N}$, then the sample probability is
\begin{align*}
    \Pr[\textsc{Query}(A)=k]=&~\Pr[\textsc{BlockSampling}=k~|~u_t=\text{Block}~l]\cdot \Pr[u_t=\text{Block}~l]\\
    = &~ \frac{v_{k}^\top A v_{k}}{\sum_{i=(l-1)d+1}^{ld} v_i^\top A v_i}\cdot\frac{\sum_{i=(l-1)d+1}^{ld} v_i^\top Av_i}{\sum_{i=1}^n v_i^\top A v_i}\\
    = &~ \frac{v_k^\top A v_k}{\sum_{i=1}^n v_i^\top A v_i}\\
    = &~ {\cal D}_A(k).
\end{align*}

The lemma is then proved.
\end{proof}

As a corollary, we get a WBSP using less space:

\begin{corollary}[Space efficient implementation of WBSP]
By plugging-in the new data structure (Algorithms~\ref{algo:ip_ds_tradeoff} and \ref{algo:ip_ds_tradeoff_query}) to \textsc{FasterRandSamplingBSS} (Algorithm~\ref{alg:BSS_faster}), we get an algorithm taking $O(|D| d^2  + \gamma^{-2}d \cdot ( d^2 \log |D| +  d^{\omega} ) )$-time and using $O(|D|d)$-space.
\end{corollary}
\begin{proof}
In the preprocessing stage of \textsc{FasterRandSamplingBSS}, we take $O(|D|d^2)$-time for Gram-Schmidt process and $O(|D|d^{\omega-1})$-time for initializing the data structure (Algorithm~\ref{algo:ip_ds_tradeoff}).

The number of iterations is $\gamma^{-2}d$. In each iteration, the matrix $E_j$ can be computed in $O(d^\omega)$-time. And querying the data structure takes $O(d^2\log(|D|/d)+d^\omega)$-time. 

Hence, the total running time is
\begin{align*}
    O\left(|D|d^2+|D|d^{\omega-1}+\gamma^{-2}d(d^2\log(|D|/d)+d^\omega)\right)=O\left(|D|d^2+\gamma^{-2}d^{\omega+1}+\gamma^{-2}d^2\log |D|\right).
\end{align*}

For the space complexity, the data structure uses $O(|D|d)$-space. The algorithm uses $O(d^2)$ extra space in preprocessing and each iteration. Hence, the total space complexity is $O(|D|d)$.
\end{proof}

%% file: distillation.tex
\section{Sketch Distillation for Fourier Sparse Signals}
\label{sec:sketch_distillation}
In Section~\ref{sec:offline_sketch}, we show an oblivious approach for sketching Fourier sparse signals. However, there are two issues of using this sketching method in Signal estimation: 1. The sketch size too large. 2. The noise in the observed signal could have much larger energy on the sketching set than its average energy.
To resolve these two issues, in this section, we propose a method called \emph{sketch distillation} to post-process the sketch obtained in Section~\ref{sec:offline_sketch} that can reduce the sketch size to $O(k)$ and prevent the energy of noise being amplified too much. However, we need some extra information about the signal $x^*(t)$: we assume that the frequencies of the noiseless signal $x(t)$ are known. But the sketch distillation process can still be done \emph{partially oblivious}, i.e., we do not need to access/sample the signal.

In Section~\ref{sec:distill_1d}, we show our distillation algorithms for one-dimensional signals. Then, we generalize the sketch distillation for high-dimensional signals in Section~\ref{sec:hd_important_sample_k} and for discrete  signals in Section~\ref{sec:distill_dis}.

\subsection{Sketch distillation for one-dimensional signals}\label{sec:distill_1d}

In this section, we show how to distill the sketch  produced by Lemma~\ref{lem:concentration_for_any_polynomial_signal_improve_by_CP19} from $O(k\log k)$-size to $O(k)$-size,  using an $\epsilon$-well-balanced sampling procedure developed in Section~\ref{sec:fast_wbsp}.

\begin{lemma}[Fast distillation for one-dimensional signal]\label{lem:important_sampling_k_faster}
Given $f_1, f_2, \cdots, f_k \in \R$.  Let $\eta=\min_{i \neq j}|f_j-f_i|$.
For any accuracy parameter $\epsilon \in (0,0.1)$, there is an algorithm \textsc{FastDistill1D} (Algorithm~\ref{algo:distillation_1d_fast}) that runs in $O(\epsilon^{-2}k^{\omega+1})$-time and 
outputs a set $S\subset [-T, T]$ of size $s=O(k/\eps^2)$ and a  weight vector $w \in \R^s_{\geq 0}$ such that, for any signal of the form $x^*(t) = \sum_{j=1}^k v_j \exp({2\pi\i  f_j  t  })$,
\begin{align*}
(1-\eps)\|x^*(t)\|_T\leq \|x^*(t)\|_{S,w} \leq (1+\eps)\|x^*(t)\|_T 
\end{align*}
holds with probability $0.99$.

Furthermore, for any noise signal $g(t)$, the following holds with high probability:
\begin{align*}
    \|g\|_{S,w}^2\lesssim \|g\|_T^2,
\end{align*}
where $\|x\|_T^2:=\frac{1}{2T}\int_{-T}^T |x(t)|^2 \d t$.
\end{lemma}
\begin{proof}
For the convenient, in the proof, we use time duration $[-T,T]$.
Let $D(t)$ be defined as follows:
\begin{flalign*}
D(t)=
\begin{cases}
{c}/(1-|t/T|) , & \text{ for } |t| \le T(1-{1}/k)\\
c \cdot k, & \text{ for } |t|\in [T(1-{1}/k), T]
\end{cases} 
\end{flalign*} 
where $c=O(T^{-1}\log^{-1}(k))$ a fixed value such that $\int_{-T}^T D(t) \d t= 1$.

First, we randomly pick up a set $S_0=\{t_1, \cdots , t_{s_0}\}$ of $s_0 = O(\eps_0^{-2}k\log(k)\log(1/\rho_0))$ i.i.d. samples from $D(t)$, and let $w'_i:=2/(T s_0 D(t_i))$ for $i\in[s_0]$ be the weight vector, where $\eps_0,\rho_0$ are parameters to be chosen later.  

By Lemma  \ref{lem:concentration_for_any_polynomial_signal_improve_by_CP19}, we know that $(S_0, w')$ gives a good weighted sketch of the signal that can preserve the norm with high probability. More specifically, with probability $1-\rho_0$,
\begin{align}
\label{eq:T_S_0_improve_dis_1d_fast}
 (1-\epsilon_0)\| x^*(t)\|^2_T\leq \| x^*(t)\|^2_{S_0,w'}\leq (1 + \epsilon_0)\| x^*(t)\|^2_T.
\end{align}

Then, we will select $s=O(k/\eps_1^2)$ elements from $S_0$ and output the corresponding weights $w_1, w_2, \cdots, w_{s}$ by applying \textsc{RandBSS+} with the following parameter: replacing $d$ by $k$, $\eps$ by $\eps_1^2$, and $D$ by $D(t_i)=w'_i/\sum_{j\in[s_0]} w'_j$ for $i\in [s_0]$.

By Theorem~\ref{thm:fast_WBSP} and the property of WBSP (Definition~\ref{def:procedure_agnostic_learning}), 
we obtain that with probability $0.995$,
\begin{align*}
    (1-\eps_1)\|x^*(t)\|_{S_0, w'}^2\leq \|x^*(t)\|_{S,w}^2 \leq (1+\eps_1)\|x^*(t)\|_{S_0, w'}^2.
\end{align*}

Combining with Eq.~\eqref{eq:T_S_0_improve_dis_1d_fast}, we conclude that
\begin{align*}
    \|x^*\|_{S,w}^2 \in&~ [1-\eps_1, 1+\eps_1]\cdot \|x^*\|_{S_0, w'}^2 \\
    \in&~[(1-\eps_0)(1-\eps_1), (1+\eps_0)(1+\eps_1)]\cdot \|x^*\|_{T}^2\\
    \in&~[1-\eps, 1+\eps]\cdot \|x^*\|_{T}^2,
\end{align*}
where the second step follows from Eq.~\eqref{eq:T_S_0_improve_dis_1d_fast} and the last stpe follows by taking $\epsilon_0=\eps_1=\eps/4$.

The overall success probability follows by taking union bound over the two steps and taking $\rho_0=0.001$.
The running time of Algorithm~\ref{algo:distillation_1d_fast} follows from Claim~\ref{cla:time_important_sampling_k_faster}. 
And the furthermore part follows from Claim~\ref{clm:preserve_noise_1d}.

The proof of the lemma is then completed.
\end{proof}

\begin{algorithm}[ht]\caption{Fast distillation for one-dimensional signal} \label{algo:distillation_1d_fast}
\begin{algorithmic}[1]
\Procedure{\textsc{WeightedSketch}}{$k, \eps, T, \mathcal{B}$}\Comment{Lemma~\ref{lem:concentration_for_any_polynomial_signal_improve_by_CP19}}
\State $c\leftarrow O(T^{-1}\log^{-1}(k))$
\State $D(t)$ is defined as follows:
\begin{flalign*}
D(t)\leftarrow
\begin{cases}
{c}/((1-|t/T|) \log k ), & \text{ if } |t| \le T(1-{1}/k),\\
c \cdot k, & \text{ if } |t|\in [T(1-{1}/k), T].
\end{cases} 
\end{flalign*} 
\State $S_0\gets$  $O(\eps^{-2}k\log(k))$ i.i.d. samples from $D$
\For{$t\in S_0$}
    \State $w_t\gets \frac{2}{T\cdot |S_0|\cdot D(t)}$
\EndFor
\State\label{step:new_dist} Set a new distribution $D'(t)\leftarrow w_t/\sum_{t'\in S_0} w_{t'}$ for all $t\in S_0$
\State \Return $D'$
\EndProcedure
\Procedure{\textsc{FastDistill1D}}{$k$, $\eps$, $F=\{f_1,\dots,f_k\}$, $T$} \Comment{Lemma \ref{lem:important_sampling_k_faster}}
\State Distribution $D'\leftarrow \textsc{WeightedSketch}(k, \eps, T, \mathcal{B})$
\State Set the function family $\mathcal{F}$ as follows:
\begin{align*}
{\mathcal F} := \Big\{f(t)=\sum_{j=1}^k v_j \exp(2\pi\i f_j t) ~\Big|~ v_j\in\C \Big\}.
\end{align*}
\State $s,\{t_1,t_2,\cdots, t_s\},w \leftarrow\textsc{RandBSS+}(k, \mathcal{F},D',(\eps/4)^2)$ \Comment{$s=O(k/\eps^2)$, Algorithm~\ref{alg:BSS_faster}}
\State \Return  $\{t_1,t_2,\cdots, t_s\}$ and $w$ %
\EndProcedure
\end{algorithmic}
\end{algorithm}

\begin{claim}[Running time of Procedure \textsc{FastDistill1D} in Algorithm \ref{algo:distillation_1d_fast}]\label{cla:time_important_sampling_k_faster}
Procedure \textsc{FastDistill1D} in Algorithm \ref{algo:distillation_1d_fast}  runs in $$O( \eps^{-2}k^{\omega+1}) $$ time. 
\end{claim}
\begin{proof}
First, it is easy to see that Procedure \textsc{WeightedSketch} takes $O(\eps^{-2} k\log(k))$-time.

By Theorem~\ref{thm:fast_WBSP} with $|D|=O(\eps^{-2}k\log(k))$, $ d=k$, we have that the running time of Procedure \textsc{RandBSS+} is
\begin{align*}
      &~O\left(k^2 \cdot \eps^{-2}k\log(k)+ \eps^{-2}k^3 \log \left(\eps^{-2}k\log(k)\right) + \eps^{-2} k^{\omega+1}\right) \\
   = &~O\left(\epsilon^{-2}k^{\omega+1}\right).
\end{align*}

Hence, the total running time of Algorithm~\ref{algo:distillation_1d_fast} is $O\left(\epsilon^{-2}k^{\omega+1}\right)$.

\end{proof}

\begin{claim}[Preserve the energy of noise]\label{clm:preserve_noise_1d}
Let $(S,w)$ be the outputs of Algorithm~\ref{algo:distillation_1d_fast}. Then, we have that
\begin{align*}
\|g(t)\|^2_{S,w}\lesssim \|g(t)\|_T^2,
\end{align*}
holds with probability $0.99$.
\end{claim}
\begin{proof}
For the convenient, in the proof, we use time duration $[-T,T]$. Algorithm~\ref{algo:distillation_1d_fast} has two stages of sampling.

In the first stage, Procedure \textsc{WeightedSketch} samples a set $S_0=\{t_1',\dots,t_{s_0}'\}$ of i.i.d. samples from the distribution $D$, and a weight vector $w'$. Then, we have
\begin{align*}
    \E\big[ \|g(t)\|^2_{S_0,w'}\big] = &~ \E\Big[ \sum_{i=1}^{s_0} w'_i |g(t_i')|^2\Big] \\
    = &~  \sum_{i=1}^{s_0} \E_{t_i' \sim D}[ w_i' |g(t_i')|^2] \\
    = &~  \sum_{i=1}^{s_0} \E_{t_i' \sim D}\Big[  \frac{2}{Ts_0 D(t_i')} |g(t_i')|^2\Big] \\
    = &~  \sum_{i=1}^{s_0} \E_{t_i' \sim \mathrm{Uniform([-T, T])}}[ s_0^{-1} |g(t_i')|^2] \\
    = &~   \E_{t \sim \mathrm{Uniform([-T, T])}}[ |g(t)|^2] \\
    = &~   \|g(t)\|^2_{T} 
\end{align*}
where the first step follows from the definition of the norm, the third step follows from the definition of $w_i$, the forth step follows from $\E_{t\sim D_0(t)}[\frac{D_1(t)}{D_0(t)}f(t)]=\E_{t\sim D_1(t)}f(t)$.

In the second stage, let $P$ denote the Procedure \textsc{RandBSS+}. With high probability, $P$ is a $\eps$-WBSP (Definition \ref{def:procedure_agnostic_learning}). By the Definition \ref{def:procedure_agnostic_learning}, each sample $t_i\sim D_i(t)$ and $w_i=\alpha_i \cdot \frac{D'(t_i)}{D_i(t_i)}$ in every iteration $i\in [s]$, where $\sum_{i=1}^s \alpha_i \leq 5/4$ and  $D'(t)=\frac{w'_t}{\sum_{t'\in S_0}w'_{t'}}$. 
As a result,
\begin{align*}
    \E_P[ \|g(t)\|^2_{S,w}] = &~ \E_P\Big[ \sum_{i=1}^s w_i |g(t_i)|^2\Big] \\
    = &~  \sum_{i=1}^s \E_{t_i \sim D_i(t_i)}[ w_i |g(t_i)|^2] \\
    = &~  \sum_{i=1}^s \E_{t_i \sim D_i(t_i)}\Big[ \alpha_i \cdot \frac{D'(t_i)}{D_i(t_i)} |g(t_i)|^2\Big] \\
    = &~  \sum_{i=1}^s \E_{t_i \sim D'(t_i)}[ \alpha_i |g(t_i)|^2] \\
    \leq &~  \underset{P}{\sup}\{\sum_{i=1}^s \alpha_i\} \E_{t \sim D'(t)}[  |g(t)|^2] \\
    = &~  \underset{P}{\sup}\{\sum_{i=1}^s \alpha_i\} \|g(t)\|^2_{S_0,w'} \cdot (\sum_{t'\in S_0}w_{t'}')^{-1} \\
    \lesssim&~ \rho^{-1}\cdot  \|g(t)\|^2_{S_0,w'}.
\end{align*}
where the first step follows from the definition of the norm, the third step follows from $w_i=\alpha_i \cdot \frac{D'(t_i)}{D_i(t_i)}$, the forth step follows from $\E_{t\sim D_0(t)}\frac{D_1(t)}{D_0(t)}f(t)=\E_{t\sim D_1(t)}f(t)$, the sixth step follows from $D'(t)=\frac{w'_t}{\sum_{t'\in S_0}w_{t'}'}$ and the definition of the norm, the last step follows from $\sum_{i=1}^s \alpha_i \leq 5/4$ and $(\sum_{t'\in S_0}w'_{t'})^{-1}=O(\rho^{-1})$ with probability at least $1-\rho/2$. 

Hence, combining the two stages together, we have
\begin{align*}
    \E\big[\E_P[\|g(t)\|_{S,w}^2]\big] \lesssim \rho^{-1}\cdot \E\big[\|g(t)\|_{S_0,w'}^2\big] = \rho^{-1}\cdot \|g\|_T^2.
\end{align*}
And by Markov inequality and union bound, we have
\begin{align*}
    \Pr\left[\|g(t)\|^2_{S,w}\lesssim \rho^{-2}\|g(t)\|_T^2\right]\leq 1-\rho.
\end{align*}
\end{proof}

\subsubsection{Sharper bound for the energy of orthogonal part of noise}
In this section, we give a sharper analysis for the energy of $g^\bot$ on the sketch, which is the orthogonal projection of $g$ to the space ${\cal F}$. More specifically, we can decompose an arbitrary function $g$ into $g^\parallel + g^\bot$, where $g^\parallel\in {\cal F}$ and $\int_{[0,T]}\ov{h(t)}g^\bot(t)\d t = 0$ for all $h\in {\cal F}$. The motivation of considering $g^\bot$ is that $g^\parallel$ is also a Fourier sparse signal and its energy will not be amplified in the Signal Estimation problem. And the nontrivial part is to avoid the blowup of the energy of $g^\bot$, which is shown in the following lemma:

\begin{lemma}[Preserving the orthogonal energy]\label{lem:guarantee_dist}
Let ${\cal F}$ be an $m$-dimensional linear function family with an orthonormal basis $\{v_1,\dots,v_m\}$ with respect to a distribution $D$.
Let $P$ be the $\eps$-WBSP that generate a sample set $S=\{t_1,\dots, t_{s}\}$ and coefficients $\alpha\in \R_{>0}^s$, where each $t_i$ is sampled from distribution $D_i$ for $i\in [s]$.
Define the weight vector $w\in \R^s$ be such that $w_i:=\alpha_i  \frac{D(t_i)}{D_i(t_i)}$ for $i\in [s]$. %

For any noise function $g(t)$ that is orthogonal to ${\cal F}$ with respect to $D$, %
the following property holds with probability 0.99:
\begin{align*}
    \sum_{i=1}^m |\langle g, v_i \rangle_{S, w}|^2\lesssim \epsilon \|g\|_D^2,
\end{align*}
where $\langle g, v\rangle_{S,w}:= \sum_{j=1}^s w_j\ov{v(t_j)}g(t_j)$.
\end{lemma}
\begin{remark}
We note that this lemma works for both continuous and discrete signals.
\end{remark}

\begin{remark}
$|\langle g, v_i\rangle_{S,w}|^2$ corresponds to the energy of $g$ on the sketch points in $S$. On the other hand, if we consider the energy on the whole time domain, we have $\langle g, v_i\rangle=0$ for all $i\in [m]$. The above lemma indicates that this part of energy could be amplified by at most $O(\epsilon)$, as long as the sketch comes from a WBSP.
\end{remark}
\begin{proof}

We can upper-bound the expectation of $\sum_{i=1}^m |\langle g, v_i \rangle_{S, w}|^2$ as follows:
\begin{align*}
\E\Big[\sum_{i=1}^m |\langle g, v_i \rangle_{S, w}|^2\Big]
 = &~ \E_{D_1,\dots,D_s}\Big[\|w\|_1^2 \sum_{i=1}^m \big| \E_{t\sim D'}[  \overline{v_j(t)}g(t)] \big|^2 \Big]\\
 = &~ \E_{D_1,\dots,D_s}\Big[\sum_{i=1}^m\big| \sum_{j=1}^s  w_j\overline{v_i(t_j)}g(t_j)] \big|^2\Big]\\
 = &~ \sum_{i=1}^m \E_{D_1,\dots,D_s} \Big[ \big| \sum_{j=1}^s w_j \overline{v_i(t_j)}  g(t_j) \big|^2 \Big] \\
 = &~ \sum_{i=1}^m \E_{D_1,\dots,D_s}  \Big[ \sum_{j=1}^s w_j^2  |v_i(t_j)|^2  |g(t_j)|^2  \Big]\\
  =&~ \sum_{j=1}^s \E_{D_j}  \Big[ \sum_{i=1}^m  w_j |v_i(t_j)|^2 \cdot  w_j |g(t_j)|^2  \Big]\\
 \le&~ \sum_{j=1}^s \sup_{t\in D_j}  \Big\{ w_j \sum_{i=1}^m |v_i(t)|^2  \Big\}\cdot     \E_{D_j}  [ w_j  |g(t_j)|^2  ],
\end{align*}
where the first step follows from Fact~\ref{fac:g_w_expand}, the second step follows from the definition of $ D'$, the third follows from the linearity of expectation, the forth step follows from Fact~\ref{fac:break_the_square}, the last step follows by pulling out the maximum value of $w_j \sum_{i=1}^k |v_i(t)|^2$ from the expectation.

Next, we consider the first term:
\begin{align*}
    \sup_{t\in D_j}  \Big\{ w_j \sum_{i=1}^m |v_i(t)|^2  \Big\} =&~ \sup_{t\in D_j}  \Big\{ \alpha_j  \frac{D(t)}{D_j(t)} \sum_{i=1}^m |v_i(t)|^2  \Big\}\\
    =&~ \alpha_j \sup_{t\in D_j}  \Big\{   \frac{D(t)}{D_j(t)}  \sup_{h \in \mathcal{F}} \big\{ \frac{|h(t)|^2}{\|h\|_D^2} \big\}  \Big\}\\ 
    = &~ \alpha_j  K_{\mathsf{IS},D_j}.
\end{align*}
where the first step follows from the definition of $w_j$, the second step follows from Fact~\ref{fac:condition_number_to_ortho_basis} that  $\underset{h \in \mathcal{F}}{\sup} \{ \frac{|h(t_j)|^2}{\|h\|_D^2} \}=\sum_{i=1}^k |v_i(t_j)|^2$, the last step follows from the definition of $K_{\mathsf{IS},D_j}$ (Eq.~\eqref{eq:def_K}).

Then, we bound the last term:
\begin{align*}
\E_{D_j}  [ w_j  |g(t_j)|^2  ]=\underset{t_j \sim D_j}{\E}  \Big[ \alpha_j  \frac{D(t_j)}{D_j(t_j)} |g(t_j)|^2  \Big]=\alpha_j  \underset{t_j \sim D}{\E}[ |g(t_j) |^2] =\alpha_j\|g\|_D^2.   
\end{align*}

Combining the two terms together, we have
\begin{align*}
   \E\Big[\sum_{i=1}^m |\langle g, v_i \rangle_{S, w}|^2 \Big]
    \le&~ \sum_{j=1}^s  ( \alpha_j K_{\mathsf{IS}, D_j}  \cdot \alpha_j   \|g\|_D^2  )\\
    \le&~ \Big(\sum_{j=1}^s \alpha_j\Big) \cdot \max_{j\in[s]} \{ \alpha_j K_{\mathsf{IS}, D_j} \} \cdot  \|g\|_D^2\\
    \leq &~ \eps \|g\|_D^2.
\end{align*}
where the last step follows from $P$ being a $\eps$-WBSP (Definition \ref{def:procedure_agnostic_learning}), which implies that $\sum_{j=1}^s \alpha_j=\frac{5}{4}$ and $\alpha_j K_{\mathsf{IS},D_j} \leq \eps/2$ for all $j\in [s]$.

Finally, by Markov's inequality, we have that
\begin{align*}
    \sum_{i=1}^m |\langle g, v_i \rangle_{S, w}|^2  \lesssim \eps \|g\|_D^2
\end{align*}
holds with probability $0.99$. 
\end{proof}

\begin{fact}\label{fac:g_w_expand}
\begin{align*}
    \sum_{i=1}^m |\langle g, v_i\rangle_{S,w}|^2 = \|w\|_1^2 \cdot \sum_{i=1}^m \Big|\E_{t\sim D'}[\ov{v_i(t)}g(t)]\Big|^2,
\end{align*}
where $D'$ is a distribution defined by $D'(t_i):=\frac{w_i}{\|w\|_1}$ for $i\in [s]$.
\end{fact}
\begin{proof}
We have:
\begin{align*}
\sum_{i=1}^m |\langle g, v_i\rangle_{S,w}|^2
 = &~ \sum_{i=1}^m  \Big| \sum_{j=1}^s w_j \overline{v_i(t_j)}  g(t_j) \Big|^2\\
 = &~ \sum_{i=1}^m \Big| \sum_{j=1}^s \frac{w_j\overline{v_i(t_j)}  g(t_j)}{\sum_{j'=1}^s w_{j'}}  \Big|^2\cdot \Big(\sum_{j'=1}^s w_{j'}\Big)^2\\
 = &~ \Big(\sum_{j'=1}^s w_{j'}\Big)^2 \cdot \sum_{i=1}^m \Big|\E_{t\sim D'}[\ov{v_i(t)}g(t)]\Big|^2. 
\end{align*}
\end{proof}

\begin{fact}\label{fac:break_the_square}
For any $i\in [m]$, we have 
\begin{align*}
    \E_{D_1,\dots,D_s} \Big[ \big| \sum_{j=1}^s w_j \overline{v_i(t_j)}  g(t_j) \big|^2 \Big] = \E_{D_1,\dots,D_s}  \Big[ \sum_{j=1}^m w_j^2  |v_i(t_j)|^2  |g(t_j)|^2  \Big].
\end{align*}
\end{fact}
\begin{proof}
We first show that for any $i\in [m]$ and $j\in [s]$,
\begin{align}
\underset{t_j \sim D_j}{\E}[w_j \overline{v_i(t_j)}  g(t_j)]=&~\underset{t_j \sim D_j}{\E}[\alpha_j  \frac{D(t_j)}{D_j(t_j)} \overline{v_i(t_j)} g(t_j)]\notag\\
=&~ \alpha_j  \underset{t_j \sim D}{\E}[\overline{v_i(t_j)} g(t_j)]\notag\\
=&~ 0.\label{eq:cross_term_expect_0}
\end{align}
where the first step follows from the definition of $w_i$, the third step follows from $g(t)$ is orthonormal with $v_i(t)$ for any $i\in [k]$. 

Then, we can expand LHS as follows:
\begin{align*}
    &~\E_{D_1,\dots,D_s} \Big[ \big| \sum_{j=1}^s w_j \overline{v_i(t_j)}  g(t_j) \big|^2 \Big]\\
    = &~ \E_{D_1,\dots,D_s} \Big[ \big( \sum_{j=1}^s w_j \overline{v_i(t_j)}  g(t_j) \big)^*\big( \sum_{j=1}^s w_j \overline{v_i(t_j)}  g(t_j) \big) \Big]\\
    = &~  \E_{D_1,\dots,D_s} \Big[\sum_{j,j'=1}^s w_{j}w_{j'}v_i(t_j) \ov{g(t_j)}\ov{v_i(t_{j'})}g(t_{j'})\Big]\\
    = &~ \sum_{j,j'=1}^s \E_{D_1,\dots,D_s} [w_{j}w_{j'}v_i(t_j) \ov{g(t_j)}\ov{v_i(t_{j'})}g(t_{j'})]\\
    = &~ \sum_{j=1}^s \E[w_j^2  |v_i(t_j)|^2  |g(t_j)|^2] + \sum_{1\leq j<j'\leq s} 2\Re \E_{D_1,\dots,D_j}[w_{j}w_{j'}v_i(t_j) \ov{g(t_j)}\ov{v_i(t_{j'})}g(t_{j'})]\\
    = &~ \RHS + \sum_{1\leq j<j'\leq s} 2\Re\E_{D_1,\dots, D_{j}}\Big[w_{j}v_i(t_j)\ov{g(t_j)} \E_{D_{j+1},\dots,D_{j'}}[w_{j'} \ov{v_i(t_{j'})}g(t_{j'})]\Big]\\
    = &~ \RHS + \sum_{1\leq j<j'\leq s} 2\Re\E_{D_1,\dots, D_{j}}[w_{j}v_i(t_j)\ov{g(t_j)} \cdot 0]\\
    = &~ \RHS,
\end{align*}
where the third step follows from the linearity of expectation, the fifth step follows from $t_j$ only depends on $t_{1},\dots, t_{j-1}$, and the sixth step follows from Eq.~\eqref{eq:cross_term_expect_0}.
\end{proof}

\begin{fact}\label{fac:condition_number_to_ortho_basis}
Let $\{v_1,\dots,v_k\}$ be an orthonormal basis of ${\cal F}$ with respect to the distribution $D$. Then, we have
\begin{align*}
    \underset{h \in \mathcal{F}}{\sup} \Big\{ \frac{|h(t)|^2}{\|h\|_D^2} \Big\}=\sum_{i=1}^k |v_i(t)|^2
\end{align*}
\end{fact}
\begin{proof}
We have:
\begin{align*}
    \underset{h \in \mathcal{F}}{\sup} \Big\{ \frac{|h(t)|^2}{\|h\|_D^2} \Big\}=&~\underset{a\in \C^k}{\sup} \Big\{ \frac{|\sum_{i=1}^k a_i v_i(t)|^2}{\|a\|_2^2}\Big\}\\
    = &~ \sup_{a\in \C^k: \|a\|_2=1} \Big|\sum_{i=1}^k a_i v_i(t)\Big|^2\\
    =&~ \sum_{i=1}^k |v_i(t)|^2,
\end{align*}
 where the first step follows from each $h\in {\cal F}$ can be expanded as $h=\sum_{i=1}^k a_i v_i$ and $\|h(t)\|_D^2=\|a\|_2^2$ (Fact~\ref{fac:orthonormal_norm}), the second step follows from the Cauchy-Schwartz inequality and taking $a=\frac{v(t)}{\|v(t)\|_2}$.
\end{proof}

\subsection{Sketch distillation for high-dimensional signals}\label{sec:hd_important_sample_k}
The goal of this section is to prove Lemma~\ref{lem:important_sampling_k_high_d}, which can reduce the sketch size of Corollary~\ref{lem:concentration_for_any_polynomial_signal_high_D} for high-dimensional signals.%

\begin{lemma}[Distillation for high-dimensional signal]\label{lem:important_sampling_k_high_d}
Given $f_1, f_2, \cdots, f_k \in \R^d$. Let $x^*(t) = \sum_{j=1}^k v_j e^{2\pi\i \langle f_j , t \rangle }$ for $t\in [0,T]^d$. Let $\eta=\min_{i \neq j}\|f_j-f_i\|_{\infty}$.
For any accuracy parameter $\eps\in(0,1)$, there is an algorithm \textsc{DistillHD} (Algorithm~\ref{algo:distillation_hd}) that runs in $\wt{O}(\eps^{-2}k^{O(d)})$-time and outputs a set $S\subset [0,T]^d$ of size $s=O(k/\eps^2)$ and a weight vector $w\in \R^s_{\geq 0}$ such that
\begin{align*}
(1-\eps)\|x^*\|_T \leq \|x^*\|_{S,w} \leq (1+\eps)\|x^*\|_T 
\end{align*}
holds with probability $0.99$.

Furthermore, for any noise function $g(t)$, with high probability, it holds that
\begin{align*}
    \|g\|_{S,w}\lesssim \|g\|_T.
\end{align*}
\end{lemma}

\begin{proof}
First, we randomly and uniformly sample a set $S_0$ of $s_0=O(\epsilon_0^{-2} k^{O(d)}\log(1/(\rho_0\epsilon_0)))$ real number in $[0, T]^d$, where $\epsilon_0,\rho_0$ are parameters to be chosen later. 

By Corollary \ref{lem:concentration_for_any_polynomial_signal_high_D}, we know that those points are good sketch of the high-dimensional signal and can preserve the norm with a large probability. More precisely, with probability $1-\rho_0$, 
\begin{align}\label{eq:T_S_0_high_D}
(1-\eps_0)\|x^*\|_T^2 \leq \|x^*\|_{S_0}^2 \leq (1+\eps_0)\|x^*\|_T^2. 
\end{align}

Then, we will select $s=O(k)$ real number from $S_0$ and output $s$ corresponding weight $w_1, w_2, \cdots, w_{s}$ by applying the Procedure \text{RandBSS+} with setting the following parameter: replacing $d$ by $k$, $\eps$ by $\eps_1^2$, $D$ by $\mathrm{Uniform}(S_0)$, and ${\cal F}$ by
\begin{align*}
    {\mathcal F} = \Big\{f(t)=\sum_{j=1}^k v_j \exp(2\pi\i\langle f_j, t \rangle) ~\big|~ v_j\in\C \Big\}.
\end{align*}
Then, by Theorem~\ref{thm:fast_WBSP} and the property of WBSP (Definition~\ref{def:procedure_agnostic_learning}), we obtain that with probability $0.995$,
\begin{align*}
    (1-\eps_1)\|x^*\|_{S_0}^2 \leq \|x^*\|_{S,w}^2 \leq (1+\eps_1)\|x^*\|_{S_0}^2.
\end{align*}

Combining with Eq.~\eqref{eq:T_S_0_high_D}, we conclude that
\begin{align*}
    \|x^*\|_{S,w}^2 \in&~ [1-\eps_1, 1+\eps_1]\cdot \|x^*\|_{S_0}^2 \\
    \in&~[(1-\eps_0)(1-\eps_1), (1+\eps_0)(1+\eps_1)]\cdot \|x^*\|_{T}^2\\
    \in&~[1-\eps, 1+\eps]\cdot \|x^*\|_{T}^2,
\end{align*}
where the second step follows from Eq.~\eqref{eq:T_S_0_high_D}, and the last step follows from $\eps_0=\eps_1=\eps/4$.

The running time of Algorithm~\ref{algo:distillation_hd} follows from Claim~\ref{cla:running_time_important_sampling_k_high_d} and the success probability follows from setting $\rho_0=0.001$. The furthermore part follows from Claim~\ref{clm:preserve_noise_high_D}.

The lemma is then proved.
\end{proof}

\begin{algorithm}[ht]\caption{Distillation for high-dimensional signal.} \label{algo:distillation_hd}
\begin{algorithmic}[1]
\Procedure{\textsc{DistillHD}}{$k, \eps, d, F=\{f_1,\dots,f_k\}, T$} \Comment{Lemma \ref{lem:important_sampling_k_high_d}}
\State $S_0\gets$ $O(\epsilon^{-2} k^{O(d)}\log(1/\epsilon))$ i.i.d. samples from $\text{Uniform}([0,T]^d)$
\State Set the function family $\mathcal{F}$ as follows:
\begin{align*}
{\mathcal F} = \Big\{f(t)=\sum_{j=1}^k v_j \exp(2\pi\i\langle f_j, t \rangle) ~\big|~ v_j\in\C \Big\}.
\end{align*}
\State $s,\{t_1,t_2,\cdots, t_s\},w \leftarrow\textsc{RandBSS+}(k, \mathcal{F},\mathrm{Uniform}(S_0),(\eps/4)^2)$ \Comment{$s=O(k/\eps^2)$, Algorithm~\ref{alg:BSS_faster}}
\State \Return  $\{t_1,t_2,\cdots, t_s\}$ and $w$ 
\EndProcedure
\end{algorithmic}
\end{algorithm}

\begin{claim}[Running time of Procedure \textsc{DistillHD} in Algorithm \ref{algo:distillation_hd}]\label{cla:running_time_important_sampling_k_high_d}
Procedure \textsc{DistillHD} in Algorithm \ref{algo:distillation_hd}  runs in time %
$$O\left(\epsilon^{-2}k^{O(d)} \log(1/\eps)\right).$$ 
\end{claim}
\begin{proof}

The first step of sampling $S_0$ takes $O(\epsilon^{-2} k^{O(d)}\log(1/\eps))$-time.

Then, by Theorem~\ref{thm:fast_WBSP} with $|D|=O(\epsilon^{-2} k^{O(d)}\log(1/\eps))$, $ d=k$, we have that the running time of \textsc{RandBSS+} is
\begin{align*}
   &~O\left(k^2 \cdot \epsilon^{-2} k^{O(d)}\log(1/\eps)+ \eps^{-2}k^3 \log \left(\epsilon^{-2} k^{O(d)}\log(1/\eps)\right) + \eps^{-2} k^{\omega+1}\right) \\
   = &~O\left(\epsilon^{-2}k^{O(d)}\log^3(k)\log(1/\eps)\right).  
\end{align*}

Hence, the total running time of Algorithm~\ref{algo:distillation_hd} is $O\left(\epsilon^{-2}k^{O(d)}\log^3(k)\log(1/\eps)\right)$.

\end{proof}

\begin{claim}[Preserve the energy of noise (high Dimension)]\label{clm:preserve_noise_high_D}
Let $(S,w)$ be the outputs of Algorithm~\ref{algo:distillation_hd}. Then, for any function $g(t)$,
\begin{align*}
\|g(t)\|^2_{S,w}\lesssim \rho^{-2}\|g(t)\|_T^2,
\end{align*}
holds with probability $1-\rho$.%
\end{claim}

\begin{proof}
Let $P$ denote the Procedure $\textsc{ImportantSampling}(k, \eps, \rho, F, T, \mathcal{B})$. Because $P$ is a $\eps$-well-balanced sampling procedure (Definition \ref{def:procedure_agnostic_learning}). By the Definition \ref{def:procedure_agnostic_learning}, we have that $t_i\sim D_i(t)$ and $w_i=\alpha_i \cdot \frac{D(t_i)}{D_i(t_i)}$ in every iteration $i\in [s]$, where $\sum_{i=1}^s \alpha_i \leq 5/4$, $D(t)=\text{Uniform}(S_0)$. 

As a result,
\begin{align*}
    \E_P[ \|g(t)\|^2_{S,w}] = &~ \E_P[ \sum_{i=1}^s w_i |g(t_i)|^2] \\
    = &~  \sum_{i=1}^s \E_{t_i \sim D_i(t_i)}[ w_i |g(t_i)|^2] \\
    = &~  \sum_{i=1}^s \E_{t_i \sim D_i(t_i)}[ \alpha_i \cdot \frac{D(t_i)}{D_i(t_i)} |g(t_i)|^2] \\
    = &~  \sum_{i=1}^s \E_{t_i \sim D(t_i)}[ \alpha_i |g(t_i)|^2] \\
    \leq &~  \underset{P}{\sup}\{\sum_{i=1}^s \alpha_i\} \E_{t \sim D(t)}[  |g(t)|^2] \\
    = &~  \underset{P}{\sup}\{\sum_{i=1}^s \alpha_i\} \|g(t)\|^2_{S_0} \\
    \leq&~ 2 \|g(t)\|^2_{S_0}.
\end{align*}
where the first step follows from the definition of the norm, the third step follows from $w_i=\alpha_i \cdot \frac{D(t_i)}{D_i(t_i)}$, the forth step follows from $\E_{t\sim D_0(t)}\frac{D_1(t)}{D_0(t)}f(t)=\E_{t\sim D_1(t)}f(t)$, the sixth step follows from $D(t)=\text{Uniform}(S_0)$ and the definition of the norm, the last step follows from $\sum_{i=1}^s \alpha_i \leq 5/4$. 

Moreover,
\begin{align*}
    \E_{S_0}[\|g(t)\|^2_{S_0}] =&~\E_{t\sim\text{Uniform}([0, T])}|g(t)|^2\\
    =&~ \|g(t)\|_T^2
\end{align*}

So, by Markov's inequality,
\begin{align*}
\Pr[\|g(t)\|^2_{S,w}\leq  \|g(t)\|^2_{S_0}/\eps_0 ] \geq 1-\eps_0/2, 
\end{align*}
and 
\begin{align*}
\Pr[\|g(t)\|^2_{S_0} \leq  \|g(t)\|_T^2/ \eps_1] \geq 1- \eps_1.
\end{align*}

Then, with probability at least $(1-\eps_0/2)(1-\eps_1)$ holds, 
\begin{align*}
    \|g(t)\|^2_{S,w}\leq  \|g(t)\|^2_{S_0}/\eps_0\leq\|g(t)\|_T^2/ (\eps_0\eps_1).
\end{align*}

Set $\eps_0=\rho/10,\eps_1=\rho/10$, we have that,
\begin{align*}
\|g(t)\|^2_{S,w}\lesssim \|g(t)\|_T^2/\rho^2,
\end{align*}
holds with probability $1-\rho$. 

\end{proof}

\subsection{Sketch distillation for discrete signals}\label{sec:distill_dis}

The goal of this section is to prove Lemma~\ref{lem:important_sampling_k_faster_discrete}, which can reduce the sketch size of Corollary  \ref{cla:concentration_for_any_polynomial_signal_discrete} for discrete signals in any dimension.

\begin{lemma}[Distillation for discrete signal]\label{lem:important_sampling_k_faster_discrete}
For any $d\geq 1$, let $n=p^d$ for some positive integer $p$. Let $x^* \in \C^{[p]^d}$, such that $\supp(\wh{x^*})\subset[p]^d$ and $|\supp(\wh{x^*})|=k$. 
For any accuracy parameter $\epsilon \in (0,0.1)$, there is an algorithm (Algorithm~\ref{algo:distill_dis})
that runs in $O(\epsilon^{-2}k^{\omega+1})$-time and outputs a set $S\subset [n]$ of size $s=O(k/\eps^2)$ and a  weight vector $w \in \R^s_{\geq 0}$ such that, 
\begin{align*}
(1-\eps)\|x^*\|_2\leq n \|x^*\|_{S,w} \leq (1+\eps)\|x^*\|_2
\end{align*}
holds with probability $0.99$.
\end{lemma}
\begin{proof}
For the convenient, in the proof, we use $x$ to denote the $x^*$.

First, we randomly pick up a set $S_0=\{t_1, \cdots , t_{s_0}\}$ of $s_0 = O(\epsilon^{-2}k\log(k/\rho))$ i.i.d. samples from $\mathrm{Uniform}([n])$, where $\epsilon_0,\rho_0$ are parameters to be chosen later.

By Corollary  \ref{cla:concentration_for_any_polynomial_signal_discrete}, with probability $1-\rho_0$,
\begin{align}
\label{eq:T_S_0_improve_discrete}
(1-\epsilon_0)\| x\|^2_2 \leq n\| x \|^2_{S_0} \leq (1 + \epsilon_0)\| x\|^2_2.
\end{align}

Then, we will select $s=O(k/\eps_1^2)$ elements from $S_0$ and output the corresponding weights $w_1, w_2, \cdots, w_{s}$ by applying Procedure \textsc{RandBSS+} with the following parameter: replacing $d$ by $k$, $\eps$ by $\eps_1^2$, %
and $D$ by 
$\mathrm{Uniform}(S_0)$.

By Theorem~\ref{thm:fast_WBSP} and Definition~\ref{def:procedure_agnostic_learning}, 
we obtain that with probability $0.995$,  
\begin{align*}
    (1-\epsilon_1)\| x\|_{S,w}^2 \leq \| x \|^2_{S_0} \leq (1 + \epsilon_1)\| x\|^2_{S,w}.
\end{align*}

Combining with Eq.~\eqref{eq:T_S_0_improve_discrete}, we conclude that
\begin{align*}
    \|x\|_{S,w}^2 \in&~ [1-\eps_1, 1+\eps_1]\cdot \|x\|_{S_0}^2 \\
    \in&~[(1-\eps_0)(1-\eps_1), (1+\eps_0)(1+\eps_1)]\cdot \|x\|_{2}^2 / n\\
    \in&~[1-\eps, 1+\eps]\cdot \|x\|_{2}^2 / n,
\end{align*}
where the second step follows from Eq.~\eqref{eq:T_S_0_improve_discrete}, and the last step follows by taking $\eps_0=\eps_1=\eps/4$.

By taking $\rho=0.001$, we get that the overall success probability is at least 0.99. 

Regarding the running time, if $d=1$, we run Procedure \textsc{DistillDisc} in Algorithm~\ref{algo:distill_dis}, whose runtime follows from  Claim~\ref{cla:time_important_sampling_k_faster_discrete}. And if $d>1$, we run Procedure \textsc{DistillDiscHD} in Algorithm~\ref{algo:distill_dis}, whose runtime follows from  Claim~\ref{cla:time_important_sampling_k_faster_discrete_hd}.

The lemma is then proved.
\end{proof}

\begin{algorithm}[ht]\caption{Distillation for discrete signal.} \label{algo:distill_dis}
\begin{algorithmic}[1]
\Procedure{\textsc{DistillDisc}}{$k, \eps, F=\{f_1, \cdots, f_k \}, n$} \Comment{Lemma \ref{lem:important_sampling_k_faster_discrete} (one-dimension)}
\State $S_0\gets$ $O(\epsilon^{-2}k\log(k))$ i.i.d. samples from $\text{Uniform}([n])$
\State Set the function family $\mathcal{F}$ as follows:
\begin{align*}
{\mathcal F} = \{f(t)=\sum_{j=1}^k v_j \exp(2\pi\i f_j t / n) | v_j\in\C \}.
\end{align*}
\State $s,\{t_1,t_2,\cdots, t_s\},w \leftarrow\textsc{RandBSS+}(k, \mathcal{F},\mathrm{Uniform}(S_0),(\eps/4)^2)$ \Comment{$s=O(k/\eps^2)$, Algorithm~\ref{alg:BSS_faster}}
\State \Return  $\{t_1,t_2,\cdots, t_s\}$ and $w$ 
\EndProcedure

\Procedure{\textsc{DistillDiscHD}}{$k, \eps, F=\{f_1, \cdots, f_k \}, p, d$} \Comment{Lemma \ref{lem:important_sampling_k_faster_discrete} (high-dimension)}
\State $S_0\gets$ $O(\epsilon^{-2}k\log(k))$ i.i.d. samples from $\text{Uniform}([p]^d)$
\State $A\gets \begin{bmatrix}{t_1'}^\top\\\vdots \\ {t'_{s_0}}^\top\end{bmatrix}\in \R^{s_0\times d}$, $B\gets \begin{bmatrix}f_1 & \cdots & f_k\end{bmatrix}\in \R^{d\times k}$
\State $C\gets A\cdot B\in \R^{s_0\times k}$
\State ${\cal F}_{ij}\gets \exp(2\pi\i C_{ij})$ for each $(i,j)\in [s_0]\times [k]$\label{ln:construct_F}
\State $s,\{t_1,t_2,\cdots, t_s\},w \leftarrow\textsc{RandBSS+}(k, \mathcal{F},\mathrm{Uniform}(S_0),(\eps/4)^2)$ \Comment{$s=O(k/\eps^2)$, Algorithm~\ref{alg:BSS_faster}}
\State \Return  $\{t_1,t_2,\cdots, t_s\}$ and $w$ 
\EndProcedure
\end{algorithmic}
\end{algorithm}

\begin{claim}[Running time of Procedure \textsc{DistillDisc} in Algorithm~\ref{algo:distill_dis}]\label{cla:time_important_sampling_k_faster_discrete}
Procedure \textsc{DistillDisc} in Algorithm \ref{algo:distill_dis} runs in $$O\left( \eps^{-2}k^{\omega+1}\right) $$ 
time. 
\end{claim}
\begin{proof}
The first step of sampling $S_0$ takes $O(\epsilon^{-2} k\log(k))$-time.

Then, by Theorem~\ref{thm:fast_WBSP} with $|D|=O(\epsilon^{-2} k\log(k))$, $ d=k$, we have that the running time of \textsc{RandBSS+} is
\begin{align*}
   &~O\left(k^2 \cdot \epsilon^{-2} k\log(k)+ \eps^{-2}k^3 \log \left(\epsilon^{-2} k\log(k)\right) + \eps^{-2} k^{\omega+1}\right) \\
   = &~O\left(\epsilon^{-2}k^{\omega+1}\right).  
\end{align*}

Hence, the total running time is $O\left(\epsilon^{-2}k^{\omega+1}\right)$. %
\end{proof}
\begin{claim}[Running time of Procedure \textsc{DistillDiscHD} in Algorithm~\ref{algo:distill_dis}]\label{cla:time_important_sampling_k_faster_discrete_hd}
Procedure \textsc{DistillDiscHD} in Algorithm \ref{algo:distill_dis} runs in $$O\left(\epsilon^{-2}k^{\omega+1} +  \epsilon^{-2}dk^{\omega-1}\log k\right) $$ 
time. 
\end{claim}
\begin{proof}
The first step of sampling $S_0$ takes $O(\epsilon^{-2} k\log(k) d)$-time.

Then, we need to implement the function family
\begin{align*}
    {\mathcal F} = \{f(t)=\sum_{j=1}^k v_j \exp(2\pi\i \langle f_j, t \rangle / p) | v_j\in\C \}.
\end{align*}
Naively, for each $f\in {\cal F}$, it takes $O(d)$-time per evaluation. We observe that in the distribution sent to \textsc{RandBSS+} is $\mathrm{Uniform}(S_0)$, which is discrete with support size $s_0=|S_0|$. And in Procedure \text{RandBSS+}, we only need to find an orthonormal basis for ${\cal F}$ with respect to this distribution, which is equivalent to orthogonalize the columns of the matrix defined at Line~\ref{ln:construct_F}. To compute the matrix ${\cal F}$, we need to multiply an $s_0$-by-$d$ matrix with a $d$-by-$k$ matrix. By fast matrix multiplication, by Fact~\ref{fac:matrix_multiplication}, it takes 
\begin{align*}
    {\cal T}_\mathrm{mat}(k,d,s_0) = \begin{cases}
    O(\epsilon^{-2}k^\omega\log k) & \text{if}~d\leq k,\\
    O(\epsilon^{-2}dk^{\omega-1}\log k) & \text{if}~d> k.
    \end{cases}
\end{align*}

For Procedure \textsc{RandBSS+}, by Theorem~\ref{thm:fast_WBSP} with $|D|=O(\epsilon^{-2} k\log(k))$, $ d=k$, we have that the running time of \textsc{RandBSS+} is
\begin{align*}
   &~O\left(k^2 \cdot \epsilon^{-2} k\log(k)+ \eps^{-2}k^3 \log \left(\epsilon^{-2} k\log(k)\right) + \eps^{-2} k^{\omega+1}\right) \\
   = &~O\left(\epsilon^{-2}k^{\omega+1}\right).  
\end{align*}

Hence, the total running time of the procedure is
\begin{align*}
    O\left(\epsilon^{-2} d k\log(k) +\epsilon^{-2}k^{\omega+1} + {\cal T}_\mathrm{mat}(k,d,s_0) \right) = O\left(\epsilon^{-2}k^{\omega+1} +  \epsilon^{-2}dk^{\omega-1}\log k\right).
\end{align*}
\end{proof}

%% file: oned.tex
\section{One-dimensional Signal Estimation}\label{sec:1d-reduction}
In this section, we apply the tools developed in previous sections to show two efficient reductions from Frequency Estimation to Signal Estimation for one-dimensional semi-continuous Fourier signals. The first reduction in Section~\ref{sec:reduction_1d_sample} is optimal in sample complexity, which takes linear number of samples from the signal but only achieves constant accuracy. The section reduction in Section~\ref{sec:reduction_1d_acc} takes nearly-linear number of samples but can achieve very high-accuracy (i.e., $(1+\eps)$-estimation error).

\subsection{Sample-optimal reduction}\label{sec:reduction_1d_sample}
The main theorem of this section is Theorem~\ref{thm:reduction_freq_signal_1d}. The optimal sample complexity is achieved via the sketch distillation in Lemma~\ref{lem:important_sampling_k_faster}.

\begin{theorem}[Sample-optimal algorithm for one-dimensional Signal Estimation]\label{thm:reduction_freq_signal_1d}
For $\eta\in\R$, let $\Lambda(\mathcal{B}) \subset \R$ denote the lattice $\Lambda(\mathcal{B}) = \{ c \eta~|~ c \in \mathbb{Z} \}$.
Suppose that $f_1, f_2, \cdots, f_k \in \Lambda(\mathcal{B})$. Let $x^*(t) = \sum_{j=1}^k v_j \exp({2\pi\i  f_j  t  })$, and let $g(t)$ denote the noise. Given observations of the form $x(t) = x^*(t) + g(t)$, $t\in [0, T]$. Let $\eta=\min_{i \neq j}|f_j-f_i|$.

Given $D, \eta \in \R_+$. %
Suppose that there is an algorithm $\textsc{FreqEst}$ that 
\begin{itemize}
    \item takes $\mathcal{S}_\mathsf{freq}$ samples,
    \item runs in $\mathcal{T}_\mathsf{freq}$-time, and
    \item outputs a set ${\cal L}$ of frequencies such that with probability $0.99$, the following condition holds:
    \begin{align*}
        \forall i\in [k],~\exists f'_i\in {\cal L}~\text{s.t.}~|f_i-f'_i|\le \frac{D}{T}.
    \end{align*}%
\end{itemize}
Then, there is an algorithm (Algorithm \ref{algo:sig_est_1d_slow}) such that
\begin{itemize}
    \item takes $O(\wt{k}+\mathcal{S}_\mathsf{freq})$ samples
    \item runs $O(\wt{k}^{\omega + 1}+\mathcal{T}_\mathsf{freq})$ time,
    \item outputs $y(t) = \sum_{j=1}^{\wt{k}} v_j' \cdot \exp(2\pi \i f_j' t)$ with $\wt{k}=O(|{\cal L}|(1+D/(T\eta)))$ such that with probability at least $0.9$, we have
    \begin{align*}
        \|y(t) - x(t)\|_T^2 \lesssim  \|g(t)\|_T^2 .
    \end{align*}
\end{itemize}  
\end{theorem}

\begin{algorithm}[ht]\caption{Signal estimation algorithm for one-dimensional signals (sample optimal version)} \label{algo:sig_est_1d_slow}
\begin{algorithmic}[1]
\Procedure{\textsc{SignalEstimationFast}}{$x, k, F, T, \mathcal{B}$} \Comment{Theorem~\ref{thm:reduction_freq_signal_1d}}
\State $\eps\gets 0.01$
\State $L \leftarrow\textsc{FreqEst}(x,  k, D, F, T, \mathcal{B})$
\State $\{f'_1,f'_2,\cdots,f'_{\wt{k}}\} \leftarrow \{f\in \Lambda(\mathcal{B})~|~\exists f' \in L, ~ |f'-f|< D/T\}$
\State $s, \{t_1,t_2,\cdots, t_s\},w  \leftarrow\textsc{FastDistill1D}(\wt{k}, \sqrt{\eps}, \{f'_i\}_{i\in[\wt{k}]}, T, \mathcal{B})$\label{ln:sig_est_distill} \Comment{$\wt{k}$, $w \in \R^{\wt{k} }$, Algorithm~\ref{algo:distillation_1d_fast}}
\State $A_{i, j} \leftarrow\exp(2\pi\i f'_j t_i)$, $A\in \C^{s \times \wt{k}}$
\State $b  \leftarrow (x(t_1),x(t_2),\cdots, x(t_s))^\top$
\State Solving the following weighted linear regression\label{ln:sig_est_regress}\Comment{Fact~\ref{fac:basic_l2_regression}}
\begin{align*}
v' \gets \underset{v' \in \C^{ \wt{k} } }{\arg\min}\| \sqrt{w} \circ (A v'- b)\|_2.  
\end{align*}
\State \Return $y(t) = \sum_{j=1}^{\wt{k}} v_j' \cdot \exp(2\pi \i f_j' t)$.
\EndProcedure
\end{algorithmic}
\end{algorithm}

\label{sec:preserve_norm_noise_faster}

\begin{proof}
First, we recover the frequencies by utilizing the algorithm $\textsc{FreqEst}$. Let $L$ be the set of frequencies output by the algorithm $\textsc{FreqEst}(x, k, D,T, F, \mathcal{B})$. 

We define $\wt{L}$ as follows:  
\begin{align*}
\wt{L}:=\left\{\wt{f}\in \Lambda(\mathcal{B})~|~\exists f' \in L, ~ |f'-\wt{f}|< D/T\right\}.
\end{align*}
We use $\wt{k}$ to denote the size of set $\wt{L}$. And we use $\wt{f}_1,\wt{f}_2,\cdots,\wt{f}_{\wt{k}}$ to denote the frequencies in the set $\wt{L}$. It is easy to see that
\begin{align*}
    \wt{k}\leq |{\cal L}|(1+D/(T\eta)).
\end{align*}

Next, we focus on recovering magnitude $v' \in \C^{\wt{k}}$.  
First we run Procedure $\textsc{FastDistill1D}$ in Algorithm~\ref{algo:distillation_1d_fast}  and obtain a set $S=\{t_1,t_2,\cdots, t_s\}\subset [0,T]$ of size  $s=O(\wt{k})$ and a weight vector $w\in \R_{>0}^s$. 
Then, we sample the signal at $t_1,\dots,t_s$ and let $x(t_1),\dots,x(t_s)$ be the samples. Consider the following weighted linear regression problem:
\begin{align}
\underset{v' \in \C^{\wt{k}} }{\min}~\left\|\sqrt{w} \circ (A v' - b )\right\|_2,  \label{eq:liner_regression_1d_sample_optimal}
\end{align}
where $\sqrt{w}:=(\sqrt{w_1},\dots,\sqrt{w_s})$, and the coefficients matrix $A\in \C^{s\times \wt{k}}$ and the target vector $b\in \C^s$ are defined as follows:
\begin{align*}
    A:=\begin{bmatrix}
    \exp(2\pi\i \wt{f}_1 t_1) & \exp(2\pi\i \wt{f}_2 t_1) & \cdots & \exp(2\pi\i \wt{f}_{\wt{k}} t_1)\\
    \exp(2\pi\i \wt{f}_1 t_2) & \exp(2\pi\i \wt{f}_2 t_2) & \cdots & \exp(2\pi\i \wt{f}_{\wt{k}} t_2)\\
    \vdots & \vdots & \ddots & \vdots\\
    \exp(2\pi\i \wt{f}_1 t_s) & \exp(2\pi\i \wt{f}_2 t_s) & \cdots & \exp(2\pi\i \wt{f}_{\wt{k}} t_s)
    \end{bmatrix}~\text{and}~b:=\begin{bmatrix}
    x(t_1)\\ x(t_2) \\ \vdots \\ x(t_s)
    \end{bmatrix}
\end{align*}

Then, we output a signal
\begin{align*}
    y(t) = \sum_{j=1}^{\wt{k}} v_j' \cdot \exp(2\pi \i \wt{f}_j t),
\end{align*}
where $v'$ is an optimal solution of Eq.~\eqref{eq:liner_regression_1d_sample_optimal}.

The running time follows from Lemma~\ref{lem:sig_est_1d_runtime}. And the estimation error guarantee $\|y(t)-x(t)\|_T\lesssim \|g(t)\|_T$ follows from Lemma~\ref{lem:sig_est_1d_error}.

The theorem is then proved.
\end{proof}

\begin{lemma}[Running time of Algorithm~\ref{algo:sig_est_1d_slow}]\label{lem:sig_est_1d_runtime}
Algorithm~\ref{algo:sig_est_1d_slow} takes $O(\wt{k}^{\omega+1})$-time, giving the output of Procedure \textsc{FreqEst}.
\end{lemma}
\begin{proof}
At Line~\ref{ln:sig_est_distill}, we run Procedure \textsc{FastDistill1D}, which takes $O(\wt{k}^{\omega+1})$-time by Lemma~\ref{lem:important_sampling_k_faster}.

At Line~\ref{ln:sig_est_regress}, we solve the weighted linear regression, which takes
\begin{align*}
    O(s\wt{k}^{\omega-1})=O(\wt{k}^\omega)
\end{align*}
time by Fact~\ref{fac:basic_l2_regression}.

Thus, the total running time is $O(\wt{k}^{\omega+1})$.
\end{proof}

\begin{lemma}[Estimation error of Algorithm~\ref{algo:sig_est_1d_slow}]\label{lem:sig_est_1d_error}
Let $y(t)$ be the output signal of Algorithm~\ref{algo:sig_est_1d_slow}. With high probability, we have
\begin{align*}
    \|y(t)-x(t)\|_T\lesssim \|g(t)\|_T.
\end{align*}
\end{lemma}
\begin{proof}
We have
\begin{align}
\|y(t)-x(t)\|_T\leq &~ \|y(t)-x^*(t)\|_{T} + \|g(t)\|_{T} \notag \\
\leq &~ (1+\eps)\|y(t)-x^*(t)\|_{S,w} + \|g(t)\|_{T} \notag \\
\leq &~ (1+\eps)\|y(t)-x(t)\|_{S,w} +(1+\eps)\|g(t)\|_{S,w} + \|g(t)\|_{T} \notag \\
\leq &~ (1+\eps)\|x^*(t)-x(t)\|_{S,w} + (1+\eps)\|g(t)\|_{S,w} + \|g(t)\|_{T} \notag\\
\lesssim &~ \|x^*(t)-x(t)\|_{S,w}+\|g(t)\|_{T} \notag \\
\lesssim &~ \|x^*(t)-x(t)\|_T+\|g(t)\|_{T} \notag \\
\lesssim &~ \|g(t)\|_T,
\end{align}
where the first step follows from triangle inequality, the second step follows from Lemma \ref{lem:important_sampling_k_faster} with $0.99$ probability, the third step follows from triangle inequality, the forth step follows from $y(t)$ is the optimal solution of the linear system, the fifth step follows from Claim~\ref{clm:preserve_noise_1d}, the sixth step follows from Lemma \ref{lem:important_sampling_k_faster}, and the last step follows from the definition of $g(t)$. 
\end{proof}
\subsection{High-accuracy reduction}\label{sec:reduction_1d_acc}
In this section, we prove Theorem~\ref{thm:reduction_freq_signal_1d_clever}, which achieves $(1+\epsilon)$-estimation error by a sharper bound on the energy of noise in Lemma~\ref{lem:guarantee_dist}.
\begin{theorem}[High-accuracy algorithm for one-dimensional Signal Estimation]\label{thm:reduction_freq_signal_1d_clever}
For $\eta\in\R$, let $\Lambda(\mathcal{B}) \subset \R$ denote the lattice $\Lambda(\mathcal{B}) = \{ c \eta~|~ c \in \mathbb{Z} \}$.
Suppose that $f_1, f_2, \cdots, f_k \in \Lambda(\mathcal{B})$. Let $x^*(t) = \sum_{j=1}^k v_j \exp({2\pi\i  f_j  t  })$, and let $g(t)$ denote the noise. Given observations of the form $x(t) = x^*(t) + g(t)$, $t\in [0, T]$. Let $\eta=\min_{i \neq j}|f_j-f_i|$.

Given $D, \eta \in \R_+$. Suppose that there is an algorithm $\textsc{FreqEst}$ that 
\begin{itemize}
    \item takes $\mathcal{S}_\mathsf{freq}$ samples,
    \item runs in $\mathcal{T}_\mathsf{freq}$-time, and
    \item outputs a set ${\cal L}$ of frequencies such that, for each $f_i$, there exists an $f'_i\in {\cal L}$ with $|f_i-f'_i|\le D/T$, holds with probability $0.99$. 
\end{itemize}
Then, there is an algorithm (Algorithm \ref{algo:sig_est_1d_accuracy}) such that
\begin{itemize}
    \item takes $O(\eps^{-1}\wt{k}\log(\wt{k})+\mathcal{S})$ samples, %
    \item runs $O(\eps^{-1}\wt{k}^{\omega}\log(\wt{k})+\mathcal{T})$ time, %
    \item outputs $y(t) = \sum_{j=1}^{\wt{k}} v_j' \cdot \exp(2\pi \i f_j' t)$ with $\wt{k}=O(|{\cal L}|(1+D/(T\eta)))$ such that with probability at least $0.9$, %
    we have
    \begin{align*}
        \|y(t) - x^*(t)\|_T^2  \leq (1+\eps)\|g(t)\|_T^2 .
    \end{align*}
\end{itemize}  
\end{theorem}

\begin{remark}
For simplicity, we state the constant failure probability. It is straightforward to get failure probability $\rho$ by blowing up a $\log(1/\rho)$ factor in both samples and running time.
\end{remark}

\begin{proof}
Let $L$ be the set of frequencies output by the Frequency Estimation algorithm $\textsc{FreqEst}$. We have the guarantee that with probability 0.99, for each true frequency $f_i$, there exists an $f'_i\in {\cal L}$ with $|f_i-f'_i|\le D/T$. Conditioning on this event, we define a set $\wt{L}$ as follows:  
\begin{align*}
\wt{L}:=\{f\in \Lambda(\mathcal{B})~|~\exists f' \in L, ~ |f'-f|< D/T\}.
\end{align*}
Since we assume that $\{f_1,\dots, f_k\}\subset \Lambda({\cal B})$, we have $\{f_1,\dots, f_k\}\subset \wt{L}$.
We use $\wt{k}$ to denote the size of set $\wt{L}$, and we denote the frequencies in $\wt{L}$ by $\wt{f}_1,\wt{f}_2,\cdots,\wt{f}_{\wt{k}}$. 

Next, we need to recover magnitude $v' \in \C^{\wt{k}}$.  

We first run Procedure $\textsc{WeightedSketch}$ in Algorithm \ref{algo:distillation_1d_fast} and obtain a set $S=\{t_1,t_2,\cdots, t_s\}\subset [0,T]$ of size  $s=O(\eps^{-2}\wt{k}\log(\wt{k}))$ and a weight vector $w\in \R_{>0}^s$. 
Then, we sample the signal at $t_1,\dots,t_s$ and let $x(t_1),\dots,x(t_s)$ be the samples. Consider the following weighted linear regression problem:
\begin{align}
\underset{v' \in \C^{\wt{k}} }{\min}~\left\|\sqrt{w} \circ (A v' - b )\right\|_2,  \label{eq:liner_regression_1d_high_acc}
\end{align}
where $\sqrt{w}:=(\sqrt{w_1},\dots,\sqrt{w_s})$, and the coefficients matrix $A\in \C^{s\times \wt{k}}$ and the target vector $b\in \C^s$ are defined as follows:
\begin{align*}
    A:=\begin{bmatrix}
    \exp(2\pi\i \wt{f}_1 t_1) & \exp(2\pi\i \wt{f}_2 t_1) & \cdots & \exp(2\pi\i \wt{f}_{\wt{k}} t_1)\\
    \exp(2\pi\i \wt{f}_1 t_2) & \exp(2\pi\i \wt{f}_2 t_2) & \cdots & \exp(2\pi\i \wt{f}_{\wt{k}} t_2)\\
    \vdots & \vdots & \ddots & \vdots\\
    \exp(2\pi\i \wt{f}_1 t_s) & \exp(2\pi\i \wt{f}_2 t_s) & \cdots & \exp(2\pi\i \wt{f}_{\wt{k}} t_s)
    \end{bmatrix}~\text{and}~b:=\begin{bmatrix}
    x(t_1)\\ x(t_2) \\ \vdots \\ x(t_s)
    \end{bmatrix}
\end{align*}
Note that if $v'$ corresponds to the true coefficients $v$, then we have $\|\sqrt{w}\circ (Av'-b)\|_2=\|\sqrt{w}\circ g(S)\|_2=\|g\|_{S, w}$. Let $v'$ be the exact solution of the weighted linear regression in Eq.~\eqref{eq:liner_regression_1d_high_acc}, i.e.,
\begin{align*}
    v':=\arg\min_{v'\in \C^{\wt{k}}}~\left\|\sqrt{w} \circ (A v' - b )\right\|.
\end{align*}
And we define the output signal to be:
$$y(t) := \sum_{j=1}^{\wt{k}} v_j' \cdot \exp(2\pi \i f_j' t).$$

The estimation error guarantee $\|y(t)-x^*(t)\|_T\leq (1+\eps) \|g(t)\|_T$ follows from Lemma~\ref{lem:error_eps}. The running time follows  from Lemma~\ref{lem:sig_est_1d_acc_runtime}.

The theorem is then proved.
\end{proof}

\begin{algorithm}[ht]\caption{Signal estimation algorithm for one-dimensional signals (high-accuracy version)} \label{algo:sig_est_1d_accuracy}
\begin{algorithmic}[1]
\Procedure{\textsc{SignalEstimationAcc}}{$x, \eps, k, F, T, \mathcal{B}$} \Comment{Theorem~\ref{thm:reduction_freq_signal_1d_clever}}
\State $L \leftarrow\textsc{FreqEst}(x,  k, D, F, T, \mathcal{B})$
\State $\{f'_1,f'_2,\cdots,f'_{\wt{k}}\} \leftarrow \{f\in \Lambda(\mathcal{B})~|~\exists f' \in L, ~ |f'-f|< D/T\}$
\State $s, \{t_1,t_2,\cdots, t_s\},w  \leftarrow\textsc{WeightedSketch}(\wt{k}, \sqrt{\eps}, T, \mathcal{B})$ \Comment{$\wt{k}$, $w \in \R^{\wt{k} }$, Algorithm~\ref{algo:distillation_1d_fast}}
\State $A_{i, j} \leftarrow\exp(2\pi\i f'_j t_i)$, $A\in \C^{s \times \wt{k}}$
\State $b  \leftarrow (x(t_1),x(t_2),\cdots, x(t_s))^\top$
\State Solving the following weighted linear regression\label{ln:sig_est_regress_fast}\Comment{Fact~\ref{fac:basic_l2_regression}}
\begin{align*}
v' \gets \underset{v' \in \C^{ \wt{k} } }{\arg\min}\| \sqrt{w} \circ (A v'- b)\|_2.  
\end{align*}
\State \Return $y(t) = \sum_{j=1}^{\wt{k}} v_j' \cdot \exp(2\pi \i f_j' t)$.
\EndProcedure
\end{algorithmic}
\end{algorithm}

\begin{lemma}[Running time of Algorithm~\ref{algo:sig_est_1d_accuracy}]\label{lem:sig_est_1d_acc_runtime}
Algorithm~\ref{algo:sig_est_1d_accuracy} takes $O(\epsilon^{-1}\wt{k}^\omega\log(\wt{k}))$-time, giving the output of Procedure \textsc{FreqEst}.
\end{lemma}
\begin{proof}
At Line~\ref{ln:sig_est_regress_fast}, the regression solver takes
\begin{align*}
    O(s\wt{k}^{\omega-1})=O(\epsilon^{-1}\wt{k}\log(\wt{k})\cdot \wt{k}^{\omega-1}) = O(\epsilon^{-1}\wt{k}^\omega\log(\wt{k}))
\end{align*}
time. The remaining part of Algorithm~\ref{algo:sig_est_1d_accuracy} takes at most $O(s)$-time.
\end{proof}

\begin{lemma}[Estimation error of Algorithm~\ref{algo:sig_est_1d_accuracy}]\label{lem:error_eps}
Let $y(t)$ be the output signal of Algorithm~\ref{algo:sig_est_1d_accuracy}. With high probability, we have
\begin{align*}
    \|y(t)-x^*(t)\|_T\leq (1+\eps) \|g(t)\|_T.
\end{align*}
\end{lemma}

\begin{proof}
Let $\cal F$ be the family of signals with frequencies in $\wt{L}$:
\begin{align*}
{\cal F}= \Big\{ h(t)=\sum_{j=1}^{\wt{k}} v_j \cdot e^{2 \pi \i \wt{f}_j t} ~\big|~ \forall v_j \in \mathbb{C}, j\in [\wt{k}]  \Big\}.
\end{align*}
Suppose the dimension of ${\cal F}$ is $m\leq k$. Let $\{u_1,u_2,\cdots,u_m\}$ be an orthonormal basis of ${\cal F}$, i.e.,
\begin{align*}
      \frac{1}{T}\int_{[0, T]} u_i(t)\overline{u_j(t)}\d t =&~{\bf 1}_{i=j}, \quad \forall i,j\in[m],
\end{align*}
On the other hand, since $u_i\in {\cal F}$, we can also expand these basis vectors in the Fourier basis. Let $V\in \C^{m\times {\wt{k}}}$ be an linear transformation\footnote{When $m<\wt{k}$, $V$ is not unique, and we take any one of such linear transformation.} such that
\begin{align*}
u_i=\sum_{j=1}^{\wt{k}} V_{i,j}\cdot \exp(2\pi\i \wt{f}_j t)~~~\forall i\in [m].
\end{align*}
Then, we have
\begin{align*}
    \begin{bmatrix}
    \exp(2\pi\i \wt{f}_1 t)\\
    \vdots\\
    \exp(2\pi\i \wt{f}_{\wt{k}} t)
    \end{bmatrix} = V^+ \cdot \begin{bmatrix}
    u_1\\
    \vdots\\
    u_m
    \end{bmatrix},
\end{align*}
where $V^+ \in \C^{\wt{k}\times m}$ is the pseudoinverse of $V$; or equivalently, the $i$-th row of $V^+$ contains the coefficients of expanding $\exp(2\pi\i \wt{f}_i t)$ under $\{u_1,\dots,u_m\}$. Define a linear operator $\alpha:{\cal F}\rightarrow \C^m$ such that for any $h(t) = \sum_{j=1}^{\wt{k}} v_j\exp(2\pi\i f_j t)$,
\begin{align*}
    \alpha(h):=V^+ \cdot v,
\end{align*}
which gives the coefficients of $h$ under the basis $\{u_1, \cdots, u_{\wt{k}}\}$.

Define an $s$-by-$m$ matrix $B$ as follows:
\begin{align*}
    B:=A\cdot V^\top = \begin{bmatrix}
    u_1(t_1) & u_2(t_1) & \cdots & u_m(t_1)\\
    u_1(t_2) & u_2(t_2) & \cdots & u_m(t_2)\\
    \vdots & \vdots & \ddots & \vdots\\
    u_1(t_s) & u_2(t_s) & \cdots & u_m(t_s)
    \end{bmatrix}.
\end{align*}
$B=AV$. It is easy to see that $\mathrm{Im}(B)=\mathrm{Im}(A)$. Thus, solving Eq.~\eqref{eq:liner_regression_1d_high_acc} is equivalent to solving:
\begin{align}
\underset{z \in \C^{m} }{\min}\|\sqrt{w} \circ (B z - b )\|_2. \label{eq:change_linear_reg_to_another_linear_reg}
\end{align}
Since $y(t)$ is an solution of Eq.~\eqref{eq:liner_regression_1d_high_acc}, we also know that $\alpha(y)$ is an solution of Eq.~\eqref{eq:change_linear_reg_to_another_linear_reg}.

For convenience, we define some notations. Let $\sqrt{W}:=\mathrm{diag}(\sqrt{w})$ and define 
\begin{align*}
    B_w:=&~ \sqrt{W} \cdot B,\\
    X_w:=&~ \sqrt{W}\cdot \begin{bmatrix} x(t_1)&  x(t_2)& \cdots& x(t_s)\end{bmatrix}^\top\\
    X^*_w:=&~ \sqrt{W}\cdot  \begin{bmatrix} x^*(t_1)&  x^*(t_2)& \cdots&   x^*(t_s)\end{bmatrix}^\top
\end{align*}
By Fact~\ref{fac:basic_l2_regression}, we know that the solution of the weighted linear regression Eq.~\eqref{eq:change_linear_reg_to_another_linear_reg} has the following closed form:
\begin{align}\label{eq:alpha_y_closed_form}
    \alpha(y)=(B^* W B)^{-1}B^* W b = (B_w^* B_w)^{-1} B_w^* X_w.
\end{align}
Then, consider the noise in the signal. Since $g$ is an arbitrary noise, let $g^{\parallel}$ be the projection of $g(x)$ to ${\cal F}$ and $g^{\bot}=g - g^{\parallel}$ be the orthogonal part to ${\cal F}$ such that
\begin{align*}
g^{\parallel}(t) \in {\cal F},~\text{and}~
\int_{[0, T]} g^{\parallel} (t)\overline{g^{\bot} (t)} \d t=0.
\end{align*}
Similarly, we also define
\begin{align*}
    g_w:=&~ \sqrt{W}\cdot  \begin{bmatrix} g(t_1)& g(t_2)& \cdots&   g(t_s)\end{bmatrix}^\top\\
    g^{\parallel}_w := &~ \sqrt{W}\cdot \begin{bmatrix}g^{\parallel}(t_1)&  g^{\parallel}(t_2)& \cdots, g^{\parallel}(t_s)\end{bmatrix}^\top,\\
    g^{\bot}_w := &~ \sqrt{W}\cdot \begin{bmatrix}g^{\bot}(t_1)&  g^{\bot}(t_2)& \cdots, g^{\bot}(t_s)\end{bmatrix}^\top.
\end{align*}

By Claim~\ref{clm:error_decompose}, the error can be decomposed into two terms:
\begin{align*}
    \|y(t)-x^*(t)\|_{T} \leq \left\| (B^*_wB_w)^{-1} B^*_w \cdot g^{\bot}_w\right\|_2 + \left\| (B^*_wB_w)^{-1} B^*_w \cdot g^{\parallel}_w \right\|_2.
\end{align*}
By Claim~\ref{clm:bound_first}, we have
\begin{align*}
    \left\| (B^*_wB_w)^{-1} B^*_w \cdot g^{\bot}_w\right\|_2^2 \lesssim&~ \eps \left\|g^{\bot}(t) \right\|_{T}^2.
\end{align*}
And by Claim~\ref{clm:bound_second}, we have
\begin{align*}
    \left\| (B^*_wB_w)^{-1} B^*_w \cdot g^{\parallel}_w \right\|_2^2=\left\| g^{\parallel} \right\|_{T}^2.
\end{align*}
Combining them together (and re-scaling $\epsilon$ be an constant factor), we have that
\begin{align*}
\|y(t)-x^*(t)\|_{T} \leq \| g^{\parallel} \|_{T}+\sqrt{\eps} \|g^{\bot} \|_{T} .
\end{align*}
Since $\| g^{\parallel} \|_{T}^2+ \|g^{\bot} \|_{T}^2 = \|g\|_{T}^2$, by Cauchy–Schwarz inequality, we have that
\begin{align*}
(\| g^{\parallel} \|_{T}+ \sqrt{\eps} \|g^{\bot} \|_{T})^2 \leq (\| g^{\parallel} \|_{T}^2 + \|g^{\bot} \|_{T}^2)\cdot (1 + \eps) = (1 + \eps) \cdot \|g\|_{T}^2.
\end{align*}
That is,
\begin{align*}
\|y(t)-x^*(t)\|^2_{T} \leq   (1+\eps) \|g(t)\|_{T}^2 .
\end{align*}

\end{proof}

\begin{claim}[Error decomposition]\label{clm:error_decompose}
\begin{align*}
    \|y(t)-x^*(t)\|_{T} \leq \left\| (B^*_wB_w)^{-1} B^*_w \cdot g^{\bot}_w\right\|_2 + \left\| (B^*_wB_w)^{-1} B^*_w \cdot g^{\parallel}_w \right\|_2.
\end{align*}
\end{claim}
\begin{proof}
Since $y,x^*\in {\cal F}$ and $\{u_1,\dots,u_{\wt{k}}\}$ is an orthonormal basis, we have $\|y-x^*\|_T = \|\alpha(y)-\alpha(x^*)\|_2$. Furthermore, by Eq.~\eqref{eq:alpha_y_closed_form}, we have $\alpha(y)=(B^*_wB_w)^{-1} B^*_w \cdot X_w$. And by Fact~\ref{fac:alpha_h_to_h}, since $x^*\in {\cal F}$, we have $\alpha(x^*)=(B^*_wB_w)^{-1} B^*_w \cdot X_w^*$. 

Thus, we have
\begin{align*}
\|\alpha(y)-\alpha(x^*)\|_2 = &~ \| (B^*_wB_w)^{-1} B^*_w \cdot (X_w - X_w^*) \|_2 \\
=&~\| (B^*_wB_w)^{-1} B^*_w \cdot g_w \|_2 \\
=&~\| (B^*_wB_w)^{-1} B^*_w \cdot (g^{\bot}_w+g^{\parallel}_w) \|_2 \\
\leq &~\| (B^*_wB_w)^{-1} B^*_w \cdot g^{\bot}_w\|_2 + \| (B^*_wB_w)^{-1} B^*_w \cdot g^{\parallel}_w \|_2 
\end{align*}
where the second step follows from the definition of $g_w$, the forth step follows from $g_w=g^\parallel + g^\bot$, and the last step follows from triangle inequality. 

Hence, we get that $\|y(t)-x^*(t)\|_{T} \leq \| (B^*_wB_w)^{-1} B^*_w \cdot g^{\bot}_w\|_2 + \| (B^*_wB_w)^{-1} B^*_w \cdot g^{\parallel}_w \|_2$.
\end{proof}

\begin{fact}\label{fac:alpha_h_to_h}
For any $h\in {\cal F}$, 
\begin{align*}
    \alpha(h)=(B^*_wB_w)^{-1} B^*_w \cdot h_w, %
\end{align*}
where $h_w=\sqrt{W}\begin{bmatrix}h(t_1) & \cdots & h(t_s)\end{bmatrix}^\top$.
\end{fact}
\begin{proof}
Suppose $h(t)=\sum_{j=1}^{\wt{k}} v_j \exp(2\pi\i \wt{f}_j t)$. We have
\begin{align*}
    B_w\alpha(h)=&~ \sqrt{W}B \cdot \alpha(h)\\
    = &~ \sqrt{W}B \cdot (V^+ v)\\
    = &~ h_w,
\end{align*}
where the second step follows from $V^+$ is a change of coordinates.

Hence, by the Moore-Penrose inverse, we have
\begin{align*}
    \alpha(h) = B_w^{\dagger} h_w = (B_w^* B_w)^{-1}B_w^* h_w.
\end{align*}
\end{proof}

\begin{claim}[Bound the first term]\label{clm:bound_first}
The following holds with high probability:
\begin{align*}
    \left\| (B^*_wB_w)^{-1} B^*_w \cdot g^{\bot}_w\right\|_2^2 \lesssim&~ \eps \left\|g^{\bot}(t) \right\|_{T}^2.
\end{align*}
\end{claim}
\begin{proof}
By Lemma~\ref{lem:concentration_for_any_polynomial_signal_improve_by_CP19}, with high probability, we have
\begin{align*}
    (1-\epsilon)\| x\|_T \leq \| x\|_{S,w} \leq (1 + \epsilon)\| x\|_T,
\end{align*}
where $(S,w)$ is the output of Procedure \textsc{WeightedSketch}. Conditioned on this event, by Lemma \ref{lem:operator_estimation}, 
\begin{align*}
    \lambda(B_w^* B_w)\in [1-\epsilon, 1+\epsilon],
\end{align*}
since $B_w$ is the same as the matrix $A$ in the lemma. 

Hence,
\begin{align*}
\| (B^*_wB_w)^{-1} B^*_w \cdot g^{\bot}_w\|_2^2 \leq&~ \lambda_{\max}((B^*_wB_w)^{-1})^2 \cdot \|B^*_w \cdot g^{\bot}_w \|_2^2\\
\leq&~  (1-\eps)^{-2}\|B^*_w \cdot g^{\bot}_w \|_2^2\\
\lesssim&~ \eps \|g^{\bot}(t) \|_{T}^2
\end{align*}
where the second step follows from $\lambda_{\max}((B_w^*B_w)^{-1})\leq (1-\epsilon)^{-1}$, and the third step follows from Lemma \ref{lem:guarantee_dist} and Corollary~\ref{cor:weightedsketch_is_WBSP}. 
\end{proof}

\begin{lemma}[Lemma 6.2 of \cite{cp19_colt}]\label{lem:agnostic_learning_single_distribution}
There exists a universal constant $C_1$ such that given any distribution $D'$ with the same support of $D$ and any $\epsilon>0$, the random sampling procedure with $m=C_1(K_{D'} \log d + \eps^{-1} K_{D'} )$ i.i.d.~random samples from $D'$ and coefficients $\alpha_1=\cdots=\alpha_m=1/m$ is an $\eps$-\emph{well-balanced sampling procedure}.
\end{lemma}

\begin{corollary}\label{cor:weightedsketch_is_WBSP}
Procedure $\textsc{WeightedSketch}$ in Algorithm \ref{algo:distillation_1d_fast} is a $\eps$-WBSP (Definition \ref{def:procedure_agnostic_learning}).
\end{corollary}

\begin{claim}[Bound the second term]\label{clm:bound_second}
\begin{align*}
    \left\| (B^*_wB_w)^{-1} B^*_w \cdot g^{\parallel}_w \right\|_2^2=\left\| g^{\parallel} \right\|_{T}^2.
\end{align*}
\end{claim}
\begin{proof}
\begin{align*}
    \| (B^*_wB_w)^{-1} B^*_w \cdot g^{\parallel}_w \|_2^2 = \| \alpha(g^{\parallel}) \|^2_2 =& \| g^{\parallel} \|_{T}^2,
\end{align*}
where the first step follows from Fact~\ref{fac:alpha_h_to_h} and $g^\parallel\in {\cal F}$, the second step follows from the definition of ${\alpha}$. 
\end{proof}

%% file: highd.tex
\section{High-dimensional Signal Estimation}\label{sec:high_d_reduction}

In this section, we show a sample-optimal reduction from Frequency Estimation to Signal Estimation for high-dimensional signals in Section~\ref{sec:red_cont_hd}, which generalize Theorem~\ref{thm:reduction_freq_signal_1d}. The key difference is that in high dimensions, we need to upper-bound the number of lattice points within a $d$-dimensional sphere, which turns out to be related to the output signal's Fourier sparsity, and the results are given in Section~\ref{sec:high_d_sparsity}.

\subsection{Sample-optimal reduction}\label{sec:red_cont_hd}

\begin{theorem}[Sample-optimal algorithm for high dimension Signal Estimation%
]\label{lem:reduction_freq_signal_high_D}
Given a basis $\mathcal{B}$ of $m$ known vectors $b_1, b_2, \cdots b_m \in \R^d$, let $\Lambda(\mathcal{B}) \subset \R^d$ denote the lattice 
\begin{align*}
    \Lambda(\mathcal{B}) = \Big\{ z \in \R^d : z = \sum_{i=1}^m c_i b_i, c_i \in \mathbb{Z}, \forall i \in [m] \Big\}
\end{align*}
Suppose that $f_1, f_2, \cdots, f_k \in \Lambda(\mathcal{B})$. Let $x^*(t) = \sum_{j=1}^k v_j e^{2\pi\i \langle f_j , t \rangle }$ and let $g(t)$ denote the noise. Given observations of the form $x(t) = x^*(t) + g(t)$, $t\in [0, T]^d$. Let $\eta=\min_{i \neq j}\|f_j-f_i\|_{\infty}$.

Given $D, \eta \in \R_+$. Suppose that there is an algorithm $\textsc{FreqEst}$ that 
\begin{itemize}
    \item takes $\mathcal{S}_\mathsf{freq}$ samples, 
    \item runs in $\mathcal{T}_\mathsf{freq}$-time,
    \item outputs a set ${\cal L}$ of frequencies such that with probability $0.99$, the following condition holds:
    \begin{align*}
        \forall i\in [k],~\exists f'_i\in {\cal L}~\text{s.t.}~\|f_i-f'_i\|_2\le \frac{D}{T}.
    \end{align*}
\end{itemize}
 
Then, there is an algorithm that 
\begin{itemize}
    \item takes $O(\wt{k}+\mathcal{S}_\mathsf{freq})$ samples
    \item runs in $O(\wt{k}^{O(d)}+\mathcal{T}_\mathsf{freq})$ time,
\item output $y(t) = \sum_{j=1}^{\wt{k}} v_j' \cdot \exp(2\pi \i \langle f_j', t \rangle)$ with $\wt{k}\leq  |L| \cdot ( D/T + \sqrt{m} \| {\cal B} \| )^m \cdot \frac{ \pi^{m/2} }{ (m/2)! } \cdot \frac{1}{|\det({\cal B})|}$ such that with probability 0.9, we have
 $$\int_{[0,T]^d} |y(t) - x(t)|^2 \d t \lesssim \int_{[0,T]^d} |g(t)|^2 \d t.$$
\end{itemize}
\end{theorem}
\begin{proof}
The algorithm is almost the same as Algorithm~\ref{algo:sig_est_1d_slow}.
First, we recover the frequencies by calling Procedure $\textsc{FreqEst}(x, k, d, T, F, \mathcal{B})$. Let $L$ be the set of frequencies output by the algorithm.

We define $\wt{L}$ as follows: %
\begin{align*}
\wt{L}:=\{f\in \Lambda(\mathcal{B})~|~\exists f' \in L, ~ \|f'-f\|_2< D/T\}.
\end{align*}
We use $\wt{k}$ to denote the size of set $\wt{L}$.
We use $f'_1,f'_2,\cdots,f'_{\wt{k}}$ to denote the frequencies in the set $\wt{L}$.
By applying Lemma \ref{lem:bound_sparsity}, we have that
\begin{align*}
     \wt{k} \leq |L| \cdot ( D/T + \sqrt{m} \| {\cal B} \| )^m \cdot \frac{ \pi^{m/2} }{ (m/2)! } \cdot \frac{1}{|\det({\cal B})|}.
\end{align*}

Next, we focus on recovering magnitude $v' \in \C^{\wt{k}}$.  
We run Procedure $\textsc{DistillHD}$ in Algorithm~\ref{algo:distillation_hd} and obtain a set $S=\{t_1,t_2,\cdots, t_s\}$ of $s=O(\wt{k})$ samples in the duration $[0, T]^d$, and a weight vector $w\in \R^s$.

Then, we consider the following weighted linear regression problem 
\begin{align*}
\underset{v' \in \C^{\wt{k}} }{\min}\| \sqrt{w} \circ (A v' - b )\|_2, 
\end{align*}
where $A\in \C^{s\times \wt{k}}$ and $b\in \C^s$ are defined as follows:
\begin{align*}
    A:=\begin{bmatrix}
    \exp(2\pi\i \langle \wt{f}_1, t_1\rangle) &   \cdots & \exp(2\pi\i \langle \wt{f}_{\wt{k}}, t_1\rangle)\\
    \vdots  & \ddots & \vdots\\
    \exp(2\pi\i \langle \wt{f}_1, t_s\rangle) &  \cdots & \exp(2\pi\i \langle \wt{f}_{\wt{k}}, t_s\rangle)
    \end{bmatrix}~\text{and}~b:=\begin{bmatrix}
    x(t_1) \\ \vdots \\ x(t_s)
    \end{bmatrix}
\end{align*}
Let $v'$ be an optimal solution of the regression and we output the signal $$y(t) := \sum_{j=1}^{\wt{k}} v_j' \cdot \exp(2\pi \i \langle f_j', t \rangle ).$$

Finally, we prove that $\|y(t)-x(t)\|_T\lesssim \|g(t)\|_T$, holds with a large constant probability. 
\begin{align}
\|y(t)-x(t)\|_T\leq &~ \|y(t)-x^*(t)\|_{T} + \|g(t)\|_{T} \notag \\
\leq &~ (1+\eps)\|y(t)-x^*(t)\|_{S,w} + \|g(t)\|_{T} \notag \\
\leq &~ (1+\eps)\|y(t)-x(t)\|_{S,w} +(1+\eps)\|g(t)\|_{S,w} + \|g(t)\|_{T} \notag \\
\leq &~ (1+\eps)\|x^*(t)-x(t)\|_{S,w} + (1+\eps)\|g(t)\|_{S,w} + \|g(t)\|_{T} \notag\\
\lesssim &~ \|x^*(t)-x(t)\|_{S,w}+\|g(t)\|_{T} \notag \\
\lesssim &~ \|x^*(t)-x(t)\|_T+\|g(t)\|_{T} \notag \\
\lesssim &~ \|g(t)\|_T,
\end{align}
where the first step follows from triangle inequality, the second step follows from Lemma \ref{lem:important_sampling_k_high_d} with $0.99$ probability, the third step follows from triangle inequality, the forth step follows from $y(t)$ is the optimal solution of the linear system, the fifth step follows from Claim~\ref{clm:preserve_noise_high_D}, the sixth step follows from Lemma \ref{lem:important_sampling_k_high_d}, and the last step follows from the definition of $g(t)$. 

The running time of the reduction follows from Lemma~\ref{lem:important_sampling_k_high_d}.
\end{proof}

\subsection{Bounding the sparsity}\label{sec:high_d_sparsity}
In this section, we show that the Fourier sparsity of the output signal can be bounded by the number of lattice points within a sphere. The intuition is that for each frequency $f'$ outputted by Procedure \textsc{FreqEst}, there could be $|B_d(f', D/T)\cap \Lambda({\cal B})|$ many candidates of true frequencies, where $B_d(x,r)$ denotes the $d$-dimensional sphere centered at $x$ with radius $r$. In Lemma~\ref{lem:bound_sparsity}, we upper-bound the sparsity for the case when $D/T$ is larger than $\lambda_1(\Lambda({\cal B}))$, the shortest vector length of the lattice. When $D/T$ is small, we show in Lemma~\ref{lem:bound_sparsity_2} that Procedure \textsc{FreqEst} finds all true frequencies.

\begin{lemma}[Bounding sparsity for large $D/T$]\label{lem:bound_sparsity}
Given a basis $\mathcal{B}$ of $m$ known vectors $b_1, b_2, \cdots b_m \in \R^d$, let $\Lambda(\mathcal{B}) \subset \R^d$ denote the lattice 
\begin{align*}
    \Lambda(\mathcal{B}) = \Big\{ z \in \R^d : z = \sum_{i=1}^m c_i b_i, c_i \in \mathbb{Z}, \forall i \in [m] \Big\},
\end{align*}
and let
\begin{align*}
    \wt{k}:=|\{f\in \Lambda(\mathcal{B})~|~\exists f' \in L, ~ \|f'-f\|_2< D/T\}|
\end{align*}
be the output sparsity.  Then, we have

\begin{itemize}
    \item (Spectral bound, which is better when $D/T < O(\|{\cal B}\|)$) 
    \begin{align*}
        \wt{k}\leq |L|\cdot (1+2D/(T\sigma_{\min}({\cal B})))^{m}.
    \end{align*} 
    \item (Volume bound, which is better when $D/T > O(\|{\cal B}\|)$)
    \begin{align*}
         \wt{k}\leq
         |L| \cdot ( D/T + \sqrt{m} \| {\cal B} \| )^m \cdot \frac{ \pi^{m/2} }{ (m/2)! } \cdot \frac{1}{|\det({\cal B})|}.
    \end{align*}
\end{itemize}

\end{lemma}
\begin{proof}

{\bf Spectral bound:}
Let $c=\begin{bmatrix}c_1 &c_2 &\cdots& c_m\end{bmatrix}^\top\in\Z^m$. Then $z={\cal B}c\in\Lambda(\mathcal{B})$, and 
\begin{align*}
    \|z\|_2 = \|{\cal B}c\|_2 \geq \sigma_{\min}({\cal B})\cdot  \|c\|_2
\end{align*}
Then we have that,
\begin{align*}
\wt{k} = &~ |\{f\in \Lambda(\mathcal{B})~|~\exists f' \in L, ~ \|f'-f\|_2< D/T\}|\\
\leq&~|L|\cdot|\{z\in \Lambda(\mathcal{B})~| ~ \|z\|_2< D/T\}| \\
\leq&~|L|\cdot|\{c\in \Z^m ~| ~ \|c\|_2\leq D/(T\sigma_{\min}({\cal B}))\}| \\
\leq&~|L|\cdot|\{c\in \Z^m ~| ~ \|c\|_{\infty} \leq D/(T\sigma_{\min}({\cal B}))\}| \\
\leq&~|L|\cdot (1+ 2D/(T\sigma_{\min}{\cal B}))^{m}.
\end{align*}
where the first step follows from $f'-f\in\Lambda(\mathcal{B})$, the second step follows from if $\|c\|_2 \geq D/(T\sigma_{\min})$, then $\|z\|_2\geq D/T$, the third step follows from $\|c\|_\infty \leq \|c\|_2$, and the last step follows from $c$ is a bounded integer vector. %

\paragraph{Volume bound:}

Using Lemma~\ref{lem:number_of_lattice_points}, we have 
\begin{align}
    \wt{k} \leq |L| \cdot ( 1 + \frac{ \sqrt{m} \| {\cal B} \| }{ D/T }  )^m \cdot \frac{ \mathrm{vol} ( {\cal B}_m(0,D/T) ) }{ \mathrm{vol}( {\cal P}( {\cal B} ) ) }.\label{eq:bound_lattice_in_ball_thomas}
\end{align}

We can upper bound volume of a ball as follows:
\begin{align}
    \mathrm{vol} ( {\cal B}_m(0,D/T) ) \leq \frac{ \pi^{m/2} }{ (m/2)! } \cdot (D/T)^{m} .\label{eq:bound_ball}
\end{align}
Combining the above two equations, we have
\begin{align*}
    \LHS 
    \leq & ~ |L| \cdot ( 1 + \frac{ \sqrt{m} \| {\cal B} \| }{ D/T }  )^m \cdot  \frac{ \pi^{m/2} }{ (m/2)! } \cdot (D/T)^{m} \cdot \frac{ 1 }{ \mathrm{vol}( {\cal P}( {\cal B} ) ) } \\
    \leq & ~ |L| \cdot ( D/T + \sqrt{m} \| {\cal B} \| )^m \cdot \frac{ \pi^{m/2} }{ (m/2)! } \cdot \frac{1}{|\det({\cal B})|}, %
\end{align*}
where the first step follows from Eq.~\eqref{eq:bound_lattice_in_ball_thomas} and Eq.~\eqref{eq:bound_ball}.

\end{proof}

\begin{lemma}[Bounding sparsity for tiny $D/T$]\label{lem:bound_sparsity_2}
Given a basis $\mathcal{B}$ of $m$ known vectors $b_1, b_2, \cdots b_m \in \R^d$, let $\Lambda(\mathcal{B}) \subset \R^d$ denote the lattice 
\begin{align*}
    \Lambda(\mathcal{B}) = \Big\{ z \in \R^d : z = \sum_{i=1}^m c_i b_i, c_i \in \mathbb{Z}, \forall i \in [m] \Big\}
\end{align*}
and let
\begin{align*}
    \wt{k}:=|\{f\in \Lambda(\mathcal{B})~|~\exists f' \in L, ~ \|f'-f\|_2< D/T\}|
\end{align*}
be the output sparsity.  If $D/T\leq \lambda_1(\Lambda({\cal B}))$\footnote{When $m$ is small, we can solve the shortest vector problem (SVP) exactly to decide the sparsity. Otherwise, we can check $D/T< \min_i \|b^*_i\|_2$ by Theorem~\ref{thm:svp_lower_bound}.}, then we have
\begin{align*}
    \wt{k}\leq |L|.
\end{align*}

\end{lemma}
\begin{proof}
Since the radius $D/T$ is at most the shortest vector length of the lattice $\Lambda({\cal B})$, for each $f'\in L$, the sphere $B_d(f', D/T)$ contains at most one lattice point.
\end{proof}

\subsection{High-accuracy reduction} %

\begin{theorem}[High-dimensional Signal Estimation algorithm]\label{lem:reduction_freq_signal_high_D_2}
Given a basis $\mathcal{B}$ of $m$ known vectors $b_1, b_2, \cdots b_m \in \R^d$, let $\Lambda(\mathcal{B}) \subset \R^d$ denote the lattice 
\begin{align*}
    \Lambda(\mathcal{B}) = \Big\{ z \in \R^d : z = \sum_{i=1}^m c_i b_i, c_i \in \mathbb{Z}, \forall i \in [m] \Big\}
\end{align*}
Suppose that $f_1, f_2, \cdots, f_k \in \Lambda(\mathcal{B})$. Let $x^*(t) = \sum_{j=1}^k v_j e^{2\pi\i \langle f_j , t \rangle }$ and let $g(t)$ denote the noise. Given observations of the form $x(t) = x^*(t) + g(t)$, $t\in [0, T]^d$. Let $\eta=\min_{i \neq j}\|f_j-f_i\|_{\infty}$.

Given $D, \eta \in \R_+$. Suppose that there is an algorithm $\textsc{FreqEst}$ that 
\begin{itemize}
    \item takes $\mathcal{S}_\mathsf{freq}$ samples, 
    \item runs in $\mathcal{T}_\mathsf{freq}$-time,
    \item outputs a set ${\cal L}$ of frequencies such that with probability $0.99$, the following condition holds:
    \begin{align*}
        \forall i\in [k],~\exists f'_i\in {\cal L}~\text{s.t.}~\|f_i-f'_i\|_2\le \frac{D}{T}.
    \end{align*}
\end{itemize}
 
Then, there is an algorithm that 
\begin{itemize}
    \item takes $O(\eps^{-1}\wt{k}^{O(d)}+\mathcal{S}_\mathsf{freq})$ samples
    \item runs in $O(\eps^{-1}\wt{k}^{O(d)}+\mathcal{T}_\mathsf{freq})$ time,
\item output $y(t) = \sum_{j=1}^{\wt{k}} v_j' \cdot \exp(2\pi \i \langle f_j', t \rangle)$ with $\wt{k}\leq  |L| \cdot ( D/T + \sqrt{m} \| {\cal B} \| )^m \cdot \frac{ \pi^{m/2} }{ (m/2)! } \cdot \frac{1}{|\det({\cal B})|}$ such that with probability 0.9, we have
 $$\int_{[0,T]^d} |y(t) - x(t)|^2 \d t \leq (1+O(\eps)) \int_{[0,T]^d} |g(t)|^2 \d t.$$
\end{itemize}
\end{theorem}
\begin{remark}
The difference between Theorem \ref{lem:reduction_freq_signal_high_D_2} and Theorem \ref{lem:reduction_freq_signal_high_D} is one is achieving $(1+\eps)$ error and the other is achieving $O(1)$ error. 
\end{remark}
\begin{proof}
We can prove this theorem by using Theorem \ref{sec:hd_energy_bound}. The proof is similar as Theorem \ref{thm:reduction_freq_signal_1d_clever}. 
\end{proof}

%% file: discrete_set_query.tex
\section{Discrete Fourier Set Query in One Dimension}
\label{sec:discrete_fourier_set_query}

In this section, we study the Fourier set-query problem, where we only care about the Fourier coefficients of a discrete signal in a given set of frequencies. We apply our framework and achieve optimal sample complexity and high-accuracy. In Section~\ref{sec:main_thm_in_discrete_set_query}, we show our main result on discrete Fourier set query. A key step to prove this result is a WBSP Composition Lemma in Section~\ref{sec:equiv_sample_procedure}, which might be of independent interest.

\label{sec:eps_net_discrete}

\subsection{Sample-optimal set query algorithm}\label{sec:main_thm_in_discrete_set_query}
In this section, we show our discrete Fourier set query result in the following theorem, which works for discrete signals in any dimension.
\begin{theorem}[Discrete Fourier Set Query]\label{thm:main_fourier_ours_clever_discrete}
For any $d\geq 1$, let $n=p^d$ where both $p$ and $d$ are positive integers. Given a vector $x \in \C^{[p]^d}$, for  $1\leq k \leq n$, any $S \subseteq [n]$, $|S| = k$, there exists an algorithm (Algorithm~\ref{algo:set_query}) %
that takes $O(\eps^{-1}k)$ samples, runs in $O(\eps^{-1}k^{\omega+1} + \eps^{-1} d k^{\omega-1}\log k)$ time, and outputs a vector $x' \in \C^{[p]^d}$ such that
\begin{align*}
\| (\wh{x}' - \wh{x})_S \|_2^2 \leq \eps \| \wh{x}_{\bar{S}} \|_2^2 
\end{align*}
holds with probability at least $0.9$. 

In particular, for $d=1$, the runtime of Algorithm~\ref{algo:set_query} is $O(\epsilon^{-1}k^{\omega+1})$.
\end{theorem}

\begin{proof}
Let $\{f_1,f_2,\cdots,f_k\}\subseteq[p]^d$ denote $S$. If $d=1$, we run Procedure $\textsc{DistillDisc}$ in Algorithm~\ref{algo:distill_dis}, and if $d>1$, we run Procedure $\textsc{DistillDiscHD}$ in Algorithm~\ref{algo:distill_dis}.
Then, we obtain a set $L=\{t_1,t_2,\cdots, t_s\}\subseteq [p]^d$ of $s=O(\epsilon^{-1}k)$ samples together with a weight vector $w\in \R^s$.

Then, we consider the following weighted linear regression problem: 
\begin{align}
\underset{v' \in \C^{k} }{\min}\| \sqrt{w} \circ (A v' - b )\|_2.  \label{eq:linear_reg_discrete_hd}
\end{align}
where $A\in \C^{s\times k}$ and $b\in \C^s$ are defined as follows:
\begin{align*}
    A:=\begin{bmatrix}
    \exp(2\pi\i f_1t_1/n) &   \cdots & \exp(2\pi\i f_k t_1/n)\\
    \vdots  & \ddots & \vdots\\
    \exp(2\pi\i f_1t_s/n) &  \cdots & \exp(2\pi\i f_kt_s/n)
    \end{bmatrix}~\text{and}~b:=\begin{bmatrix}
    x(t_1) \\ \vdots \\ x(t_s)
    \end{bmatrix}
\end{align*}
Let $v'$ be an optimal solution of Eq.~\eqref{eq:linear_reg_discrete_hd}. And we output a vector
\begin{align*}
    \wh{x}'_{f_i}=v'_i~~~\forall i\in [k].
\end{align*}

The running time follows from Lemma~\ref{lem:set_query_time}, and the estimation error guarantee follows from Lemma~\ref{lem:set_query_error}.

The proof of the theorem is then completed.
\end{proof}

\begin{algorithm}[ht]\caption{Discrete signal set-query algorithm.} \label{algo:set_query}
\begin{algorithmic}[1]
\Procedure{\textsc{SetQuery}}{$x$, $n$, $k$, $S$, $\eps$} \Comment{Theorem~\ref{thm:main_fourier_ours_clever_discrete} (one-dimension)}
\State $\{f_1,f_2,\cdots,f_{k}\} \gets S$
\State $s, \{t_1,t_2,\cdots, t_s\},w  \leftarrow\textsc{DistillDisc}(k, \sqrt{\eps}, F, n)$\label{ln:set_distill_disc} \Comment{Algorithm \ref{algo:distill_dis}}
\State $A_{i, j} \leftarrow\exp(2\pi\i f_j t_i/n)$, $A\in \C^{s \times k}$
\State $b  \leftarrow (x(t_1),x(t_2),\cdots, x(t_s))^\top$
\State Solving the following weighted linear regression\label{ln:set_regress}\Comment{Fact~\ref{fac:basic_l2_regression}}
\begin{align*}
v' \gets \underset{v' \in \C^{ \wt{k} } }{\arg\min}\| \sqrt{w} \circ (A v'- b)\|_2.  
\end{align*}
\State \Return $\hat{x}'$ such that $\hat{x}'_{f_j}=v'_j$ for $j\in [k]$
\EndProcedure

\Procedure{\textsc{SetQueryHD}}{$x$, $n$, $k$, $S$, $\eps$} \Comment{Theorem~\ref{thm:main_fourier_ours_clever_discrete} (high-dimension)}
\State $\{f_1,f_2,\cdots,f_{k}\} \gets S$
\State $s, \{t_1,t_2,\cdots, t_s\},w  \leftarrow\textsc{DistillDiscHD}(k, \sqrt{\eps}, F, n)$\label{ln:set_distill_disc_hd} \Comment{Algorithm~\ref{algo:distill_dis}}
\State $F_{\mathrm{batch}} = [f_1,f_2,\cdots, f_k]\in [p]^{d\times k}$
\State $T_{\mathrm{batch}} = [t_1,t_2,\cdots, t_s]\in [p]^{d\times s}$
\State \label{ln:batch_hd}$U = F_{\mathrm{batch}}^{\top}T_{\mathrm{batch}} \in \Z^{k\times s} $ \Comment{Fact \ref{fac:matrix_multiplication}}
\State $A_{i, j} \leftarrow\exp(2\pi\i U_{j, i}/p)$, $A\in \C^{s \times k}$
\State $b  \leftarrow (x(t_1),x(t_2),\cdots, x(t_s))^\top$
\State Solving the following weighted linear regression\label{ln:set_regress_hd}\Comment{Fact~\ref{fac:basic_l2_regression}}
\begin{align*}
v' \gets \underset{v' \in \C^{ \wt{k} } }{\arg\min}\| \sqrt{w} \circ (A v'- b)\|_2.  
\end{align*}
\State \Return $\hat{x}'$ such that $\hat{x}'_{f_j}=v'_j$ for $j\in [k]$
\EndProcedure
\end{algorithmic}
\end{algorithm}

\begin{lemma}[Running time of Algorithm~\ref{algo:set_query}]\label{lem:set_query_time} The time complexity of 
Algorithm~\ref{algo:set_query} is as follows:
\begin{itemize}
    \item Procedure \textsc{SetQuery} runs in $O(\epsilon^{-1}k^{\omega+1})$-time.
    \item Procedrue \textsc{SetQueryHD} runs in $O(\epsilon^{-1}k^{\omega+1} + \epsilon^{-1} d k^{\omega-1}\log k )$-time.
\end{itemize}
\end{lemma}
\begin{proof}
We first show the time complexity of Procedure \textsc{DistillDisc}.
At Line~\ref{ln:set_distill_disc}, Procedure \textsc{DistillDisc} takes $O(\epsilon^{-1}k^{\omega+1})$-time by Lemma~\ref{lem:important_sampling_k_faster_discrete}.

At Line~\ref{ln:set_regress}, by Fact~\ref{fac:basic_l2_regression}, it takes $O(\epsilon^{-1}k\cdot k^{\omega-1})=O(\epsilon^{-1}k^\omega)$-time.

Thus, the total running time is $O(\epsilon^{-1}k^{\omega+1})$.

Then, we show the time complexity of Procedure \textsc{DistillDiscHD}.
At Line~\ref{ln:set_distill_disc_hd}, Procedure \textsc{DistillDiscHD} takes $O(\epsilon^{-1}k^{\omega+1} + \epsilon^{-1}dk^{\omega-1}\log k)$-time by Lemma~\ref{lem:important_sampling_k_faster_discrete}.

At Line~\ref{ln:batch_hd}, by Fact~\ref{fac:matrix_multiplication}, it takes the time $\Tmat(k,d,s)$. We know that $s \geq k$. We can consider two cases.
\begin{itemize}
    \item In case 1, $d \leq k$, we can just simply bound the time by $\Tmat(k,k,s) = O( k^{\omega} \cdot (s/k) ) = O( k^{\omega-1} s) = O(\epsilon^{-1} k^{\omega} )$. (In this regime, this part running time is dominated by Line~\ref{ln:set_distill_disc_hd})
    \item In case 2, $d \geq k$, we can just bound the time by $\Tmat(k,d,s) = O( k^{\omega} \cdot (d/k) \cdot (s/k) ) = ds k^{\omega-2} = O( \epsilon^{-1} d k^{\omega-1} ) $
\end{itemize}

At Line~\ref{ln:set_regress_hd}, by Fact~\ref{fac:basic_l2_regression}, it takes $O(\epsilon^{-1}k\cdot k^{\omega-1})=O(\epsilon^{-1}k^\omega)$-time.

Thus, the total running time is $O(\epsilon^{-1}k^{\omega+1} + \epsilon^{-1} d k^{\omega-1}\log k)$.
\end{proof}

\begin{lemma}[Estimation error of Algorithm~\ref{algo:set_query}]\label{lem:set_query_error}
Let $\wh{x}'$ be the output of Algorithm~\ref{algo:set_query} (with $d=1$ or $d>1$). Then, with high probability,
\begin{align*}
    \| (\wh{x}' - \wh{x})_S \|_2^2 \lesssim\| \wh{x}_{\bar{S}} \|_2^2.
\end{align*}
\end{lemma}
\begin{proof}

Let $D:=\mathrm{Uniform}([p]^d)$. Recall that $n=p^d$.
Let $\cal F$ be the family of length-$n$ discrete signals with frequencies in $S$:
\begin{align*}
{\cal F}=\Big\{ \sum_{j=1}^{k} v_j e^{2 \pi \i\langle f_j, t \rangle/ p} ~\big|~ v_j \in \mathbb{C} \Big\}
\end{align*}
Then, it is well-known that $\{v_j(t)=\exp(2\pi\i \langle f_j, t\rangle/p)\}_{j\in [k]}$ forms an orthonormal basis for ${\cal F}$ with respect to the distribution $D$, i.e.,
\begin{align*}
    \E_{t\sim D}[\ov{v_i(t)}v_j(t)]={\bf 1}_{i=j} ~~~\forall i,j\in [k].
\end{align*}

Now, we define some notations. Let $\alpha:{\cal F}\rightarrow \C^k$ be a linear operator such that for any $h(t) = \sum_{j=1}^{k} a_j\exp(2\pi\i \langle f_j, t\rangle/p)$, 
\begin{align*}
    \alpha(h):=\begin{bmatrix}
    a_1& a_2& \cdots & a_k
    \end{bmatrix}^\top.
\end{align*}
Suppose the true discrete signal $x(t)=\sum_{j=1}^n v_j \exp(2\pi\i \langle j, t\rangle/p)$. Define 
\begin{align*}
    x_S(t):=&~ \sum_{f\in S} v_f \exp(2\pi \i \langle f, t\rangle/p),\\
    x_{\overline{S}}(t):=&~ \sum_{f\in \overline{S}} v_f \exp(2\pi \i \langle f, t\rangle/p).
\end{align*}
Let $\sqrt{W}\in \R^{s\times s}$ denote the diagonal matrix $\mathrm{diag}(\sqrt{w_1},\dots,\sqrt{w_s})$. Define
\begin{align*}
    A_w:=&~ \sqrt{W}\cdot A,\\
    X_w:=&~\sqrt{W}\cdot \begin{bmatrix}x(t_1) & \cdots & x(t_s)\end{bmatrix}^\top,\\
    X_w^S:=&~\sqrt{W}\cdot \begin{bmatrix}x_S(t_1) & \cdots & x_S(t_s)\end{bmatrix}^\top,\\
    X_w^{\ov{S}}:= &~ \sqrt{W}\cdot \begin{bmatrix}x_{\ov{S}}(t_1) & \cdots & x_{\ov{S}}(t_s)\end{bmatrix}^\top.\\
\end{align*}
Notice that for any $h=\sum_{i=1}^k a_i \exp(2\pi \i \langle f_i,t\rangle/p)\in {\cal F}$,
\begin{align*}
    A_w \alpha(h) = \sqrt{W}\cdot \begin{bmatrix}
    \exp(2\pi\i f_1t_1/n) &   \cdots & \exp(2\pi\i f_k t_1/n)\\
    \vdots  & \ddots & \vdots\\
    \exp(2\pi\i f_1t_s/n) &  \cdots & \exp(2\pi\i f_kt_s/n)
    \end{bmatrix}\begin{bmatrix}a_1\\\vdots \\ a_k\end{bmatrix} = \begin{bmatrix}
    \sqrt{w_1}h(t_1)\\
    \vdots\\
    \sqrt{w_s}h(t_s)
    \end{bmatrix}.
\end{align*}
Thus, by Moore-Penrose inverse, we have
\begin{align}
    \alpha(h)=(A^*_wA_w)^{-1} A^*_w \cdot \begin{bmatrix}
    \sqrt{w_1}h(t_1)\\
    \vdots\\
    \sqrt{w_s}h(t_s)
    \end{bmatrix}. \label{eq:alpha_f_in_cal_F_is_solution_hd}
\end{align}

Let $x'(t):=\sum_{j=1}^k \wh{x}'_{f_j} \exp(2\pi\i \langle f_j, t\rangle/p)$ be the output signal in the time domain.
Then we claim that %
\begin{align*}
\|x'-x_S\|^2_{D} = &~ \|\alpha(x')-\alpha(x_S)\|^2_2 \\
=&~\| (A^*_wA_w)^{-1} A^*_w \cdot (X_w - X_w^S) \|^2_2 \\
=&~\| (A^*_wA_w)^{-1} A^*_w \cdot X^{\ov{S}}_w \|^2_2 \\
 \leq&~ \lambda_{\max}((A^*_wA_w)^{-1})^2 \cdot \|A^*_w \cdot X^{\ov{S}}_w \|_2^2\\
\leq&~  \|A^*_w \cdot X^{\ov{S}}_w \|_2^2,
\end{align*}
where the first step follows from the definition of $\alpha$, the second step follows from $\alpha(x')=v'$ being the optimal solution of Eq.~\eqref{eq:linear_reg_discrete_hd} and Eq.~\eqref{eq:alpha_f_in_cal_F_is_solution_hd} for $x_S$, the third step follows from $x=x_S + x_{\ov{S}}$, the fifth  step follows from Lemma \ref{lem:operator_estimation} and Lemma \ref{lem:important_sampling_k_faster_discrete} and holds with high probability.

Notice that $x_{\ov{S}}$ is orthogonal to ${\cal F}$. And by Lemma~\ref{lem:equiv_sampling}, we know that $(L,w)$ is generated by an $\epsilon$-WBSP. Hence, by Lemma~\ref{lem:guarantee_dist}, we have
\begin{align*}
    \|x'-x_S\|^2_{D} \leq \|A^*_w \cdot X^{\ov{S}}_w \|_2^2 \lesssim \eps \|x_{\overline{S}}\|_{D}^2.
\end{align*}

By Parseval's theorem (Theorem~\ref{thm:parseval}), we conclude that
\begin{align*}
    \|\hat{x}'-\hat{x}_S\|_2^2\leq \epsilon\|\hat{x}_{\ov{S}}\|_2^2
\end{align*}
holds with high probability.

\end{proof}

\subsection{Composition of two WBSPs}\label{sec:equiv_sample_procedure}
In this section, we prove the following key lemma on the composition of two WBSPs for discrete signals.

\begin{lemma}[WBSP Composition Lemma]\label{lem:equiv_sampling}
Let $m_0,m_1,n\in \Z_+$, $m_1\leq m_0\leq n$. Let $\{f_1,\cdots,f_k\} \subseteq [n]$. Let ${\cal F}$ be the family of discrete $k$-sparse signals in $t\in [n]$: 
\begin{align*}
{\cal F}=\Big\{ v_0 + \sum_{j=1}^{k} v_j \cdot \exp(2\pi\i f_j t / n) ~|~ \forall v_j \in \mathbb{C}, j=\{0,\dots,k\} \Big\}
\end{align*}
Define the following two WBSPs for ${\cal F}$:
\begin{itemize}
    \item Let $P_1$ be an $\epsilon$-WBSP generating $m_1$ samples, with input distribution $D_1$, coefficients $\alpha_{1}$, and output distributions $D_{1,i}$ for $i\in [m_1]$.
    \item Let $P_2$ be an $\epsilon$-WBSP generating $m_2$ samples, with input distribution $D_2$, coefficients $\alpha_2$, and output distributions $D_{2,i}$ for $i\in [m_2]$.
\end{itemize}
We can composite $P_1$ and $P_2$ by taking $D_2(x_i):=\frac{w_{1,i}}{\sum_{j\in [m_1]}w_{1,j}}$ for $i\in [m_1]$.
Let $P_1\circ P_2$ denote the composition of $P_1,P_2$.

Then, if $P_1$ satisfies $D_{1,i}=D_1=\mathrm{Uniform}([n])$, then $P_1\circ P_2$ is an $O(\eps)$-WBSP generating $m_2$ samples, with input distribution $D_1$, coefficients $w_2$, and output distributions $D_{1}$ for all $i\in [m_2]$. 
\end{lemma}
\begin{proof}

Let $S_1=\{x_1,\dots,x_{m_1}\}$ denote the set sampled by $P_1$ and $S_2=\{x_1',\dots,x'_{m_2}\}$ denote the set sampled by $P_2$. Then, we have $S_2\subset S_1$. In the followings, we show that $P_1\circ P_2$ satisfies all the stated properties.

\paragraph{Input distribution and the first WBSP property.}We first show that $P_1\circ P_2$ satisfies the first property of WBSP in Definition~\ref{def:procedure_agnostic_learning} with respect to distribution $D_1$, that is,
\begin{align*}
    \|f\|_{S_2, w_2}\in [1-O(\sqrt{\eps}), 1+O(\sqrt{\eps})]\|f\|_{D_1}^2~~~\forall f\in {\cal F}.
\end{align*}
By definition of $\epsilon$-WBSP (Definition~\ref{def:procedure_agnostic_learning}), we have for any $f\in {\cal F}$,
\begin{align}\label{eq:compose_wbsp_P1}
    &\|f\|_{S_1, w_1}^2\in [1-\sqrt{\eps}, 1+\sqrt{\eps}] \cdot \|f\|_{D_1}^2, ~\text{and}\\
    &\|f\|_{S_2, w_2}^2\in [1-\sqrt{\eps}, 1+\sqrt{\eps}] \cdot \|f\|_{D_2}^2.\notag
\end{align}
By the definition of $D_2$, we have $\|f\|_{D_2}^2 = \|f\|_{S_1, w_1}^2$ (assuming $\|w_1\|_1=1$ without loss of generality). Thus, we get that
\begin{align}\label{eq:compose_wbsp_f}
    \|f\|_{S_2, w_2}^2\in &~ [1-\sqrt{\eps}, 1+\sqrt{\eps}]\|f\|_{S_1, w_1}^2 \notag\\
    \in &~ [1-\sqrt{\eps}, 1+\sqrt{\eps}] \cdot [1-\sqrt{\eps}, 1+\sqrt{\eps}] \|f\|_{D_1}^2\notag\\
    \in &~ [1-3\sqrt{\eps}, 1+3\sqrt{\eps}]\|f\|_{D_1}^2.
\end{align}

\paragraph{Coefficients.}Then, consider the equivalent coefficients $\alpha_3$ of $P_1\circ P_2$. Let $D_{3,i}$ be the output distribution of the $i$-th sample $x'_i$ produced by $P_1\circ P_2$. By Fact~\ref{fac:double_sample},
\begin{align*}
    D_{3,i}(x'_i) = \sum_{j=1}^{m_1} D_{2,i}(x_j)\cdot D_{1,j}(x'_i) = D_1(x'_i),
\end{align*}
where the second step follows from the assumption that $D_{1,j}=D_1$ for all $j\in [m_1]$. Thus, we have $D_{3,i}=D_1$ for all $i\in [m_2]$. Since its weight vector is $w_2$ and input distribution is $D_1$, by definition, we have for $i\in [m_2]$,
\begin{align*}
    \alpha_{3,i}= w_{2,i}\cdot \frac{D_{3,i}(x_i')}{D_1(x_i')}
    = w_{2,i}.
\end{align*}
Thus, the coefficients of $P_1\circ P_2$ is $w_{2}$.

\paragraph{The second WBSP property.} We first bound $\sum_{i=1}^{m_2}\alpha_{3,i}$. Since $\alpha_3 = w_2$, we just need to bound $\sum_{i=1}^{m_2}w_{2,i}$. Let $f_1:=\begin{bmatrix}1 & 1 & \cdots & 1\end{bmatrix}^\top\in \C^n$. Then, it is easy to see that $f_1\in {\cal F}$ with $v_0=1$ and $v_i=0$ for $i\in[k]$. By Eq.~\eqref{eq:compose_wbsp_f}, we have
\begin{align*}
    \|f_1\|_{S_2, w_2}^2 = &~ \sum_{i=1}^{m_2}w_{2,i} \\
    \in &~ [1-\sqrt{\eps}, 1+\sqrt{\eps}] \cdot \|f_1\|_{D_1}^2\\
    = &~ [1-\sqrt{\eps}, 1+\sqrt{\eps}],
\end{align*}
where the last step follows from $\|f_1\|_{D_1}^2=\sum_{i=1}^n D_1(i)=1$. Hence,
\begin{align*}
    \sum_{i=1}^{m_2}\alpha_{3,i} = \sum_{i=1}^{m_2}w_{2,i} \leq 1+\sqrt{\epsilon} \leq \frac{5}{4}.
\end{align*}

We also need to show that $\alpha_{3,i}K_{\mathsf{IS}, D_{3,i}}=O(\epsilon)$ for all $i\in [m_2]$. By definition, we have
\begin{align}\label{eq:compose_bound_k3}
    K_{\mathsf{IS}, D_{3,i}} = &~ \underset{t}{\sup} \bigg\{ \frac{D_1(t)}{D_{3,i}(t)} \cdot \underset{f \in \mathcal{F}}{\sup} \big\{ \frac{|f(t)|^2}{\|f\|_{D_1}^2} \big\}\bigg\}\notag\\
    = &~ \sup_t \sup_{f\in {\cal F}}\Big\{\frac{|f(t)|^2}{\|f\|_{D_1}^2}\Big\}\notag\\
    \leq &~ k,
\end{align}
where the second step follows from $D_{3,i}=D_1$ and the last step follows from the energy bound (Theorem \ref{thm:discrete_energy_bound}) and the assumption that $D_1=\mathrm{Uniform}([n])$. 

Since $P_2$ is an $\eps$-WBPS, we have
\begin{align*}
    K_{\mathsf{IS}, D_{2,i}} = &~ \underset{t}{\sup} \bigg\{ \frac{D_2(t)}{D_{2,i}(t)} \cdot \underset{f \in \mathcal{F}}{\sup} \big\{ \frac{|f(t)|^2}{\|f\|_{D_2}^2} \big\}\bigg\}\\
    = &~ \underset{t}{\sup} \bigg\{ \frac{D_2(t)}{D_{2,i}(t)} \cdot \underset{f \in \mathcal{F}}{\sup} \big\{ \frac{|f(t)|^2}{\|f\|_{S_1,w_1}^2}\big\}\bigg\}\\
    \geq &~ (1+\sqrt{\eps})^{-1}\cdot \underset{t}{\sup} \bigg\{ \frac{D_2(t)}{D_{2,i}(t)} \cdot \underset{f \in \mathcal{F}}{\sup} \big\{ \frac{|f(t)|^2}{\|f\|_{D_1}^2}\big\}\bigg\},
\end{align*}
where the second step follows from $\|f\|_{D_2}=\|f\|_{S_1,w_1}$, the third step follows from Eq.~\eqref{eq:compose_wbsp_P1}. 
And for all $i\in [m_2]$, 
\begin{align*}
    \alpha_{2,i}K_{\mathsf{IS}, D_{2,i}}=O(\epsilon),
\end{align*}
which implies that
\begin{align*}
    \alpha_{2,i}\cdot \underset{t}{\sup} \bigg\{ \frac{D_2(t)}{D_{2,i}(t)} \cdot \underset{f \in \mathcal{F}}{\sup} \big\{ \frac{|f(t)|^2}{\|f\|_{D_1}^2}\big\}\bigg\} = O(\epsilon).
\end{align*}
Since $D_1$ is uniform, we know that $\{\exp(2\pi\i f_jt)\}_{j\in [k]}$ form an orthonormal basis with respect to $D_1$. Thus, by Fact~\ref{fac:condition_number_to_ortho_basis}, for any $t\in [n]$,
\begin{align*}
        \underset{f \in \mathcal{F}}{\sup} \{ \frac{|f(t)|^2}{\|f\|_{D'}^2} \} = \sum_{j=1}^k |\exp(2\pi\i f_j t)|^2 = k.
\end{align*}
Hence, we get that
\begin{align*}
    \alpha_{2,i}\cdot \underset{t}{\sup} \bigg\{ \frac{D_2(t)}{D_{2,i}(t)}\bigg\} = O(\epsilon/k)
\end{align*}

Therefore,
\begin{align*}
    \alpha_{3,i}K_{\mathsf{IS},D_{3,i}}\leq &~ w_{2,i} \cdot k\\
    = &~ \alpha_{2,i}\cdot \frac{D_{2}(x_i)}{D_{2,i}(x_i)} \cdot k\\
    \leq &~ \alpha_{2,i} \cdot \underset{t}{\sup} \bigg\{ \frac{D_2(t)}{D_{2,i}(t)}\bigg\} \cdot k\\
    = &~ O(\epsilon/k)\cdot k\\
    = &~ O(\epsilon).
\end{align*}
where the first step follows from $\alpha_3=w_2$ and Eq.~\eqref{eq:compose_bound_k3}, the second step follows from the definition of $w_{2,i}$.

Thus, we prove that $P_1\circ P_2$ is an $O(\epsilon)$-WBSP with input distribution $D_1$, output distributions $D_1$, coefficients $w_2$.
\end{proof}

\begin{fact}[Double-sampling distribution]\label{fac:double_sample}
For $i\in [n]$, let $D_i$ be a distribution over the domain $G$. Suppose we first sample $x_i$ from $D_i$ for each $i \in [n]$. Let $w_1,\cdots, w_n\in\R_+$ such that $\sum_{i=1}^n w_i = 1$. Conditioned on the samples $\{x_1,\dots,x_n\}$, let $D'$ be a distribution over these samples such that $D'(x_i)=w_i$. Then, we sample an $x'$ from $D'$. 

Then, the distribution of $x'$ is $D''$, where
\begin{align*}
    D''(x)=\sum_{i=1}^n w_i D_i(x)~~~\forall x\in G.
\end{align*}
\end{fact}
\begin{proof}
Notice that the second sampling process is equivalent to sample an index ${\sf i}\in [n]$.
Hence, for any $a\in G$, 
\begin{align*}
    \Pr[x'=a] =&~ \sum_{j=1}^n \Pr[{\sf i}=j]\cdot \Pr_{D_j}[x_j=a~|~{\sf i}=j]\\
    =&~ \sum_{j=1}^n w_j \cdot \Pr_{D_j}[x_j=a] \\
    =&~ \sum_{j=1}^n w_j D_j(a)\\
    = &~ D''(a).
\end{align*}
where the first step follows from law of total probability, and the second step follows from sampling $x_j$ from $D_j$ is independent to sampling the index ${\sf i}$ from $D'$.
\end{proof}

%% file: lowerror.tex
\section{High-Accuracy Fourier Interpolation Algorithm}\label{sec:high_precision_interpolate}

In this section, we propose an algorithm for one-dimensional continuous Fourier interpolation problem, which significantly improves the accuracy of the algorithm in \cite{ckps16}. 

This section is organized as follows. In Sections~\ref{sec:high_acc:fourier_poly} and \ref{sec:high_acc:filter}, we provide some technical tools for Fourier-sparse signals, low-degree polynomials and filter functions. In Section~\ref{sec:high_sensitive}, we design a high sensitivity frequency estimation method using these tools. In Section~\ref{sec:high_acc:7_eps}, we combine the frequency estimation with our Fourier set query framework, and give a  $(7+\epsilon)$-approximate Fourier interpolation algorithm. Then, in Section~\ref{sec:sig_noise_cancel}, we build a sharper error control, and in Section~\ref{sec:hashtobins}, we analysis the \textsc{HashToBins} procedure. Based on these result, in Section~\ref{sec:high_acc:ultra_high}, we develop the ultra-high sensitivity frequency estimation method. In Section \ref{sec:high_acc:1_sqrt_2}, we show the a  $(1+\sqrt{2}+\epsilon)$-approximate Fourier interpolation algorithm.

\subsection{Technical tools \RN{1}: Fourier-polynomial equivalence}
\label{sec:high_acc:fourier_poly}
In this section, we show that low-degree polynomials and Fourier-sparse signals can be transformed to each other with arbitrarily small errors.

The following lemma upper-bounds the error of using low-degree polynomial to approximate Fourier-sparse signal.
\begin{lemma}[Fourier signal to polynomial, \cite{ckps16}]\label{lem:low_degree_approximates_concentrated_freq}
For any $\Delta>0$ and any $\delta>0$, let $x^*(t)=\sum_{j \in [k]} v_j e^{2 \pi \i f_j t}$ where $|f_j| \le \Delta$ for each $j\in [k]$. There exists a polynomial $P(t)$ of degree at most
\[ d=O(T \Delta + k^3 \log k + k \log 1/\delta) \] such that
\[ \|P - x^*\|^2_T \le \delta \|x^*\|^2_T.\]
\end{lemma}

As a corollary, we can expand a Fourier-sparse signal under the \emph{mixed Fourier-monomial basis} (i.e., $\{e^{2\pi \i f_i t}\cdot t^j\}_{i\in [k], j\in [d]}$).
\begin{corollary}[Mixed Fourier-polynomial approximation]\label{cor:low_degree_approximates_concentrated_freq_ours}
For any $\Delta>0$, $\delta>0$, $n_j\in \Z_{\geq 0}, j\in [k], \sum_{j\in [k]}n_j=k$. Let
\begin{align*}
    x^*(t)=\sum_{j \in [k]} e^{2 \pi \i f_j t} \sum_{i=1}^{n_j} v_{j, i} e^{2\pi\i f'_{j, i} t}, 
\end{align*}
where $|f'_{j, i}| \le \Delta$ for each $j\in [k], i\in [n_j]$. There exist $k$ polynomials $P_j(t)$ for $j \in [k]$ of degree at most
\[ d=O(T \Delta + k^3 \log k + k \log 1/\delta) \] such that
\[ \Big\|\sum_{j\in [k]} e^{2 \pi \i f_j t} P_j(t) - x^*(t)\Big\|^2_T \le \delta \|x^*(t)\|^2_T.\]
\end{corollary}

The following lemma bounds the error of approximating a low-degree polynomial using Fourier-sparse signal.
\begin{lemma}[Polynomial to Fourier signal, \cite{ckps16}]\label{lem:polynomial_to_FT}
For any degree-$d$ polynomial $Q(t) = \overset{d}{\underset{j=0}{\sum}} c_j t^j$, any $T>0$ and any $\epsilon>0$, there always exist $\gamma>0$ and $$x^*(t)=\sum_{j=1}^{d+1} \alpha_j e^{2\pi \i (\gamma j) t}$$ with some coefficients $\alpha_0,\cdots,\alpha_d$ such that
\begin{equation*}
\forall t \in [0,T], |x^*(t) - Q(t)| \le \epsilon. %
\end{equation*}
\end{lemma}

\subsection{Technical tools \RN{2}: filter functions}\label{sec:high_acc:filter}
In this section, we introduce the filter functions $H$ and $G$ designed by \cite{ckps16}, and we generalize their constructions to achieve higher sensitivity.

We first construct the $H$-filter, which uses $\rect$ and $\sinc$ functions.
\begin{fact}[$\rect$ function Fourier transform]
For $s>0$, let $\rect_s(t):={\bf 1}_{|t|\leq s/2}$. Then, we have
\begin{align*}
    \wh{\rect_s}(f)=\sinc(sf) = \frac{\sin(sf)}{\pi s f}.
\end{align*}
\end{fact}

\begin{definition}\label{def:H1}

Given $s_1,s_2>0$ and an even number $\ell\in \mathbb{N}_+$, we define the filter function $H_1(t)$ and its Fourier transform $\wh{H}_1(f)$ as follows:
\begin{eqnarray}
H_1(t) &= & s_0 \cdot (\sinc^{\ell}( s_1 t)) \star \rect_{s_2}(t) \notag  \\ 
\widehat{H}_1(f) & = & s_0 \cdot (\rect_{s_1}\star \cdots \star \rect_{s_1})(f) \cdot \sinc\left( f s_2\right) \notag %
\end{eqnarray}
where $s_0=C_0 s_1\sqrt{\ell}$ is a normalization parameter such that $H_1(0)=1$, and $\star$ means convolution. 
\end{definition}

\begin{definition}[$H$-filter's construction, \cite{ckps16}]\label{def:def_of_filter_H}
Given any $ 0< s_1,s_3 < 1$, $0<\delta <1$, we define $H_{s_1,s_3,\delta}(t)$ from the filter function $H_1(t)$ (Definition~\ref{def:H1}) as follows:
\begin{itemize}
\item let $\ell := \Theta( k \log ( k /\delta ) )$, $ s_2  := 1 -\frac{2}{s_1}$, and
\item shrink $H_1$ by a factor $s_3$ in time domain, i.e.,
\end{itemize}
\begin{eqnarray}
H_{s_1,s_3,\delta}(t) &:= & H_1( t/s_3) \label{eq:definition_H} \\
\widehat{H_{s_1,s_3,\delta}}(f) & = & s_3 \widehat{H_1}(s_3 f) \label{eq:definition_hatH}
\end{eqnarray}
We call the ``filtered cluster"
around a frequency $f_0$ to be the support of $(\delta_{f_0} \star \wh{H_{s_1,s_3,\delta}})(f)$ in the frequency domain %
and use 
\begin{equation}
\Delta_h=|\supp(\widehat{H_{s_1,s_3,\delta}} )| = \frac{ s_1 \cdot \ell }{s_3} %
\end{equation} to denote the width of the cluster.
\end{definition}

\begin{lemma}[High sensitivity $H$-filter's properties]\label{lem:property_of_filter_H}
Given $\epsilon\in (0,0.1)$, $s_1, s_3\in (0,1)$ with $\min (\frac{1}{1-s_3}, s_1) \geq \wt{O}(k^4)/\eps$, and $\delta\in (0,1)$. Let the filter function $H:=H_{s_1,s_3,\delta}(t)$ defined in Definition~\ref{def:def_of_filter_H}. Then, $H$ satisfies the following properties:
\begin{eqnarray*}
&\mathrm{Property~\RN{1}} : &H(t) \in [ 1 - \delta, 1], \text{~when~} |t| \leq  ( \frac{1}{2} - \frac{2}{s_1} ) s_3.\\
&\mathrm{Property~\RN{2}} : &H(t) \in [0,1], \text{~when~}  (\frac{1}{2} - \frac{2}{s_1}) s_3 \leq |t| \leq \frac{1}{2} s_3.\\
&\mathrm{Property~\RN{3}} : &H(t) \leq s_0 \cdot (s_1( \frac{|t|}{s_3}-\frac{1}{2})+2)^{-\ell},\text{~when~} |t| > \frac{1}{2} s_3.\\
&\mathrm{Property~\RN{4}} : &\supp(\widehat{H}) \subseteq [-\frac{s_1  \ell}{2 s_3}, \frac{s_1 \ell }{2 s_3}].
\end{eqnarray*}

For any exact $k$-Fourier-sparse signal $x^*(t)$, we shift the interval from $[0,T]$ to $[-1/2,1/2]$ and consider $x^{*}(t)$ for $t \in [-1/2,1/2]$ to be our observation, which is also $x^*(t) \cdot \rect_1(t)$.
\begin{eqnarray*}
&\mathrm{Property~\RN{5}} : &\int_{-\infty}^{+\infty} \bigl|x^*(t) \cdot H(t) \cdot (1- \rect_1(t) ) \bigr|^2 \mathrm{d} t < \delta \int_{-\infty}^{+\infty} | x^*(t) \cdot \rect_1(t) |^2 \mathrm{d} t.\\
&\mathrm{Property~\RN{6}} : &\int_{-\infty}^{+\infty} |x^*(t) \cdot H(t) \cdot \rect_1(t) |^2 \mathrm{d} t \in  [1-\epsilon, 1]\cdot \int_{-\infty}^{+\infty} |x^*(t) \cdot \rect_1(t) |^2 \mathrm{d} t.
\end{eqnarray*}
\end{lemma}

\begin{remark}\label{rmk:trivial_bound_on_H}
By Property \RN{1}, and \RN{2}, and \RN{3}, we have that $H(t)\leq 1$ for $t\in [0,T]$.%
\end{remark}

\begin{proof}
The proof of Property I - V easily follows from \cite{ckps16}. We prove Property VI in below.

 First, because of for any $t$, $|H_1(t)|\leq 1$, thus we prove the upper bound for $\text{LHS}$,
\begin{equation*}
\int_{-\infty}^{+\infty} |x^*(t) \cdot H(t) \cdot \rect_{1}(t)|^2 \mathrm{d} t \leq \int_{-\infty}^{+\infty} |x^*(t) \cdot 1 \cdot \rect_{1}(t)|^2 \mathrm{d} t.
\end{equation*}

Second, as mentioned early, we need to prove the general case when $s_3 = 1 - 1/\poly(k)$. Define interval $S = [-s_3( \frac{1}{2}- \frac{1}{s_1}), s_3( \frac{1}{2}- \frac{1}{s_1}) ]$, by definition, $S \subset [-1/2, 1/2]$. Then define $\overline{S} =[-1/2,1/2]\setminus S$, which is $ [-1/2, -s_3( \frac{1}{2}- \frac{1}{s_1}) ) \cup (s_3( \frac{1}{2}- \frac{1}{s_1}) , 1/2]$. By Property \RN{1}, we have
\begin{equation}\label{eq:eq1_proof_of_property_6}
\int_S |x^*(t) \cdot H(t)|^2 \mathrm{d} t \geq (1-\delta)^2 \int_S |x^*(t)|^2 \mathrm{d} t
\end{equation}
Then we can show
\begin{align}\label{eq:eq2_proof_of_property_6} 
 &~ \int_{\overline{S}} |x^*(t) |^2 \mathrm{d} t \notag \\ %
 \leq &~ |\overline{S}|\cdot \underset{t\in [-1/2,1/2]}{\max} |x^*(t)|^2   \notag \\
 \leq &~ (1-s_3(1-\frac{2}{s_1})) \cdot {O}(k^2) \int_{-\frac{1}{2}}^{\frac{1}{2}} |x^*(t)|^2 \mathrm{d} t \notag \\
 \leq &~ \eps \int_{-\frac{1}{2}}^{\frac{1}{2}} |x^*(t)|^2 \mathrm{d} t 
\end{align}
where the first step follows from $\overline{S}\subset [-1/2, 1/2]$, the second step follows from Theorem~\ref{thm:worst_case_sFFT_improve}, 
the third step follows from $(1-s_3(1-\frac{2}{s_1})) \cdot {O}(k^2) \leq \eps $.

Combining Equations (\ref{eq:eq1_proof_of_property_6}) and (\ref{eq:eq2_proof_of_property_6}) gives a lower bound for $\text{LHS}$,
\begin{align*}
 & \int_{-\infty}^{+\infty} |x^*(t) \cdot H(t) \cdot \rect_{1}(t)|^2 \mathrm{d} t\\
 \geq &~ \int_S |x^*(t) H(t)|^2 \mathrm{d} t \\
 \geq &~ (1-2\delta) \int_S |x^*(t) |^2 \mathrm{d} t \\
 = &~ (1-2\delta) \int_{S\cup \overline{S}} |x^*(t) |^2 \mathrm{d} t - (1-2\delta) \int_{\overline{S} } |x^*(t) |^2 \mathrm{d} t \\
 \geq &~ (1-2\delta) \int_{S\cup \overline{S}} |x^*(t) |^2 \mathrm{d} t - (1-2\delta) \epsilon \int_{S \cup \overline{S} } |x^*(t)|^2 \mathrm{d} t \\
 = &~ (1-2\delta - \epsilon ) \int_{-\frac{1}{2}}^{\frac{1}{2}} |x^*(t)|^2 \mathrm{d} t \\
 \geq &~ (1-2\epsilon) \int_{-\infty}^{+\infty} | x^*(t) \cdot \rect_1(t) |^2 \mathrm{d} t,
\end{align*}
where the first step follows from $S\subset [-1/2,1/2]$, the second step follows from Eq.~\eqref{eq:eq1_proof_of_property_6}, the third step follows from $S\cap \overline{S}=\emptyset$, the forth step follows from Eq.~\eqref{eq:eq2_proof_of_property_6}, the fifth step follows from $S\cup \overline{S}=[-1/2,1/2]$, the last step follows from $\epsilon \gg \delta$.

\end{proof}

As remarked in \cite{ckps16}, to match $(H(t),\widehat{H}(f))$ on $[-1/2,1/2]$ with signal $x(t)$ on $[0,T]$, we will scale the time domain from $[-1/2,1/2]$ to $[-T/2,T/2]$ and shift it to $[0,T]$. Then, in frequency domain, the Property \RN{4} in Lemma~\ref{lem:property_of_filter_H} becomes
\begin{equation}\label{eq:def_of_delta_h}
\supp( \widehat{H}(f) ) \subseteq [ -\frac{\Delta_h}{2} , \frac{\Delta_h}{2} ], \text{~where~} \Delta_h = \frac{s_1\ell}{s_3 T}.
\end{equation}

We also need another filter function, $G$, whose construction and properties are given below.

\begin{definition}[$G$-filter's construction, \cite{ckps16}]\label{def:define_G_filter}
Given $B >1$, $\delta >0$, $\alpha>0$. Let $l := \Theta( \log(k/\delta) )$.
Define $G_{B,\delta,\alpha}(t)$ and its Fourier transform $\wh{G_{B,\delta,\alpha}}(f)$ as follows:
\begin{eqnarray*}
G_{B,\delta,\alpha}(t) %
:= &~ b_0 \cdot (\rect_{ \frac{B}{(\alpha \pi)}} (t) )^{\star  l} \cdot  \sinc(t \frac{\pi}{2B}), \\
\widehat{G_{B,\delta,\alpha}}(f) %
:= &~ b_0 \cdot ( \sinc(\frac{B}{\alpha \pi} f) )^{\cdot l} * \rect_{\frac{\pi}{2B}}(f),
\end{eqnarray*}
where $b_0 = \Theta(B \sqrt{l}/\alpha)$ is the normalization factor such that $\wh{G}(0)=1$. %
\end{definition}

\begin{lemma}[$G$-filter's properties, \cite{ckps16}]\label{lem:property_of_filter_G}
Given $B >1$, $\delta >0$, $\alpha>0$, let $G:=G_{B,\delta,\alpha}(t)$ be defined in Definition~\ref{def:define_G_filter}. Then, $G$ satisfies the following properties:%
\begin{eqnarray*}
&\mathrm{Property~\RN{1}} : &\widehat{G}(f) \in [1 - \delta/k, 1] , \text{~if~} |f|\leq (1-\alpha)\frac{2\pi}{2B}.\\
&\mathrm{Property~\RN{2}}: &\widehat{G}(f) \in [0,1], \text{~if~}  (1-\alpha)\frac{2\pi}{2B} \leq |f| \leq \frac{2\pi}{2B}.\\
&\mathrm{Property~\RN{3}} : &\widehat{G}(f) \in [-\delta /k, \delta/k], \text{~if~}  |f| >  \frac{2\pi}{2B}
.\\
&\mathrm{Property~\RN{4}} : &\supp(G(t) ) \subset [\frac{l}{2} \cdot \frac{-B}{\pi\alpha}, \frac{l}{2} \cdot \frac{B}{\pi\alpha}].\\
&\mathrm{Property~\RN{5}} : & \underset{t}{\max} |G(t)| \lesssim  \poly(B,l).%
\end{eqnarray*}
\end{lemma}

\subsection{High sensitivity frequency estimation}\label{sec:high_sensitive}

In this section, we show a high sensitivity frequency estimation. Compared with the result in \cite{ckps16}, we relax the condition of the frequencies that can be recovered by the algorithm.
\begin{definition}[Definition 2.4 in \cite{ckps16}]\label{def:heavy_clusters}
Given $x^*(t)= \overset{k}{\underset{j=1}{\sum} } v_j e^{2 \pi \i f_j t}$, any $\N>0$, and a filter function $H$ with bounded support in frequency domain. Let  %
$L_j$ denote the interval of $ ~\supp(\wh{e^{2 \pi \i f_j t} \cdot H})$ for each $j\in [k]$.
Define an equivalence relation $\sim$ on the frequencies $f_i$ as follows:
\begin{align*}
    f_i\sim f_j ~~\text{iff}~~L_i\cap L_j\ne \emptyset~~~\forall i,j\in [k]. 
\end{align*}
Let $S_1,\ldots,S_n$ be the equivalence classes under this relation for some $n\leq k$. 

Define $C_i := \underset{f\in S_i}{ \cup } L_i $ for each $i\in [n]$. We say $C_i$ is an ${\cal N}$-heavy cluster iff $$\int_{C_i} |\wh{H \cdot x^*}(f)|^2 \mathrm{d} f \ge T \cdot \N^2/k.$$ 

\end{definition}

The following claim gives a tight error bound for approximating the true signal $x^*(t)$ by the signal $x_{S^*}(t)$ whose frequencies are in heavy-clusters. It improves the Claim 2.5 in \cite{ckps16}.

\begin{claim}[Approximation by heavy-clusters]\label{cla:guarantee_removing_x**_x*_ours}
Given $x^*(t)= \overset{k}{ \underset{j=1}{\sum} } v_j e^{2 \pi \i f_j t}$ and any $\N>0$,  let $C_1,\cdots,C_l$ be the $\N$-heavy clusters from Definition \ref{def:heavy_clusters}. For
\begin{equation*}
{S^*}=\left\{j \in [k]\bigg{|}f_j \in C_1 \cup \cdots C_l \right\},
\end{equation*}
we have $x_{S^*}(t)= \underset{j\in {S^*}}{\sum} v_j e^{2 \pi \i f_j t}$ approximating $x^*$ within distance $\|x_{S^*} - x^* \|_T^2 \leq (1+\eps) \N^2.$
\end{claim}
\begin{proof}
Let $H$ be the filter function defined as in Definition \ref{def:def_of_filter_H}. 

Let 
\begin{align*}
x_{\overline{{S^*}}}(t) := \underset{j\in [k] \backslash {S^*}}{\sum} v_j e^{2 \pi \i f_j t}.
\end{align*}

Notice that $\|x^*-x_{S^*}\|_T^2=\|x_{\overline{{S^*}}}\|^2_T$.

By Property \RN{6} in Lemma \ref{lem:property_of_filter_H} with setting $\eps=\eps/2$, let $\eps_0=\eps/2$, we have 
\begin{align*}
  (1-\eps_0) \cdot T \|x_{\overline{{S^*}}}\|_T^2 = &~ (1-\eps_0) \int_{0}^T |x_{\overline{{S^*}}}(t)|^2 \mathrm{d} t \\
  = &~ (1-\eps_0) \int_{0}^T |x_{\overline{{S^*}}}(t) \cdot\rect_T (t) |^2 \mathrm{d} t \\
  \leq &~ \int_{-\infty}^{+\infty} |x_{\overline{{S^*}}}(t) \cdot H(t)\cdot \rect_T (t)|^2 \mathrm {d} t ,\\
  \leq &~ \int_{-\infty}^{+\infty} |x_{\overline{{S^*}}}(t) \cdot H(t)|^2 \mathrm {d} t ,
\end{align*}
where the first step follows from the definition of the norm, the second step follows from the definition of $\rect_T(t)=1,\forall t\in[0, T]$, the third step follows from Lemma \ref{lem:property_of_filter_H}, the forth step follows from $\rect_T (t)\leq 1$.

From Definition \ref{def:heavy_clusters}, we have 
\begin{align*}
\int_{-\infty}^{+\infty} |x_{ \overline{{S^*}} }(t) \cdot H(t)|^2 \mathrm {d} t   =& ~ \int_{-\infty}^{+\infty} |\wh{x_{ \overline{{S^*}} } \cdot H}(f)|^2 \mathrm {d} f \\
= & ~ \int_{[-\infty,+\infty]\setminus C_1 \cup \cdots \cup C_l} |\widehat{x^{*} \cdot H}(f)|^2 \mathrm {d} f\\
 \le& ~ (k-l) \cdot T\N^2/k.
\end{align*}
where the first step follows from Parseval's theorem, the second step follows from Definition \ref{def:heavy_clusters}, Property~\RN{4} of Lemma \ref{lem:property_of_filter_H}, the definition of ${S^*}$, thus, 
$\supp(\widehat{x_{S^*} \cdot H}(f))=C_1 \cup \cdots \cup C_l$, $\supp(\widehat{x_{S^*} \cdot H}(f)) \cap \supp(\widehat{x_{\overline{{S^*}}} \cdot H}(f)))=\emptyset$, the last step follows from Definition \ref{def:heavy_clusters}. 

Overall, we have $(1-\eps_0) \|x_{\overline{{S^*}}}\|^2_T \leq \N^2$. Thus, $\|x_{S^*}(t) - x^*(t) \|_T^2 \leq (1 - l/k)(1+\eps) \N^2$. 
\end{proof}

Due to the noisy observations, not all frequencies in heavy-clusters are recoverable. Thus, we define the recoverable frequency as follows:

\begin{definition}[Recoverable frequency]\label{def:recoverable_freq}
A frequency $f$ is \emph{$({\cal N}_1,{\cal N}_2)$-recoverable} if $f$ is in an  ${\cal N}_1$-heavy cluster $C$ that satisfies:
\begin{align*}
    \int_C |\wh{x\cdot H}(f)|^2\geq T{\cal N}_2^2/k.
\end{align*}
\end{definition}

The following lemma shows that most frequencies in the heavy-clusters are actually recoverable.
\begin{lemma}[Heavy-clusters are almost recoverable]\label{lem:S_to_x_star}
Let $x^*(t) = \sum_{j=1}^k v_j e^{2\pi\i f_j t}$ and $x(t)= x^*(t) +g(t)$ be our observable signal. Let $\N^2 := \| g \|_T^2 + \delta \| x^* \|_T^2$. Let $C_1,\cdots,C_l$ are the $2\N$-heavy clusters from Definition \ref{def:heavy_clusters}. Let $S^*$ denotes the set of frequencies $f^*\in \{f_j\}_{j\in[k]}$ such that, $f^*\in C_i$ for some $i\in [l]$.
Let $S\subset S^*$ be the set of $(2{\cal N}, {\cal N})$-recoverable frequencies.

Then we have that,
\begin{align*}
    \|x_{S}-x^*\|_T \leq (3 - l/k +\eps) \N.
\end{align*}

\end{lemma}
\begin{proof}

If a cluster $C_i$ is $2{\cal N}$-heavy but not ${\cal N}$-recoverable, then it holds that:
\begin{align}
\int_{C_i} | \wh{H \cdot x^*}(f)|^2 \d f \ge 4T\N^2/k \ge 4 \int_{C_i} | \wh{H \cdot x}(f)|^2 \d f \label{eq:miss_cluster_constant_4_eric_fix}
\end{align}
where the first steps follows from $C_i\subset \bigcup_{f_j\in S^*} C_j$, the second step follows from $C_i	\not\subset \bigcup_{f_j\in S} C_j$. 

So, 
\begin{align}
\int_{C_i} | \wh{H \cdot g}(f)|^2 \d f = &~ \int_{C_i} | \wh{H \cdot (x-x^*)}(f)|^2 \d f \notag \\
\ge &~ \left(\sqrt{\int_{C_i} | \wh{H \cdot x^*}(f)|^2 \d f}-\sqrt{\int_{C_i} | \wh{H \cdot x}(f)|^2 \d f}\right)^2\notag \\
\ge &~ \frac{1}{4}\int_{C_i} | \wh{H \cdot x^*}(f)|^2 \d f\label{eq:eric_fix:g_large}
\end{align}
where the first step follows from $g(t)=x(t)-x^*(t)$, and the second step follows from triangle inequality, the last step follows from Eq.~\eqref{eq:miss_cluster_constant_4_eric_fix}. 

Let $C':=\bigcup_{f_j\in S^*\backslash S} C_j$, i.e., the union of heavy but not recoverable clusters. Then, we have
\begin{align}
 \|\wh{H \cdot g} \|_2^2 \ge \sum_{C_i\in C'} \int_{C_i} |\wh{H \cdot g(f)} |^2\d f \ge  \frac{1}{4}\sum_{C_i\in C'} \int_{C_i} | \wh{H \cdot x^*}(f)|^2\d f\label{eq:eric_fix:g_l2_large}
\end{align}
 where the first step follows from the definition of the norm and $C_i\cap C_j=\emptyset,\forall i\ne j$, the second step follows from Eq.~\eqref{eq:eric_fix:g_large}.

Then we have that
\begin{align}
T\|x_{S* \backslash S}\|_T^2 \le &~ \frac{T}{1-\eps/2}\|x_{S* \backslash S}\cdot H\|_T^2 \notag\\
\le &~ (1+\eps)\sum_{C_i\in C'} \int_{C_i} | \wh{H \cdot x^*}(f)|^2\d f \notag\\
\le &~ 4(1+\eps) \|\wh{H \cdot g} \|_2^2\notag \\
= &~ 4(1+\eps) T\|H \cdot g\|_T^2 \notag\\
\leq &~ 4(1+\eps) T\|g\|_T^2\notag\\
\leq &~ 4(1+\eps) T\N^2.\notag    
\end{align}
where the first step follows from Property \RN{6} of $H$ in Lemma \ref{lem:property_of_filter_H}  (taking $\epsilon$ there to be $\epsilon/2$), the second step follows from $\eps\in [0,1]$ and the definition of $C_i$, the third step follows from Eq.~\eqref{eq:eric_fix:g_l2_large}, the forth step follows from $g(t)=0,\forall t\not\in[0, T]$,  the fifth step follows from Remark \ref{rmk:trivial_bound_on_H}, the last step follows from the definition of $\N^2$. Thus, we get that:
\begin{align}\label{eq:eric_fix:last_step}
    \|x_{S* \backslash S}\|_T \leq (2 - l/k +\epsilon){\cal N},
\end{align}
which follows from $\sqrt{1+\epsilon}\leq 1+\epsilon/2$.

Finally, we can conclude that 
\begin{align*}
\|x_S - x^*\|_T \leq &~ \|x_S - x_{S^*}\|_T + \|x_{S^*} - x^*\|_T \\
= &~ \|x_{S* \backslash S}\|_T + \|x_{S^*} - x^*\|_T \\
\leq &~ \|x_{S* \backslash S}\|_T + (1+\eps)\N\\
\leq &~ (3 - l/k +2\eps)\N,
\end{align*}
where the first step follows from triangle inequality, the second step follows from the definition of $ x_{S* \backslash S}$, the third step follows from Claim \ref{cla:guarantee_removing_x**_x*_ours}, the last step follows from Eq.~\eqref{eq:eric_fix:last_step}. The lemma follows by re-scaling $\epsilon$ to $\epsilon/2$.
\end{proof}

\subsection{\texorpdfstring{$(9 +\eps)$-approximate Fourier interpolation algorithm}{}~}
\label{sec:high_acc:7_eps}
The goal of this section is to prove Theorem~\ref{thm:main_ours}, which gives a Fourier interpolation algorithm with approximation error $(9+\eps)$. It improves the constant (more than 1000) error algorithm in \cite{ckps16}.

\begin{claim}[Mixed Fourier-polynomial energy bound, \cite{ckps16}]
\label{cla:max_bounded_Q}
For any $$u(t) \in \mathrm{span}\left\{e^{2\pi \i {f}_{i} t} \cdot t^j~\bigg{|}~ j \in \{0,\cdots,d\},i\in [k] \right\},$$ 
we have that
\begin{equation*}
\max_{t \in [0,T]}~|u(t)|^2 \lesssim (k d)^{4} \log^{3} (k d) \cdot \|u\|^2_T
\end{equation*}
\end{claim}

\begin{claim}[Condition number of Mixed Fourier-polynomial]\label{cla:max_bounded_Q_condition_number}
Let ${\cal F}$ is a linear function family as follows:
$${\cal F} := \mathrm{span}\left\{e^{2\pi \i {f}_{i} t} \cdot t^j~\bigg{|}~ j \in \{0,\cdots,d\},i\in [k] \right\},$$ 
Then the condition number of $\mathrm{Uniform}[0, T]$ with respect to ${\cal F}$ is as follows:
\begin{equation*}
K_{\mathrm{Uniform}[0, T]}:=\sup_{t \in [0, T]} \sup_{f \in {\mathcal F}}  \frac{|f(t)|^2}{\|f\|_T^2} = O((k d)^{4} \log^{3} (k d))
\end{equation*}
\end{claim}

The following definition extends the well-balanced sampling procedure (Definition~\ref{def:procedure_agnostic_learning}) to high probability.

\begin{definition}[($\epsilon,\rho$)-well-balanced sampling procedure]
\label{def:procedure_agnostic_learning_rho}
Given a linear family $\mathcal{F}$ and underlying distribution $D$, let $P$ be a random sampling procedure that terminates in $m$ iterations ($m$ is not necessarily fixed) and provides a coefficient $\alpha_i$ and a distribution $D_i$ to sample $x_i \sim D_i$ in every iteration $i \in [m]$.

We say $P$ is an $\eps$-WBSP if it satisfies the following two properties:
\begin{enumerate}
\item With probability $1-\rho$, for weight $w_i=\alpha_i \cdot \frac{D(x_i)}{D_i(x_i)}$ of each $i \in [m]$, $$
\sum_{i=1}^m w_i \cdot |h(x_i)|^2 \in \left[1-10\sqrt{\eps}, 1+10\sqrt{\eps} \right] \cdot \|h\|_D^2 \quad \forall h \in \mathcal{F}.
$$

\item The coefficients always have $\sum_{i=1}^m \alpha_i \le \frac{5}{4}$ and $\alpha_i \cdot K_{\mathsf{IS},D_i} \le \frac{\eps}{2}$ for all $i \in [m]$.  
\end{enumerate}
\end{definition}

The following lemma is a generalization of Lemma \ref{lem:agnostic_learning_single_distribution}, showing an $(\eps,\rho)$-WBSP for mixed Fourier-polynomial family.
\begin{lemma}[WBSP for mixed Fourier-polynomial family]\label{lem:rho_wbsp}
Given any distribution $D'$ with the same support of $D$ and any $\epsilon>0$, the random sampling procedure with $m=O(\eps^{-1} K_{\mathsf{IS},D'} \log (d/\rho))$ i.i.d.~random samples from $D'$ and coefficients $\alpha_1=\cdots=\alpha_m=1/m$ is an $(\eps,\rho)$-WBSP.
\end{lemma}
\begin{proof}
By Lemma \ref{lem:general_distribution} with setting $\eps=\sqrt{\eps}$, we have that, as long as $m\geq O(\frac{1}{\eps} \cdot K_{\mathsf{IS},D'} \log \frac{d}{\rho})$, then with probability $1-\rho$, 
\begin{align*}
    \|A^*A-I\|_2 \leq \sqrt{\eps}
\end{align*}

By Lemma \ref{lem:operator_estimation}, we have that, for every $ h \in {\mathcal F}  $, 
\begin{align*} 
  \sum_{j=1}^s w_j \cdot |h(x_j)|^2 \in [1 \pm \eps] \cdot \|h\|_D^2,
\end{align*}
where $S$ is the $m$ i.i.d. random samples from $D'$, $w_i=\alpha_i D(x_i)/D'(x_i)$.

Moreover, $ \sum_{i=1}^m\alpha_i = 1 \leq 5/4$ and 
\begin{align*}
    \alpha_i \cdot K_{\mathsf{IS},D'} = ~ \frac{K_{\mathsf{IS},D'} }{m} 
    \leq ~ \frac{\eps}{  \log(d/\rho)}
    \leq ~ {\eps},
\end{align*}
where the first step follows from the definition of $\alpha_i$, the second step follows from the definition of $m$, the third step follows from $\log(d/\rho)>1$.
\end{proof}

Then, similar to Theorem~\ref{thm:reduction_freq_signal_1d_clever}, we can solve the Signal Estimation problem for mixed Fourier-polynomial signals.
\begin{lemma}[Mixed Fourier-polynomial signal estimation]\label{lem:magnitude_recovery_for_CKPS}
Given $d$-degree polynomials $P_j(t),j\in [k]$ and frequencies $f_j,j\in[k]$. Let $x_S(t) = \sum_{j=1}^k P_j(t) \exp({2\pi\i  f_j  t  })$, and let $g(t)$ denote the noise. Given observations of the form $x(t):=x_S(t) + g'(t)$ for arbitrary noise $g'$ in time duration $t\in [0, T]$. 

Then, there is an algorithm such that
\begin{itemize}
    \item takes $O(\eps^{-1}\poly(kd)\log(1/\rho))$ samples from $x(t)$, 
    \item runs $O(\eps^{-1}\poly(kd)\log(1/\rho))$ time, 
    \item outputs $y(t) = \sum_{j=1}^k P'_j(t) \exp({2\pi\i  f_j  t  })$ with $d$-degree polynomial $P'_j(t)$, such that with probability at least $1-\rho$, we have
    \begin{align*}
        \|y - x_S\|_T^2  \leq (1+\eps)\|g'\|_T^2 .
    \end{align*}
\end{itemize}  
\end{lemma}
\begin{proof}[Proof sketch]
The proof is almost the same as Theorem~\ref{thm:reduction_freq_signal_1d_clever} where we follow the three-step Fourier set-query framework. Claim~\ref{cla:max_bounded_Q} gives the energy bound for the family of mixed Fourier-polynomial signals, which implies that uniformly sampling $m=\wt{O}(\epsilon^{-1}|L|^4d^4)$ points in $[0,T]$ forms an oblivious sketch for $x^*$. Moreover, by Lemma~\ref{lem:rho_wbsp}, we know that it is also an $(\eps, \rho)$-WBSP, which gives the error guarantee. Then, we can obtain a mixed Fourier-polynomial signal $y(t)$ by solving a weighted linear regression. 
\end{proof}

Now, we are ready to prove the main result of this section, a $(9+\eps)$-approximate Fourier interpolation algorithm.
\begin{theorem}[Fourier interpolation with $(9 +\eps)$-approximation error]\label{thm:main_ours}
Let $x(t) = x^*(t) + g(t)$, where $x^*$ is $k$-Fourier-sparse signal with  frequencies in $[-F, F]$.  Given samples of $x$ over $[0, T]$ we can output $y(t)$ such that with probability at least $1-2^{-\Omega(k)}$, 
  \[
  \|y - x^*\|_T \leq (7+\eps)\|g\|_T + \delta \|x^*\|_T.
  \]
  Our algorithm uses $\poly(k,\eps^{-1},\log (1/\delta) ) \log(FT)$
  samples and $\poly(k,\eps^{-1},\log(1/\delta)) \cdot \log^2(FT)$ time.  The output $y$
  is $\poly(k,\log(1/\delta))\eps^{-1.5}$-Fourier-sparse signal. 
\end{theorem}
\begin{proof}
Let $\N^2 := \| g(t) \|_T^2 + \delta \| x^*(t) \|_T^2$ be the heavy cluster parameter. %

First, by Lemma \ref{lem:S_to_x_star}, there is a set of frequencies $S\subset [k]$ and $x_{S}(t)= \underset{j\in S}{\sum} v_j e^{2 \pi \i f_j t}$ such that
\begin{align}
    \|x_S - x^* \|_T \leq (3+O(\eps))\N. \label{eq:fourier_intor:x_S_x*}
\end{align}
Furthermore, each $f_j$ with $j\in S$ belongs to an $\N$-heavy cluster $C_j$ with respect to the filter function $H$ defined in Definition~\ref{def:def_of_filter_H}.

By Definition \ref{def:heavy_clusters} of heavy cluster, it holds that
\begin{align*}
    \int_{C_j} | \widehat{H\cdot x^*}(f) |^2 \mathrm{d} f \geq T\N^2/k.
\end{align*}
By Definition \ref{def:heavy_clusters}, we also have $|C_j|\leq k\cdot \Delta_h$, where $\Delta_h$ is the bandwidth of $\wh{H}$. 

Let $\Delta\in \R_+$, and $\Delta > k\cdot \Delta_h$, which implies that $C_j\subseteq [f_j-\Delta, f_j+\Delta]$. Thus, we have
\begin{align*}
\int_{f_j-\Delta}^{f_j+\Delta} | \widehat{H\cdot x^*}(f) |^2 \mathrm{d} f \geq T\N^2/k.
\end{align*}

Now it is enough to recover only $x_S$, instead of $x^*$. 

By applying Theorem \ref{thm:frequency_recovery_k_cluster_ours_low_prob}, 
there is an algorithm that outputs a set of frequencies $L\subset \R$ such that, $|L|=O(k)$, and with probability at least $1-2^{-\Omega(k)}$,  for any $f_j$ with $j\in {S_f}$, there is a $\wt{f}\in L$ such that,
\begin{align*}
|f_j-\widetilde{f} |\lesssim \Delta \sqrt{\Delta T}.
\end{align*}

We define a map $p:\R\rightarrow L$ as follows:
\begin{align*}
    p(f):=\arg \min_{\wt{f}\in L} ~ |f-\wt{f}|~~~\forall f\in \R.
\end{align*}
Then, $x_S(t)$ can be expressed as
\begin{align*}
x_{S_f}(t)= &~\sum_{j\in {S_f}}v_je^{2\pi \i f_j t}\\
= &~ \sum_{j\in {S_f}}v_j e^{2\pi \i \cdot p(f_j)t} \cdot e^{2\pi \i \cdot (f_j - p(f_j))t}\\
= &~ \sum_{\wt{f} \in L} e^{2 \pi \i \widetilde{f} t} \cdot \sum_{j\in {S_f}:~ p(f_j)=\wt{f}} v_j e^{2\pi \i ( f_j - \widetilde{f})t},
\end{align*}
where the first step follows from the definition of $x_S$, the last step follows from interchanging the summations.

For each $\wt{f}_i\in L$, by Corollary \ref{cor:low_degree_approximates_concentrated_freq_ours} with $ x^*=x_{S_f}, \Delta= \Delta \sqrt{\Delta T}$, we have that there exist degree $ d=O(T \Delta \sqrt{\Delta T} + k^3 \log k + k \log 1/\delta)$ polynomials $P_i(t)$ corresponding to $\wt{f}_i\in L$ such that, 
\begin{align}
\|x_{S_f}(t)-\sum_{\wt{f}_i \in L} e^{2 \pi \i \widetilde{f}_i t} P_i(t)\|_T \leq \delta \|x_{S_f}(t)\|_T\label{eq:fourier:xS_sum}
\end{align}

Define the following function family: 
\begin{align*}
  \mathcal{F} := \mathrm{span}\Big\{e^{2\pi \i \widetilde{f} t} \cdot t^j~{|}~ \forall \wt{f}\in L, j \in \{0,1,\dots,d\} \Big\}.
\end{align*}
Note that $\sum_{\wt{f}_i \in L} e^{2 \pi \i \widetilde{f}_i t} P_i(t)\in {\cal F}$.

By Claim \ref{cla:max_bounded_Q_condition_number}, for function family $\cal F$, $K_{\mathrm{Uniform[0, T]}} = O((|L| d)^{4} \log^{3} (|L| d))$. 

By Lemma \ref{lem:rho_wbsp}, we have that, choosing a set $W$ of $O(\eps^{-1} K_{\mathrm{Uniform[0, T]}} \log(|L|d/\rho))$
i.i.d. samples uniformly at random over duration $[0, T]$ is a $(\eps,\rho)$-WBSP.

By Lemma \ref{lem:magnitude_recovery_for_CKPS}, there is an algorithm that runs in $O(\eps^{-1}|W|(|L|d)^{\omega-1}\log(1/\rho))$-time using samples in $W$, and outputs $y'(t)\in {\cal F}$ such that, with probability $1-\rho$, 
\begin{align}
    \|y'(t) - \sum_{\wt{f}_i \in L} e^{2 \pi \i \widetilde{f}_i t} P_i(t)\|_T\leq (1+\eps)\|x(t)-\sum_{\wt{f}_i \in L} e^{2 \pi \i \widetilde{f}_i t} P_i(t)\|_T\label{eq:fourier:y_sum}
\end{align}

Then by Lemma \ref{lem:polynomial_to_FT}, we have that there is a $O(kd) $-Fourier-sparse signal $y(t)$, such that
\begin{align}\label{eq:approx_y}
    \|y(t)-y'(t)\|_T \leq \delta'
\end{align}
where $\delta'>0$ is any positive real number, thus, $y$ can be arbitrarily close to $y'$. 

Moreover, the sparsity of $y(t)$ is $kd = k O(T \Delta \sqrt{\Delta T} + k^3 \log k + k \log 1/\delta) = \eps^{-1.5}\poly(k,\log(1/\delta))$.

Therefore, the total approximation error can be upper bounded as follows:
\begin{align*}
&~ \|y-x^*\|_T\\
\leq &~ \|y-y'\|_T+ \Big\|y'-\sum_{\wt{f}_i \in L} e^{2 \pi \i \widetilde{f}_i t} P_i(t)\Big\|_T + \Big\|\sum_{\wt{f}_i \in L} e^{2 \pi \i \widetilde{f}_i t} P_i(t) - x^*\Big\|_T\tag{Triangle inequality}\\
\leq &~  (1+o(1))\Big\|y-\sum_{\wt{f}_i \in L} e^{2 \pi \i \widetilde{f}_i t} P_i(t)\Big\|_T + \Big\|\sum_{\wt{f}_i \in L} e^{2 \pi \i \widetilde{f}_i t} P_i(t) - x^*\Big\|_T\tag{Eq.~\eqref{eq:approx_y}}\\
\leq &~ (1+\eps)\Big\|x-\sum_{\wt{f}_i \in L} e^{2 \pi \i \widetilde{f}_i t} P_i(t)\Big\|_T + \Big\|\sum_{\wt{f}_i \in L} e^{2 \pi \i \widetilde{f}_i t} P_i(t) - x^*\Big\|_T\tag{Eq.~\eqref{eq:fourier:y_sum}}\\
    \leq &~ (1+2\eps)\|g\|_T + (2+\eps) \Big\|\sum_{\wt{f}_i \in L} e^{2 \pi \i \widetilde{f}_i t} P_i(t) - x^*\Big\|_T\tag{Triangle inequality}\\
    \leq &~ (1+2\eps)\|g\|_T+(2+\eps) \Big\|\sum_{\wt{f}_i \in L} e^{2 \pi \i \widetilde{f}_i t} P_i(t) - x_{S_f}\Big\|_T + (2+\eps) \|x_{S_f} - x^*\|_T\tag{Triangle inequality}\\
    \leq &~ (1+2\eps)\|g\|_T + (2+\eps)\delta \|x_{S_f}\|_T + (2+\eps) \|x_{S_f} - x^*\|_T\tag{Eq.~\eqref{eq:fourier:xS_sum}}\\
    \leq &~ (1+2\eps)\|g\|_T + O(\delta)\|x^*\|_T + (2+\eps)(1+\delta)\|x_{S_f}-x^*\|_T\tag{Triangle inequality}\\
    \leq &~ (1+2\eps)\|g\|_T + O(\delta)\|x^*\|_T + (2+\eps)(1+\delta)(\|x_{S_f} - x_S\|_T + \| x_S - x^*\|_T)\tag{Triangle inequality}\\
    \leq &~ (1+2\eps)\|g\|_T + O(\delta)\|x^*\|_T + (2 + \eps + O(\delta))(4 +O(\eps)){\cal N}\tag{Eq.~\eqref{eq:fourier_intor:x_S_x*} and Lemma~\ref{lem:bound_x_S_f}}\\
    = &~ (1+2\eps)\|g\|_T + O(\delta)\|x^*\|_T + (8 +O(\eps+\delta)){\cal N},
\end{align*}

Since we take $$\N = \sqrt{\| g \|_T^2 + \delta \| x^* \|_T^2}\leq \|g\|_T+\sqrt{\delta}\|x^*\|_T,$$
we have
\begin{align*}
    \|y-x^*\|_T \leq (9 +O(\eps))\|g\|_T + O(\sqrt{\delta})\|x^*\|_T.
\end{align*}
By re-scaling $\eps$ and $\delta$, we prove the theorem.

\end{proof}
\subsection{Sharper error control by signal-noise cancellation effect}\label{sec:sig_noise_cancel}

In this section, we significantly improve the error analysis in Section~\ref{sec:high_sensitive}. Our key observation is the \emph{signal-noise cancellation effect}: if there is a frequency $f^*$ in a ${\cal N}_1$-heavy cluster but not $({\cal N}_1,{\cal N}_2)$-recoverable for some ${\cal N}_2<{\cal N}_1$, then it indicates that the contribution of $f^*$ to the signal $x^*$'s energy are cancelled out by the noise $g$. 

In the following lemma, we improving Lemma~\ref{lem:S_to_x_star} by considering $g$'s effect in the gap between heavy-cluster signal and recoverable signal.

\begin{lemma}[Sharper error bound for recoverable signal, an improved version of Lemma~\ref{lem:S_to_x_star}]\label{lem:use_g_bound_x_under_H}
Let $x^*(t) = \sum_{j=1}^k v_j e^{2\pi\i f_j t}$ and $x(t)= x^*(t) +g(t)$ be our observable signal. Let $\N_1^2 := \| g(t) \|_T^2 + \delta \| x^*(t) \|_T^2$. Let $C_1,\cdots,C_l$ are the $\N_1$-heavy clusters from Definition \ref{def:heavy_clusters}. Let $S^*$ denotes the set of frequencies $f^*\in \{f_j\}_{j\in[k]}$ such that, $f^*\in C_i$ for some $i\in [l]$.
Let $S\subset S^*$ be the set of $({\cal N}_1, \sqrt{\eps_2}{\cal N}_1)$-recoverable frequencies (Definition~\ref{def:recoverable_freq}).

Then we have that,
\begin{align*}
    \|H\cdot x_{S^*}-H\cdot x_S\|^2_T + \|H\cdot x-H\cdot x_S\|^2_T  \leq (1+O(\sqrt{\eps_2}))\|x-x_{S^*}\|_T^2.
\end{align*}
\end{lemma}
\begin{proof}

Let $g'(t):=g(t)+x^*(t)-x_{S^*}(t)=x(t)-x_{S^*}(t)$. 

In order for cluster $C_i$ to be missed, we must have that
\begin{align}
\int_{C_i} | \wh{H \cdot x_{S^*}}(f)|^2 \d f \ge T\N_1^2/k \ge \frac{1}{\eps_2} \int_{C_i} | \wh{H \cdot x}(f)|^2 \d f\label{eq:separation_between_S_star_to_x}
\end{align}
where the first steps follows from $C_i\subset \cup_{f_j\in S^*} C_j$, %
the second step follows from $C_i	\not\subset \cup_{f_j\in S} C_j$. 

Thus, 
\begin{align}
\int_{C_i} | \wh{H \cdot g'}(f)|^2 \d f = &~ \int_{C_i} | \wh{H \cdot (x-x_{S^*})}(f)|^2 \d f \notag\\
\ge &~ \left(\sqrt{\int_{C_i} | \wh{H \cdot x_{S^*}}(f)|^2 \d f}-\sqrt{\int_{C_i} | \wh{H \cdot x}(f)|^2 \d f}\right)^2\notag\\
\ge &~ (\frac{1}{\sqrt{\eps_2}}-1)^2\int_{C_i} | \wh{H \cdot x}(f)|^2 \d f \notag \\
\ge &~ \frac{1}{2{\eps_2}}\int_{C_i} | \wh{H \cdot x}(f)|^2 \d f, \label{eq:g_is_large_compare_x}
\end{align}
where the first step follows from the definition of $g'$, the second step follows from triangle inequality, the third step follows from Eq.~\eqref{eq:separation_between_S_star_to_x}, the last step follows from $\eps_2\leq 0.1$. 

\paragraph{Bound $\|H\cdot x - H\cdot x_S\|_T$.}
Let $I' = \cup_{f_j\in S^*\backslash S} C_j$, then we have that, 
\begin{align}
T \|H\cdot x-H\cdot x_S\|^2_T \leq&~ \int_{-\infty}^\infty |H\cdot x(t)-H\cdot x_S(t)|^2 \d t \notag \\
=&~ \int_{-\infty}^\infty |(\wh{H\cdot x}-\wh{H \cdot x_S})(f)|^2 \d f\notag\\
=&~ \int_{I'} |(\wh{H\cdot x}-\wh{H \cdot x_S})(f)|^2 \d f + \int_{\overline{I'}} |(\wh{H\cdot x}-\wh{H \cdot x_S})(f)|^2 \d f \label{eq:bound_H_x_H_x_S}
\end{align}
where the first step follows from the definition of the norm, the second step follows from Parseval's theorem, the third step follows from $I'\cup \overline{I'}=[-\infty,\infty]$.

\paragraph{Bound $\|H\cdot x_{S^*}-H\cdot x_S\|_T$}
We can upper-bound it as follows:
\begin{align}
T \|H\cdot x_{S^*}-H\cdot x_S\|^2_T \leq&~ \int_{-\infty}^\infty |H\cdot x_{S^*}(t)-H\cdot x_S(t)|^2 \d t \notag\\
=&~ \int_{-\infty}^\infty |(\wh{H\cdot x_{S^*}}-\wh{H\cdot x_S})(f)|^2 \d f \notag\\
=&~ \int_{I'} |(\wh{H\cdot x_{S^*}}-\wh{H\cdot x_S})(f)|^2 \d f + \int_{\overline{I'}} |(\wh{H\cdot x_{S^*}}-\wh{H\cdot x_S})(f)|^2 \d f \notag \\
=&~\int_{I'} |(\wh{H\cdot x_{S^*}}-\wh{H\cdot x_S})(f)|^2 \d f \label{eq:bound_H_x_star_H_x_S}
\end{align}
where the first step follows from the definition of the norm, the second step follows from Parseval's theorem, the third step follows from $I'\cup \overline{I'}=[-\infty,\infty]$, the last step follows from $(\cup_{f_j\in S^*/S} C_j )\cap \overline{I'}=\emptyset$.

\paragraph{Putting it all together.}
By Eqs.~\eqref{eq:bound_H_x_H_x_S} and \eqref{eq:bound_H_x_star_H_x_S}, we get that
\begin{align*}
    &~T\|H\cdot x_{S^*}-H\cdot x_S\|_T^2 + T\|H\cdot x - H\cdot x_S\|_T^2\\
    \leq &~  \int_{I'} |(\wh{H\cdot x_{S^*}}-\wh{H\cdot x_S})(f)|^2 \d f + \int_{I'} |(\wh{H\cdot x}-\wh{H \cdot x_S})(f)|^2 \d f + \int_{\overline{I'}} |(\wh{H\cdot x}-\wh{H \cdot x_S})(f)|^2 \d f.
\end{align*}

For the first integral, we have
\begin{align}
\sqrt{\int_{I'} |(\wh{H\cdot x_{S^*}}-\wh{H\cdot x_S})(f)|^2 \d f} = &~ \sqrt{\int_{I'} |\wh{H\cdot x_{S^*}}(f)|^2 \d f }\notag\\
\leq&~ \sqrt{\int_{I'} |\wh{H\cdot x}(f)|^2 \d f} + \sqrt{\int_{I'} |\wh{H\cdot g'}(f)|^2 \d f}\notag \\
\leq&~ \sqrt{2\eps_2 \int_{I'} |\wh{H\cdot g'}(f)|^2\d f} + \sqrt{\int_{I'} |\wh{H\cdot g'}(f)|^2 \d f}\notag \\
\leq&~ (1+\sqrt{2\eps_2}) \sqrt{\int_{I'} |\wh{H\cdot g'}(f)|^2 \d f}, \label{eq:bound_x_S_star_x_S}
\end{align}
where the first step follows from $(\cup_{f_j\in S} C_j )\cap I'=\emptyset$, the second step follows from triangle inequality, the third step follows from Eq.~\eqref{eq:g_is_large_compare_x}, the last step is straightforward.

For the second integral, we have
\begin{align}
\int_{I'} |(\wh{H\cdot x}-\wh{H \cdot x_S})(f)|^2 \d f =&~\int_{I'} |\wh{H\cdot x}(f)|^2 \d f \notag \\
\leq&~ 2\eps_2 \int_{I'} |\wh{H\cdot g'}(f)|^2 \d f, \label{eq:miss_area_negligible}
\end{align}
where the first step follows from $(\cup_{f_j\in S} C_j )\cap I'=\emptyset$, the second step follows from Eq.~\eqref{eq:g_is_large_compare_x}. 

For the third integral, together with the $\int_{I'}|\wh{H\cdot g'}(f)|^2\d f$ term in the first integral's upper bound (Eq.~\eqref{eq:bound_x_S_star_x_S}), we have
\begin{align}
&~ \int_{\overline{I'}} |(\wh{H\cdot x}-\wh{H\cdot x_S})(f)|^2 \d f + \int_{I'} |\wh{H\cdot g'}(f)|^2 \d f\notag\\
=&~ \int_{\overline{I'}} |\wh{H\cdot (x_{S^*} + g' - x_S)}(f)|^2 \d f + \int_{I'} |\wh{H\cdot g'}(f)|^2 \d f \notag\\
 =&~ \int_{\overline{I'}} |\wh{H\cdot g'}(f)|^2 \d f + \int_{I'} |\wh{H\cdot g'}(f)|^2 \d f \notag\\
 =&~ \int_{-\infty}^\infty |\wh{H\cdot g'}(f)|^2 \d f \notag\\
 =&~ \int_{-\infty}^\infty |H\cdot g'(t)|^2 \d t \notag\\
 =&~ T\|H\cdot g'(t)\|_T^2 \notag \\
 \leq&~ T\|g'\|_T^2, \label{eq:merge_two_noise_into_one_noise}
\end{align}
where the first step follows from the definition of $g'$, the second step follows from $(\cup_{f_j\in S^*} C_j )\cap \overline{I'}=(\cup_{f_j\in S} C_j )$, the third step follows from $I'\cup \overline{I'}=[-\infty,\infty]$, the forth step follows from Parseval's theorem, the fifth step follows from $g'(t)=0,\forall t\not\in [0,T]$, the last step follows from $H(t)\leq 1$ by Remark \ref{rmk:trivial_bound_on_H}. 

Furthermore, we have that
\begin{align}
\int_{I'} |\wh{H\cdot g'}(f)|^2 \d f \leq \int_{-\infty}^\infty |\wh{H\cdot g'}(f)|^2 \d f \leq T\|g'(t)\|_T^2 .\label{eq:bound_H_g_prime}
\end{align}

Therefore, we conclude that
\begin{align*}
&~T\|H\cdot x_{S^*}-H\cdot x_S\|^2_T + T\|H\cdot x-H\cdot x_S\|^2_T \\
\leq&~ T\|H\cdot x_{S^*}-H\cdot x_S\|^2_T  + \int_{I'} |(\wh{H\cdot x}-\wh{H \cdot x_S})(f)|^2 \d f + \int_{\overline{I'}} |(\wh{H\cdot x}-\wh{H \cdot x_S})(f)|^2 \d f \\
\leq&~ \int_{I'} |(\wh{H\cdot x_{S^*}}-\wh{H\cdot x_S})(f)|^2 \d f + \int_{I'} |(\wh{H\cdot x}-\wh{H \cdot x_S})(f)|^2 \d f + \int_{\overline{I'}} |(\wh{H\cdot x}-\wh{H \cdot x_S})(f)|^2 \d f \\
\leq &~ (1+\sqrt{\eps_2})^2 {\int_{I'} |\wh{H\cdot g'}(f)|^2 \d f} + \int_{I'} |(\wh{H\cdot x}-\wh{H \cdot x_S})(f)|^2 \d f  + \int_{\overline{I'}} |(\wh{H\cdot x}-\wh{H\cdot x_S})(f)|^2 \d f   \\
\leq &~ (1+\sqrt{\eps_2})^2 {\int_{I'} |\wh{H\cdot g'}(f)|^2 \d f} + 2\eps_2 \int_{I'} |\wh{H\cdot g'}(f)|^2 \d f + \int_{\overline{I'}} |(\wh{H\cdot x}-\wh{H\cdot x_S})(f)|^2 \d f   \\
= &~ O(\sqrt{\eps_2} ) \int_{I'} |\wh{H\cdot g'}(f)|^2 \d f + \int_{I'} |\wh{H\cdot g'}(f)|^2 \d f + \int_{\overline{I'}} |(\wh{H\cdot x}-\wh{H\cdot x_S})(f)|^2 \d f \\
\le &~ O(\sqrt{\eps_2} ) T\|g'\|_T^2 + \int_{I'} |\wh{H\cdot g'}(f)|^2 \d f + \int_{\overline{I'}} |(\wh{H\cdot x}-\wh{H\cdot x_S})(f)|^2 \d f \\
\leq &~ O(\sqrt{\eps_2} ) T\|g'\|_T^2 + T\|g'\|_T^2\\
= &~ (1+O(\sqrt{\eps_2} ))T\|g'\|_T^2
\end{align*}
where the first step follows from Eq.~\eqref{eq:bound_H_x_H_x_S}, the second step follows from Eq.~\eqref{eq:bound_H_x_star_H_x_S}, the third step follows from Eq.~\eqref{eq:bound_x_S_star_x_S}, the forth step follows from Eq.~\eqref{eq:miss_area_negligible}, the fifth step follows from $(1+\sqrt{2\eps_2})^2 \leq 1+O(\sqrt{\eps_2} )$, the sixth step follows from Eq.~\eqref{eq:bound_H_g_prime}, the seventh step follows from Eq.~\eqref{eq:merge_two_noise_into_one_noise}, the last step is straightforward.

The lemma is then proved.
\end{proof}

As a consequence, we can easily bound $\|x_{S^*}-x_S\|_T$ as follows.
\begin{corollary}\label{cor:use_g_bound_x}
Let $S^*$ and $S$ be defined as in Lemma~\ref{lem:use_g_bound_x_under_H}.
Then, we have that,
\begin{align*}
\| x_{S^*}- x_S\|^2_T \leq (1+O(\sqrt{\eps_2}))\|x-x_{S^*}\|_T^2
\end{align*}
\end{corollary}
\begin{proof}
We have that,
\begin{align*}
\| x_{S^*}- x_S\|^2_T \leq (1+2\eps)\| H\cdot x_{S^*}- H\cdot x_S\|^2_T \leq (1+2\eps)(1+O(\sqrt{\eps_2}))\|x-x_{S^*}\|_T^2
\end{align*}
where the first step follows from Lemma \ref{lem:property_of_filter_H} Property \RN{6}, the second step follows from Lemma \ref{lem:use_g_bound_x_under_H} and $\eps=\eps_2$. 
\end{proof}

In Lemma~\ref{lem:use_g_bound_x_under_H}, we introduce an extra term $\|H\cdot x - H\cdot x_S\|_T$. The following lemma shows that this term appears in the approximation error $\|x-x_S\|_T$, which can be used to upper-bound the Signal Estimation's error.
\begin{lemma}[Decomposing the approximation error of recoverable signal]\label{lem:decompose_noise_by_filter}
Let $x^*(t) = \sum_{j=1}^k v_j e^{2\pi\i f_j t}$ and $x(t)= x^*(t) +g(t)$ be our observable signal. Let $\N_1^2 := \| g(t) \|_T^2 + \delta \| x^*(t) \|_T^2$. Let $C_1,\cdots,C_l$ are the $\N_1$-heavy clusters from Definition \ref{def:heavy_clusters}. Let $S^*$ denotes the set of frequencies $f^*\in \{f_j\}_{j\in[k]}$ such that, $f^*\in C_i$ for some $i\in [l]$, and 
\begin{align*}
\int_{C_i} | \widehat{x^*\cdot H}(f) |^2 \mathrm{d} f \geq T\N_1^2/k,    
\end{align*}

 Let $S$ denotes the set of frequencies $f^*\in S^*$ such that, $f^*\in C_{j}$ for some $j\in [l]$, and 
\begin{align*}
\int_{C_{j}} | \widehat{x\cdot H}(f) |^2 \mathrm{d} f \geq \eps_2 T\N_1^2/k,    
\end{align*}

Then we have that,
\begin{align*}
    \|x-x_S\|_T\leq \|H(x-x_S)\|_T+\|g\|_T+O({\eps})\|x^*-x_S\|_T.
\end{align*}
\end{lemma}
\begin{proof}
We first decompose $\|x-x_S\|_T$ into the part that passes through the filter $H$ and the part that does not pass through $H$:
\begin{align*}
\|x-x_S\|^2_T\leq &~  \|H(x-x_S)\|^2_T + \|(1-H)(x-x_S)\|^2_T\\
\leq &~ \|H(x-x_S)\|^2_T + \|(1-H)(x-x^*)\|^2_T +\|(1-H)(x^*-x_S)\|^2_T\\
\leq &~ \|H(x-x_S)\|^2_T + \|(1-H)g\|^2_T +\|(1-H)(x^*-x_S)\|^2_T,
\end{align*}
where the first step follows from triangle inequality, the second step follows from triangle inequality, the last step follows from the definition of $g$.

For the second term, we have that
\begin{align*}
\|(1-H)g\|^2_T\leq \|g\|^2_T,
\end{align*}
by Remark \ref{rmk:trivial_bound_on_H}.

For the third term, we have that, 
\begin{align*}
\|(1-H)(x^*-x_S)\|^2_T=\|x^*-x_S\|^2_T-\|H(x^*-x_S)\|^2_T\leq \eps\|x^*-x_S\|^2_T,
\end{align*}
where the first step follows from $1-H>0$, the second step follows from $x^*-x_S$ is $k$-Fourier-sparse, thus combine Property \RN{6} of Lemma \ref{lem:property_of_filter_H}, we have that $\|H(x^*-x_S)\|^2_T\ge (1-\eps)\|x^*-x_S\|_T^2$.

Combining them together, we prove the lemma.
\end{proof}

\subsection{Technical tools \RN{3}: \textsc{HashToBins}}\label{sec:hashtobins}

In this section, we provide some definitions and technical lemmas for the \textsc{HashToBins} procedure, which will be very helpful for frequency estimation. %

\textsc{HashToBins} partitions the frequency coordinates into $B = O(k) $ bins and collects rotated magnitudes in each bins. Ideally, each bins only contains a single ground-truth frequency, which allows us to recover its magnitude. 

More specifically, \textsc{HashToBins} first randomly hashes the frequency coordinates into the interval $[0, 1]$. After equally dividing $[0, 1]$ into $O(k)$ small bins, each coordinate lays in a different bin. This step can be implemented by multiplying the signal in the frequency domain with a \emph{period pulse function} $G^{(j)}_{\sigma, b}$. Then, even if the signal does not have frequency gap, the \textsc{HashToBins} procedure can still partition it into several one-cluster signals with high probability.%

\begin{definition}[Hash function, \cite{ckps16}]
Let $\pi_{\sigma,b}(f) = \sigma(f+b) \pmod {1}$ and $h_{\sigma,b}(f) = \mathrm{round} (\pi_{\sigma,b}(f) \cdot {B})$ be the hash function that maps frequency $f \in [-F,F]$ into bins $\{0,\cdots,B-1\}$.
\end{definition}
\begin{claim}[Collision probability,  \cite{ckps16}]\label{cla:PS15_hash_claims}
For any $\Delta_0>0$, let $\sigma$ be a sample uniformly at random from $[\frac{1}{4B\Delta_0}, \frac{1}{2B\Delta_0}]$. Then, we have:

\begin{enumerate}[label={\Roman*.}]
\item If $4\Delta_0 \leq|f^+ - f^-| < {2(B-1)\Delta_0} $, then $\mathsf{Pr}[h_{\sigma,b}(f^+) = h_{\sigma,b}(f^-)]=0$.

\item If ${2(B-1)\Delta_0} \leq |f^+ - f^-|$, then
$\mathsf{Pr}[h_{\sigma,b}(f^+) = h_{\sigma,b}(f^-)] \lesssim \frac{1}{B}$.
\end{enumerate}
\end{claim}

\begin{definition}[Filter for bins]\label{def:G_j_sigma_b}
Given $B >1$, $\delta >0$, $\alpha>0$, let $G(t):=G_{B,\delta,\alpha}(2\pi t)$ where $G_{B,\delta,\alpha}$ is defined in Definition~\ref{def:define_G_filter}. For any $\sigma>0, b\in\R$ and $j\in [B]$. 
define
\begin{align*}
G^{(j)}_{\sigma,b}(t) := &~ \frac{1}{\sigma} G(t/\sigma) e^{2\pi\i t(j/B-\sigma b)/\sigma},
\end{align*}
and its Fourier transformation:
\begin{align*}
\widehat{G}^{(j)}_{\sigma,b}(f)  = \sum_{i\in \Z}\widehat{G}(i+ \frac{j}{B} -\sigma f -\sigma b).
\end{align*}
\end{definition}

\begin{definition}[$(\epsilon_0,\Delta_0)$-one-cluster signal,  \cite{ckps16}]\label{def:one_cluster}
We say that a signal $z(t)$ is an $(\epsilon_0,\Delta_0)$-one-cluster signal around $f_0$ iff $z(t)$ and $\wh{z}(f)$ satisfy the following two properties: 
\begin{eqnarray*}
\mathrm{Property ~\RN{1}} &:& \int_{f_0-\Delta_0}^{f_0+\Delta_0} | \widehat{z}(f) |^2 \mathrm{d} f \geq (1-\epsilon_0) \int_{-\infty}^{+\infty} | \widehat{z}(f) |^2 \d f \\
\mathrm{Property ~\RN{2}} &:& \int_0^T | z(t) |^2 \mathrm{d} t \geq (1-\epsilon_0) \int_{-\infty}^{+\infty} |z(t) |^2 \d t.
\end{eqnarray*}
\end{definition}

\begin{definition}[Well-isolation,  \cite{ckps16}]\label{def:k_signal_recovery_z}
  We say that a frequency $f^*$ is \emph{well-isolated} under the
  hashing $(\sigma, b)$ if, for $j = h_{\sigma, b}(f^*)$ and $\overline{I_{f^*}} = (-\infty, \infty) \setminus (f^* -
  \Delta_0, f^* + \Delta_0)$,
  \[
  \int_{\overline{I_{f^*}}} \big|(\wh{H\cdot x}\cdot \wh{G}_{\sigma,b}^{(j)})(f)\big|^2\d f \lesssim \epsilon_0 \cdot T\N_2^2/k,
  \]
where $\N_2^2 :
  = \eps_1\eps_2(\| g(t) \|_T^2 + \delta \| x^*(t) \|_T^2)$.
\end{definition}

\begin{lemma}[Well-isolation implies one-cluster signal, a variation of Lemma 7.20 in \cite{ckps16}]
\label{lem:full_proof_of_3_properties_true_for_z} %
Let $f^*$ satisfy $$\int_{f^*-\Delta}^{f^*+\Delta} | \widehat{x^*\cdot H}(f) |^2 \mathrm{d} f \geq  T\N_2^2/k,$$ where $\N_2^2 :
  = \eps_1\eps_2(\| g(t) \|_T^2 + \delta \| x^*(t) \|_T^2)$. Let  $\wh{z} = \widehat{x^* \cdot H} \cdot \widehat{G}^{ (j)}_{\sigma,b}$ where $j=h_{\sigma,b}(f^*)$. If $f^*$ is well-isolated, then $z$ and $\wh{z}$ satisfying Property \RN{2} of one-cluster signal (Definition \ref{def:one_cluster}), i.e.,
\[ \int_0^T | z(t) |^2 \mathrm{d} t \geq (1-\epsilon_0) \int_{-\infty}^{+\infty} |z(t) |^2 \mathrm{d} t,\]
\end{lemma}

\begin{lemma}[Well-isolation by randomized hashing, \cite{ckps16}]\label{lem:often-well-isolated}
  Given $B= \Theta(k/(\eps_0\eps_1\eps_2))$ and $\sigma \in [\frac{1}{4B \Delta_0},\frac{1}{2B \Delta_0}]$ chosen uniformly at random. 
  Let $f^*$ be any frequency.  Then $f^*$ is well-isolated by a
  hashing $(\sigma, b)$ with probability at least $0.9$.
\end{lemma}
\begin{proof}
  Let $S'=\{f_i\}_{i\in[k]} \cap \overline{I_{f^*}}$. 
 By Claim \ref{cla:PS15_hash_claims}, with probability at least $(1-1/B)^k\ge 1-k/B \ge 1-\eps_0\eps_1\eps_2 \ge 0.99$, for all the frequencies $ f\in S'$, we have that $ h_{\sigma,b}(f^*)\neq h_{\sigma,b}(f)$. 
 
Hence, 
\begin{align}    \int_{\overline{I_{f^*}}} |\widehat{x^*\cdot H} \cdot \widehat{G}^{ (j)}_{\sigma,b} (f)|^2\d f \lesssim &~  \frac{\delta^2}{k^2} \int_{\overline{I_{f^*}}} |\widehat{x^*\cdot H}(f)|^2\d f \notag \\
\leq &~ \frac{\delta^2}{k^2} \int_{-\infty}^\infty |\widehat{x^*\cdot H}(f)|^2\d f \notag \\
= &~ \frac{\delta^2}{k^2} \int_{-\infty}^\infty |{x^*\cdot H}(t)|^2\d t \notag \\
= &~ \frac{\delta^2}{k^2} \int_{[-\infty,\infty]\backslash [0, T]} |{x^*\cdot H}(t)|^2\d t + \frac{\delta^2}{k^2} \int_{[0, T]} |{x^*\cdot H}(t)|^2\d t \notag \\
\leq &~  \frac{\delta^2}{k^2} \int_{[-\infty,\infty]\backslash [0, T]} |{x^*\cdot H}(t)|^2\d t +\frac{\delta^2}{k^2} T \|{x^*}\|_T^2 \notag \\
\leq &~ \frac{\delta^2(1+\delta)}{k^2} T \|{x^*}\|_T^2 \label{eq:well_isolated_1}
\end{align}
where the first step follows by the Property \RN{3} in the Lemma \ref{lem:property_of_filter_G} that $|\wh{G}(f)|\leq \delta / k$, which implies that $|\wh{G}^{(j)}_{\sigma, b}(f)| \leq O(\delta / k)$ for $f\in S'$, the second step follows from $\overline{I_{f^*}} \subset [-\infty,\infty]$, the third step follows from Parseval's theorem, the forth step is straight forward, the fifth step follows from the property \RN{6} of Lemma \ref{lem:property_of_filter_H}, the sixth step follows from \RN{5} of Lemma \ref{lem:property_of_filter_H}.

Moreover, let $I'$ denote the set of frequencies that  hash into the same bin as $f^*$, then we have that, 
\begin{align}    \int_{\overline{I_{f^*}}} |\widehat{g\cdot H} \cdot \widehat{G}^{ (j)}_{\sigma,b} (f)|^2\d f \leq &~ \int_{I'} |\widehat{g\cdot H} \cdot \widehat{G}^{ (j)}_{\sigma,b} (f)|^2\d f + \int_{\overline{I'}} |\widehat{g\cdot H} \cdot \widehat{G}^{ (j)}_{\sigma,b} (f)|^2\d f\notag \\
\lesssim &~ \int_{I'} |\widehat{g\cdot H}(f) |^2\d f + \int_{\overline{I'}} |\widehat{g\cdot H} \cdot \widehat{G}^{ (j)}_{\sigma,b} (f)|^2\d f \notag \\
\lesssim &~ \int_{I'} |\widehat{g\cdot H}(f) |^2\d f + \frac{\delta^2}{k^2}\int_{\overline{I'}} |\widehat{g\cdot H}(f) |^2\d f \notag \\
\leq &~ \int_{I'} |\widehat{g\cdot H}(f) |^2\d f + \frac{\delta^2 T}{k^2}\|g\|_T^2 \label{eq:well_isolated_2}
\end{align}
where the first step follows from $I'\cup \overline{I'}=[-\infty,\infty]$, the second step follows from for any $f\in \R$, $\widehat{G}^{ (j)}_{\sigma,b}(f)\lesssim 1$, the third step follows from for any $f\in \overline{I'}$, $\widehat{G}^{ (j)}_{\sigma,b}(f)\lesssim \delta/k$,
the last step follows from $$\int_{\overline{I'}} |\widehat{g\cdot H}(f) |^2\d f\leq \int_{-\infty}^\infty |\widehat{g\cdot H}(f) |^2\d f=\int_{-\infty}^\infty |{g\cdot H}(t) |^2\d t=T \|{g\cdot H} \|_T^2 \leq T \|{g} \|_T^2.$$
where the first step follows from $\overline{I'}\in[-\infty,\infty]$, the second step follows from Parseval's theorem, the third step follows from $g(t)=0,\forall t\not\in[0, T]$, the last step follows from Remark \ref{rmk:trivial_bound_on_H}.

Next, we consider
\begin{align*}
\E_{\sigma, b}\left[\int_{I'} |\widehat{g\cdot H}(f) |^2\d f\right] \eqsim &~ \frac{1}{B} \int_{-\infty}^\infty |\widehat{g\cdot H}(f) |^2\d f \notag \\
\lesssim &~ \frac{\eps_0\eps_1\eps_2}{k}T \|{g} \|_T^2
\end{align*}
where the first step follows from $\sigma, b$ are chosen randomly, the second step follows from $ \int_{-\infty}^\infty |\widehat{g\cdot H}(f) |^2\d f\leq T \|{g} \|_T^2$.

Thus, by Markov inequality, with probability at least $0.99$, 
\begin{align}\label{eq:well_isolated_3}
    \int_{I'} |\widehat{g\cdot H}(f) |^2\d f\lesssim \frac{\eps_0\eps_1\eps_2}{k}T \|{g} \|_T^2.
\end{align}

Finally, we can conclude that
\begin{align*}
\int_{\overline{I_{f^*}}} |(\wh{H\cdot x}\cdot \wh{G}_{\sigma,b}^{(j)})(f)|^2\d f =&~ \int_{\overline{I_{f^*}}} |(\wh{H\cdot (x^*+g)}\cdot \wh{G}_{\sigma,b}^{(j)})(f)|^2\d f \\
\leq &~ 2\int_{\overline{I_{f^*}}} |\widehat{x^*\cdot H} \cdot \widehat{G}^{ (j)}_{\sigma,b} (f)|^2\d f + 2\int_{\overline{I_{f^*}}} |\widehat{g\cdot H} \cdot \widehat{G}^{ (j)}_{\sigma,b}(f)|^2\d f\notag \\
\lesssim &~ \frac{\delta^2(1+\delta)}{k^2} T \|{x^*}\|_T^2 + 2\int_{\overline{I_{f^*}}} |\widehat{g\cdot H} \cdot \widehat{G}^{ (j)}_{\sigma,b}(f)|^2\d f\notag \\
\lesssim &~ \frac{\delta^2(1+\delta)}{k^2} T \|{x^*}\|_T^2 + \frac{\delta^2 T}{k^2}\|g\|_T^2  +\int_{I'} |\widehat{g\cdot H}(f) |^2\d f \notag \\
\lesssim&~ \frac{\delta^2(1+\delta)}{k^2} T \|{x^*}\|_T^2 +  \frac{\delta^2 T}{k^2}\|g\|_T^2  +\frac{\eps_0\eps_1\eps_2}{k}T \|{g} \|_T^2 \notag \\
= &~ \frac{\delta(1+\delta)}{\eps_0\eps_1\eps_2 k } \eps_0\eps_1\eps_2 T \delta\|{x^*}\|_T^2/k +  (\frac{\delta^2 }{\eps_0\eps_1\eps_2 k} +1 )\eps_0\eps_1\eps_2 T\|{g} \|_T^2/k \notag \\
\leq &~  \eps_0\eps_1\eps_2 T \delta\|{x^*}\|_T^2/k +  2\eps_0\eps_1\eps_2 T\|{g} \|_T^2/k \notag \\
\lesssim &~  \epsilon_0 \cdot T\N_2^2/k,
\end{align*}
where the first step follows from  the definition of $g$, the second step follows from $(a+b)^2\leq 2a^2+2b^2$, the third step follows from Eq.~\eqref{eq:well_isolated_1}, the forth step follows from Eq.~\eqref{eq:well_isolated_2}, the fifth step follows from  Eq.~\eqref{eq:well_isolated_3}, the sixth step is straightforward, the seventh step follows from $ \frac{\delta(1+\delta)}{\eps_0\eps_1\eps_2 k }\leq 1$ and $ (\frac{\delta^2 }{\eps_0\eps_1\eps_2 k} +1 ) \leq 2$, the last step follows from the definition of $\N_2^2$. %

\end{proof}

\subsection{High signal-to-noise ratio (SNR)  band approximation}
In the this  section, we will give the upper bound of $\| x_{S_f}(t) - x_{S}(t)\|_T$.
\begin{definition}[High SNR and  Recoverable Set]\label{def:S_f}
For $j \in [B]$, let $z^*_j(t) : = (x^* \cdot H) \cdot G^{(j)}_{\sigma,b}$, we define the set as follows
\begin{align*}
        S_{g_1} :=  \Big \{ j \in [B] ~{|}~ \| g_j(t) \|^2_T  \leq ( 1 - c \eps) \cdot \| z_j^*(t)\|^2_T \Big \}
\end{align*}
where $c$ is constant.
And we also give the definition  of recoverable set which is the same with $s$ above
\begin{align*}
    S_{g_2} := \Big \{ j \in [B] ~{|}~  \exists f_0, h_{\sigma,b}(f_0) = j ~{\text{and}}~ \int_{f^*-\Delta}^{f^*+\Delta} | \widehat{x\cdot H}(f) |^2 \mathrm{d} f \geq  T\N_2^2/k \Big \}
\end{align*}
where $\N_2^2 :
  = \eps_1\eps_2(\| g(t) \|_T^2 + \delta \| x^*(t) \|_T^2$.
  
And then we define a High SNR and recoverable set as follows
\begin{align*}
    S_g := S_{g_1} \cap S_{g_2}
\end{align*}
Let $S_f := \{ j \in [k] | h_{\sigma,b}(f_g) \in S_g\} \cap S$.
We have  $x_{S_f}(t) := \sum_{j \in S_f} v_j e^{2 \pi i f_j t}$
\end{definition}

\begin{remark}
    In the left part of the paper, we focus on the frequency in set $S_f$ which is a subset of the recoverable frequency set $S$.
\end{remark}

The following lemma shows that for any recoverable frequency (i.e., those satisfy Eq.~\eqref{eq:heavyfrequency}), \textsc{HashToBins} will output a one-cluster and high signal-to-noise ratio signal around it with high probability.

\begin{lemma}[\textsc{HashToBins} for recoverable frequency]\label{lem:z_satisfies_two_properties}
Let $f^*\in [-F,F]$ satisfy:
\begin{equation}
\label{eq:heavyfrequency}
    \int_{f^*-\Delta}^{f^*+\Delta} | \widehat{x\cdot H}(f) |^2 \mathrm{d} f \geq  T\N_2^2/k,
\end{equation}
where $\N_2^2 :
  = \eps_1\eps_2(\| g(t) \|_T^2 + \delta \| x^*(t) \|_T^2)$.

For a random hashing $(\sigma,b)$, let $j=h_{\sigma,b}(f^*)$ be the bucket that $f^*$ maps to under the hash such that $z=(x \cdot H)*G^{(j)}_{\sigma,b}$ and $\wh{z}=\wh{x \cdot H} \cdot \wh{G}^{(j)}_{\sigma,b}$. With probability at least $0.9$, $z(t)$ is an $(\epsilon_0,\Delta_0)$-one-cluster signal around $f^*$.
\end{lemma}
\begin{proof}
The proof consists of two parts. In part 1, we prove that $z(t)$ satisfies Property I of the one-cluster signal around $f^*$ (Definition~\ref{def:one_cluster}). In part 2, we prove that $z(t)$ satisfies Property II of Definition~\ref{def:one_cluster}.

{\bf Part 1.}
  Let region $I_{f^*} = (f^* - \Delta, f^* + \Delta)$ with complement
  $\overline{I_{f^*}} = (-\infty, \infty)\setminus I_{f^*}$.

Next, with probability at least $0.99$, we have that 
 \begin{align*}
 \int_{I_{f^*}} | \widehat{z}(f) |^2 \mathrm{d} f \geq (1-\delta/k)\int_{I_{f^*}} | \widehat{x\cdot H}(f) |^2 \mathrm{d} f \gtrsim T\N_2^2/k
  \end{align*}
  where the probability  follows from $\Delta_0 > 1000 \Delta$, the first step follows from Property I of $G$ in Lemma \ref{lem:property_of_filter_G}, the second step follows from Eq.~\eqref{eq:heavyfrequency}.
  
  On the other hand, $f^*$ is well-isolated with probability $0.9$, thus by the definition of well-isolated, we have that
  \begin{equation*}
    \int_{\overline{I_{f^*}}} | \widehat{z}(f) |^2 \mathrm{d} f \lesssim  \epsilon_0 T\N_2^2/k.
  \end{equation*}
  Hence, $\wh{z}$ satisfies the
  Property I (in Definition \ref{def:one_cluster}) of one-mountain recovery.
  
 {\bf Part 2.}
  By Lemma~\ref{lem:full_proof_of_3_properties_true_for_z}, we know that $(x^*\cdot H)*G^{(j)}_{\sigma,b}$ always satisfies Property \RN{2} (in Definition \ref{def:one_cluster}):
  \begin{align*}
      \int_0^T | x^*(t) H(t) * G_{\sigma,b}^{(j)}(t) |^2 \mathrm{d} t \geq (1-\epsilon_0) \int_{-\infty}^{+\infty} |x^*(t) H(t) * G_{\sigma,b}^{(j)}(t) |^2 \mathrm{d} t
  \end{align*}
As a result, by $[-\infty,\infty]=[-\infty,0]\cup[0, T]\cup[T,\infty]$,
\begin{align}
      \eps_0\int_{-\infty}^{+\infty} | x^*(t) H(t) * G_{\sigma,b}^{(j)}(t) |^2 \mathrm{d} t \geq  \int_{-\infty}^{0} |x^*(t) H(t) * G_{\sigma,b}^{(j)}(t) |^2 \mathrm{d} t + \int_{T}^\infty |x^*(t) H(t) * G_{\sigma,b}^{(j)}(t) |^2 \mathrm{d} t \label{eq:outside_0T_is_small}
  \end{align}

  Then, we claim that
 \begin{align}\label{eq:xHG_infty_is_at_least_xH_inside_region}
\int_{-\infty}^{\infty} | x(t) \cdot H(t) * G_{\sigma,b}^{(j)}(t) |^2 \mathrm{d} t = &~ \int_{-\infty}^{\infty} | \widehat{ x\cdot H } (f) \cdot \widehat{G}_{\sigma,b}^{(j)}(f) |^2 \mathrm{d} f \notag \\
\geq &~ \int_{f^*-\Delta}^{f^*+\Delta} | \widehat{ x\cdot H } (f) \cdot \widehat{G}_{\sigma,b}^{(j)}(f)|^2 \mathrm{d} f \notag \\
\gtrsim &~ \int_{f^*-\Delta}^{f^*+\Delta} | \widehat{ x\cdot H } (f) |^2 \mathrm{d} f  \notag \\
\geq &~ T \N_2^2/k ,
\end{align}
where the first step follows from Parseval's theorem, the second step follows from $[f^*-\Delta,f^*+\Delta]\subset [-\infty,\infty]$, the third step holds with probability at least $0.99$ and follows from $\Delta_0 > 1000 \Delta$ and Property \RN{1} of Lemma \ref{lem:property_of_filter_G}, the last step follows from the definition of $f^*$. 

By Definition~\ref{def:S_f}, we have that 
\begin{align}\label{eq:bound_noise_H_convG}
    \int_{-\infty}^{+\infty} | g(t) \cdot H(t) * G_{\sigma,b}^{(j)}(t) |^2 \mathrm{d} t 
    = & ~  \int_{0}^{T} | g(t) \cdot H(t) * G_{\sigma,b}^{(j)}(t) |^2 \mathrm{d} t \\
     \leq  & ~ (1 - c \eps) \int_{0}^{T} | z_j^*(t) |^2 \mathrm{d} t \notag \\
    \leq  & ~ (1- c \eps) \int_{-\infty}^{+\infty} | z_j^*(t) |^2 \mathrm{d} t \notag \\
    \leq & ~  \int_{-\infty}^{+\infty} (1 - c \eps) | x(t) \cdot H(t) * G_{\sigma,b}^{(j)}(t) |^2 \mathrm{d} t \notag 
\end{align}
where the first step from $g(t)=0,\forall t\not\in[0, T]$, the second step follows from Definition~\ref{def:S_f},
the third step follows from simple algebra,
the last step is due to Definition of $z_j^*(t)$.

Then, we claim that
\begin{align}
\sqrt{\int_{-\infty}^\infty |x^*\cdot H * G^{(j)}_{\sigma,b}|^2 \d t} \leq &~ \sqrt{\int_{-\infty}^\infty |(x^*+g)\cdot H * G^{(j)}_{\sigma,b}|^2 \d t}  + \sqrt{\int_{-\infty}^\infty |g\cdot H * G^{(j)}_{\sigma,b}|^2 \d t} \notag\\
\leq &~ (1+ \sqrt{\eps_0\eps_2})\sqrt{\int_{-\infty}^\infty |(x^*+g)\cdot H * G^{(j)}_{\sigma,b}|^2 \d t} \notag\\
\lesssim&~ \sqrt{\int_{-\infty}^\infty |(x^*+g)\cdot H * G^{(j)}_{\sigma,b}|^2 \d t} \label{eq:bound_x_star_H_convG}
\end{align}
where the first step follows from triangle inequality, the second step follows from Eq.~\eqref{eq:bound_noise_H_convG}, the last step follows from $\eps_0,\eps_2\leq 1$.

Next, we consider
\begin{align}
\sqrt{\int_T^\infty |(x^*+g)\cdot H * G^{(j)}_{\sigma,b}|^2 \d t} \leq &~ \sqrt{\int_T^\infty |x^*\cdot H * G^{(j)}_{\sigma,b}|^2 \d t} + \sqrt{\int_T^\infty |g\cdot H * G^{(j)}_{\sigma,b}|^2 \d t} \notag \\
\leq &~ \sqrt{\eps_0\int_{-\infty}^\infty |x^*\cdot H * G^{(j)}_{\sigma,b}|^2 \d t} + \sqrt{\int_T^\infty |g\cdot H * G^{(j)}_{\sigma,b}|^2 \d t} \notag \\
\leq &~ \sqrt{\eps_0\int_{-\infty}^\infty |x^*\cdot H * G^{(j)}_{\sigma,b}|^2 \d t} + \sqrt{\epsilon_0 \int_{-\infty}^\infty |x \cdot H * G^{(j)}_{\sigma,b}|^2 \d t} \notag \\ 
\lesssim &~  \sqrt{\eps_0 \int_{-\infty}^\infty |x \cdot H * G^{(j)}_{\sigma,b}|^2 \d t},\label{eq:eps_delta_final_result_1}
\end{align}
where the first step follows from triangle inequality, the second step follows from Eq.~\eqref{eq:outside_0T_is_small}, the third step follows from Eq.~\eqref{eq:bound_noise_H_convG}, the forth step follows from Eq.~\eqref{eq:bound_x_star_H_convG}. %

Similarly, 
\begin{align}
\sqrt{\int_{-\infty}^0 |(x^*+g)\cdot H * G^{(j)}_{\sigma,b}|^2 \d t} \lesssim &~  \sqrt{\eps_0 \int_{-\infty}^\infty |x \cdot H * G^{(j)}_{\sigma,b}|^2 \d t}\label{eq:eps_delta_final_result_2}
\end{align}

Combine equations above, we have that,
\begin{align*}
&~\sqrt{\int_{-\infty}^0 |(x^*+g)\cdot H * G^{(j)}_{\sigma,b}|^2 \d t + \int_{T}^\infty |(x^*+g)\cdot H * G^{(j)}_{\sigma,b}|^2 \d t} \\ 
\leq &~ \sqrt{\int_{-\infty}^0 |(x^*+g)\cdot H * G^{(j)}_{\sigma,b}|^2 \d t }+ \sqrt{\int_{T}^\infty |(x^*+g)\cdot H * G^{(j)}_{\sigma,b}|^2 \d t} \\
\lesssim &~  \sqrt{\eps_0 \int_{-\infty}^\infty |x \cdot H * G^{(j)}_{\sigma,b}|^2 \d t}
\end{align*}
where the first step follows from $\sqrt{a+b}\leq \sqrt{a}+\sqrt{b}$, the second step follows from Eq.~\eqref{eq:eps_delta_final_result_1} and Eq.~\eqref{eq:eps_delta_final_result_2}.

Hence, we have that $z=(x^*+g)\cdot H * G^{(j)}_{\sigma,b}$ satisfies Property \RN{2} (in Definition \ref{def:one_cluster}) with probability $0.95$.

\end{proof}

\subsection{Ultra-high sensitivity frequency estimation}
\label{sec:high_acc:ultra_high}

In this section, we improve the high sensitivity frequency estimation in Section~\ref{sec:high_sensitive} with even higher sensitivity, using the results in previous sections. More specifically, we show how to estimate the frequencies of the signal $x_S$ whose frequencies are only $\epsilon^2{\cal N}$-heavy, while in section~\ref{sec:high_sensitive} the recoverable signal's frequencies are ${\cal N}$-heavy.

\begin{lemma}[Frequency estimation for one-cluster signal, \cite{ckps16}]
\label{lem:findfrequency}
For a sufficiently small constant $\eps_0>0$, any $f_0 \in [-F,F]$, and $\Delta_0>0$, given an $(\eps_0,\Delta_0)$-one-cluster signal $z(t)$ around $f_0$, Procedure \textsc{FrequencyRecovery1Cluster}, 
returns $\wt{f_0}$ with $|\widetilde{f}_0 -f_0| \lesssim \Delta_0 \cdot \sqrt{\Delta_0 T} $ with probability at least $1-2^{-\Omega(k)}$.
\end{lemma}

The following theorem shows the algorithm for ultra-high sensitivity frequency estimation. 
\begin{theorem}[Ultra-high sensitivity frequency estimation algorithm with low success probability]\label{thm:frequency_recovery_k_cluster_ours_low_prob}
Let $x^*(t) = {\sum_{j=1}^k}  v_j e^{2\pi\i f_j t}$ and $x(t)= x^*(t) +g(t)$ be our observable signal where $\|g(t) \|_T^2 \le c\|x^*(t)\|_T^2$ for a sufficiently small constant $c$. Then
  Procedure \textsc{FrequencyRecoveryKCluster} %
  returns a set $L$ of $O(k/(\eps_0\eps_1\eps_2)) $
  frequencies that cover all $\N_2$-heavy clusters and have high SNR (See Definition~\ref{def:S_f}) of $x^*$, which uses $\poly(k, \eps^{-1},\eps_0^{-1},\eps_1^{-1},\eps_2^{-1},  \log(1/\delta) ) \log(FT)$ samples and $\poly(k, \eps^{-1},\eps_0^{-1},\eps_1^{-1}, \eps_2^{-1},  \log(1/\delta)) \log^2 (FT)$ time. 

In particular, for $\Delta_0=\eps^{-1}\poly(k,\log(1/\delta))/T$ and $\N_2^2 :
  = \eps_1\eps_2(\| g(t) \|_T^2 + \delta \| x^*(t) \|_T^2)$, with probability $0.9$, for any $f^*$ with
  \begin{equation}%
    \int_{f^*-\Delta}^{f^*+\Delta} | \widehat{x\cdot H}(f) |^2 \mathrm{d} f \geq  T\N_2^2/k,
  \end{equation}
there  exists an $\wt{f} \in L$ satisfying
  \begin{equation*}
  |f^*-\widetilde{f} |
  \lesssim \Delta_0 \sqrt{\Delta_0 T}.
  \end{equation*}
\end{theorem}

\begin{proof}
By Lemma \ref{lem:z_satisfies_two_properties} and Lemma \ref{lem:findfrequency}, we prove the theorem. 
\end{proof}

\begin{theorem}[Ultra-high sensitivity frequency estimation algorithm with high success probability]\label{thm:frequency_recovery_k_cluster_ours}
Let $x^*(t) = {\sum_{j=1}^k}  v_j e^{2\pi\i f_j
    t}$ and $x(t)= x^*(t) +g(t)$ be our observable signal where $\|g(t) \|_T^2 \le
  c\|x^*(t)\|_T^2$ for a sufficiently small constant $c$. Then
  Procedure \textsc{FrequencyRecoveryKCluster} %
  returns a set $L$ of $O(k/(\eps_0\eps_1\eps_2)) $
  frequencies that covers all $\N_2$-heavy clusters of $x^*$, which uses $\poly(k, \eps^{-1},\eps_0^{-1},\eps_1^{-1}, \eps_2^{-1},  \log(1/\delta) ) \log(FT)$ samples and $\poly(k, \eps^{-1},\eps_0^{-1},\eps_1^{-1}, \eps_2^{-1},  \log(1/\delta)) \log^2 (FT)$ time. 

In particular, for $\Delta_0=\eps^{-1}\poly(k,\log(1/\delta))/T$ and $\N_2^2 :
  = \eps_1\eps_2(\| g(t) \|_T^2 + \delta \| x^*(t) \|_T^2)$, with probability $1- 2^{-\Omega(k)}$, for any $f^*$ with
  \begin{equation}%
    \int_{f^*-\Delta}^{f^*+\Delta} | \widehat{x\cdot H}(f) |^2 \mathrm{d} f \geq  T\N_2^2/k,
  \end{equation}
there  exists an $\wt{f} \in L$ satisfying
  \begin{equation*}
  |f^*-\widetilde{f} |
  \lesssim \Delta_0 \sqrt{\Delta_0 T}.
  \end{equation*}
\end{theorem}

The following lemma shows the approximation error guarantee for the recoverable signal $x_S$ of the ultra-high sensitivity frequency estimation algorithm (Theorem~\ref{thm:frequency_recovery_k_cluster_ours}).

\begin{lemma}[Recoverable signal's approximation error guarantee]\label{lem:useful_delta_bounding}
Let $x^*(t) = \sum_{j=1}^k v_j e^{2\pi\i f_j t}$ and $x(t)= x^*(t) +g(t)$ be our observable signal. Let $\N_1^2 := \eps_1(\| g(t) \|_T^2 + \delta \| x^*(t) \|_T^2)$. Let $C_1,\cdots,C_l$ are the $\N_1$-heavy clusters from Definition \ref{def:heavy_clusters}. Let $S^*$ denotes the set of frequencies $f^*\in \{f_j\}_{j\in[k]}$ such that, $f^*\in C_i$ for some $i\in [l]$, and 
\begin{align*}
\int_{C_i} | \widehat{x^*\cdot H}(f) |^2 \mathrm{d} f \geq T\N_1^2/k,    
\end{align*}

 Let $S$ denotes the set of frequencies $f^*\in S^*$ such that, $f^*\in C_{j}$ for some $j\in [l]$, and 
\begin{align*}
\int_{C_{j}} | \widehat{x\cdot H}(f) |^2 \mathrm{d} f \geq \eps_2 T\N_1^2/k,    
\end{align*}

Then, we have that,
\begin{align*}
\|x-x_S\|_T + \|x_S - x^*\|_T \leq (1+\sqrt{2}+O(\sqrt{\eps}))\|g\|_T + O(\sqrt{\delta})\|x^*\|_T.
\end{align*}
\end{lemma}
\begin{proof}

Following from the fact that $\sqrt{1+\eps}= 1+O(\eps)$ for $\eps < 1$, we have $$\N_1 = \sqrt{\eps_1(\| g \|_T^2 + \delta \| x^* \|_T^2)}\leq \sqrt{\eps_1}\|g\|_T+\sqrt{\delta\eps_1}\|x^*\|_T$$

We have that
\begin{align}
 \|x^*-x_S\|_T \leq &~ \|x_{S^*}-x_S\|_T + \|x^*-x_{S^*}\|_T \notag \\ 
 \leq &~ (1+O(\sqrt{\epsilon_2}))\|x-x_{S^*}\|_T + \|x^*-x_{S^*}\|_T \notag\\
 \leq &~ (1+O(\sqrt{\epsilon_2})) \|x-x^*\|_T + (2+O(\sqrt{\epsilon_2}))\|x^*-x_{S^*}\|_T \notag\\
 \leq &~ (1+O(\sqrt{\epsilon_2}))\|g\|_T + (2+O(\sqrt{\epsilon_2}+\epsilon))\N_1 %
 \label{eq:useful_delta_bounding_1}
\end{align}
where the first step follows from triangle inequality, the second step follows from Corollary \ref{cor:use_g_bound_x}, the third step follows from triangle inequality, the forth step follows from Claim \ref{cla:guarantee_removing_x**_x*_ours}.

Thus, we have that
\begin{align}
\|x-x_{S^*}\|_T \leq&~ \|x-x^*\|_T + \|x^*-x_{S^*}\|_T \notag \\
\leq&~ \|g\|_T +\|x^*-x_{S^*}\|_T\notag\\
\leq &~ \|g\|_T+(1+\epsilon)\N_1 \label{eq:useful_delta_bounding_2}
\end{align}
where the first step follows from triangle inequality, the second step follows from the definition of $g$, the third step follows from Claim \ref{cla:guarantee_removing_x**_x*_ours}. %

Therefore,
\begin{align*}
&~\|x-x_S\|_T + \|x_S - x^*\|_T \\
\leq &~ (\|H(x-x_S)\|_T + \|g\|_T + O(\eps)\|x^*-x_S\|_T) + \|x_S - x^*\|_T\\
\leq &~ (\|H(x-x_S)\|_T + \|g\|_T + O(\eps)\|x^*-x_S\|_T) +\|x_S-x_{S^*}\|_T +\|x_{S^*} - x^*\|_T\\
\leq&~ (\|H(x-x_S)\|_T + \|g\|_T + O(\eps)\|x^*-x_S\|_T) + (1+2\eps)\|H(x_S - x_{S^*})\|_T + \|x_{S^*}-x^*\|_T \\
= &~ \|g\|_T + O(\eps)\|x^*-x_S\|_T + (1+O(\eps))(\|H(x-x_S)\|_T+\|H(x_S - x_{S^*})\|_T ) + \|x_{S^*}-x^*\|_T\\
\leq&~  \|g\|_T + O(\eps)\|x^*-x_S\|_T + (1+O(\eps))(\|H(x-x_S)\|_T+\|H(x_S - x_{S^*})\|_T ) + (1+\eps)\N_1 \\
\leq&~  \|g\|_T + O(\eps)\|x^*-x_S\|_T + (1+O(\eps))\sqrt{2}\sqrt{\|H(x-x_S)\|^2_T+\|H(x_S - x_{S^*})\|^2_T } + (1+\eps)\N_1 \\
\leq&~  \|g\|_T + O(\eps)\|x^*-x_S\|_T + (1+O(\eps))(1+O(\sqrt{\eps_2}))\sqrt{2}\|x-x_{S^*}\|_T + (1+\eps)\N_1 \\
\leq&~  \|g\|_T + O(\eps)((1+O(\sqrt{\epsilon_2}))\|g\|_T + (2+O(\sqrt{\epsilon_2}+\epsilon))\N_1 ) \\
& ~ + (1+O(\eps))(1+O(\sqrt{\eps_2}))\sqrt{2}\|x-x_{S^*}\|_T + (1+\eps)\N_1 \\
\leq&~  \|g\|_T + O(\eps)((1+O(\sqrt{\epsilon_2}))\|g\|_T + (2+O(\sqrt{\epsilon_2}+\epsilon))\N_1 )\\
&~+ (\sqrt{2}+O(\eps+\sqrt{\eps_2}))(\|g\|_T + (1+\eps)\N_1) + (1+\eps)\N_1 \\
\leq &~ (1+\sqrt{2}+O(\sqrt{\eps}))\|g\|_T + O(\sqrt{\delta})\|x^*\|_T,
\end{align*}
where the first step follows from Lemma \ref{lem:decompose_noise_by_filter}, the second step follows from triangle inequality, the third step follows from $x_S - x_{S^*}$ being $k$-Fourier-sparse and Property \RN{6} of Lemma \ref{lem:property_of_filter_H}, the forth step change the order of the terms, the fifth step follows from Claim \ref{cla:guarantee_removing_x**_x*_ours}, the sixth step follows from  $\|H(x-x_S)\|_T+\|H(x_S-x_{S^*})\|_T\leq \sqrt{2}\sqrt{\|H(x-x_S)\|_T^2+\|H(x_S-x_{S^*})\|_T^2}$, the seventh step follows from Lemma \ref{lem:use_g_bound_x_under_H}, 
the eighth step follows from Eq.~\eqref{eq:useful_delta_bounding_1}, the ninth step follows from Eq.~\eqref{eq:useful_delta_bounding_2}, the last step follows from $\eps=\eps_0=\eps_1=\eps_2$. 
\end{proof}
The following lemma shows that the recoverable signal $x_S(t)$'s energy is close to the observation signal $x(t)$.
\begin{lemma}[Recoverable signal's energy]\label{lem:bound_x_S}
Let $x^*(t) = \sum_{j=1}^k v_j e^{2\pi\i f_j t}$ and $x(t)= x^*(t) +g(t)$ be our observable signal. Let $\N_1^2 := \eps_1(\| g(t) \|_T^2 + \delta \| x^*(t) \|_T^2)$. Let $C_1,\cdots,C_l$ are the $\N_1$-heavy clusters from Definition \ref{def:heavy_clusters}. Let $S^*$ denotes the set of frequencies $f^*\in \{f_j\}_{j\in[k]}$ such that, $f^*\in C_i$ for some $i\in [l]$, and 
\begin{align*}
\int_{C_i} | \widehat{x^*\cdot H}(f) |^2 \mathrm{d} f \geq T\N_1^2/k,    
\end{align*}

 Let $S$ denotes the set of frequencies $f^*\in S^*$ such that, $f^*\in C_{j}$ for some $j\in [l]$, and 
\begin{align*}
\int_{C_{j}} | \widehat{x\cdot H}(f) |^2 \mathrm{d} f \geq \eps_2 T\N_1^2/k,    
\end{align*}

Then, we have that, 
\begin{align*}
    \|x_S\|_T \lesssim \|g\|_T + \|x^*\|_T
\end{align*}
\end{lemma}
\begin{proof}
We have that, 
\begin{align*}
    \|x_S\|_T \leq&~ \|x_{S^*}-x^*\|_T+ \|x_S-x_{S^*}\|_T  + \|x^*\|_T\\
    \lesssim&~ \|x_{S^*}-x^*\|_T+ \|x-x_{S^*}\|_T  + \|x^*\|_T\\
    \lesssim&~ \|x_{S^*}-x^*\|_T+ \|x-x^*\|_T  + \|x^*\|_T\\
    \leq &~  \|g\|_T  + \|x^*\|_T,
\end{align*}
where the first step follows from triangle inequality, the second step follows from  Corollary \ref{cor:use_g_bound_x}, the third step follows from triangle inequality, the forth step follows from Claim \ref{cla:guarantee_removing_x**_x*_ours}. 
\end{proof}

\subsection{High SNR and Recoverable signals}
\begin{lemma}[High SNR and recoverable approximation error guarantee]\label{lem:hign_snr_bounding}
Let $x^*(t) = \sum_{j=1}^k v_j e^{2\pi\i f_j t}$ and $x(t)= x^*(t) +g(t)$ be our observable signal. Let $\N_1^2 := \eps_1(\| g(t) \|_T^2 + \delta \| x^*(t) \|_T^2)$. Let $C_1,\cdots,C_l$ are the $\N_1$-heavy clusters from Definition \ref{def:heavy_clusters}. Let $S^*$ denotes the set of frequencies $f^*\in \{f_j\}_{j\in[k]}$ such that, $f^*\in C_i$ for some $i\in [l]$, and 
\begin{align*}
\int_{C_i} | \widehat{x^*\cdot H}(f) |^2 \mathrm{d} f \geq T\N_1^2/k,    
\end{align*}

 Let $S$ denotes the set of frequencies $f^*\in S^*$ such that, $f^*\in C_{j}$ for some $j\in [l]$, and 
\begin{align*}
\int_{C_{j}} | \widehat{x\cdot H}(f) |^2 \mathrm{d} f \geq \eps_2 T\N_1^2/k,    
\end{align*}
And $S_f$ is defined in Definition~\ref{def:S_f}.
Then, we have that,
\begin{align}\label{eq:lower_bound_gHG}
\| x_{S_f} - x_S \|_T \leq  (1 + O(\eps)) \cdot \| g(t) \|_T
\end{align}
\end{lemma}
\begin{proof}
By Definition~\ref{def:S_f}, we have that
\begin{align*}
    S_f \subseteq S.
\end{align*}
And then for any $f \in S \setminus S_f, j = h_{\sigma,b}(f)$, we have that
\begin{align}
    \| (g \cdot H (t)) * G^{(j)}_{\sigma , b}(t) \|^2_T 
    \geq & ( 1 -  c \cdot \eps ) \|  (x^* \cdot H (t)) * G^{(j)}_{\sigma , b}(t) \|^2_T \notag 
\end{align}
where the first step follows from Definition~\ref{def:S_f}, the second step is from simple algebra.

Let $\mathcal{T} = S \setminus S_f$. And for any $j \in [B]$, 
if $j \in [B] \setminus S_g$, $\mathcal{T}_j = \{ i \in S | h_{\sigma,b}(f_i) = j\}$. Otherwise, $\mathcal{T}_j = \emptyset$.
Moreover, we have that for any $f \in \supp(\widehat{x}_{\mathcal{T}_j}*\widehat{H})$,
\begin{align}
    \widehat{G}^{(j)}_{\sigma,b}(f) \geq 1 - \frac{\delta}{k}
\end{align}

From $\mathrm{Property~\RN{6}}$ of Lemma~\ref{lem:property_of_filter_H}, we have that
\begin{align}\label{eq:upper_bound_H}
     (1-\epsilon) \cdot \int_{-\infty}^{+\infty} |x^*(t) |^2 \mathrm{d} t \leq \int_{-\infty}^{+\infty} |x^*(t) \cdot H(t) |^2 \mathrm{d} t \leq  \int_{-\infty}^{+\infty} |x^*(t) |^2 \mathrm{d} t.
\end{align}
By Lemma~\ref{lem:full_proof_of_3_properties_true_for_z}, we know that $(x^*\cdot H)*G^{(j)}_{\sigma,b}$ always satisfies Property \RN{2} (in Definition \ref{def:one_cluster}):
  \begin{align}\label{eq:upper_bound_g_x}
      & ~ T \| x^*(t) H(t) * G_{\sigma,b}^{(j)}(t) \|^2_T \notag \\
      = & ~\int_0^T | x^*(t) H(t) * G_{\sigma,b}^{(j)}(t) |^2 \mathrm{d} t \notag \\
      \geq & ~ (1-\epsilon_0) \int_{-\infty}^{+\infty} |x^*(t) H(t) * G_{\sigma,b}^{(j)}(t) |^2 \mathrm{d} t \notag \\
      = & ~  (1-\epsilon_0) \int_{-\infty}^{+\infty} |(\hat{x}^*(f) * \hat{H}(f)) *\hat{G}_{\sigma,b}^{(j)}(f) |^2 \mathrm{d} f \notag \\
    = & ~  (1-\epsilon_0)( \int_{-\infty}^{+\infty} |(\hat{x}^*(f) * \hat{H}(f)) \cdot \hat{G}_{\sigma,b}^{(j)}(f) |^2 \mathrm{d} f + 
       \int_{-\infty}^{+\infty} |(\hat{x}^*(f) * \hat{H}(f)) \cdot \hat{G}_{\sigma,b}^{(j)}(f) |^2 \mathrm{d} f
       )\notag \\
       \geq & ~ (1-\epsilon_0) \cdot \int_{-\infty}^{+\infty} |(\hat{x}^*(f) * \hat{H}(f)) \cdot \hat{G}_{\sigma,b}^{(j)}(f) |^2 \mathrm{d} f \notag \\
        \geq & ~ (1-\epsilon_0) \cdot \int_{-\infty}^{+\infty} |(\hat{x}^*(f) * \hat{H}(f))|^2 \mathrm{d} f
  \end{align}
where the first step follows from the definition of the norm, the second step is from Lemma~\ref{lem:full_proof_of_3_properties_true_for_z}, the third step is due to Parseval's Theorem, the forth step is based on the Large Offset event not happening, the fifth step is based on simple algebra, the last step is because of Lemma~\ref{lem:full_proof_of_3_properties_true_for_z}. 

We also have that
\begin{align}\label{eq:upper_bound_sf_s}
    & ~ T \| x_{S_f}(t) - x_{S}(t) \|_T^2 \notag \\
    = & ~ T \| x_{\mathcal{T}} \|_T^2 \notag \\
    \leq  & ~ (T/(1 - \eps)^2)\cdot \| x_{\mathcal{T}}(t) \cdot H(t) \|^2_T \notag \\
    = & ~ \frac{1}{1 - \eps^2} \cdot\int_{0}^T |x_{\mathcal{T}}(t) \cdot H(t) |^2 \mathrm{d} t \notag \\
    \leq & ~ \frac{1}{1 - \eps^2} \cdot \int_{-\infty}^{\infty} |x_{\mathcal{T}}(t) \cdot H(t) |^2 \mathrm{d} t \notag \\
    = & ~ \frac{1}{1 - \eps^2} \cdot \int_{-\infty}^{\infty} |\hat{x}_{\mathcal{T}}(f) * \hat{H}(f) |^2 \mathrm{d} f \notag \\
    = & ~ \frac{1}{1 - \eps^2} \cdot \sum_{j = 1}^B \int_{-\infty}^{\infty} |\hat{x}_{\mathcal{T}_j}(f) * \hat{H}(f) |^2 \mathrm{d} f \notag \\
    \leq & ~ \frac{k^2}{(1 - \eps)^2(k - \delta)^2} \cdot \sum_{j \in B \setminus S_g} T \| (x^*(t) \cdot H(t)) * G^{(j)}_{\sigma,b}(t))\|^2_T \notag \\
    \leq & ~ \frac{k^2 }{(1 - c \eps)(1 - \eps)^2(k - \delta)^2} \cdot\sum_{j \in B \setminus S_g} T \| (g(t) \cdot H(t)) * G^{(j)}_{\sigma,b}(t))\|^2_T
\end{align}
where the first step follows from Definition of $\mathcal{T}$, the second step follows from Eq.~\eqref{eq:upper_bound_H}, the third step is based on definition of norm, the forth step follows from simple algebra, the fifth step follows from Parseval's Theorem, the six step is due to Large Offset event not happening, the seventh step is due to Lemma~\ref{lem:full_proof_of_3_properties_true_for_z}, the eighth step follows from Eq.~\eqref{eq:lower_bound_gHG}.

In the following, we have that 
\begin{align}\label{upper_bound_gHG}
    & ~ \sum_{j \in [B]} T \cdot \| (g(t) \cdot H(t)) * G^{(j)}_{\sigma,b}(t)\|^2_T \notag \\
    \leq & ~ \sum_{j \in [B]} \int_{0}^T | (g^*(t) \cdot H(t)) * G^{(j)}_{\sigma,b}(t) |^2 \mathrm{d} t \notag \\
    \leq & ~ \sum_{j \in [B]} \int_{-\infty}^\infty | (g^*(t) \cdot H(t)) * G^{(j)}_{\sigma,b}(t) |^2 \mathrm{d} t \notag \\
     = & ~ \sum_{j \in [B]} \int_{-\infty}^\infty 
    |(\hat{g}(f) * \hat{H}(f)) \cdot \hat{G}_{\sigma,b}^{(j)}(f) |^2 \mathrm{d} f \notag \\
    \leq & ~ \frac{k^2 }{(k - \delta)^2}
    \int_{-\infty}^\infty 
    |\hat{g}(f) * \hat{H}(f)|^2  \mathrm{d} f \notag \\
    = & ~ 
    \frac{k^2 }{(k - \delta)^2} \cdot \int_{-\infty}^\infty 
    |g(t) \cdot H(t)|^2  \mathrm{d} t \notag \\
    = & ~ 
    \frac{k^2 }{(k - \delta)^2} \cdot \int_{0}^T
    |g(t) \cdot H(t)|^2  \mathrm{d} t \notag \\
    \leq & ~ 
    \frac{k^2 }{(k - \delta)^2} \cdot \int_{0}^T
    |g(t)|^2  \mathrm{d} t \notag \\
    = & ~ T \frac{k^2 }{(k - \delta)^2} \| g(t) \|^2_T
\end{align}
where the first step is due to the definition of norm, the second step follows from $g(t) = 0$ when $t \notin [0,T]$, the third step follows from Parseval's Theorem, the forth step is because of Lemma~\ref{lem:full_proof_of_3_properties_true_for_z}, the fifth step is from Parseval's Theorem, the sixth step is based on  $g(t) = 0$ when $t \notin [0,T]$, the seventh step is from $|H(t)|^2 \leq 1$, the last step is from the definition of norm.
We have that
\begin{align*}
    & ~ T \| x_{S_f}(t) - x_S(t) \|_T^2   \\
    \leq & ~ \frac{k^2 }{(1  -  c \eps) (1 - \eps)^2(k - \delta)^2} \cdot \sum_{j \in B \setminus S_g} T \| (g^*(t) \cdot H(t)) * G^{(j)}_{\sigma,b}(t))\|^2_T \\ 
    \leq & ~ \frac{k^2 }{(1  -  c \eps) (1 - \eps)^2(k - \delta)^2} \sum_{j \in [B]} T \| (g^*(t) \cdot H(t)) * G^{(j)}_{\sigma,b}(t))\|^2_T \\
    \leq & ~ \frac{k^4}{(1  -  c \eps) (1 - \eps)^2(k - \delta)^4}T \|g(t) \|^2_T \\
    \leq & ~ (1 + O(\eps)) T \|g(t) \|^2_T
\end{align*}
where the first step follows from Eq.~\eqref{eq:upper_bound_sf_s}, the second step follows from simple algebra, the third step is due to Eq.\eqref{eq:upper_bound_g_x}, the forth step is because of the reason that $\delta$ is much smaller than $\eps$ and $\eps < 1$.
\end{proof}

\begin{lemma}[High SNR signal's energy]\label{lem:bound_x_S_f}
Let $x^*(t) = \sum_{j=1}^k v_j e^{2\pi\i f_j t}$ and $x(t)= x^*(t) +g(t)$ be our observable signal. Let $\N_1^2 := \eps_1(\| g(t) \|_T^2 + \delta \| x^*(t) \|_T^2)$. Let $C_1,\cdots,C_l$ are the $\N_1$-heavy clusters from Definition \ref{def:heavy_clusters}. Let $S^*$ denotes the set of frequencies $f^*\in \{f_j\}_{j\in[k]}$ such that, $f^*\in C_i$ for some $i\in [l]$, and 
\begin{align*}
\int_{C_i} | \widehat{x^*\cdot H}(f) |^2 \mathrm{d} f \geq T\N_1^2/k,    
\end{align*}

 Let $S$ denotes the set of frequencies $f^*\in S^*$ such that, $f^*\in C_{j}$ for some $j\in [l]$, and 
\begin{align*}
\int_{C_{j}} | \widehat{x\cdot H}(f) |^2 \mathrm{d} f \geq \eps_2 T\N_1^2/k,    
\end{align*}
Let $S_f$ be defined in Definition~\ref{def:S_f}. Then, we have that, 
\begin{align*}
    \|x_{S_f}\|_T \leq (1 + O(\eps)) \|g\|_T + \|x^*\|_T
\end{align*}
\end{lemma}
\begin{proof}
We have that, 
\begin{align*}
    \|x_{S_f}\|_T \leq&~ \|x_{S_f} - x_{S}\|_T + \|x_{S^*}-x^*\|_T+ \|x_S-x_{S^*}\|_T  + \|x^*\|_T\\
    \lesssim&~ \|x_{S_f} - x_{S}\|_T + \|x_{S^*}-x^*\|_T+ \|x-x_{S^*}\|_T  + \|x^*\|_T\\
    \lesssim&~ \|x_{S_f} - x_{S}\|_T + \|x_{S^*}-x^*\|_T+ \|x-x^*\|_T  + \|x^*\|_T\\
    \leq &~  \|x_{S_f} - x_{S}\|_T + \|g\|_T  + \|x^*\|_T \\
   \leq  & ~ (1 + O(\eps)) \|g\|_T  + \|x^*\|_T,
\end{align*}
where the first step follows from triangle inequality, the second step follows from  Corollary \ref{cor:use_g_bound_x}, the third step follows from triangle inequality, the forth step follows from Claim \ref{cla:guarantee_removing_x**_x*_ours}, where the last step follows from Lemma~\ref{lem:hign_snr_bounding}.
\end{proof}

\subsection{\texorpdfstring{$(3+\sqrt{2}+\eps)$}{(3+sqrt(2)+epsilon)}-approximate algorithm}\label{sec:high_acc:1_sqrt_2}

In this section, we prove the main result: a $(3+\sqrt{2}+\epsilon)$-approximate Fourier interpolation algorithm, which significantly improves the accuracy of \cite{ckps16}'s result.

\begin{theorem}[Fourier interpolation with $(3+\sqrt{2}+\eps)$-approximation error]\label{thm:main_ours_better}
Let $x(t) = x^*(t) + g(t)$, where $x^*$ is $k$-Fourier-sparse signal with  frequencies in $[-F, F]$.  Given samples of $x$ over $[0, T]$ we can output $y(t)$ such that with probability at least $1-2^{-\Omega(k)}$, 
 \[
\|y - x^*\|_T \leq (3+\sqrt{2}+\eps)\|g\|_T + \delta \|x^*\|_T.
 \]
 Our algorithm uses $\poly(k,\eps^{-1},\log (1/\delta) ) \log(FT)$
  samples and $\poly(k,\eps^{-1},\log(1/\delta)) \cdot \log^2(FT)$ time.  The output $y$
  is $\poly(k,\eps^{-1},\log(1/\delta))$-Fourier-sparse signal. 
\end{theorem}
\begin{proof}
Let $\N_2^2 :=\eps_1\eps_2( \| g(t) \|_T^2 + \delta \| x^*(t) \|_T^2)$, $\N_1^2 :=\eps_1( \| g(t) \|_T^2 + \delta \| x^*(t) \|_T^2)$ be the heavy cluster parameter. 

First, by Lemma \ref{cla:guarantee_removing_x**_x*_ours}, there is a set of frequencies $S^*\subset [k]$ and $x_{S^*}(t)= \underset{j\in S^*}{\sum} v_j e^{2 \pi \i f_j t}$ such that
\begin{align}
    \|x_{S^*} - x^* \|_T^2 \leq (1+\eps)\N_1^2. \label{eq:2_plus_eps_fourier_intor:x_S_x*}
\end{align}
Furthermore, each $f_j$ with $j\in S^*$ belongs to an $\N_1$-heavy cluster $C_j$ with respect to the filter function $H$ defined in Definition~\ref{def:def_of_filter_H}. 

By Definition \ref{def:heavy_clusters} of heavy cluster, it holds that
\begin{align*}
    \int_{C_j} | \widehat{H\cdot x^*}(f) |^2 \mathrm{d} f \geq T\N_1^2/k.
\end{align*}
By Definition \ref{def:heavy_clusters}, we also have $|C_j|\leq k\cdot \Delta_h$, where $\Delta_h$ is the bandwidth of $\wh{H}$.

Let $\Delta\in \R_+$, and $\Delta > k\cdot \Delta_h$, which implies that $C_j\subseteq [f_j-\Delta, f_j+\Delta]$. Thus, we have
\begin{align*}
\int_{f_j-\Delta}^{f_j+\Delta} | \widehat{H\cdot x^*}(f) |^2 \mathrm{d} f \geq T\N_1^2/k.
\end{align*}

By Corollary \ref{cor:use_g_bound_x}, there is a set of frequencies $S\subset S^*$ and $x_{S}(t)= \underset{j\in S}{\sum} v_j e^{2 \pi \i f_j t}$ such that
\begin{align*}
    \|x_S-x_{S^*}\|_T^2 \leq (1+O(\sqrt{\eps_2}))\|x-x_{S^*}\|_T^2.
\end{align*}

Let $g'=x-x_{S^*}$.

In the following part, we will only focus on recovering the high SNR frequency.
Let $S_f$ be defined in Definition~\ref{def:S_f}. It's to know $S_f \subset S$
By applying Theorem  \ref{thm:frequency_recovery_k_cluster_ours}, there is an algorithm that outputs a set of frequencies $L\subset \R$ such that, $|L|=O(k/(\eps_0\eps_1\eps_2))$, 
and with probability at least $1-2^{-\Omega(k)}$,  for any $f_j$ with $j\in S_f$, there is a $\wt{f}\in L$ such that,
\begin{align*}
|f_j-\widetilde{f} |\lesssim \Delta \sqrt{\Delta T}.
\end{align*}

We define a map $p:\R\rightarrow L$ as follows:
\begin{align*}
    p(f):=\arg \min_{\wt{f}\in L} ~ |f-\wt{f}|~~~\forall f\in \R.
\end{align*}
Then, $x_S(t)$ can be expressed as
\begin{align*}
x_{S_f}(t)= &~\sum_{j\in S_f}v_je^{2\pi \i f_j t}\\
= &~ \sum_{j\in S_f}v_j e^{2\pi \i \cdot p(f_j)t} \cdot e^{2\pi \i \cdot (f_j - p(f_j))t}\\
= &~ \sum_{\wt{f} \in L} e^{2 \pi \i \widetilde{f} t} \cdot \sum_{j\in S_f :~ p(f_j)=\wt{f}} v_j e^{2\pi \i ( f_j - \widetilde{f})t},
\end{align*}
where the first step follows from the definition of $x_S(t)$, the last step follows from interchanging the summations.

For each $\wt{f}_i\in L$, by Corollary \ref{cor:low_degree_approximates_concentrated_freq_ours} with $ x^*=x_S, \Delta= \Delta \sqrt{\Delta T}$, we have that there exist degree $ d=O(T \Delta \sqrt{\Delta T} + k^3 \log k + k \log 1/\delta)$ polynomials $P_i(t)$ corresponding to $\wt{f}_i\in L$ such that, 
\begin{align}
\|x_{S_f}(t)-\sum_{\wt{f}_i \in L} e^{2 \pi \i \widetilde{f}_i t} P_i(t)\|_T \leq \sqrt{\delta} \|x_{S_f}(t)\|_T\label{eq:2_plus_eps_fourier:xS_sum}
\end{align}

Define the following function family: 
\begin{align*}
  \mathcal{F} := \mathrm{span}\Big\{e^{2\pi \i \widetilde{f} t} \cdot t^j~{|}~ \forall \wt{f}\in L, j \in \{0,1,\dots,d\} \Big\}.
\end{align*}
Note that $\sum_{\wt{f}_i \in L} e^{2 \pi \i \widetilde{f}_i t} P_i(t)\in {\cal F}$.

By Claim \ref{cla:max_bounded_Q_condition_number}, for function family $\cal F$, $K_{\mathrm{Uniform[0, T]}} = O((|L| d)^{4} \log^{3} (|L| d))$. 

By Lemma \ref{lem:rho_wbsp}, we have that, choosing a set $W$ of $O(\eps^{-1} K_{\mathrm{Uniform[0, T]}} \log(|L|d/\rho))$
i.i.d. samples uniformly at random over duration $[0, T]$ is a $(\eps,\rho)$-WBSP.

By Lemma \ref{lem:magnitude_recovery_for_CKPS}, there is an algorithm that runs in $O(\eps^{-1}|W|(|L|d)^{\omega-1}\log(1/\rho))$-time using samples in $W$, and outputs $y'(t)\in {\cal F}$ such that, with probability $1-\rho$, 
\begin{align}
    \Big\|y'(t) - \sum_{\wt{f}_i \in L} e^{2 \pi \i \widetilde{f}_i t} P_i(t)\Big\|_T\leq (1+\eps)\Big\|x(t)-\sum_{\wt{f}_i \in L} e^{2 \pi \i \widetilde{f}_i t} P_i(t)\Big\|_T\label{eq:2_plus_eps_fourier:y_sum}
\end{align}

Then by Lemma \ref{lem:polynomial_to_FT}, we have that there is a 
$(k d) $-Fourier-sparse signal $y(t)$, such that
\begin{align}\label{eq:2_plus_eps_approx_y}
    \|y-y'\|_T \leq \delta'
\end{align}
where $\delta'>0$ is any positive real number, thus, $y$ can be arbitrarily close to $y'$.

Moreover, the sparsity of $y(t)$ is 
\begin{align*}
kd = k O(T \Delta \sqrt{\Delta T} + k^3 \log k + k \log 1/\delta) = \poly(k, \eps^{-1},\log(1/\delta)).
\end{align*}

Therefore, the total approximation error can be upper bounded as follows:
\begin{align*}
    &~ \|y-x^*\|_T\\
\leq &~ \|y-y'\|_T+ \Big\|y'-\sum_{\wt{f}_i \in L} e^{2 \pi \i \widetilde{f}_i t} P_i(t)\Big\|_T + \Big\|\sum_{\wt{f}_i \in L} e^{2 \pi \i \widetilde{f}_i t} P_i(t) - x^*\Big\|_T\tag{Triangle inequality}\\
\leq &~ (1+0.1\eps) \Big\|y'-\sum_{\wt{f}_i \in L} e^{2 \pi \i \widetilde{f}_i t} P_i(t)\Big\|_T + \Big\|\sum_{\wt{f}_i \in L} e^{2 \pi \i \widetilde{f}_i t} P_i(t) - x^*\Big\|_T\tag{Eq.~\eqref{eq:2_plus_eps_approx_y}}\\
\leq &~ (1+2\eps) \Big \|x-\sum_{\wt{f}_i \in L} e^{2 \pi \i \widetilde{f}_i t} P_i(t)\Big\|_T + \Big\|\sum_{\wt{f}_i \in L} e^{2 \pi \i \widetilde{f}_i t} P_i(t) - x^*\Big\|_T\tag{Eq.~\eqref{eq:2_plus_eps_fourier:y_sum}}\\
\leq &~ (1+2\eps)(\|x-x_{S_f}\|_T + \|x_{S_f} - x^*\|_T) + 2(1+2\eps)\|\sum_{\wt{f}_i \in L} e^{2 \pi \i \widetilde{f}_i t} P_i(t)-x_{S_f}\|_T \tag{Triangle Inequality}\\
\leq &~ (1+2\eps)(\|x-x_{S}\|_T + 2\|x_{S_f} - x_{S}\|_T + \|x_{S} - x^*\|_T) + 2(1+2\eps)\|\sum_{\wt{f}_i \in L} e^{2 \pi \i \widetilde{f}_i t} P_i(t)-x_{S_f}\|_T \tag{Triangle Inequality}\\
\leq &~ (1+2\eps)(\|x-x_S\|_T + \|x_S - x^*\|_T) + O(\sqrt{\delta}) \|x_{S_f}(t)\|_T + 2(1 + 2 \eps)\|x_{S_f} - x_{S}\|_T \tag{Eq.~\eqref{eq:2_plus_eps_fourier:xS_sum}} \\
\leq &~ (1+2\eps)(1+\sqrt{2}+O(\sqrt{\eps}))\|g\|_T + O(\sqrt{\delta})\|x^*\|_T+ O(\sqrt{\delta}) \|x_{S_f}(t)\|_T +  2(1 + 2 \eps)\|x_{S_f} - x_{S}\|_T \tag{Lemma \ref{lem:useful_delta_bounding}} \\
\leq &~ (1+2\eps)(1+\sqrt{2}+O(\sqrt{\eps}))\|g\|_T + O(\sqrt{\delta})\|x^*\|_T+ O(\sqrt{\delta}) (\|g\|_T+\|x^*\|_T) + 2(1 + 2 \eps)(1+O(\eps)) \|g(t)\|_T \tag{Lemma \ref{lem:bound_x_S}}\\
\leq &~(3 + \sqrt{2}+O(\sqrt{\eps}))\|g\|_T + O(\sqrt{\delta})\|x^*\|_T
\end{align*}

By re-scaling $\eps$ and $\delta$, we prove the theorem.

\end{proof}

%% file: noiseless.tex
\section{High Dimensional Reduction Under Noiseless Assumption}
\label{sec:hd_reduction_noiseless}

This section is organized as follows:
\begin{itemize}
\item Section \ref{sec:nearly_independent} shows that Fourier basis is linear independent.
\item Section \ref{sec:reduction_under_noiseless} shows that, there is straightforward algorithm for signal estimation under noiseless setting.
\end{itemize}

\subsection{Fourier basis is linear independent on randomly sampled points}\label{sec:nearly_independent}

\begin{lemma}\label{lem:fourier_basis_indep}
Given a basis $\mathcal{B}$ of $m$ known vectors $b_1, b_2, \cdots b_m \in \R^d$, let $\Lambda(\mathcal{B}) \subset \R^d$ denote the lattice 
\begin{align*}
    \Lambda(\mathcal{B}) = \Big\{ z \in \R^d : z = \sum_{i=1}^m c_i b_i, c_i \in \mathbb{Z}, \forall i \in [m] \Big\}
\end{align*}
Suppose that $f_1, f_2, \cdots, f_k \in \Lambda(\mathcal{B})$. 
Randomly samples a vector $o\sim {\cal N}(0, I_d)$. Let $t=k^{-1}\cdot (T/2+o\cdot {\min}({\max}_{j\in[k]}|\langle f_j, o\rangle|^{-1}, k^{-1}T))$. Let $t_i:=(i-1)\cdot t$ for $i\in[k]$. 

Let $v_j=(\exp(2\pi\i \langle f_j, t_1 \rangle), \exp(2\pi\i \langle f_j, t_2 \rangle),\cdots, \exp(2\pi\i \langle f_j, t_k \rangle))$ for $j\in [k]$. We have that $v_1, v_2,\cdots, v_k$ are linear independent with probability $1$.
\end{lemma}
\begin{proof}

We have that
\begin{align*}
    u_j:=\begin{bmatrix}
    e^{2\pi\i \langle f_1, 0\cdot t\rangle}\\
    e^{2\pi\i \langle f_1, 1\cdot t\rangle}\\
    \vdots\\
    e^{2\pi\i \langle f_1, (k-1)\cdot t\rangle}\\
    \end{bmatrix}=\begin{bmatrix}
    1\\
    e^{2\pi\i \langle f_1, t\rangle}\\
    \vdots\\
    (e^{2\pi\i \langle f_1, t\rangle})^{k-1}
    \end{bmatrix}=\begin{bmatrix}
        1\\
        w_1\\
        \vdots\\
        w_1^{k-1}
    \end{bmatrix},
\end{align*}
where $w_1:=e^{2\pi\i\langle f_1, t\rangle}$. Similarly, we can define $w_j:=e^{2\pi \i \langle f_j, t\rangle}$. And we have
\begin{align*}
    \begin{bmatrix}
    | & | & & |\\
    u_1 & u_2 & \cdots & u_k\\
    | & | & & |
    \end{bmatrix}=\begin{bmatrix}
    1 & 1 & \cdots & 1\\
    w_1 & w_2 & \cdots & w_k\\
    w_1^2& w_2^2 & \cdots &w_k^2\\
    \vdots & \vdots & \ddots & \vdots\\
    w_1^{k-1}& w_2^{k-1} & \cdots& w_k^{k-1}
    \end{bmatrix},
\end{align*}
which is a Vandermonde matrix. Hence, they are linearly independent as long as $w_1,\dots,w_k$ are distinct. Or equivalently, 
\begin{align*}
    \langle f_1, t\rangle \mod {2\pi}, \cdots, \langle f_k, t\rangle \mod {2\pi}
\end{align*}
are distinct.

Next, we will show that $w_1,\dots,w_k$ are distinct with probability $1$. We first show that: 
\begin{align*}
    \langle f_i,t\rangle\leq \langle f_i, o \rangle / {\max}_{j\in[k]}|\langle f_j,o\rangle| \in [-1, 1] \subset [-\pi, \pi].
\end{align*}

Then, $\langle f_i,t\rangle = \langle f_j, t\rangle \mod {2\pi}$ is equivalent to $\langle f_i,t\rangle = \langle f_j, t\rangle$, and is equivalent to $\langle f_i - f_j, t\rangle=0$, and is equivalent to $\langle f_i - f_j, o\rangle=0$. However, $\langle f_i-f_j, o\rangle$ follows from ${\cal N}(0, \|f_i-f_j\|_2^2)$. By our assumption, $\|f_i-f_j\|_2^2\ne 0$. Thus,
\begin{align*}
    \Pr_{t\sim {\cal N}(0, I_d)}[\langle f_i-f_j, t\rangle=0]=0.
\end{align*}
Therefore, by union bound,
\begin{align*}
    \Pr_{t\sim {\cal N}(0, I_d)}[\exists i\ne j\in [k]:\langle f_i-f_j, t\rangle=0]=0.
\end{align*}
\end{proof}

\subsection{Reduction}
\label{sec:reduction_under_noiseless}

\begin{lemma}[High Dimension Noiseless]\label{lem:reduction_freq_signal_high_D_noiseless}%
Given a basis $\mathcal{B}$ of $m$ known vectors $b_1, b_2, \cdots b_m \in \R^d$, let $\Lambda(\mathcal{B}) \subset \R^d$ denote the lattice 
\begin{align*}
    \Lambda(\mathcal{B}) = \Big\{ z \in \R^d : z = \sum_{i=1}^m c_i b_i, c_i \in \mathbb{Z}, \forall i \in [m] \Big\}
\end{align*}
Suppose that $f_1, f_2, \cdots, f_k \in \Lambda(\mathcal{B})$. Let $x^*(t) = \sum_{j=1}^k v_j e^{2\pi\i \langle f_j , t \rangle }$. Given $x^*(t)$ is observable for $t\in [0, T]^d$. Let $\eta=\min_{i \neq j}\|f_j-f_i\|_{\infty}$.

Given $D, \eta \in \R_+$. Suppose that there is an algorithm $\textsc{FrequencyEstimation}(x^*, k, \rho, d, F, T, \mathcal{B})$ that 
\begin{itemize}
    \item takes $\mathcal{S}(k, \rho, d, F, T, \eta)$ samples, 
    \item runs $\mathcal{T}(k, \rho, d, F, T, \eta)$ time,
    \item output a set $L$ of frequencies such that, for each $f_i$, there is $f'_i\in L$, $|f_i-f'_i|\le D/T$, holds with probability $1-\rho$.
\end{itemize}
 
Let $C=|L| \cdot( D/T + \sqrt{m} \| {\cal B} \| )^m \cdot \frac{ \pi^{m/2} }{ (m/2)! } \cdot \frac{1}{|\det({\cal B})|} $. If $ C\leq d$, then, there is an algorithm ($\textsc{SignalEstimation}(x, k, d, F, T, \mathcal{B})$ %
) that 
\begin{itemize}
\item takes $O(\wt{k}+\mathcal{S}(k, d, F, T, \eta))$ samples, 
\item runs $O(\wt{k}^{\omega}+\wt{k}^{2}d+\mathcal{T}(k, d, F, T, \eta))$ time,
\item output $y(t) = \sum_{j=1}^{\wt{k}} v_j' \cdot \exp(2\pi \i \langle f_j', t \rangle)$ such that
\begin{itemize}
\item $\wt{k}\leq C$,
\item $y(t)=x(t)$, holds with probability $(1-\rho)^2$.
\end{itemize}
\end{itemize}
\end{lemma}
\begin{proof}
First, we recover the frequencies by utilizing the algorithm $\textsc{FrequencyEstimation}(x, k, d, T, F, \mathcal{B})$. Let $L$ be the set of frequencies output by the algorithm $\textsc{FrequencyEstimation}(x, k, d, T, F, \mathcal{B})$. 

We define $\wt{L}$ as follows:  
\begin{align*}
\wt{L}:=\{f\in \Lambda(\mathcal{B})~|~\exists f' \in L, ~ |f'-f|< D/T\}.
\end{align*}
We use $\wt{k}$ to denote the size of set $\wt{L}$.
We use $f'_1,f'_2,\cdots,f'_{\wt{k}}$ to denote the frequencies in the set $\wt{L}$.

By applying Lemma \ref{lem:bound_sparsity}, we have that
\begin{align*}
     \wt{k}=|\wt{L}| \leq |L| \cdot ( D/T + \sqrt{m} \| {\cal B} \| )^m \cdot \frac{ \pi^{m/2} }{ (m/2)! } \cdot \frac{1}{|\det({\cal B})|}.
\end{align*}

Next, we focus on recovering magnitude $v' \in \C^{\wt{k}}$.  
Let $\rho_1:=\rho$. First, we randomly samples a vector $o\sim {\cal N}(0, I_d)$. Let $t=k^{-1}\cdot (T/2+o\cdot {\min}({\max}_{j\in[k]}|\langle f_j, o\rangle|^{-1}, \rho_1 T))$. Let $t_i:=(i-1)\cdot t$ for $i\in[k]$. We have that $|o\cdot {\min}({\max}_{j\in[k]}|\langle f_j, o\rangle|^{-1}, \rho_1 T)| \leq o\rho_1 T $. Then with probability $1-\rho_1$, $|o|\leq 1/(2\rho_1)$, then $t_i\in[0, T]^d,\forall i\in[k]$.

Consider the matrix $A\in \C^{\wt{k} \times \wt{k}}$, where the $(i,j)$-th entry in $A$ is $A_{i, j}=\exp(2\pi\i \langle f'_j, t_i \rangle)$ for each $i \in [\wt{k}]$ and $j \in[\wt{k}]$. Let $b = (x^*(t_1),x^*(t_2),\cdots, x^*(t_{\wt{k}}))^\top \in \C^{\wt{k}}$.

Then, if the following linear system solvable, then we solve the following linear system:
\begin{align}
A v' = b  .\label{eq:linear_system_noiseless}
\end{align}
We output $y(t) = \sum_{j=1}^{\wt{k}} v_j' \cdot \exp(2\pi \i \langle f_j', t \rangle )$, and we have that $y(t)=x^*(t)$.

Finally, by Lemma \ref{lem:fourier_basis_indep}, we have that Eq.~\eqref{eq:linear_system_noiseless} is solvable with a large probability.

\end{proof}

%% file: exact_exponential_algo.tex
\section{Semi-continuous Approximation}
\label{sec:exact_recover_algo}

In this section, we justify the usefulness of the semi-continuous setting by showing that for any $k$-Fourier-sparse, it can be approximated by a $k$-Fourier-sparse semi-continuous signal. This section is organized as follows: %
\begin{itemize}
    \item In Sections~\ref{sec:gaussian_properties} and \ref{sec:CFT}, we give some technical tools on the Gaussian multiplier.
    \item In Section \ref{sec:core_lemma}, we prove the main result of this section (Theorem~\ref{thm:bound_x'_x}), which is incomparable to the existence result in \cite{ckps16}. 
    \item In Section \ref{sec:our_exact_algo}, we give a fast optimal-sparsity Fourier interpolation algorithm with a different error guarantee. 
    \item In Section~\ref{sec:generate_frequency_gap}, we also show that the frequency gap of the approximation can be increased to $\Theta(1/T)$ if we slightly blow up the sparsity of the approximation signal (Corollary~\ref{cor:Fourier-existence-general-case}).
\end{itemize}

\subsection{Properties related to Gaussians}
\label{sec:gaussian_properties}

\begin{definition}[Gaussian Multiplier]\label{def:gaussian-multiplier-2}
For parameters $\mu, \sigma$, we define 
\begin{align*}
M_{\mu, \sigma^2}(x) = e^{-\frac{(x - \mu)^2}{2 \sigma^2}}
\end{align*}
i.e. it is a Gaussian scaled so that its maximum value is $1$.
\end{definition}

We stat a standard result for Gaussian multiplier (see \cite{llm21} for example).
\begin{lemma}\label{lem:approx-constant-2}
Let $0 < \eps < 0.1$ be a parameter.  Let $c$ be a real number such that $0 < c \leq (\log (1/\eps))^{-1/2}$. Let $M$ be defined as in Definition \ref{def:gaussian-multiplier-2}. Define
\begin{align*}
f(x) := \sum_{j = - \infty}^{\infty} \frac{c}{\sqrt{2\pi}} M_{ c j \sigma , \sigma^2}(x) \,.
\end{align*}
Then 
\begin{align*}
1 - \eps^{10} \leq f(x) \leq 1 + \eps^{10}
\end{align*}
for all $x$. Furthermore, let $\alpha \leq 1$ be a parameter and $c = \alpha (\log(1/\eps))^{-1/2}$. Then 
\begin{align*}
1 - \eps^{10/\alpha^2} \leq f(x) \leq 1 + \eps^{10/\alpha^2}
\end{align*}
for all $x$.
\end{lemma}
\begin{proof}
Without loss of generality $\sigma = 1$.  Now the function $f$ is $c$-periodic and even, so we may consider its Fourier expansion %
\begin{align*}
f(x) 
= & ~ a_0 + 2 \sum_{j=1}^{\infty} a_j \cos (2j \pi x / c) \\
= & ~ a_0 + 2 a_1 \cos \left( \frac{2\pi  x}{c}\right) + 2a_2 \cos \left( \frac{4\pi  x}{c}\right)+ \dots 
\end{align*}
and we will now compute the Fourier coefficients.  First note that 
\begin{align*}
a_0  = & ~ \frac{1}{c}\int_0^c f(x) \d x \\
= & ~ \frac{1}{\sqrt{2\pi}}\sum_{j = -\infty}^{\infty} \int_{c(j+1)}^{cj} M_{0, 1}(x) \d x\\
= & ~ 1 .
\end{align*}
where the second step follows from the definition of $ f$.

Next, for any $j \geq 1$, 
\begin{align*}
a_j = & ~ \frac{1}{c} \int_0^c f(x) \cos \left( \frac{2\pi j x}{c}\right) \d x \\
=& ~ \frac{1}{c} \int_0^c \sum_{j = - \infty}^{\infty} \frac{c}{\sqrt{2\pi}} M_{ c j, 1}(x) \cos \left( \frac{2\pi j x}{c}\right) \d x \\
= & ~ \frac{1}{\sqrt{2\pi}} \sum_{l = -\infty}^{\infty} \int_{c(l+1)}^{cl} M_{0, 1}(x) \cos \left( \frac{2\pi j x}{c}\right) \d x \\ 
= & ~ \frac{1}{\sqrt{2\pi}} \sum_{l = -\infty}^{\infty} \int_{c(l+1)}^{cl} \exp(-\frac{x^2}{2}) \cos \left( \frac{2\pi j x}{c}\right) \d x \\ 
= & ~ \frac{1}{\sqrt{2\pi}} \int_{-\infty}^\infty \frac{1}{2}\left( \exp(\frac{-x^2}{2} + \frac{2\pi \i j x}{c}) + \exp(\frac{-x^2}{2} - \frac{2\pi \i j x}{c}) \right) \d x \\
= & ~ \exp(-\frac{2\pi^2 j^2}{c^2}) \,.
\end{align*}
where the second step follows  from the definition of $ f$. The fourth step follows from the definition of $ M$.

Then we can claim that
\begin{align*}
|f(x)-1| \leq & ~ 2\sum_{j=1}^{\infty} \exp(-\frac{2\pi^2 j^2}{c^2}) \\
\leq & ~ 2\sum_{j=1}^{\infty} \exp(-15 j/c^2)\\
= & ~ \frac{2\exp(-15/c^2)}{1-\exp(-15/c^2)} \\
\leq & ~ \frac{2\eps^{15/\alpha^2}}{1-\eps^{15}} \\
\leq & ~ \eps^{10/\alpha^2}\\
\leq & ~ \eps^{10}
\end{align*}
where the first step follows from $a_0=1$, the forth step follows from $\alpha \leq 1$ and $c = \alpha (\log(1/\eps))^{-1/2}$, the fifth step follows from $\eps \in (0, 0.1)$, 
\end{proof}

\begin{claim}
\label{approximate_rectangle_pulse_with_uniform_spaced_gaussian-2}
Let $\eps \in (0, 0.01)$ be a parameter. Let $\eps_0 = 2\eps$, $c = 0.01 / \sqrt{\log (1/\eps) }$, $K = \lceil \frac{1 + 0.5\eps }{c\eps^2} \rceil$. Let $M:\R\rightarrow\R$ be defined as in Definition~\ref{def:gaussian-multiplier-2}. Define $f:\R\rightarrow\R$
\begin{align*}
f(x) = \sum_{j = - K}^{K} \frac{c}{\sqrt{2\pi}} M_{\mu_j,\sigma^2}(x), \text{~~~where~~~} \mu_j= {c j \eps^2 l}, ~~~\sigma^2=\eps^4 l^2 \,.
\end{align*}
Then the following properties is satisfied
\begin{itemize}
\item Part 1. $ f(x) \in [0, 1 + \eps_0]$, for all $x$
\item Part 2. $ f(x) \in [1-\eps_0, 1 + \eps_0]$, for all $x\in [-l, l]$
\item Part 3. $ f(x) \in [0, \eps_0]$, for all $|x| \geq (1+ \eps)l $ %
\end{itemize}
\end{claim}
\begin{proof}
 Using Lemma \ref{lem:approx-constant-2} with setting the following parameter choice $\sigma = \eps^2 l$, $c=0.01(\log(1/\eps))^{-1/2}$, we define
\begin{align*}
f_0(x) = \sum_{j = - \infty}^{\infty} \frac{c}{\sqrt{2\pi}} M_{\mu_j,\sigma^2}\,.
\end{align*}
 we can get that $1-\eps \leq f_0(x) \leq 1 + \eps$.
 
 We will then upper bound the perturbation between $f$ and $f_0$. Firstly, we will provide a upper bound for $M$ when $x/(\eps^2l)-cj  \geq 10$
 \begin{align}
 \frac{c}{\sqrt{2\pi}} M_{\mu_j,\sigma^2}(x) 
 =   & ~ \frac{c}{\sqrt{2\pi}} \exp(-(x-cj\eps^2l)^2/(2\eps^4l^2)) \notag \\
 =& ~  \frac{c}{\sqrt{2\pi}} \exp(-(x/(\eps^2l)-cj)^2/2)\notag \\ \label{eq:upper_bound_M_cjeps^2l_eps^4l^2}
  \leq & ~ \exp(-|x/(\eps^2l)-cj|)
 \end{align}
 where the first step follows from the definition of $M$, the third step follows from $x/(\eps^2l)-cj  \geq 10$.

Next we will claim that
\begin{align}
\sum_{j=0}^{\infty} \exp(-cj)  = & ~\frac{1}{1-\exp(-c)} \notag \\
 \leq & ~\frac{2}{c}\notag \\
 \leq & ~ \frac{200}{\eps} \label{eq:bound_sum_exp_cj}
\end{align}
where the first step follows from $c > 0$ and the sum of geometric sequence, the second step follows from $1-\exp(-x) > x/2$ when $x < 0.1$, the third step follows from $c \geq 0.01 \eps$. %

Because of the symmetry of $M$, We will also claim that 
\begin{align}
    M_{\mu,\sigma^2}(x) = M_{-\mu,\sigma^2}(-x) \label{eq:symmetry_M_mu_sigma}
\end{align}

{\bf {Part 1.}}

Because $M_{\mu,\sigma} = \exp(-(x-\mu)^2/\sigma^2)\geq 0$, we can obtain that $f(x) \geq 0$. Besides, we will claim that
\begin{align*}
    f(x) \leq & ~ f_0(x) \\
    \leq & ~ 1 + \eps \\
    \leq & ~ \eps_0
\end{align*}
where the first step follows from $M_{\mu,\sigma}\geq 0$.

{\bf {Part 2.}} 

Firstly, we upper bound $f$ when $x \in [-l, l]$. To bound $\sum_{j=K+1}^{\infty} \frac{c}{\sqrt{2\pi}} M_{\mu_j,\sigma^2}(x) $, we obtain
\begin{align}
    cj - x/(\eps^2l)  \geq & ~ cj - 1/(\eps^2)\notag \\
     \geq & ~ c(K+1)-1/\eps^2 \notag\\
     \geq & ~ 0.5/\eps \notag \\
     \geq & ~ 10 \label{eq:bound_cj_x_eps_l}
\end{align}
where the first step follows from $x\in[-l,l]$, the second step follows from $j \geq K+1$, the third step follows from $K + 1 \geq (1+0.5\eps)/(c\eps^2)$. %
 
Then we can upper bound 
 \begin{align}
 \sum_{j=K+1}^{\infty} \frac{c}{\sqrt{2\pi}} M_{\mu_j,\sigma^2}(x)  \leq & ~ \sum_{j=K+1}^{\infty}\exp(x/(\eps^2l)-cj)  \notag\\
  \leq  & ~\frac{200}{\eps}\exp(x/(\eps^2l)-c(K+1))  \notag\\
  \leq  & ~\frac{200}{\eps}\exp(-0.5/\eps)  \notag\\
  \leq & ~ 0.1\eps \label{eq:bound_sum_M_K_part1}
 \end{align}
 where the first step follows from Eq.~\eqref{eq:upper_bound_M_cjeps^2l_eps^4l^2} and Eq.~\eqref{eq:bound_cj_x_eps_l}, the second step follows from Eq.~\eqref{eq:bound_sum_exp_cj}, the third step follows from Eq.~\eqref{eq:bound_cj_x_eps_l}, the last step follows from $\eps \in (0, 0.01)$.

We conclude 
 \begin{align*}
|f(x) - 1|  = & ~ \Big| f_0(x) -1 -  \sum_{j=K+1}^{\infty} \frac{c}{\sqrt{2\pi}} M_{\mu_j,\sigma^2}(x) - \sum_{j=-\infty}^{-K-1} \frac{c}{\sqrt{2\pi}} M_{\mu_j,\sigma^2}(x) \Big| \\
 \leq & ~ |f_0(x) -1|+ \Big| \sum_{j=K+1}^{\infty} \frac{c}{\sqrt{2\pi}} M_{\mu_j,\sigma^2}(x) \Big| + \Big| \sum_{j=-\infty}^{-K-1} \frac{c}{\sqrt{2\pi}} M_{\mu_j,\sigma^2}(x) \Big|  \\
 \leq & ~ \eps + 0.1 \eps + 0.1\eps \\
 \leq & ~\eps_0
 \end{align*}
where the first step follows from the definition of $f$ and $f_0$, the second step follows from the triangle inequality, the third step follows from $|f_0(x)-1| \leq \eps$, Eq.\eqref{eq:bound_sum_M_K_part1} and Eq.~\eqref{eq:symmetry_M_mu_sigma}.

{\bf {Part 3.}} 

Secondly, we  will provide the upper bound for $f$ when $|x| > (1+\eps)l$. Without loss of generality, we can only consider $x > (1+\eps)l$ because of Eq.~\eqref{eq:symmetry_M_mu_sigma}. To bound $ \sum_{j=-K}^{K}\frac{c}{\sqrt{2\pi}} M_{\mu_j,\sigma^2}(x) $. We obtain
 \begin{align}
     x/(\eps^2l) - cj  \geq  & ~ x/(\eps^2l) -cK \notag \\
     \geq  & ~ (1+\eps)/\eps^2 -cK\notag \\
      \geq  & ~ 0.25/\eps \notag \\
      \geq  & ~ 10 \label{eq:bound_x_epsl^2_cj}
 \end{align}
 where the first step follows from $j \leq K$, the second step follows from $x> (1+\eps)l$, the third step follows from $K \leq (1+0.5\eps)/(c\eps^2)+1$ and $c \leq 0.01 \leq 0.25 / \eps$.
 
 Then we can upper bound
 \begin{align}
 \sum_{j=-K}^{K}\frac{c}{\sqrt{2\pi}} M_{\mu_j,\sigma^2}(x) \leq & ~\sum_{j=-K}^{K}\exp(cj-x/(\eps^2l))  \notag\\
 \leq & ~\frac{200}{\eps}\exp(cK-x/(\eps^2l))  \notag\\
  \leq & ~\frac{200}{\eps}\exp(-0.25/\eps)  \notag\\
 \leq & ~\eps \notag \\
 \leq & ~\eps_0
 \end{align}
 where the first step follows from Eq.~\eqref{eq:upper_bound_M_cjeps^2l_eps^4l^2} and Eq.~\eqref{eq:bound_x_epsl^2_cj}, the second step follows from Eq.~\eqref{eq:bound_sum_exp_cj}, the third step follows from Eq.~\eqref{eq:bound_x_epsl^2_cj}, the fourth step follows from $\eps \in (0, 0.01)$.

\end{proof}

\begin{lemma}
\label{lem:def_C_0_plus}
Given $ c\in\R, ~\gamma,l,\sigma\in\R_+$, $ \mu\in\R$. Let $ M:\R\rightarrow\R$ be defined as in Definition \ref{def:gaussian-multiplier-2}. For $|a-b|\leq \gamma$, %
We have  that
\begin{align*}
|{{M}_{\mu,\sigma^2}(a)\exp(c \i a)}-{{M}_{\mu,\sigma^2}(b)\exp(c \i b)}| \leq     ({\sigma}^{-1}+|c|)\gamma.
\end{align*}
\end{lemma}
\begin{proof}

We have that
\begin{align}
&~|{M}_{\mu,\sigma^2}(b)-{{M}_{\mu,\sigma^2}(a)}| \notag \\
=&~ |\exp(-(b-\mu)^2/(2\sigma^2))-\exp(-(a-\mu)^2/(2\sigma^2))|\notag \\
=&~ |b-a||(\xi-\mu)\exp(-(\xi-\mu)^2/(2\sigma^2))/(\sigma^2)|\notag \\
\leq &~ \gamma /\sigma\label{eq:bound_delta_M}
\end{align}
where the first step follows from the definition of $M$, the second step defines $\xi\in[a,b] $ and follows from Lagrange's Mean Value Theorem, the last step follows from $ |x\exp(-x^2/2)|\leq 1$.
We have that
\begin{align*}
 &~|{{M}_{\mu,\sigma^2}(a)\exp(c \i a)}-{{M}_{\mu,\sigma^2}(b)\exp(c \i b)}| \\
 \leq &~ |{{M}_{\mu,\sigma^2}(a)\exp(c \i a)}-{{M}_{\mu,\sigma^2}(b)\exp(c \i a)}|+|{{M}_{\mu,\sigma^2}(b)\exp(c \i a)}-{{M}_{\mu,\sigma^2}(b)\exp(c \i b)}|\\
\leq &~ |{M}_{\mu,\sigma^2}(a)-{M}_{\mu,\sigma^2}(b)|+{M}_{\mu,\sigma^2}(b)\cdot|{\exp(c \i a)}-{\exp(c \i b)}|\\
\leq &~ \frac{\gamma}{\sigma}+{M}_{\mu,\sigma^2}(b)\cdot|c(a-b)|\\ 
\leq &~ \frac{\gamma}{\sigma}+|c|\gamma\\
\leq &~ (\frac{1}{\sigma}+|c|)\gamma
\end{align*}
where the first follows from triangle inequality, the second follows from $|\exp(\i \theta)|=1$, the third step follows from Eq.~\eqref{eq:bound_delta_M} and that arc length is greater than chord length, the fourth step follows from $M\leq 1$ and the assume in the statement. 

\end{proof}

\subsection{Continuous Fourier transform}
\label{sec:CFT}

\subsubsection{Bounding the tails}

The goal of this section is to prove Lemma~\ref{lem:tail_bound-2},

\begin{lemma}
\label{lem:tail_bound-2}
Let $ \mu\in\R,\alpha\in\Z_+, \sigma$, $F,F_0,\eps_1 \in\R_+$. Let $ \sigma\in(0,1)$, $ \eps_1\in(0,0.1)$, $F>1$. Let $M:\R\rightarrow\R$ be defined as in Definition \ref{def:gaussian-multiplier-2}. Let $x:\R\rightarrow\C$ be a function. 
Let $ \supp(\wh{x})\subseteq [-F,F]$. 
If 
\begin{align*}
F_0>{\sigma}^{-1} \cdot\log^{1/2}(4F/ \eps_1)+F,
\end{align*}

then we have that,
\begin{align*}
\Big|\int_{-\infty}^{\infty}(\wh{M}_{\mu,\sigma^2}(f)* \wh{x}(f))^\alpha\d f -\int_{-F_0}^{F_0}(\wh{M}_{\mu,\sigma^2}(f)* \wh{x}(f))^\alpha \d f \Big| \leq&~ \|\wh{x}\|^\alpha_\infty \eps_1/{F^{\alpha}}.%
\end{align*}
\end{lemma}
\begin{proof}

First, we will calculate $\wh{M}$. Then, we provide a bound on the tail of $\wh{M}$. Finally, we conclude our proof.

We can claim that
\begin{align*}
 \wh{M}_{\mu,\sigma^2}(f) = \sqrt{2\pi\sigma^2} \cdot \exp(-2\pi\i\mu f)\cdot M_{0,1/(4\pi^2\sigma^2)} .
\end{align*}
Taking the $|\cdot|$ on both sides of the above equation, we get
\begin{align*}
    | \wh{M}_{\mu,\sigma^2}(f) | = \sqrt{2\pi\sigma^2} \cdot | M_{0,1/(4\pi^2\sigma^2)} |
\end{align*}

We start with

\begin{align}\label{eq:tail_bound-2:eq1}
 \int_{F_0-F}^{\infty}|\wh{M}_{\mu,\sigma^2}(f)|^{\alpha} \d f=&~ \sqrt{2\pi\sigma^2} \cdot \int_{F_0-F}^{\infty}| M_{0,1/(4\pi^2\sigma^2)}|^\alpha \d f \notag \\
 =&~ \sqrt{2\pi\sigma^2} \cdot \int_{F_0-F}^{\infty} \exp(-2\alpha\pi^2\sigma^2 f^2) \d f \notag \\
 =&~ \frac{1}{\sqrt{\alpha \pi}} \cdot \int_{\sqrt{2\alpha}\pi\sigma (F_0-F)}^{\infty} \exp(-\xi^2) \d \xi \notag \\
 \leq&~ \frac{1}{\sqrt{\alpha \pi}} \cdot \int_{\sqrt{2\alpha}\pi\sigma (F_0-F)}^{\infty} \exp(-\sqrt{2\alpha}\pi\sigma (F_0-F) \xi ) \d \xi \notag \\
 =&~ \frac{1}{\sqrt{\alpha \pi}} \cdot \frac{1}{\sqrt{2\alpha}\pi\sigma (F_0-F)} \exp(-2\alpha\pi^2\sigma^2 (F_0-F)^2 ) \notag \\
\leq&~ \eps_1/(4F)^{2\alpha}   .
\end{align}
where the second step follows definition of $M$, the third step follows from $\xi = \sqrt{2\alpha}\pi \sigma f$, the forth step follows from $\xi \leq \sqrt{2\alpha} \pi \sigma (F_0 - F)$, the last step follows from lower bound on $F_0$ in the Lemma statement.

Finally, we can bound LHS in the statement  as follows
\begin{align}\label{eq:tail_bound-2:eq2}
&~ \Big| \int_{-\infty}^{\infty}(\wh{M}_{\mu,\sigma^2}(f) *\wh{x}(f))^\alpha\d f - \int_{-F_0}^{F_0} ( \wh{M}_{\mu,\sigma^2}(f) *\wh{x}(f))^\alpha \d f \Big| \notag \\
= &~ \Big| \int_{-\infty}^{F_0}(\wh{M}_{\mu,\sigma^2}(f)* \wh{x}(f))^\alpha\d f +\int_{F_0}^{\infty}(\wh{M}_{\mu,\sigma^2}(f)* \wh{x}(f))^\alpha \d f \Big| \notag \\
= &~ \Big| \int_{-\infty}^{F_0}(\int_{-F}^{F}\wh{M}_{\mu,\sigma^2}(f-\xi)\cdot \wh{x}(\xi)\d \xi)^\alpha\d f +\int_{F_0}^{\infty}(\int_{-F}^{F}\wh{M}_{\mu,\sigma^2}(f-\xi)\cdot \wh{x}(\xi)\d \xi)^\alpha\d f \Big| \notag \\
\leq &~ \Big| \int_{F_0}^{\infty}(\int_{-F}^{F}\wh{M}_{\mu,\sigma^2}(f-\xi)\cdot \wh{x}(\xi)\d \xi)^\alpha\d f \Big|+\Big|\int_{-\infty}^{-F_0}(\int_{-F}^{F}\wh{M}_{\mu,\sigma^2}(f-\xi)\cdot \wh{x}(\xi)\d \xi)^\alpha\d f\Big|\notag \\
\leq &~  \int_{F_0}^{\infty}(\int_{-F}^{F}|\wh{M}_{\mu,\sigma^2}(f-\xi)\cdot \wh{x}(\xi)|\d \xi)^\alpha\d f + \int_{-\infty}^{-F_0}(\int_{-F}^{F}|\wh{M}_{\mu,\sigma^2}(f-\xi)\cdot \wh{x}(\xi)|\d \xi)^\alpha\d f \notag \\
= &~  \|\wh{x}\|^\alpha_\infty\cdot   (\int_{F_0}^{\infty}(\int_{-F}^{F}|\wh{M}_{\mu,\sigma^2}(f-\xi)|\d \xi)^\alpha\d f + \int_{-\infty}^{-F_0}(\int_{-F}^{F}|\wh{M}_{\mu,\sigma^2}(f-\xi)|\d \xi)^\alpha\d f )
\end{align}
where the third step follows from triangle inequality, the fourth step follows from triangle inequality.

For the second term in the above Eq.~\eqref{eq:tail_bound-2:eq2}
\begin{align}\label{eq:tail_bound-2:eq3}
\int_{F_0}^{\infty}(\int_{-F}^{F}|\wh{M}_{\mu,\sigma^2}(f-\xi)|\d \xi)^\alpha\d f 
\leq &~      \int_{F_0}^{\infty}(\int_{-F}^{F}|\wh{M}_{\mu,\sigma^2}(f-F)|\d \xi)^\alpha\d f  \notag \\
= &~      \int_{F_0}^{\infty}|2F\wh{M}_{\mu,\sigma^2}(f-F)|^\alpha\d f \notag \\
= &~ (2F)^\alpha   \cdot  \int_{F_0}^{\infty}|\wh{M}_{\mu,\sigma^2}(f-F)|^\alpha\d f  \notag \\
= &~ (2F)^\alpha   \cdot  \int_{F_0-F}^{\infty}|\wh{M}_{\mu,\sigma^2}(\xi)|^\alpha\d \xi \notag \\
\leq & ~  (2F)^{\alpha }  \cdot \eps_1/(4F)^{2\alpha} \notag\\
\leq&~ \eps_1 /(2 F^{\alpha})
\end{align}
where the first step follows from what $ \wh{M}_{\mu,\sigma^2}(f-\xi) \leq \wh{M}_{\mu,\sigma^2}(f-F)$ because $f-\xi \geq f- F \geq F_0-F>0$, the forth step follows from $f-F= \xi$, %
second last step follows from Eq.~\eqref{eq:tail_bound-2:eq1}.

Similarly, we have that
\begin{align}
\label{eq:tail_bound-2:eq4}
\int_{-\infty}^{-F_0}(\int_{-F}^{F}|\wh{M}_{\mu,\sigma^2}(f-\xi)|\d \xi)^\alpha\d f 
\leq   \eps_1 /2.
\end{align}

Combining Eq.~\eqref{eq:tail_bound-2:eq2}, \eqref{eq:tail_bound-2:eq3} and \eqref{eq:tail_bound-2:eq4} completes the proof.
\end{proof}

\subsubsection{Bounding the convolution}
The goal of this section is to prove Lemma~\ref{lem:bound_M*f}.

\begin{lemma}%
\label{lem:bound_M*f}
Given $\mu',F, \sigma',\gamma\in\R_+$. Let $M:\R\rightarrow\R$ be defined as in Definition \ref{def:gaussian-multiplier-2}. Let $x:\R\rightarrow\C$ be a function.
For simplicity, let $ M=M_{\mu',\sigma'^2}$. 
We have that for $\forall f\in\R$, %
\begin{align*}
|\wh{M}*\wh{x}(f)-\sum_{\mu \in \Z \gamma}  \int_{\mu-\gamma/2}^{\mu+\gamma/2}\wh{x}(\xi)\d \xi \cdot \wh{M}(f)* \delta(f-\mu)|  \leq (\sigma'^2+\sigma'|\mu'|)\gamma^2\cdot\|\wh{x}\|_1.
\end{align*}
\end{lemma}
\begin{proof}
We will separate $\wh{M}*\wh{x}(f)$ into $\wh{M}(f)*(\wh{x}(f)\cdot \rect_{\gamma/2}(f-\mu))$ where $\mu\in\gamma\Z$. This decomposes $ \wh{x}$ into different intervals. Then, we get a bound for each interval.

We can rewrite $\wh{M}_{\mu',\sigma'^2}(f)$ in the following sense,
\begin{align*}
\wh{M}_{\mu',\sigma'^2}(f) = & ~ \sqrt{2\pi\sigma'^2} \cdot \exp(-2\pi^2\sigma'^2f^2) \cdot \exp(-2\pi\i\mu'f)\\
= & ~ \sqrt{2\pi\sigma'^2} \cdot M_{0,1/(4\pi^2\sigma'^2)}(f) \cdot \exp(-2\pi\i\mu'f)
\end{align*}
where the second step follows from the definition of $M$.

First, we consider the first term in the LHS of our statement
\begin{align}
\wh{M}(f)*(\wh{x}(f)\cdot \rect_{\gamma/2}(f-\mu))=&~\wh{M}(f)* \Big(   \int_{\mu-\gamma/2}^{\mu+\gamma/2} \wh{x}(\xi)\delta(f-\xi)  \d \xi \Big) \notag \\
=&~ \int_{\mu-\gamma/2}^{\mu+\gamma/2}\wh{x}(\xi)\cdot \wh{M}(f)*    \delta(f-\xi)   \d \xi\notag \\
=&~ \int_{\mu-\gamma/2}^{\mu+\gamma/2}\wh{x}(\xi)\cdot \wh{M}(f-\xi)   \d \xi \label{eq:eq_bound_M*x_1}
\end{align}
where the first step follows from definition of $\rect$ function, the last step follows from $\wh{M}(f) * \delta(f-\xi) = \wh{M}(f-\xi)$.

Now, we consider the second term in the LHS of our statement
\begin{align}
\int_{\mu-\gamma/2}^{\mu+\gamma/2}\wh{x}(\xi)\d \xi \cdot \wh{M}(f)* \delta(f-\mu)= &~\int_{\mu-\gamma/2}^{\mu+\gamma/2}\wh{x}(\xi)\cdot \wh{M}(f)*\delta(f-\mu)   \d \xi \notag \\
=&~ \int_{\mu-\gamma/2}^{\mu+\gamma/2}\wh{x}(\xi)\cdot \wh{M}(f-\mu)   \d \xi \label{eq:eq_bound_M*x_2}.
\end{align}
where the last step follows from $\wh{M}(f) * \delta(f-\mu) = \wh{M}(f-\mu) $.

By Lemma \ref{lem:def_C_0_plus} (with %
$\gamma=\gamma/2$, $\sigma =1/ (2\pi\sigma')$, $\mu=0$, $c=-2\pi\mu'$),  %
we have that 
\begin{align}
&~ |\wh{M}(f-\xi)   - \wh{M}(f-\mu)| \notag\\
= &~ \sqrt{2\pi\sigma'^2}|M_{0,1/(4\pi^2\sigma'^2)}(f-\xi)\exp(-2\pi\i\mu' (f-\xi))   - M_{0,1/(4\pi^2\sigma'^2)}(f-\mu)\exp(-2\pi\i\mu' (f-\mu))| \notag \\
\leq &~  \sqrt{2\pi\sigma'^2} (2\pi\sigma'+2\pi|\mu'|)\gamma \notag \\
\lesssim &~ (\sigma'^2+\sigma'|\mu'|)\gamma\label{eq:eq_bound_M*x_3}
\end{align}
where the first step from the calculation of $ \wh{M}$, the second step follows from Lemma \ref{lem:def_C_0_plus}.

Thus, we can claim that $\forall f\in\R$, %
\begin{align}
 &~|\wh{M}(f)*(\wh{x}(f)\cdot \rect_{\gamma/2}(f-\mu))- \int_{\mu-\gamma/2}^{\mu+\gamma/2}\wh{x}(f)\d f \cdot \wh{M}(f)* \delta(f-\mu)|\notag \\
=&~ |\int_{\mu-\gamma/2}^{\mu+\gamma/2}\wh{x}(\xi)\cdot \wh{M}(f-\xi)   \d \xi-\int_{\mu-\gamma/2}^{\mu+\gamma/2}\wh{x}(\xi)\cdot \wh{M}(f-\mu)   \d \xi|\notag \\
=&~ |\int_{\mu-\gamma/2}^{\mu+\gamma/2}\wh{x}(\xi)\cdot (\wh{M}(f-\xi)   - \wh{M}(f-\mu)) \d \xi|\notag \\
\leq&~ \int_{\mu-\gamma/2}^{\mu+\gamma/2}|\wh{x}(\xi)|\cdot |\wh{M}(f-\xi)   - \wh{M}(f-\mu)| \d \xi\notag \\
\leq &~   \underset{\xi \in [\mu-\gamma/2,\mu+\gamma/2]}{\max}\{|\wh{M}(f-\xi)   - \wh{M}(f-\mu)|\} \cdot\int_{\mu-\gamma/2}^{\mu+\gamma/2}|\wh{x}(\xi)|\d \xi \notag \\
\lesssim&~ (\sigma'^2+\sigma'|\mu'|)\gamma \cdot\int_{\mu-\gamma/2}^{\mu+\gamma/2}|\wh{x}(\xi)|\d \xi .
\label{eq:bound_whM*whx_delta}
\end{align}
where the first step follows from from Eq.~\eqref{eq:eq_bound_M*x_1}, Eq.~\eqref{eq:eq_bound_M*x_2}, the third step follows from triangle inequality, and the last step follows from  Eq.~\eqref{eq:eq_bound_M*x_3}.

As a result, we have that for $f$ %

\begin{align*}
&~|\wh{M}*\wh{x}(f)- \sum_{\mu \in \Z \gamma}  \int_{\mu-\gamma/2}^{\mu+\gamma/2}\wh{x}(\xi)\d \xi \cdot \wh{M}(f)* \delta(f-\mu)|\notag \\
\leq &~|\sum_{\mu\in\Z\gamma} \wh{M}(f)*(\wh{x}(f)\cdot \rect_{\gamma/2}(f-\mu))- \sum_{\mu \in \Z \gamma}  \int_{\mu-\gamma/2}^{\mu+\gamma/2}\wh{x}(\xi)\d \xi \cdot \wh{M}(f)* \delta(f-\mu)|\notag \\
\leq &~ \sum_{\mu\in\Z\gamma}|\wh{M}(f)*(\wh{x}(f)\cdot \rect_{\gamma/2}(f-\mu))-\int_{\mu-\gamma/2}^{\mu+\gamma/2}\wh{x}(f)\d f \cdot \wh{M}(f)* \delta(f-\mu)|\notag \\
\lesssim &~ (\sigma'^2+\sigma'|\mu'|)\gamma\cdot \sum_{\mu\in\Z\gamma}\int_{\mu-\gamma/2}^{\mu+\gamma/2}|\wh{x}(\xi)|\d \xi  \notag \\
\leq &~ (\sigma'^2+\sigma'|\mu'|)\gamma\cdot\|\wh{x}\|_1 
\end{align*}
where the first step follows from $\sum_{\mu\in\Z\gamma}\rect_{\gamma/2}(f-\mu)=1 $, the second step follows from triangle inequality, the third step follows from Eq.~\eqref{eq:bound_whM*whx_delta}, the last step follows from the definition of $\ell_1$ norm.

Thus we complete proof.
\end{proof}

\paragraph{Choice of parameters} The following lemma shows how to take the parameters in this section.

\begin{lemma}
\label{lem:bound_small_number}
Let $\eps_0, \eps, F, c_0, \sigma',T ,F_0, K,\eps_1 \in\R_+$ such that 
\begin{itemize}
\item $\eps_0=0.01$
\item $\eps_1\leq \eps^2/T$
\item $\sigma'=\sqrt{2}\eps_0^2 T$
\item $c_0=0.01/\sqrt{\log(1/\eps_0)}$
\item $K = \lceil \frac{1 + 0.5\eps_0 }{c_0\eps_0^2} \rceil \leq 2 /(c_0 \eps_0^2)$
\item $F_0={\sigma'}^{-1} \cdot\log^{1/2}(4F/ \eps_1)+F$
\end{itemize}

We have that 
\begin{itemize}
\item {\bf Part 1.} $c_0 K \lesssim 1$
\item {\bf Part 2.} %
    $F_0\lesssim (1/ T) \cdot \log^{1/2}(F T / \eps )+F$
\end{itemize}
\end{lemma}
\begin{proof}
We will prove them separately.

{\bf {Part 1.}}

We can show
\begin{align*}
    c_0 K  =&~ c_0 \lceil \frac{1 + 0.5\eps_0 }{c_0\eps_0^2} \rceil
    \leq~  \frac{2c_0}{c_0 \eps_0^2}
    \lesssim~ 1.
\end{align*}
where the first step follows from the definition of $K$, the third step follows from the definition of $ \eps_0$.

{\bf {Part 2.}}

We have that
\begin{align*}
    F_0=&~{\sigma'}^{-1} \cdot\log^{1/2}(4F/ \eps_1)+F\notag\\
    \lesssim&~ 1/(\eps_0^2 T) \cdot\log^{1/2}(F/ \eps_1 )+F\notag\\
    \lesssim&~  1/T \cdot \log^{1/2}(F/  \eps_1 )+F\notag\\
    \lesssim&~  1/ T \cdot \log^{1/2}(FT/ \eps^2  )+F\notag \\
    \lesssim&~  1/ T \cdot \log^{1/2}(FT/ \eps )+F
\end{align*}
where the second step follows from the definition of $\sigma'$, the third step follows from the definition of $ \eps_0$, the fourth step follows from the definition of $\eps_1 \geq \eps^2 / T $.

\end{proof}

\subsection{Semi-continuouse approximation of Fourier-sparse signals}
\label{sec:core_lemma}
The main theorem of this section is stated and proved below.
\begin{theorem}[Sparse Signal is Semi-continuous]
\label{thm:bound_x'_x}
Given $\gamma,\eps\in (0,0.1), F,T\in\R_+$. Let $x:\R\rightarrow\C$ be a function that such that $x$ is $k$-Fourier-sparse 
and $\supp(\wh{x})\subseteq [-F,F]$. 
Let $\gamma > 0$ and $\eps_1 > 0$. 
Then there is an algorithm output a $k'$-Fourier-sparse signal ($k'\leq k$), %
\begin{align*}
x'(t) = \sum_{i=1}^{k'} v_i \exp(2\pi\i f_i t) 
\end{align*}
such that 
\begin{align*}
\|x'-x\|_T^2 \lesssim (F_0 T^3 \gamma^2 + \eps_1/(F^2T)) \cdot \| \wh{x} \|_1^2 %
\end{align*}
where $\gamma =\underset{i\neq j}{\min}|f_i-f_j|$, $f_i \in[-F,F],~ f_i \in \gamma\Z,~ \forall i \in [k']$ and $F_0= \Omega(T^{-1} \log^{1/2} (F/\eps_1) + F)$. 

Further, if $\eps_1 \leq \eps^2/T$ and $\gamma \leq \eps /\sqrt{F_0T^3}$, then we have 
\begin{align*}
\|x'-x\|_T^2 \lesssim \eps^2 \| \wh{x} \|_1^2 
\end{align*}
\end{theorem}

\begin{proof}
First, we will introduce our choice for $x'$. Then, we bound $|\wh{M}_j*\wh{x} -\wh{M}_j*\wh{x}' | $. Finally, show the relationship between $|\wh{M}_j*\wh{x} -\wh{M}_j*\wh{x}' | $ and the LHS in our statement by utilizing $M$ to bridge the integral calculated in the time domain $ \|x'-x\|_T$ and the integral calculated in the frequency domain $ \|\wh{x}\|_1$.

Let $x'$ be chosen as 
\begin{align*}
 x'(t) = \sum_{\mu \in \Z\gamma}   \int_{\mu-\gamma/2}^{\mu+\gamma/2}\wh{x}(f)\d f \cdot \exp(2\pi\i \mu t).
\end{align*}
Then,
\begin{align}
    \|\wh{ x'}\|_1=\sum_{\mu \in \Z\gamma}   \int_{\mu-\gamma/2}^{\mu+\gamma/2}\wh{x}(f)\d f=\int_{-\infty}^{\infty}\wh{x}(f)\d f= \|\wh{ x}\|_1.\label{eq:wtx'}
\end{align}

Let $ c_0,\eps_0 \in \R_+$, $K\in\Z_+$ such that $\eps_0=0.01\in(0,0.1),~c_0=0.01/\sqrt{\log(1/\eps_0)},~K = \lceil \frac{1 + 0.5\eps_0 }{c_0\eps_0^2} \rceil \leq 2/(c_0 \eps_0^2)$. %

For simplicity, let $M_j(t)=M_{\mu_j',\sigma'^2}(t)=M_{c_0j\eps_0^2T,2\eps_0^4T^2} (t)$, then 
\begin{align*}
\wh{M}_j(f)=\sqrt{2\pi\sigma'^2}\exp(-2\pi^2\sigma'^2f^2)\exp(-2\pi\i\mu'_j f).
\end{align*}
and
\begin{align*}
\mu'_j\leq&~ c_0 K \eps_0^2 T \\
\leq&~ c_0 (\frac{2 }{c_0\eps_0^2})\eps_0^2 T  \\
= &~ 2T.
\end{align*}
where the second step follows from upper bound on $K$.

We have that
\begin{align*}
\wh{x}'(f) = \sum_{\mu \in \Z\gamma}   \int_{\mu-\gamma/2}^{\mu+\gamma/2}\wh{x}(f)\d f \cdot \delta(f-\mu).
\end{align*}

So, we can convolute $ \wh{M}$ at the both sides and get
\begin{align}
\wh{M}_j(f)* \wh{x}'(f) = \sum_{\mu \in \Z\gamma}   \int_{\mu-\gamma/2}^{\mu+\gamma/2}\wh{x}(f)\d f \cdot \wh{M}_j(f)*\delta(f-\mu).\label{eq:whM*whx'}
\end{align}

By Lemma \ref{lem:bound_M*f},
we have that for $\forall f \in\R$,
\begin{align}
&~ |\wh{M}_j*\wh{x} -\wh{M}_j*\wh{x}' | \notag \\
= &~ \Big| \wh{M}_j*\wh{x}   - \sum_{\mu \in \Z\gamma}    \int_{\mu-\gamma/2}^{\mu+\gamma/2}\wh{x}(f)\d f  (\wh{M}_j(f)* \delta(f-\mu)) \Big| \notag\\
\leq&~ (\sigma'^2+\sigma'|\mu'_j|)\gamma\cdot\|\wh{x}\|_1\notag \\
\lesssim &~ T^2 \gamma\cdot\|\wh{x}\|_1.\label{eq:bound_sparse_delta}
\end{align}
where the first step follows from Eq.~\eqref{eq:whM*whx'}, the second step follows from  Lemma \ref{lem:bound_M*f}, the third step follows the choice of $ \mu'_j$ and $ \sigma$ and that $\eps_0$ is  a constant. %

Let $F_0$ be defined as 
\begin{align*}
F_0>{\sigma'}^{-1} \cdot\log^{1/2}(4F/ \eps_1)+F.
\end{align*}

Next, we can bound $|x'(t) -x(t)|^2$, 
\begin{align}
\int_{-\infty}^{\infty} {M}_j^2(t) \cdot |x'(t) -x(t)|^2\d t 
= &~  \int_{-\infty}^{\infty}  |\wh{M}_j *\wh{x}'(f) -\wh{M}_j *\wh{x}(f)|^2\d f \notag \\
\leq &~  \int_{-F_0}^{F_0}  |\wh{M}_j *\wh{x}'(f) -\wh{M}_j *\wh{x}(f)|^2\d f + \eps_1 \|\wh{x}'(f)-\wh{x}(f)\|^2_\infty \notag \\
\leq &~  \int_{-F_0}^{F_0}  |\wh{M}_j *\wh{x}'(f) -\wh{M}_j *\wh{x}(f)|^2\d f + 4\eps_1 \|\wh{x}\|^2_\infty \notag \\
\lesssim  &~  2F_0 T^4\gamma^2 \cdot\|\wh{x}\|_1^2 + 4\eps_1 \|\wh{x}\|^2_\infty \notag \\
\lesssim  &~ \underbrace{ (F_0 T^4\gamma^2  + \eps_1 / F^2 )\cdot\|\wh{x}\|_1^2 }_{ \mathrm{err} } 
\label{eq:bound_x_delta_step3}
\end{align}
where the second step follows from Lemma \ref{lem:tail_bound-2}, the third step follows from $ \|\wh{x}'(f)-\wh{x}(f)\|^2_\infty\leq (\|\wh{x}'(f)\|_\infty+\|\wh{x}(f)\|_\infty)^2=(2\|\wh{x}(f)\|_\infty)^2= 4 \|\wh{x}\|^2_\infty$ due to Eq.~\eqref{eq:wtx'}, the forth step follows from Eq.~\eqref{eq:bound_sparse_delta}, the fifth step follows from $ \|\wh{x}\|_\infty\leq \|\wh{x}\|_1$. %

Then we can upper bound the LHS $\int_{0}^{T} |x'(t) -x(t)|^2 \d t$ as follows:
\begin{align*}
& ~ \int_{0}^{T} |x'(t) -x(t)|^2 \d t  \notag \\
= & ~ 
\int_{0}^{T} \rect_T(t)|x'(t) -x(t)|^2 \d t  \notag \\
\leq & ~ \int_{0}^{T}\sum_{j=-K}^{K} (\frac{c_0}{(1-2\eps_0)\sqrt{2\pi}}){M}_j^2(t) \cdot |x'(t) -x(t)|^2\d t \notag \\
\leq & ~\int_{0}^{T}\sum_{j=-K}^{K} c_0 {M}_j^2(t) \cdot |x'(t) -x(t)|^2\d t \notag \\
\leq & ~ c_0 \cdot \sum_{j=-K}^{K}\int_{-\infty}^{\infty} {M}_j^2(t) \cdot |x'(t) -x(t)|^2\d t \\
\lesssim &~ c_0 K\cdot \mathrm{err} \\
\lesssim &~ \mathrm{err}
\end{align*}
where the second step follows from Claim \ref{approximate_rectangle_pulse_with_uniform_spaced_gaussian-2} and ${M}_j^2(t)={M}_{c_0j\eps_0^2T,\eps_0^4T^2}(t)$, the third step follows from $ \eps_0=0.01$, the forth step follows from relaxing integral range, the fifth step follows from Eq.~\eqref{eq:bound_x_delta_step3}, the last step follows from $ c_0 K \lesssim 1$ due to Lemma \ref{lem:bound_small_number}.

\end{proof}

\subsection{Fast optimal-sparsity Fourier sparse recovery}\label{sec:our_exact_algo}

\begin{corollary}[Our result]\label{cor:net_k_sparse_FT_ours}
For any $F>0,T>0,\eps>0$. Let $x^*(t)=\sum_{j=1}^k v_j \exp({2\pi \i f_j t})$ with $|f_j| \le F$ for $j\in[k]$. For observation $x(t)=x^*(t) + g(t)$, there exists an algorithm that takes $$m=O(\eps^{-1}k^2\log^3(k)\log(FT/(\delta \rho)))$$ random samples $t_1,\dotsc,t_m \in [0, T]$, runs in $({\eps}^{-1}FT)^{O(k)} $ time, and outputs $y(t)=\sum_{j=1}^k \wt{v}_j \exp({2\pi \i \wt{f}_j t})$ such that
\begin{align*}
\|y(t)-x^*(t)\|_T \leq (1+\eps) \|g(t)\|_T + \delta \|\wh{x}^*(f)\|_1,
\end{align*}
holds with probability $1-\rho$.
\end{corollary} 
\begin{proof}
Let $N_{f}=O(\frac{\delta}{T \sqrt{  FT\log(1/\delta)}}) \cdot \mathbb{Z}  \cap [-F,F]$ denote a net of frequencies. Because  
\begin{align*}
 \frac{\delta}{T \sqrt{F_0T}}
 \gtrsim &~  \frac{\delta}{T \sqrt{\log^{1/2} (F/\eps_1) + FT}} \\
 \geq &~  \frac{\delta}{T \sqrt{\log^{1/2} (FT/\delta^2) + FT}} \\
 \geq &~  \frac{\delta}{T \sqrt{\log (FT) + \log(1/\delta) + FT}} \\
 \geq &~  \frac{\delta}{T \sqrt{  FT\log(1/\delta)}} 
\end{align*}
where the first step follows from $F_0= \Omega(T^{-1} \log^{1/2} (F/\eps_1) + F)$, the second step follows from setting $\eps_1 = \delta^2/T$.

By Theorem \ref{thm:bound_x'_x}, for any signal $x^*(t)=\sum_{j=1}^k v_j \exp({2 \pi \i f_j t})$, there exists a $k$-Fourier-sparse signal $\wt{x}(t)=\sum_{j=1}^k v'_j \exp({2 \pi \i f'_j t})$ such that,
$$
\|x^*(t)-\wt{x}(t)\|_T \le \delta \|\wh{x}^*(f)\|_1
$$
and $f'_1,\cdots,f'_k \subseteq N_f$. 

Because we have that
\begin{align*}
\|y(t)-x^*(t)\|_T \leq&~ \|y(t)-\wt{x}(t)\|_T + \|\wt{x}(t)-x^*(t)\|_T \\
\leq &~\|y(t)-\wt{x}(t)\|_T + \delta \|\wh{x}^*(f)\|_1.
\end{align*}

Let $S$ be the set of i.i.d samples from $D(t)$ of size $O(\eps^{-1}k\log^3(k)\log(1/\rho_0))$, $w$ be the corresponding weight in Algorithm \ref{alg:k_sparse_FT} Procedure \textsc{SparseFT} line \ref{line:output_S_w_exact_algo}. By Lemma \ref{lem:concentration_for_any_polynomial_signal_improve_by_CP19}, we have that, for any $\cal F$, with probability at least $1-\rho_0$, 
\begin{align*}
    (1-\sqrt{\eps})\|x\|_T \leq \|x\|_{S,w}\leq (1+\sqrt{\eps})\|x\|_T.
\end{align*}
In total, we enumerate $({\delta}^{-1}FT)^{O(k)}$ function family ${\cal F}$ in Algorithm \ref{alg:k_sparse_FT} Procedure \textsc{SparseFT} line \ref{line:for_subset_exact_algo}. By taking $\rho_0 = \rho {({\delta}(FT)^{-1})^{O(k)}}$, we have that the total success probability is at least
\begin{align*}
    (1-\rho_0)^{({\delta}^{-1}FT)^{O(k)}} \geq 1-{({\delta}^{-1}FT)^{O(k)}} \cdot \rho_0 \geq 1- \rho
\end{align*}
Thus, by Lemma \ref{lem:agnostic_learning_single_distribution}, with probability at least $1-\rho$, sampling $S$ and $w$ forms a $\eps$-WBSP for every $\cal F$.

Finally, we bound $\|y(t)-\wt{x}(t)\|_T$ as follows,
\begin{align*}
    \|y(t)-\wt{x}(t)\|_T \leq &~ (1+O(\eps)) \|x(t)-\wt{x}(t)\|_T \\
    \leq &~ (1+O(\eps))( \|x(t)-x^*(t)\|_T + \|x^*(t)-\wt{x}(t)\|_T) \\
    \leq &~ (1+O(\eps))( \|g(t)\|_T + \delta \|\wh{x}^*(f)\|_1)
\end{align*}
where the first step follows from the proof of Theorem \ref{thm:reduction_freq_signal_1d_clever}.

Combine the results above we have that
\begin{align*}
    \|y(t)-x^*(t)\|_T \leq&~ \|y(t)-\wt{x}(t)\|_T + \delta \|\wh{x}^*(f)\|_1\\
    \leq &~ (1+O(\eps)) \|g(t)\|_T + O(\delta) \|\wh{x}^*(f)\|_1.
\end{align*}

Note that our sample complexity is $|S| = O(\eps^{-1}k\log^3(k)\log(1/\rho_0)) = O(\eps^{-1}k^2\log^3(k)\log(FT/(\delta\rho)))$.

\end{proof}

\begin{algorithm}[t]
\caption{Recover $k$-sparse FT}\label{alg:k_sparse_FT}
\begin{algorithmic}[1]
\Procedure{\textsc{SparseFT}}{$x,k,F,T,\eps,\delta,\rho$} \Comment{Corollary \ref{cor:net_k_sparse_FT_ours}}
\State $m \leftarrow O(\eps^{-1}k^2\log^3(k)\log(FT/(\delta\rho)))$
\State\label{line:output_S_w_exact_algo} $S, w\leftarrow \textsc{WeightedSketch}(m, k, T)$\Comment{Algorithm~\ref{algo:distillation_1d_fast}}
\State We observe the signal $x(t)$ for each $t\in S$
\State $N_{f} \leftarrow O(\frac{\eps}{T\sqrt{FT\log(1/\eps)}}) \cdot \mathbb{Z}  \cap [-F,F]$ %
\For{ $\{f'_1,\dotsc,f'_k\} \in \binom{ N_f}{[k]}$}\label{line:for_subset_exact_algo}
\State Let ${\cal F}= {\text{span}\{\exp({2 \pi \i  f'_1 t}),\cdots,\exp({2 \pi \i  f'_k t})\}}$
\State\label{step:in_enum_freq} $h(t) \leftarrow {\argmin}_{h\in {\cal F}}\|h(t)-x(t)\|_{S,w}$  
\If{$\|h(t)-x(t)\|_{S,w} \le \|\wt{f}(t)-x(t)\|_{S,w}$}
\State $\wt{f}(t)\leftarrow h(t)$
\EndIf
\EndFor 
\State \Return $\wt{f}(t)$
\EndProcedure
\end{algorithmic}
\end{algorithm}
\begin{lemma}[Running time of Lemma \ref{cor:net_k_sparse_FT_ours}]
Procedure \textsc{SparseFT} in Algorithm \ref{alg:k_sparse_FT} runs in $O(({\delta}^{-1}{FT})^{O(k)} \log(1/\rho))$ times.
\end{lemma}
\begin{proof}
In each call of the Procedure \textsc{SparseFT} in Algorithm \ref{alg:k_sparse_FT},  

\begin{itemize}
    \item In the for loop, it repeats the line \ref{step:in_enum_freq} for $ ({\delta}^{-1}{\log^{0.5}(1/\delta)(FT)^{1.5}})^k$ times. 
    \item Note that each \ref{step:in_enum_freq} of Procedure \textsc{SparseFT} in Algorithm \ref{alg:k_sparse_FT} is solving linear regression. This part takes $ O(\eps^{-1}k^{\omega+1}\log^3(k)\log(FT/(\delta\rho)))$ time.
\end{itemize}

So, the time complexity of Procedure \textsc{SparseFT} in Algorithm \ref{alg:k_sparse_FT} is 
\begin{align*}
O(({\delta}^{-1}{\log^{0.5}(1/\delta)(FT)^{1.5}})^k) \cdot O(\eps^{-1}k^{\omega+1}\log^3(k)\log(FT/(\delta\rho)))
=  O(({\eps}^{-1}{FT})^{O(k)} \log(1/\rho)).
\end{align*}
\end{proof}

%% file: constant_gap.tex
\subsection{Semi-continuous approximation with a constant frequency gap}\label{sec:generate_frequency_gap}
 
In this section, we show that the semi-continuous approximation result in previous section can be further improved in terms of the frequency gap.

We first consider the one-sparse case in the following lemma.
\begin{lemma}\label{lem:Fourier-existence}
Let $0 < \delta < 0.1$ be a parameter. Let $x^*= v^* \exp(2\pi \i f^* t)$ be a function such that $f^*\in [-c/T, c/T]$. Then, there exists $k>\log(1/\delta)$, $F \lesssim k/T$, $f_1, \dots , f_k \in c/T \Z \cap [-F, F]$, and $v_1, \dots , v_k \in \C$
and for the function 
\begin{align*}
\wt{x}(t) = \sum_{j = 1}^{k} v_j \exp(2\pi \i f_j t)
\end{align*}
we have that
\begin{align*}
\|x^*- \wt{x}\|_T \lesssim \delta |v^*|
\end{align*}
\end{lemma}

\begin{proof}
Note that $x^*(t)$ can be written as 
\begin{align*}
  x^*(t) = \int_{-F}^F \wh{x}^*(f) \exp(2 \pi \i ft) \d f.  
\end{align*}
Now consider the Taylor expansion of 
\begin{align*}
\exp(2 \pi \i f t)  = \sum_{j = 0}^{\infty} \frac{(2 \pi \i f t)^j}{j!}
\end{align*}
Note that since $ f\in [-F, F]$, $t\in [0, T]$, $FT \lesssim k$, we have 
\begin{align}
\sum_{j = k + 1}^{\infty} \lvert \frac{(2 \pi \i ft)^j}{j!} \rvert  \leq &~ \sum_{j = k + 1}^{\infty} \lvert \frac{(2 e\pi \i ft)^j}{j^j} \rvert  \notag \\
\leq &~ \sum_{j = k + 1}^{\infty} \frac{1}{\exp(j)}  \notag \\ 
\leq &~ \frac{1}{\exp(k)} \cdot \frac{1}{1-1/e} \notag \\
\lesssim &~ \delta_1 \label{eq:bound_taylor_tail}
\end{align}
where the first step follows from Stirling's formula $n!\geq \sqrt{2\pi} n^{n+0.5}\exp(-n)$, the second step follows from $ 2e^2\pi f \geq j $, the last step follows from $k\geq \log(1/\delta_1)$.

In particular, if we define the approximator  of exponential function
\begin{align*}
g_{f}(x) = \sum_{j = 0}^{k} \frac{(2 \pi \i  f t)^j}{j!}
\end{align*}
then over the interval $t \in [0,T]$, $f \in [-F, F]$, $FT \lesssim k$, 
\begin{align}\label{eq:taylor-bound}
|\exp(2 \pi \i f t) - g_{f}(t)| =&~ |\sum_{j = k + 1}^{\infty} \frac{(2 \pi \i ft)^j}{j!} |\notag  \\
\leq&~ \sum_{j = k + 1}^{\infty} \lvert \frac{(2 \pi \i ft)^j}{j!} \rvert \notag \\
\leq&~ \delta_1.
\end{align}
where the first step follows from the definition of $g_{f}(x)$, the second step follows from the triangle inequality, the last step follows from Eq.~\eqref{eq:bound_taylor_tail}.

Next, let ${\cal V}_k(f)=(1,f,\cdots,f^k)$. For $ f^* \in [-c/T,c/T]$, we can write the vector
\begin{align*}
{\cal V}_{k} (f^*)= w_1 {\cal V}_{k}(f_1) + \dots + w_{k} {\cal V}_{k}(f_k)
\end{align*}
for some real numbers (depending on $ f^*$) $w_1, \dots , w_{k}$, because ${\cal V}_{k}(f_1),\cdots,{\cal V}_{k}(f_k)$ is linear independence. Thus,
\begin{align}\label{eq:decompose_exp}
g_{ f^*}(t) = w_1 g_{f_1}(t) + \dots + w_{k} g_{f_k}(t)
\end{align}
for the same weights.  Now note that by (\ref{eq:taylor-bound}), for all $t \in [0,T]$,
\begin{align*}
\lvert x^*(t) - v^* g_{f^*}(t) \rvert =&~ \lvert v^* (g_{f^*}(t)- \exp(2\pi\i f^*t)) \rvert  \\
\leq&~ |v^*| \delta_1.
\end{align*}

Then, because $g_{f}(t)$ can be expressed by $g_{f_j}(t),j\in[k]$, we have that
\begin{align*}
\lvert x^*(t) - \sum_{j = 0}^{k} v^* g_{f_j}(t) w_j   \rvert = \lvert x^*(t) - v^* g_{f^*}(t) \rvert
\end{align*}
which is follows from Eq.~\eqref{eq:decompose_exp}.

Note that 
\begin{align*}
\sum_{j = 1}^{k} \lvert w_j \rvert \leq  C_1
\end{align*}

So by Eq.~\eqref{eq:taylor-bound}, we will transform the approximator of exponential function back to exponential function. For $\forall t \in [0,T]$,
\begin{align*}
  \lvert \sum_{j = 1}^{k}  (g_{f_j}(t) - \exp(2\pi \i f_j t)  )  (v^* w_j  )   \rvert \leq C_1 |v^*| \delta_1
\end{align*}

Therefore, we can conclude that for $\forall t \in [0, T]$
\begin{align*}
  \lvert x^*(t) -  \sum_{j = 1}^{k}(v^*w_j)  \exp(2\pi \i f_j t) \rvert \leq  C_1 |v^*|  \delta_1
\end{align*}
and setting
\begin{align*}
h(x) = \sum_{j = 1}^{k} ( v^*w_j  ) \exp(2\pi \i f_j t) 
\end{align*}
immediately leads to the desired conclusion.
\end{proof}

Lemma \ref{lem:Fourier-existence} immediately gives the following corollary by taking linear summation over $k$ frequencies.

\begin{corollary}\label{cor:Fourier-existence-general-case}
Let $0 < \delta < 0.1$ be a parameter. Let $x^*$ be any $k$-Fourier-sparse signal. Then, there exists $\wt{k} \lesssim k \log(k/\delta)$, universal constant $c\in (0, 1)$, $f_1,\cdots,f_{k} \in c/T \Z$, and $v_1, \dots , v_k \in \C$
and for the function 
\begin{align*}
\wt{x}(t) = \sum_{j=1}^{\wt{k}} v_j \exp(2\pi \i f_j t)
\end{align*}
we have that
\begin{align*}
\|x^*- \wt{x}\|_T \lesssim \delta \|\wh{x}^*(f)\|_1
\end{align*}
\end{corollary}

%% file: lowerror_app.tex
\section{Improving Fourier Interpolation Precision in a Smaller Range}\label{sec:shrinking_range}
In this section, we show that the approximation error of the Fourier interpolation algorithm developed in Section~\ref{sec:high_precision_interpolate} can be further improved, if we only care about the signal in a shorter time duration $[0,(1-c)T]$ for $c\in (0,1)$. The main result of this section is Theorem~\ref{thm:main_ours_better_lose_in_interval}.

\subsection{Control noise}

\begin{lemma}\label{lem:useful_delta_bounding_shrink}
Let $x^*(t) = \sum_{j=1}^k v_j e^{2\pi\i f_j t}$ and $x(t)= x^*(t) +g(t)$ be our observable signal. Let $\N_1^2 := \eps_1(\| g(t) \|_T^2 + \delta \| x^*(t) \|_T^2)$. Let $C_1,\cdots,C_l$ are the $\N_1$-heavy clusters from Definition \ref{def:heavy_clusters}. Let $S^*$ denotes the set of frequencies $f^*\in \{f_j\}_{j\in[k]}$ such that, $f^*\in C_i$ for some $i\in [l]$, and 
\begin{align*}
\int_{C_i} | \widehat{x^*\cdot H}(f) |^2 \mathrm{d} f \geq T\N_1^2/k,    
\end{align*}

 Let $S$ denotes the set of frequencies $f^*\in S^*$ such that, $f^*\in C_{j}$ for some $j\in [l]$, and 
\begin{align*}
\int_{C_{j}} | \widehat{x\cdot H}(f) |^2 \mathrm{d} f \geq \eps_2 T\N_1^2/k,    
\end{align*}

Then, we have that,
\begin{align*}
\|x-x_S\|_{T'} + \|x_S - x^*\|_{T'} \leq (\sqrt{2}+O(\sqrt{\eps}+c))\|g\|_T + O(\sqrt{\delta})\|x^*\|_T.
\end{align*}
\end{lemma}
\begin{proof}

Following from the fact that $\sqrt{1+\eps}= 1+O(\eps)$ for $\eps < 1$, we have $$\N_1 = \sqrt{\eps_1(\| g \|_T^2 + \delta \| x^* \|_T^2)}\leq \sqrt{\eps_1}\|g\|_T+\sqrt{\delta\eps_1}\|x^*\|_T.
$$

We have that
\begin{align}
\|x-x_{S^*}\|_T \leq&~ \|x-x^*\|_T + \|x^*-x_{S^*}\|_T \notag \\
\leq&~ \|g\|_T +\|x^*-x_{S^*}\|_T\notag\\
\leq &~ \|g\|_T+(1+\epsilon)\N_1, \label{eq:shrink_useful_delta_bounding_2}
\end{align}
where the first step follows from triangle inequality, the second step follows the definition of $g$, the third step follows from Claim \ref{cla:guarantee_removing_x**_x*_ours}. %

Therefore,
\begin{align*}
&~\|x-x_S\|_{T'} + \|x_S - x^*\|_{T'} \\
\leq &~ \|x-x_S\|_{T'} +\|x_S-x_{S^*}\|_{T'} +\|x_{S^*} - x^*\|_{T'}\\
\leq &~ \|x-x_S\|_{T'} +\|x_S-x_{S^*}\|_{T'} +(1+2c)\|x_{S^*} - x^*\|_{T'}\\
\leq&~ (1+2\delta)\|H(x-x_S)\|_{T'} + (1+2\delta)\|H(x_S - x_{S^*})\|_{T'} + (1+2c)\|x_{S^*}-x^*\|_T \\
\leq &~  (1+O(\delta))(1+2c)(\|H(x-x_S)\|_T+\|H(x_S - x_{S^*})\|_T ) + (1+\eps)(1+O(c))\N_1 \\
\leq &~  (1+O(\delta))(1+2c)\sqrt{2}\sqrt{\|H(x-x_S)\|^2_T+\|H(x_S - x_{S^*})\|^2_T } + (1+\eps)(1+O(c))\N_1 \\
\leq&~   (1+O(\delta))(1+O(\sqrt{\eps_2}))(1+O(c))\sqrt{2}\|x-x_{S^*}\|_T + (1+\eps)(1+O(c))\N_1 \\
\leq&~  (\sqrt{2}+O(\delta+\sqrt{\eps_2}+c))(\|g\|_T + (1+\eps)\N_1) + (1+\eps)(1+O(c))\N_1 \\
\leq &~ (\sqrt{2}+O(\sqrt{\eps}+c))\|g\|_T + O(\sqrt{\delta})\|x^*\|_T,
\end{align*}
where the first step follows from triangle inequality, the second step follows from for any function $x:\R\rightarrow\C$, $(1-c)\|x\|_{T'} \leq \|x\|_{T}$, the third step follows from Property \RN{1} of Lemma \ref{lem:property_of_filter_H} and $ (1-c)/2<(\frac{1}{2}-\frac{2}{s_1})s_3$, the forth step follows from Claim \ref{cla:guarantee_removing_x**_x*_ours}, the fifth step follows from   $\|H(x-x_S)\|_T+\|H(x_S-x_{S^*})\|_T\leq \sqrt{2}\sqrt{\|H(x-x_S)\|_T^2+\|H(x_S-x_{S^*})\|_T^2}$, the sixth step follows from Lemma \ref{lem:use_g_bound_x_under_H}, the seventh step follows from Eq.~\eqref{eq:useful_delta_bounding_1}, the last step follows from $\eps=\eps_0=\eps_1=\eps_2$. 
\end{proof}

\paragraph{Parameters setting}
By Section C.3 in \cite{ckps16}, we choose parameters for filter function $(H(t), \widehat{H}(f))$ as follows:
\begin{itemize}
    \item By Eq.~\eqref{eq:eq2_proof_of_property_6} in the proof of Property \RN{6} of filter function  $(H(t), \widehat{H}(f))$, we need $(1-s_3(1-\frac{2}{s_1})) \cdot \wt{O}(k^4) \leq \eps $, thus we have that $\min (\frac{1}{1-s_3}, s_1) \geq \wt{O}(k^4)/\eps$.
    \item In the proof of Property \RN{5} of filter function $(H(t),\wh{H}(f))$, we set $\ell \gtrsim k \log(k/\delta)$.
    \item In the proof of Lemma \ref{lem:useful_delta_bounding_shrink}, we set $ (1-c)/2<(\frac{1}{2}-\frac{2}{s_1})s_3$. Thus, we have that $ \min(s_3, 1-\frac{4}{s_1})\geq 1-\frac{c}{2}$ or equivalently $\min (\frac{1}{1-s_3}, s_1/4) \geq \frac{2}{c}$.
    \item $\Delta_h$ is determined by the parameters of filter $(H(t), \widehat{H}(f))$ in Eq.~\eqref{eq:def_of_delta_h}: $\Delta_h \eqsim \frac{s_1 \ell}{s_3 T}$. Combining the setting of $s_1$, $s_3$ $\ell$, we should set $\Delta_h \geq \max(\wt{O}(k^5 \log(1/\delta))/(\eps T), O({k \log(k/\delta)}/{(c T)}))$. 
\end{itemize}

\subsection{\texorpdfstring{$(\sqrt{2}+\eps)$-approximation ratio}{}~}

\begin{corollary}[Corollary of Theorem \ref{thm:frequency_recovery_k_cluster_ours}]\label{thm:frequency_recovery_k_cluster_ours_shrink}
Let $x^*(t) = {\sum_{j=1}^k}  v_j e^{2\pi\i f_j
    t}$ and $x(t)= x^*(t) +g(t)$ be our observable signal where $\|g(t) \|_T^2 \le
  c_0\|x^*(t)\|_T^2$ for a sufficiently small constant $c_0$. Then
  Procedure \textsc{FrequencyRecoveryKCluster} %
  returns a set $L$ of $O(k/(\eps_0\eps_1\eps_2)) $
  frequencies that covers all $\N_2$-heavy clusters of $x^*$, which uses $\poly(k, c^{-1},\eps^{-1},\eps_0^{-1},\eps_1^{-1}, \eps_2^{-1},  \log(1/\delta) ) \log(FT)$ samples and $\poly(k, c^{-1},\eps^{-1},\eps_0^{-1},\eps_1^{-1}, \eps_2^{-1},  \log(1/\delta)) \log^2 (FT)$ time. 

In particular, for $\Delta_0=c^{-1}\eps^{-1}\poly(k,\log(1/\delta))/T$ and $\N_2^2 :
  = \eps_1\eps_2(\| g(t) \|_T^2 + \delta \| x^*(t) \|_T^2)$, with probability $1- 2^{-\Omega(k)}$, for any $f^*$ with
  \begin{equation}%
    \int_{f^*-\Delta}^{f^*+\Delta} | \widehat{x\cdot H}(f) |^2 \mathrm{d} f \geq  T\N_2^2/k,
  \end{equation}
there  exists an $\wt{f} \in L$ satisfying
  \begin{equation*}
  |f^*-\widetilde{f} |
  \lesssim \Delta_0 \sqrt{\Delta_0 T}.
  \end{equation*}
\end{corollary}
\begin{remark}
The proof is similar with the proof of Theorem \ref{thm:frequency_recovery_k_cluster_ours}.
\end{remark}

\begin{theorem}[$(\sqrt{2}+\eps)$-approximate Fourier interpolation algorithm with shrinking range]
\label{thm:main_ours_better_lose_in_interval}
Let $x(t) = x^*(t) + g(t)$, where $x^*$ is $k$-Fourier-sparse signal with  frequencies in $[-F, F]$. Let $T'=T(1-c)$. Given samples of $x$ over $[0, T]$, we can output $y(t)$ such that with probability at least $1-2^{-\Omega(k)}$, 
 \[
\|y - x^*\|_{T'} \leq (\sqrt{2}+\eps+c)\|g\|_T + \delta \|x^*\|_T.
 \]
 Our algorithm uses $\poly(k,\eps^{-1},c^{-1},\log (1/\delta) ) \log(FT)$
  samples and $\poly(k,\eps^{-1},c^{-1},\log(1/\delta)) \cdot \log^2(FT)$ time.  The output $y$
  is $\poly(k,\eps^{-1},c^{-1},\log(1/\delta))$-Fourier-sparse signal. 
\end{theorem}
\begin{remark}
\cite{llm21} considers the following guarantee: for parameter $\delta>0$, $\| y(t) - x^*(t) \|_{(1-c)T} \leq \alpha \| g (t) \|_T + \delta \|\wh{x}^*(f)\|_1$, where $\alpha$ is the approximation ratio, $c$ is the shrinking factor. \cite{llm21} uses $\poly(k,\log (1/\delta) ) \log(FT)$ samples and $\poly(k,c^{-1},\log(1/\delta)) \cdot \log^2(FT)$ time, and recovers a $k\poly(1/c, \log(k/\delta))$-Fourier-sparse signal but only guarantee for $c\in (0, 1)$ and $\alpha = \poly(\log( k/(\delta c)))$. 
\end{remark}
\begin{proof}
Let $\N_1^2 :=\eps_1( \| g(t) \|_T^2 + \delta \| x^*(t) \|_T^2)$ be the heavy cluster parameter. 

First, by Lemma \ref{cla:guarantee_removing_x**_x*_ours}, there is a set of frequencies $S^*\subset [k]$ and $x_{S^*}(t)= \underset{j\in S^*}{\sum} v_j e^{2 \pi \i f_j t}$ such that
\begin{align}
    \|x_{S^*} - x^* \|_T^2 \leq (1+\eps)\N_1^2. \label{eq:shrink_2_plus_eps_fourier_intor:x_S_x*}
\end{align}
Furthermore, each $f_j$ with $j\in S$ belongs to an $\N_1$-heavy cluster $C_j$ with respect to the filter function $H$ defined in Definition~\ref{def:def_of_filter_H}. 

By Definition \ref{def:heavy_clusters} of heavy cluster, it holds that
\begin{align*}
    \int_{C_j} | \widehat{H\cdot x^*}(f) |^2 \mathrm{d} f \geq T\N_1^2/k.
\end{align*}
By Definition \ref{def:heavy_clusters}, we also have $|C_j|\leq k\cdot \Delta_h$, where $\Delta_h$ is the bandwidth of $\wh{H}$. 

Let $\Delta\in \R_+$, and $\Delta > k\cdot \Delta_h$, which implies that $C_j\subseteq [f_j-\Delta, f_j+\Delta]$. Thus, we have
\begin{align*}
\int_{f_j-\Delta}^{f_j+\Delta} | \widehat{H\cdot x^*}(f) |^2 \mathrm{d} f \geq T\N_1^2/k.
\end{align*}

By Corollary \ref{cor:use_g_bound_x}, there is a set of frequencies $S\subset S^*$ and $x_{S}(t)= \underset{j\in S}{\sum} v_j e^{2 \pi \i f_j t}$ such that
\begin{align*}
    \|x_S-x_{S^*}\|_T^2 \leq (1+O(\sqrt{\eps_2}))\|x-x_{S^*}\|_T^2.
\end{align*}
Let $g'=x-x_{S^*}$.

Now it is enough to recover only $x_S$, instead of $x^*$. 

By applying Theorem \ref{thm:frequency_recovery_k_cluster_ours_shrink}, there is an algorithm that outputs a set of frequencies $L\subset \R$ such that, $|L|=O(k/(\eps_0\eps_1\eps_2))$, and with probability at least $1-2^{-\Omega(k)}$,  for any $f_j$ with $j\in S$, there is a $\wt{f}\in L$ such that,
\begin{align*}
|f_j-\widetilde{f} |\lesssim \Delta \sqrt{\Delta T}.
\end{align*}

We define a map $p:\R\rightarrow L$ as follows:
\begin{align*}
    p(f):=\arg \min_{\wt{f}\in L} ~ |f-\wt{f}|~~~\forall f\in \R.
\end{align*}
Then, $x_S(t)$ can be expressed as
\begin{align*}
x_S(t)= &~\sum_{j\in S}v_je^{2\pi \i f_j t}\\
= &~ \sum_{j\in S}v_j e^{2\pi \i \cdot p(f_j)t} \cdot e^{2\pi \i \cdot (f_j - p(f_j))t}\\
= &~ \sum_{\wt{f} \in L} e^{2 \pi \i \widetilde{f} t} \cdot \sum_{j\in S:~ p(f_j)=\wt{f}} v_j e^{2\pi \i ( f_j - \widetilde{f})t},
\end{align*}
where the first step follows from the definition of $x_S(t)$, the last step follows from interchanging the summations.

For each $\wt{f}_i\in L$, by Corollary \ref{cor:low_degree_approximates_concentrated_freq_ours} with $ x^*=x_S, \Delta= \Delta \sqrt{\Delta T}$, we have that there exist degree $ d=O(T \Delta \sqrt{\Delta T} + k^3 \log k + k \log 1/\delta)$ polynomials $P_i(t)$ corresponding to $\wt{f}_i\in L$ such that, 
\begin{align}
\|x_S(t)-\sum_{\wt{f}_i \in L} e^{2 \pi \i \widetilde{f}_i t} P_i(t)\|_T \leq \delta \|x_S(t)\|_T\label{eq:shrink_2_plus_eps_fourier:xS_sum}
\end{align}

Define the following function family: 
\begin{align*}
  \mathcal{F} := \mathrm{span}\Big\{ e^{2\pi \i \widetilde{f} t} \cdot t^j~{|}~ \forall \wt{f}\in L, j \in \{0,1,\dots,d\} \Big\}.
\end{align*}
Note that $\sum_{\wt{f}_i \in L} e^{2 \pi \i \widetilde{f}_i t} P_i(t)\in {\cal F}$.

By Claim \ref{cla:max_bounded_Q_condition_number}, for function family $\cal F$, $K_{\mathrm{Uniform[cT/2,T(1-c/2)]}} = O((|L| d)^{4} \log^{3} (|L| d))$. 

By Lemma \ref{lem:rho_wbsp}, we have that, choosing a set $W$ of $O(\eps^{-1} K_{\mathrm{Uniform[cT/2,T(1-c/2)]}} \log(|L|d/\rho))$
i.i.d. samples uniformly at random over duration $[0, T]$ is a $(\eps,\rho)$-WBSP.

By Lemma \ref{lem:magnitude_recovery_for_CKPS}, there is an algorithm that runs in $O(\eps^{-1}|W|(|L|d)^{\omega-1}\log(1/\rho))$-time using samples in $W$, and outputs $y'(t)\in {\cal F}$ such that, with probability $1-\rho$, 
\begin{align}
\|y'(t) - \sum_{\wt{f}_i \in L} e^{2 \pi \i \widetilde{f}_i t} P_i(t)\|_{T'}\leq (1+\eps)\|x(t)-\sum_{\wt{f}_i \in L} e^{2 \pi \i \widetilde{f}_i t} P_i(t)\|_{T'}\label{eq:shrink_2_plus_eps_fourier:y_sum}
\end{align}

Then by Lemma \ref{lem:polynomial_to_FT}, we have that there is a $O(kd) $-Fourier-sparse signal $y(t)$, such that
\begin{align}\label{eq:shrink_2_plus_eps_approx_y}
    \|y(t)-y'(t)\|_{T'} \leq \delta'
\end{align}
where $\delta'>0$ is any positive real number. Thus, $y$ can be arbitrarily close to $y'$.

Moreover, the sparsity of $y(t)$ is $kd = k O(T \Delta \sqrt{\Delta T} + k^3 \log k + k \log 1/\delta) = \poly(k, \eps^{-1}, c^{-1},\log(1/\delta))$.

Therefore, the total approximation error can be upper bounded as follows:
\begin{align*}
    &~ \|y-x^*\|_{T'}\\
    \leq &~ \|y-y'\|_{T'}+ \Big\|y-\sum_{\wt{f}_i \in L} e^{2 \pi \i \widetilde{f}_i t} P_i(t)\Big\|_{T'} + \Big\|\sum_{\wt{f}_i \in L} e^{2 \pi \i \widetilde{f}_i t} P_i(t) - x^*\Big\|_{T'}\\
    \leq &~  (1+0.1\eps)\Big\|y-\sum_{\wt{f}_i \in L} e^{2 \pi \i \widetilde{f}_i t} P_i(t)\Big\|_{T'} + \Big\|\sum_{\wt{f}_i \in L} e^{2 \pi \i \widetilde{f}_i t} P_i(t) - x^*\Big\|_{T'}\\
    \leq &~ (1+2\eps)\Big\|x-\sum_{\wt{f}_i \in L} e^{2 \pi \i \widetilde{f}_i t} P_i(t)\Big\|_{T'} + \Big\|\sum_{\wt{f}_i \in L} e^{2 \pi \i \widetilde{f}_i t} P_i(t) - x^*\Big\|_{T'}\\
     \leq &~ (1+2\eps)(\Big\|x-x_S\Big\|_{T'} + \Big\|x_S - x^*\Big\|_{T'})+2(1+\eps)\|x_S-\sum_{\wt{f}_i \in L} e^{2 \pi \i \widetilde{f}_i t} P_i(t)\|_{T'}\\
    \leq &~ (1+2\eps)(\Big\|x-x_S\Big\|_{T'} + \Big\|x_S - x^*\Big\|_{T'}) + O(\delta) \|x_S(t)\|_T\\
    \leq &~ (1+2\eps)(\sqrt{2}+O(\sqrt{\eps}+c))\|g\|_T + O(\sqrt{\delta})\|x^*\|_T+ O(\delta) \|x_S(t)\|_T \\
    \leq &~ (1+2\eps)(\sqrt{2}+O(\sqrt{\eps}+c))\|g\|_T + O(\sqrt{\delta})\|x^*\|_T+ O(\delta) (\|g\|_T+\|x^*\|_T) \\
    \leq &~ (\sqrt{2}+O(\sqrt{\eps}+c))\|g\|_T + O(\sqrt{\delta})\|x^*\|_T,
\end{align*}
where the first step follows from triangle inequality, the second step follows from Eq.~\eqref{eq:shrink_2_plus_eps_approx_y}, the third step follows from Eq.~\eqref{eq:shrink_2_plus_eps_fourier:y_sum}, the forth step follows from Triangle Inequality again, the fifth step follows from Eq.~\eqref{eq:shrink_2_plus_eps_fourier:xS_sum}, the sixth step follows from Lemma \ref{lem:useful_delta_bounding_shrink}, the seventh step follows from Lemma \ref{lem:bound_x_S}, and the last step is straightforward.

By re-scaling $\eps$ and $\delta$, we prove the theorem.

\end{proof}